\PassOptionsToPackage{bookmarks=false}{hyperref}
\documentclass[
oldfontcommands,
	12pt,				
	openright,			
	oneside,			
 a4paper,			
	english,
 bookmarks=false,
	brazil				
	]{abntex2}

 \usepackage[utf8]{inputenc}	

\usepackage{lmodern}			
\usepackage[T1]{fontenc}			

\DeclareUnicodeCharacter{2212}{-}
\usepackage{indentfirst}			
\usepackage{color}				
\usepackage{graphicx}			
\usepackage{relsize}				

\usepackage[final]{pdfpages}		
\usepackage{amsmath,amssymb,amsthm,bm,amstext,epstopdf}	
\usepackage{pdflscape} 		
\usepackage{multirow}			
\usepackage{adjustbox}			

\usepackage{lipsum}	 

\usepackage[fixlanguage]{babelbib}
\selectbiblanguage{english}

\usepackage[utf8]{inputenc}

\usepackage{mathtools}
\usepackage{graphicx}
\usepackage{mathrsfs}
\usepackage{indentfirst}
\usepackage{dsfont}
\usepackage{braket}
\usepackage{amssymb,amsfonts,amsmath,amstext,amsthm}
\usepackage{epsfig,epstopdf}
\usepackage{color}
\usepackage{relsize}            
\usepackage[misc,clock,geometry]{ifsym} 
\usepackage{braket}
\usepackage{url} 
\usepackage{float}
\usepackage{graphicx}
\restylefloat*{figure}

\theoremstyle{plain}
\newtheorem{thm}{Theorem}
\newtheorem{definition}{Definition}

\newtheorem{lemma}{Lemma}

\usepackage[hyperpageref]{backref}	 

\usepackage{cite}
\definecolor{blue}{RGB}{41,5,195}

\definecolor{internal}{rgb}{0.6,0.,0.5}
\definecolor{citecolor}{rgb}{0.,0.5,0.3}
\definecolor{urlcol}{rgb}{0.05,0.,0.2}
\makeatletter
\hypersetup{     		
		colorlinks=true,       		
    	linkcolor=internal,          	
    	citecolor=citecolor,        		
    	filecolor=magenta,      		
		urlcolor=urlcol,
		bookmarksdepth=4
  } 
\makeatother

\renewcommand*{\backrefalt}[4]{
	\ifcase #1 %
 No citations in the text.%
	 \or
	Cited in page #2.%
 \else
	 #1 citations in pages #2.%
	\fi}%

\usepackage{mathdots}
\usepackage{yhmath}
\usepackage{cancel}
\usepackage{color}
\usepackage{siunitx}
\usepackage{array}
\usepackage{multirow}
\usepackage{amssymb}
\usepackage{textcomp,gensymb}
\usepackage{tabularx}
\usepackage{booktabs}
\usepackage{relsize}
\usepackage{graphicx,color}          
\usepackage{epsfig}                  
\usepackage{wrapfig}
\usepackage{subfig}
\usepackage{braket}

\usepackage[normalem]{ulem}

\usepackage[misc,clock,geometry]{ifsym}

\newcommand{\on}[1]{\operatorname{#1}}

\usepackage{tocloft,etoolbox}

\cftsetindents{section}{0em}{5.2 em}
\cftsetindents{subsection}{5.2em}{3em}

\begin{document}

\selectlanguage{brazil}

\includepdf{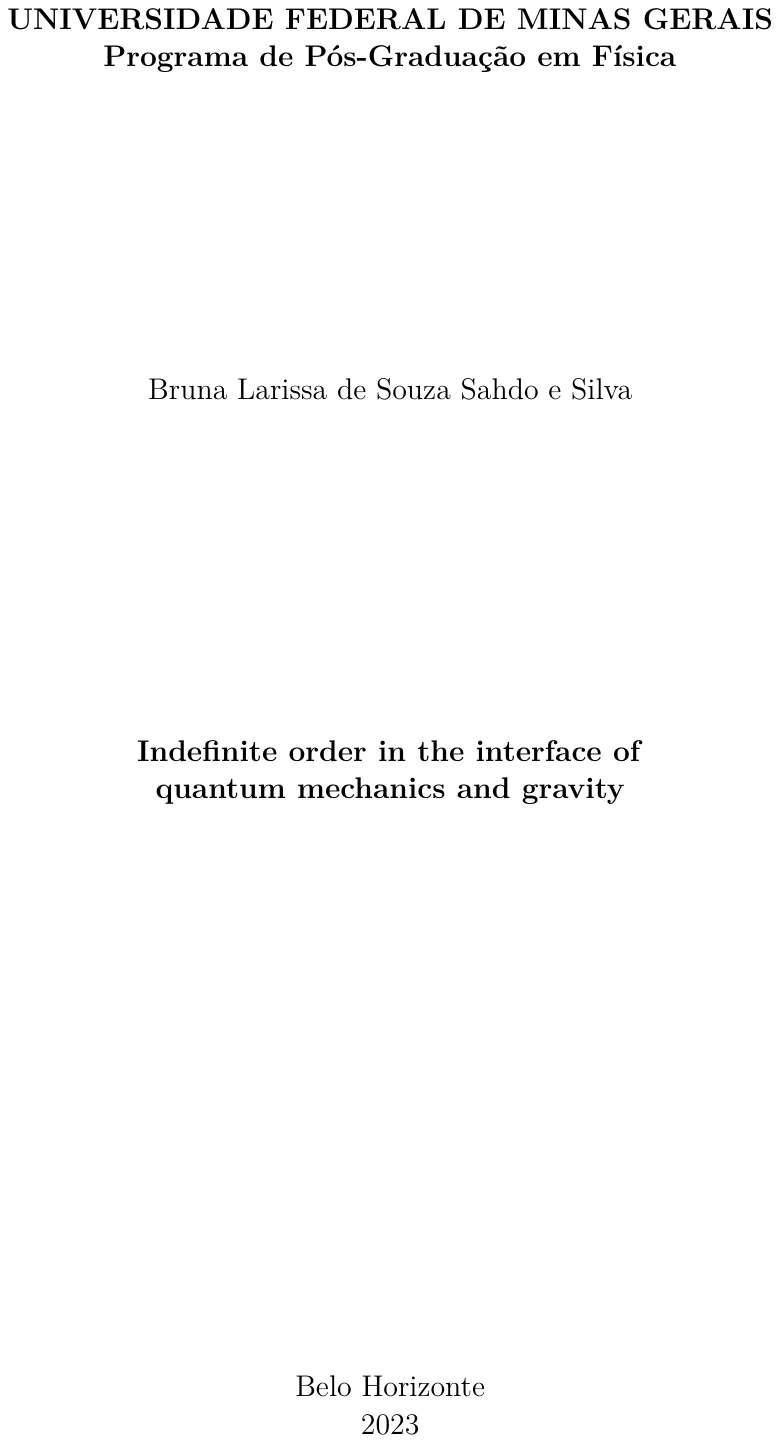}
\newpage

\clearpage
\pagenumbering{arabic} 

\selectlanguage{brazil}

\includepdf{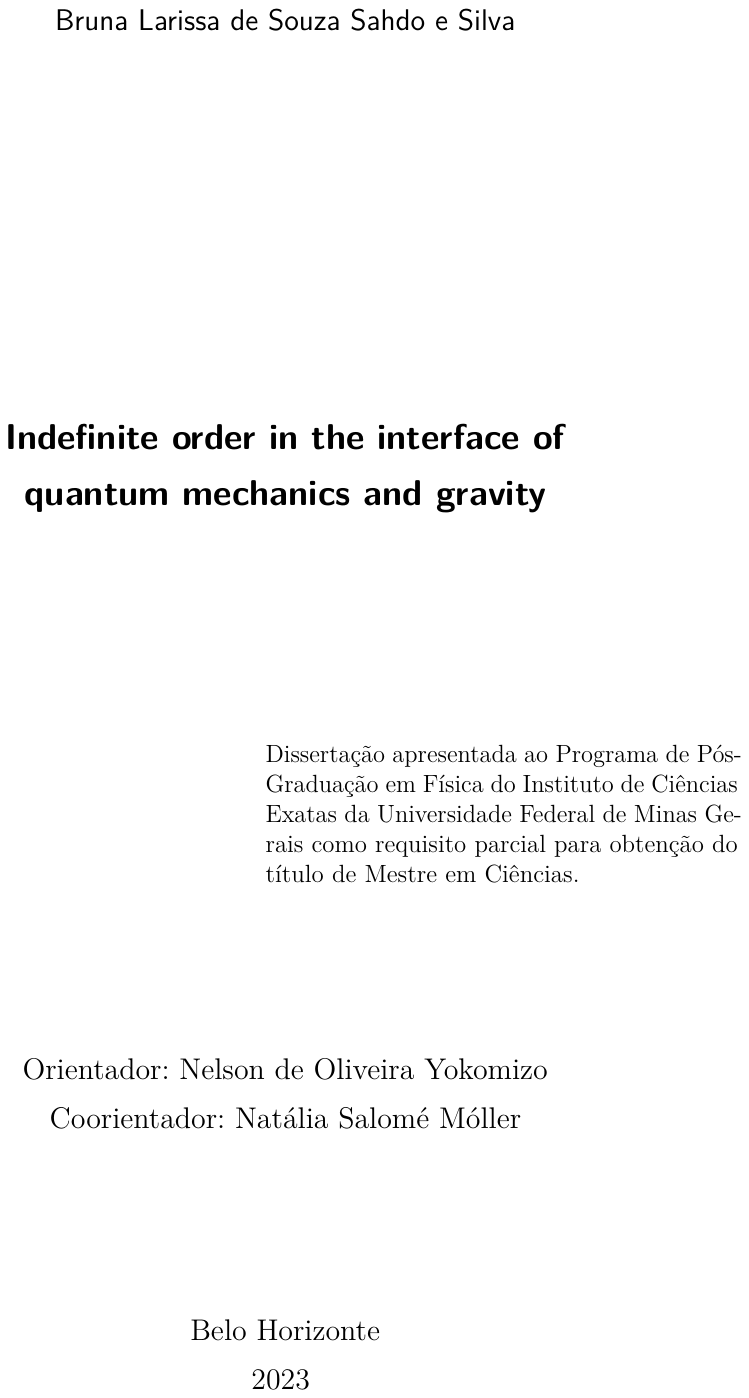}



     \includepdf{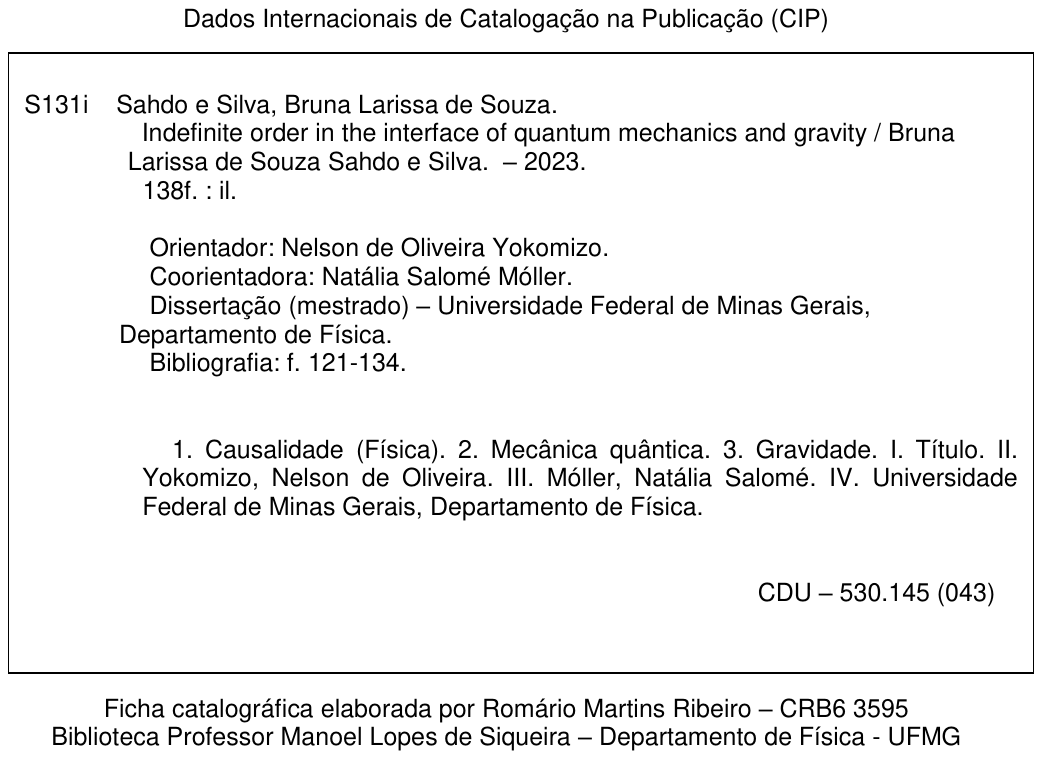}


%

   \includepdf[pages=1-2]{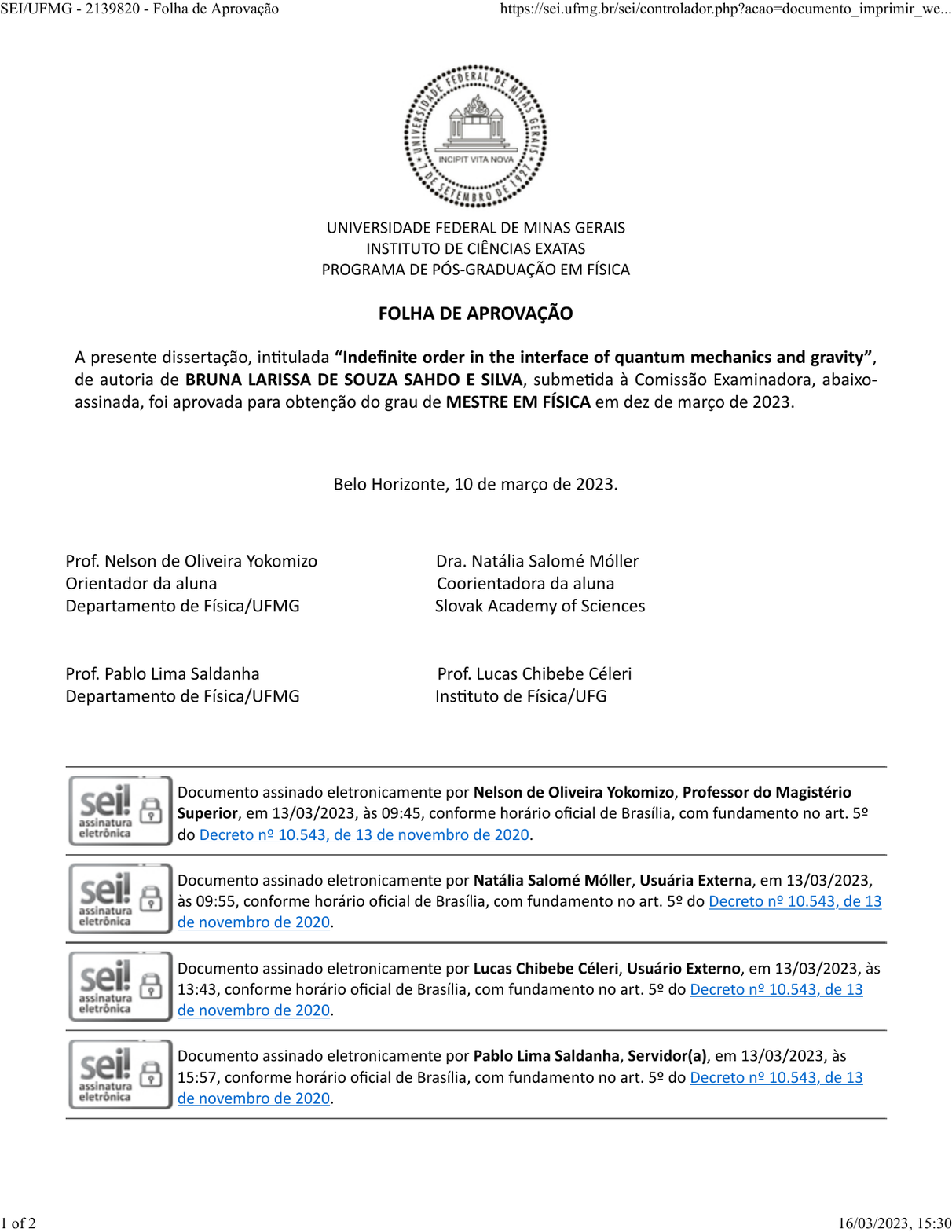}

%


\begin{agradecimentos}
Agradeço aos meus pais Márcia e Jesse, e a minha irmã Aline por todo o suporte ao longo dos anos com minhas idas e vindas entre Manaus e Belo Horizonte.

Aos meus orientadores Nelson Yokomizo e Natália S. Móller pela melhor orientação que eu poderia ter recebido neste mestrado, ainda mais levando em conta a pandemia. Agradeço pela parceria, compreensão, incentivo e, enfim, pela amizade. Espero que continuem tendo muito sucesso e que possamos trabalhar juntos no futuro novamente.

Aos professores do Departamento de Física do ICEX com quem tive aula por terem, em sua maior parte, proporcionado a mim e aos meus colegas  uma boa experiência de aprendizado apesar das dificuldades do ensino virtual na pandemia.

Aos ex e atuais integrantes do Grupo de Física Teórica Fundamental (GFT) por, desde a época graduação, trazerem as mais interessantes dúvidas e discussões sobre física, apresentarem seus trabalhos excepcionais e serem ótimos parceiros de cafézinho.

Aos meus amigos em Belo Horizonte, Pedro (Bruni), João, Ana e Amanda, que estiveram ao meu lado vivendo as experiências da faculdade e fora dela, compartilhando frustrações e alegrias e criando boas memórias ao longo dos anos. Agradeço também especialmente ao Bruni por nossas conversas espontâneas sobre tópicos em física e matemática, que mantêm vivas as nossas paixões por ambas as áreas até hoje.

Ao Davi por me conhecer tão bem e me ajudar nos momentos difíceis com a maior leveza, pela amizade de tantos anos e, recentemente, pela companhia e apoio moral na rotina de ficar até tarde escrevendo nossos trabalhos de fim de curso. Agradeço pelo carinho e desejo tudo o que há de melhor na sua caminhada.

Ao Gustavo por me ouvir e manter as conversas ainda que nossas rotinas se distanciem um pouco, por torcer por mim e por sempre arrumar um tempinho para tomar um café, passear e conversar pessoalmente quando estou em Manaus.

À CAPES pelo suporte financeiro.
\end{agradecimentos}

\begin{epigrafe}
 \vspace*{\fill}
	\begin{flushright}
 \textit{``We all like to congregate at boundary conditions.
\\
Where land meets water. Where earth meets air. 
\\
Where body meets mind. Where space meets time. 
\\
We like to be on one side, and look at the other.''
\\
\textbf{Douglas Adams}\\ (Mostly Harmless)}
	\end{flushright}
\end{epigrafe}
---


\setlength{\absparsep}{18pt} 
\begin{resumo}[Resumo]
\selectlanguage{brazil}

Há tempos procura-se entender como as características da Teoria Quântica e da Relatividade Geral se unem para explicar a física na sua interface. Um dos motivos pelos quais essa é uma tarefa difícil é a discrepância entre as formas de abordar o tempo e a causalidade em cada teoria. Por exemplo, a estrutura causal na relatividade é dinâmica, determinada pela distribuição de massa no espaço-tempo, enquanto que, no formalismo quântico, ela é fixa e deve ser estabelecida \emph{a priori}. Nesta dissertação, discutimos a noção de ordem indefinida, que aparece pela primeira vez em uma generalização abstrata da Teoria Quântica. Tal generalização é feita com o intuito de eliminar a incompatibilidade da teoria com a Relatividade Geral no que diz respeito à causalidade. Para isso, o formalismo remove a exigência de que haja uma estrutura causal global e, portanto, a ordem entre operações em protocolos não precisa estar bem definida. Um típico exemplo de ordem indefinida é o processo do \textit{quantum switch}, que realiza uma superposição quântica da ordem em que duas operações são aplicadas em um sistema. As probabilidades do \textit{quantum switch} já foram reproduzidas experimentalmente com fótons, em protocolos completamente descritos pela mecânica quântica usual. Já que esses experimentos são compatíveis com a estrutura causal do espaço-tempo, isso gerou incertezas sobre o que se pode concluir da realização de um processo com ordem indefinida dependendo do contexto. Retornaremos às motivações iniciais, apresentando como cenários gravitacionais a energias baixas poderiam fazer surgir ordem indefinida. Nisto estão inclusas as formulações de um \textit{quantum switch} em um cenário de gravidade quântica e um \textit{quantum switch} em uma métrica clássica de Schwarzschild. Assim, o \textit{quantum switch} pode ser usado como uma base comum para discutir diferenças entre os contextos. A última proposta, de um \textit{quantum switch} em uma métrica clássica, é um trabalho original. Além de ser um exemplo de ordem indefinida, a sua realização na gravidade da Terra é proposta como um teste das previsões da mecânica quântica em espaços-tempos curvos, um regime que até hoje não foi experimentalmente testado.

   \vspace{\onelineskip}
 
   \noindent 
 \textbf{Palavras-chave}: Ordem causal indefinida; \textit{Quantum switch}; Relógios quânticos; Mecânica quântica em espaços-tempos curvos.
\end{resumo}

\begin{resumo}[Abstract]
 \begin{otherlanguage*}{english}
Researchers have long been aiming to understand how the characteristics of Quantum Theory and General Relativity combine to account for regimes in their interface. One reason why this is a hard task is how differently the theories approach time and causality. For instance, causal structure in relativity is dynamical, determined by the distribution of mass in spacetime, while, in the quantum formalism, it is supposed to be fixed and given in advance. In this master's thesis, we discuss the notion of indefinite order, which
first appears in an abstract generalization of Quantum Theory. The purpose of that generalization is eliminating the incompatibility of the theory with General Relativity with respect to causality. For this, the demand for global causal structure is removed, in principle allowing cases for which the order of operations in protocols is not necessarily well-defined. One epitomical example of indefinite order is the quantum switch process, which realizes a
quantum superposition of orders of two operations on a target system. The quantum switch probabilities have been reproduced in experimental optical setups that are fully described by regular quantum mechanics. Since these experiments are compatible with spacetime causal structure, this generated uncertainty on what conclusions can be drawn from the realization of indefinite order depending on the context. Here, we return to the initial motivations and present how scenarios involving gravity in low energies could lead to indefinite order. This includes the formulation of a quantum switch in a quantum gravity scenario and of a quantum switch in a classical Schwarzschild metric. Thus, the quantum switch will be used as a common 
ground for discussing differences between all setups. The latter proposal of a quantum switch in a classical metric is an original work that, aside from being an example of indefinite order, proposes the realization of the protocol in Earth’s gravity as a test of quantum mechanics on curved spacetimes, a regime which has not yet been explored experimentally.

   \vspace{\onelineskip}
 
 \end{otherlanguage*}

\end{resumo}


\selectlanguage{english}
\tableofcontents*

\textual

\selectlanguage{english}

\chapter{Introduction}

    One of the main goals of a physical theory is to describe how systems \emph{change in time}.  Way before relativistic physics, the fundamental explanations for the movement of stars, falling apples, the atomic world, electric and magnetic effects were all made with respect to a global time parameter.
    
         When a global time exists, it determines which physical events can exert influence on others, i.e. causal relations. An event is said to be in the causal future of another event if it has the \emph{possibility} of being influenced by it. To talk about possibility of influence in physics, we assume that one can freely choose what happens at the first event, say A. Then, by analysing the virtual changes predicted to occur in the physical theory at the other event, say B, under that free variation, we can determine the causal relation between them. More specifically, we inspect whether what happens at B is dependent on what happens at A considering all possibilities and, if that is true, B is in the causal future of A. In a theory with global time in which any velocity is allowed, an event A$=(x_A,t_A)$ is in the causal future of another B$=(x_B,t_B)$ if it happens later in the arrow of time ($t_A>t_B$).

 The Theory of Relativity changed the paradigm that time and space should be treated as part of a passive background on top of which physics is formulated. Not only time becomes a relative quantity already in Special Relativity, but also causal relations are modified: systems cannot travel faster than light, and that reduces the set of events that can be reached by a physical influence coming from a fixed event. In the theory, the set of events in spacetime that can either influence or be influenced by an event A is called the \emph{lightcone} of A. And we say that events outside the lightcone are \emph{spacelike separated} from A and are fundamentally unable to establish a causal relation with A.
 
 If two events are inside each other's ligthcones, relativity says that every observer agrees upon their order, even if the perceived elapsed time between them may vary. If they are spacelike separated, however, an observer may perceive A happening before B, while another observer perceives B before A or even both happening simultaneously. Then, although each observer has a perception of time, perceived time order does not fully determine causal future and past like it did before, in newtonian mechanics. The lightcones of all events are the objects really characterizing the causal structure of the theory.

 Furthermore, if we consider General Relativity, spacetime curvature (gravity) can shape lightcones differently, because it limits how physical information travels. For instance, gravity bends the otherwise straight paths naturally followed by light between 2 points, and those define the lightcones' surfaces. We can visualize that in Fig.~\ref{fig:BHlightcones}, which illustrates local and global lightcones in a spacetime generated by a black hole. Now, the spacetime matter configuration shaping the lightcones is determined by solving Einstein's equations. Therefore, causal relations in General Relativity behave as extra physical variables of the theory rather than part of a fixed structure on top of which the physical variables live.

 \begin{figure}
     \centering
     \includegraphics[scale=0.7]{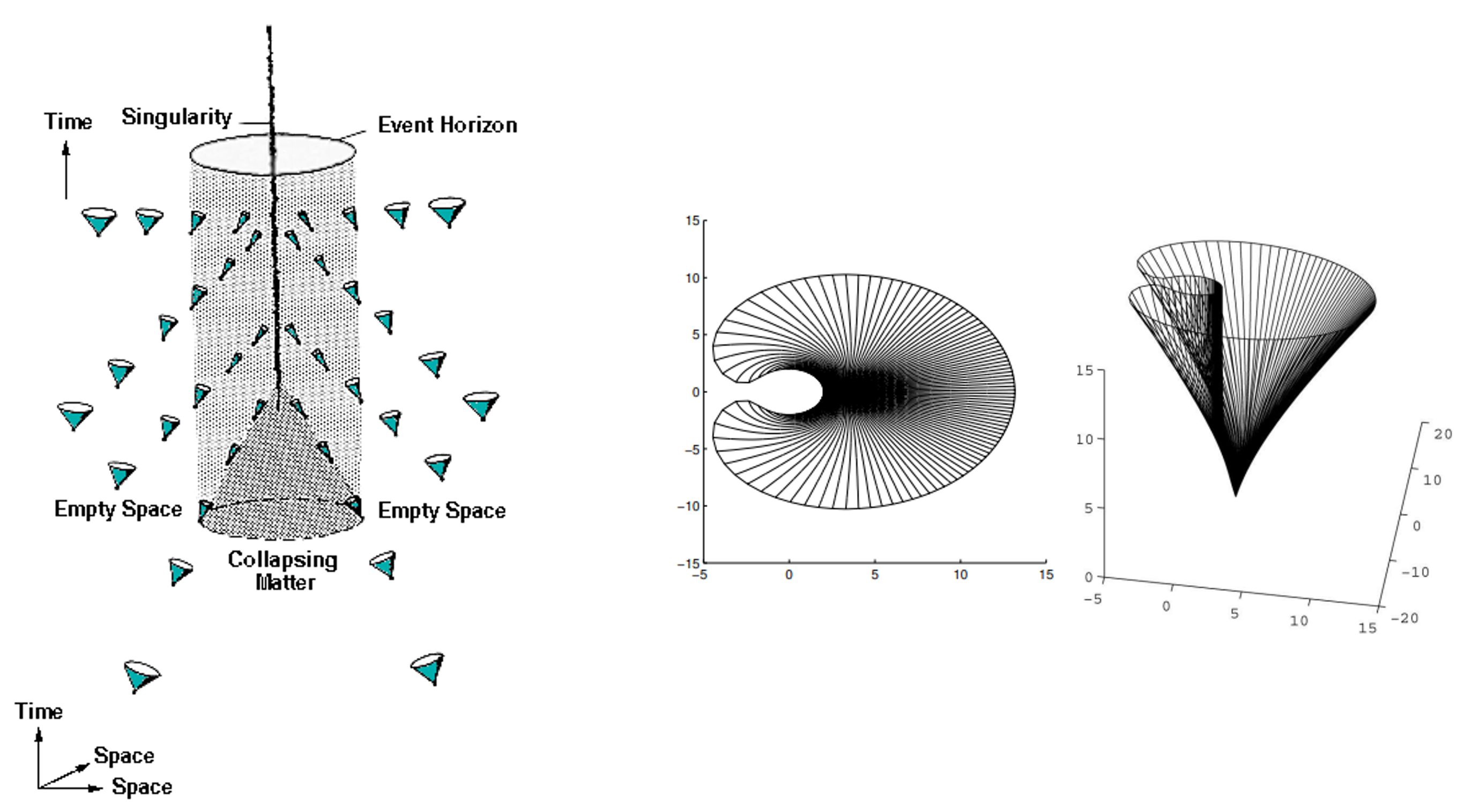}
     \caption{Black hole spacetime: an example of how lightcones can have modified shapes in gravitational configurations. a - representation of the local lightcones near the collapse of matter that forms a black hole (in Eddington-Finkelstein coordinates), from reference \protect\cite{Penrose1979}. b - top down and rotated views of the global future lightcone of a point in the proximity of a black hole, from reference \protect\cite{elmabrouk2013visualizing}.}
     \label{fig:BHlightcones}
 \end{figure}

 Quantum Theory, on the other hand, usually describes changes with respect to a global time parameter, for instance, through the quantum wave equation. In the expectation of making the theories more similar and possibly join their descriptions, we could wonder: what would it be like if causal structure was treated as a physical variable in quantum mechanics too? Physical quantities in Quantum Theory are described by observables, like position, momentum and spin. Measurements associated to them give out inherently probabilistic results and we cannot generally assign definite values to them independently of the measurement context. In other words, they are subject to quantum uncertainty. That would also happen to causal relations if they were treated within the quantum formalism.

 In a regime in which both General Relativity and Quantum Theory contributions are significant, we could speculate the appearance of this indefiniteness of causal relations, assuming the degrees of freedom shaping spacetime are described quantum mechanically~\cite{quantumCausality,tbell}. This could result in pairs of events for which order relations are indefinite, like a superposition of A being in the past of B and B being in the past of A.
 
 The idea of describing such possibly indefinite causal relations in Quantum Theory has been explored abstractly using quantum information methods~\cite{Oreshkov,Chiribella,Hardy}. The generalized causal structures resulting from these frameworks rely on operational approaches. For instance, operationally an event is not expressed as a point $(x,t)$ in a background, since this is usually a theory-dependent notion. Statements must be made based mainly on possible experimental outcomes probabilities. This shift brings consequences for the physical interpretation of the elements in the frameworks. For instance, it has been suggested that, in a certain context, indefinite order/quantum causal structure may appear within regular quantum mechanics scenarios, without gravity~\cite{Araujo2015,Oreshkov2019timedelocalized}. The central topic in this discussion is the quantum switch~\cite{Chiribella}, a theoretical protocol with indefinite order. The probabilities predicted for it are reproducible in photonic laboratories~\cite{ReviewExp,Rubino,Procopio,Goswami,Wei,Taddei,RubinoAgain}. Later, theoretical versions of that protocol involving quantum and classical gravity were constructed as well~\cite{tbell,Voji,QSonEarth}. By examining the proposals, one can possibly investigate the role of indefinite order as an indicator of incompatibility with causal structure for each context. This is what we intend to do in the second half of this work, after understanding how causal relations are dealt with in regular quantum mechanics and in the abstract process formalism in~\cite{Oreshkov}.

 The structure of this thesis is as follows. In chapter~\ref{Chap CausalityinQT}, we will introduce the mathematical definition of causal structure and how to approach causal relations from the point of view of a general probabilistic theory, along with a Bell-type causal inequality to illustrate the meaning of definite order~\cite{Oreshkov,Chiribellanetworks,Voji}, and then we will better understand how causal structure appears within Quantum Theory by doing an overview on operational quantum dynamics~\cite{Nielsen,MilzPollock,Kraus,Chiribellanetworks}. In chapter~\ref{Chap Process Matrix}, we present motivations and construct explicitly the bipartite process matrix formalism~\cite{Oreshkov,Chiribella}, a generalization of Quantum Theory which does not assume definite causal structure. We also introduce the quantum switch~\cite{Chiribella}, a typical process with indefinite order, and quickly discuss its implementations~\cite{ReviewExp}. From there, we study the formalism of ideal clocks on curved spacetimes~\cite{Zych} in chapter~\ref{Chap Clocks}. With the tools provided by it, we are able to describe, in chapter~\ref{Chap:QGravitySwitch}, the gravitational quantum switch, a thought experiment happening in a quantum gravity scenario~\cite{tbell}. Finally, in chapter~\ref{Chap QSonEarth}, an original work is presented regarding the formulation of a quantum switch. It uses gravitational time dilation as a resource to produce indefinite order of operations in a classical spherical spacetime. Other than being an example of indefinite order, the realization of the protocol in the gravity of Earth could be used to probe the regime of quantum mechanics on curved spacetimes.

\chapter{Causality in the quantum framework}\label{Chap CausalityinQT}     
A first step in learning about causality when Quantum Theory (QT) and General Relativity (GR) are both relevant is to understand how it appears in the mathematical formulation of each theory. In General Relativity, causality shows up right away encoded in lightcones. For each point in spacetime, the associated lightcone is defined as the union of the region to which it can send a signal with the region from which it can receive a signal that is no faster than light. As mentioned in the introduction, distinct distributions of mass and energy in spacetime shape the lightcones differently because, according to Eintein's equations, the distribution determines curvature and consequently dictates the possible paths followed by any signals. In the quantum realm, causality is not usually brought up in such a direct manner. The goal of this chapter is to understand how causality presents itself in Quantum Theory from an information-theoretic perspective.

First of all, we will introduce the mathematical definition of causal structure and a couple of examples~\cite{Chiribellanetworks,Voji}. Then, we present a way to address causal relations in a general probabilistic framework which is well known to the quantum information community, the signaling conditions for probabilities~\cite{Costa,Oreshkov,Chiribellanetworks}. They represent a notion of future and past that depends only on the probabilities directly acquired from experimental outcomes, and not on a specific physical theory. To illustrate this, we present the Bell-style causal inequality from references~\cite{Oreshkov,Costa}. Its violation would imply that the results obtained by two agents in a certain task are not compatible with causal structure. Next, we finally consider Quantum Theory through an overview on operational quantum dynamics~\cite{Nielsen,MilzPollock,Kraus,Chiribellanetworks,ChirQuantumCircuitArc,Chiribella}, the study of general evolutions of quantum systems.The mathematical objects that characterize how a general quantum state transforms into a final output state are called quantum operations. The description can also be generalized to account for multiple inputs/outputs, generating the so-called quantum networks or quantum combs~\cite{Chiribellanetworks,ChirQuantumCircuitArc,ChirSupermap}. These objects, whose definitions are grounded in QT, contain the information on how systems are allowed to evolve. They enable a more clear discussion on how causality appears in the general structure of the theory. Operational dynamics is also a particularly convenient topic since it introduces useful mathematical tools for the next chapter.

\section{Operational view on causal relations}\label{Operationalview}

In quantum scenarios, it is commonly assumed that particles and fields live on a background, sometimes with a global time coordinate. Causal structure is then known from the description General Relativity (or even Classical Theory) gives to that background. But if we intend to address problems that arise in the interface between the theories, where we know they might break, it is better to not rely on that. We should define what is past and future independently.

But how to do that? The safest way is to work with statements about what can be measured in principle, adopting an operational view. When comparing experimental data with theory predictions in physics, a fundamental assumption is made: recreating initial conditions of system and measurement apparatuses in laboratory counts as making the \emph{same} single measurement over and over, and the probabilities acquired from the frequency of results reveal a tendency that is a \emph{property} of the system+measurement ensemble\footnote{This is related to the objective/statistical interpretations of probability, such as the frequency interpretation, which was followed by the propensity interpretation by K. Popper~\cite{propensity_Popper}. Experiments with deterministic results can be seen as special cases.}. In the end, results of experiments are described in terms of probabilities. The role of a physical theory is to model system+measurement as mathematical entities from which probability distributions can be derived and then tested, thanks to that lower level assumption. In the spirit of neutrality, we can try to construct the notion of causal relations with statements about experimental outcomes. With this, causality can be analyzed for all imaginable general probabilistic theories under the paradigm of the assumption above, including the established physical theories to date. In particular, we may use it to talk about causality in Quantum Theory and its eventual generalizations.

The notion of causal structure can be mathematically described as follows.
\begin{definition}
A partial order relation over a set V is a binary relation $\prec$ that is
\begin{align*}
&\text{Reflexive}: x \prec x
\\
&\text{Antissymetric}: (x \prec y \text{ and } y\prec x) \implies x=y
\\
&\text{Transitive}: (x\prec y  \text{ and } y\prec z) \implies x \prec z.
 \end{align*}
$\forall x,y,z \in V$. A \textbf{causal structure} on a finite set V is characterized by a partial order relation $\prec$ over V. The elements of a set equipped with such relation will be called events.
\end{definition}
If we think about this partial order as representing the event on the left-hand side being in the past of (or equal to) the event on the right-hand side, it is natural to ask for the properties above to hold. Because of this intuitive notion of how causal relations should behave, a finite set\footnote{If the set was not finite, the partial order would be required to be locally finite in order to characterize a causal set. Only finite causal sets will be considered throughout the text.} equipped with a partial order is called a causal set. For example, if we construct a set V whose elements are a number of events in a Minkowski spacetime, the possibility of sending a signal between each two of those events naturally obeys those rules. Then, we can introduce a causal structure with the partial order defined by A$\prec$ B iff event A is in the past lightcone of or equal to event B. 

Another example of causal set is a circuit. A circuit is a mathematical object capturing the idea of a diagram of operations linked by arrows that indicate the order of their application on a system.
\begin{definition}
A \textbf{directional acyclic graph} is a pair $G=(N,E)$ where N is a set, whose elements are called \textbf{nodes}, and $E \subset \{(u,v) \mid u,v \in N\}$ is a set of ordered pairs of nodes, called \textbf{edges}, such that for any $x\in N$ there is no set $\{x=u_1,u_2...,u_k=x\}\subset N$ with $(u_i,u_{i+1}) \in E$  $\forall i \in\{1,...,k\}$. In words, there is no sequence of consecutive edges in $E$, a \textbf{path} for short, going from any node x to itself (acyclic property).
\end{definition}
\begin{definition}
A \textbf{circuit} $\mathcal{C}$ over a set of operations $\mathcal{G}$ is a pair $(G,g)$ where $G=(N,E)$ is a directional acyclic graph and $g:N \to \mathcal{G}$ is a function which assigns to each node an operation.\end{definition}
For any circuit $\mathcal{C}:=((N,E),g)$, we can define a natural partial order in $N$: $x\prec y$ iff there is a path from $x$ to $y$. This relation automatically obeys the properties listed before, characterizing a causal structure. The definition reproduces our intuitive idea of a circuit with no loops, which can be visualized as a diagram of directed wires connecting operation gates like the ones in Fig.~\ref{fig:circuit}.
\begin{figure}
\centering
\includegraphics[scale=0.8]{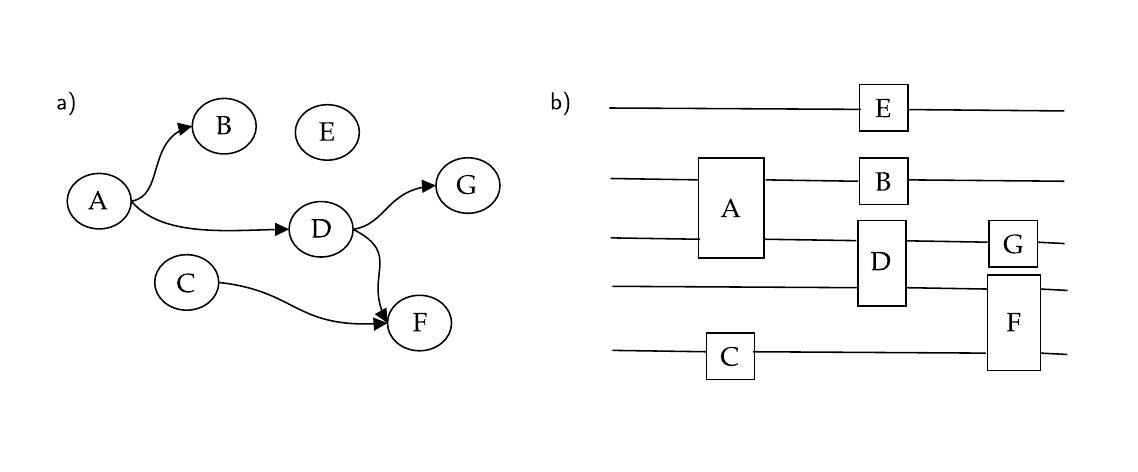}
\caption{\label{fig:circuit} Two ways one can schematically represent a circuit. The letters represent operations from the set $\mathcal{G}$ which label the nodes in the graph. a) The edges are represented by arrows connecting the labeled nodes. b) The edges are represented by wires connecting gates. The figure is supposed to be read from left to right.}
\end{figure}
  
   If the operations of a circuit are taken from the quantum operation formalism, we call it a quantum circuit. Quantum circuits are a common tool to describe the time evolution of quantum systems routinely employed in quantum computation. In this case, the description of a protocol begins with a given circuit, and one knows in advance in which order to apply operations in the mathematical framework to match the physical situation. If we assign spacetime locations to a collection of operations, the causal relations of spacetime induce a circuit structure for them. Thus, a causal structure is usually required \emph{a priori} in quantum mechanics, if not in the form of an underlying spacetime, by at least assuming a predefined circuit structure for the events of interest.

To understand how causal structure limits the statistics of experiments, we introduce the notion of signaling in a general probabilistic theory.

Consider a situation in which two experimenters, Alice and Bob, each perform one measurement over a physical system. Let $a$ and $b$ represent the general settings of the measurement apparatuses and $x$ and $y$ the outcomes measured by Alice and Bob, respectively. An example of this would be Alice and Bob acting on a spin $1/2$ system with Alice's measurement being $a$ = spin in the $z$ direction, obtaining result $x=+1/2$, and Bob's being $b$ = spin in the $x$ direction with result $y=-1/2$.

This class of bipartite problems is characterized\footnote{The reasoning made here and in the next chapter is partially heuristic, but one can find similarities with approaches that 
depart from basic operational elements, like preparations and effects, generating a class of probabilistic theories~\cite{Kraus,BarretGPT,MullerGPT}. The construction is mainly used to search for principles that uniquely identify QT among this sea of theories and to study correlations in a general setting~\cite{Hardy5axioms,DerivQTChir,Bellnonlocality_review}. Here, we will be specially concerned with signaling properties, as seen next.} by the form of a conditional probability function $P$ over the possible pairs of settings and possible outcomes for each pair. The symbol $P(xy \vert ab)$ represents the probability for Alice and Bob to measure results $x$ and $y$, given they make measurements with settings $a$ and $b$. Assuming that events are defined here such that each occurrence of an operation or measurement corresponds to one event, we will denote the two events in this problem by A: Alice realizes her measurement and B: Bob realizes his measurement.

\begin{definition}\label{cannotSignal}
We say A \textbf{cannot signal} to B if
 $$P(y\vert ab):=\sum\limits_{x}P(xy\vert ab)=P(y\vert b),$$
where $P(y\vert b)$ symbolizes that $P(y\vert ab)$ has no explicit dependence on the variable a, representing the settings of Alice's measurement. We denote this property by \text{A}$\nprec$\text{B}.
 If, otherwise, $P(y\vert ab)$ depends on $a$, we say \textbf{signaling is possible} from A to B, denoting \text{A}$\prec$\text{B}.
 
\end{definition}
Then, A not being able to signal to B means that the probability for Bob to obtain result $y$ from his measurement $b$ does not depend on the measurement $a$ done by Alice. This relates to the notion of A not being in the causal past of B: if Alice acts on a system before Bob, then she \emph{could} in principle choose a setting to modify the state of the system before Bob's measurement and influence his results. More generally, this should be true even if they manipulate distinct systems, because if Alice's action is in the causal past of Bob's there is a possibility of influence by definition. For example, her system \emph{could} interact (even if indirectly through other systems) with Bob's system before it gets to him and influence the results. Thus, if A was in the causal past of B, the probability function that accounts for all outcomes of \emph{any} of Bob's possible measurements would have to have explicit dependence on Alice's choice. Interestingly, this is similar to Einstein's construction of lightcones, substituting the possibility of sending a classical signal with the possibility of influencing a conditional probability to determine what lies in the ``future''.

In quantum mechanics, the most studied types of bipartite problems are the ones for which A cannot signal to B and B cannot signal to A. These are called non-signaling conditions and are used as a translation of spacelike separation for experiment statistics. That is the case when Alice and Bob share a bipartite system in an entangled state and measure their respective parts independently, with no channel to exchange information. Since the formulation of Bell-type inequalities and experimental tests for their violation~\cite{Bell,CHSH,Clauser_ExpBell,Aspect_ExpBell,Aspect2_ExpBell,Zeilinger_ExpBell}, it is understood\footnote{Thanks to scientists like Alain Aspect, John F. Clauser and Anton Zeilinger, awarded the Nobel prize in Physics 2022 for their work!} that although the tensor product description of spacelike separated systems by Quantum Theory always obeys non-signaling conditions, it still produces correlations that are more general than those allowed by Classical Theory for this class of problems. Thus, there is interest in specifying where the set of non-signaling quantum correlations stands between the set of classical (local-realistic) and the set of all non-signaling correlations\cite{Bellnonlocality_review,Nosignaling_quantumset,AlmostQuantumCorrel}.

If an underlying causal structure exists, it means there can be signaling at most in one direction. That is, a circuit either has a directed path from A to B, from B to A or the events are not linked, and similarly for the spacetime case with lightcones. At most, we could be uncertain of which one is the case. For instance, if Alice chooses to make her measurement at different times depending on the roll of a dice, she could in some rounds be at Bob's past and in others be at his future. But a classical probability distribution would describe that uncertainty. Expressing this mathematically, we get the following condition on the bipartite conditional probabilities:
\begin{equation}\label{classord}
    P(xy\vert ab) = q P^{\text{A}\nprec \text{B}}(xy\vert ab) + (1-q)P^{\text{B}\nprec \text{A}}(xy\vert ab) \qquad 0 \leq q \leq 1,
    \end{equation}
where $P^{\text{A}\nprec \text{B}}(xy\vert ab)$ is a probability function such that A cannot signal to B, that is, $P^{\text{A}\nprec \text{B}}(y\vert ab)= P^{\text{A}\nprec \text{B}}(y\vert b)$, and similarly for $P^{\text{B} \nprec \text{A}}$. We are using the convention that the probabilities for which A$\nprec$B and B$\nprec$A both hold are of the type $P^{\text{B}\nprec \text{A}}$. The condition above is, therefore, the restriction on probabilities imposed by demanding definite causality.

A remark to be made here is: it is not clear until now what is the mathematical nature of events A and B. Should they be regarded merely as elements of a causal set or do they have more structure, like events in General Relativity? We only established that the bipartite Alice and Bob case contains two events, each corresponding to the realization of one measurement. In reference~\cite{Costa}, two possibilities are discussed. The first is to consider that each operation is localized in an arbitrarily small region of spacetime, ideally a single event. Hence, there would be some well-defined local notion of spacetime, at least for the points where the measurements happen, and the types of probabilities (signaling or non-signaling) would allow us to deduce the causal relation between those spacetime events. The second possibility is to consider that the operations happen in closed laboratories. A closed laboratory could be \emph{``pictured as a finite region of spacetime bounded by two spacelike surfaces such that physical systems can enter in the laboratory only from the past surface and can exit only from the future one, while no exchange of information is possible through the time-like boundaries of the region''}~\cite{Costa}. This is to avoid the ``arbitrarily small'' part and consider regions instead of points. For this case, each of the events A and B still have the status of being in a classical spacetime location, with no need to acknowledge a global spacetime structure. This discussion is important because we want the analysis to be as operational as possible and assume the minimum without relying on specific physical theories to proceed.  

It is argued, however, that not even the notion of spacetime regions for the laboratories is necessary. For instance, assuming Quantum Theory is valid for the measurements, the set of allowed quantum operations that an agent can perform on a single system can be used as an abstract definition of closed laboratory, a concept explored further in~\cite{Giarmatzipaper}. This definition making no reference to localization can generate some broad interpretations for quantum experiments, as we will later comment in chapter~\ref{Chap Process Matrix}. In this text, we will mostly refer to closed laboratories, whether they are thought of as regions where spacetime is locally definite or in terms of the information-theoretic definition, following the approach in references~\cite{Oreshkov,Costa}.

Considering these definite notions of localization of laboratories, the general probability in~(\ref{classord}) can be interpreted as a situation for which the localization of A and B may not be known with certainty, but only because of a classical ignorance. That is, similarly to when we do not know the positions of each particle of a gas in statistical mechanics, we could possibly not be aware of what is the location of the two events with certainty. Surely, though, assuming a definite causal structure, for each round of experiment the possibility of signaling from A to B precludes signaling from B to A. So. independently of the exact nature of the events, we were able to formulate \emph{a necessary condition for two events in a bipartite experiment to be in a causal structure} in terms of probabilities~\eqref{classord}.

\label{causalineqsec}\section{Causal Inequality}
Let us introduce a thought bipartite experiment to conceptually clarify the condition in~(\ref{classord}). If a causal structure is assumed for this protocol, we can arrive at a statement about some evaluations of the probability function $P(xy\vert ab)$, a \emph{causal inequality}, as proposed in~\cite{Oreshkov}. The violation of a causal inequality is considered to be a task impossible to accomplish if the events of a probabilistic experiment are in a causal structure, in a very similar way that the violation of a Bell inequality~\cite{Bell} would be impossible for classical states assuming local realism and measurement independence.

Consider again an experiment involving two agents, Alice and Bob, in their closed laboratories. In every round of the experiment, Alice and Bob will each receive a physical system once, perform an operation and send the system out of the laboratory. As discussed before, the laboratories are isolated while the experimenters make their operations. 

Suppose signaling is possible from Alice to Bob and that Bob can only receive information through the system that enters his laboratory. Hence, Alice's operation must have happened before the system entered Bob's lab, because it was able to influence it. In this case, if the laboratories are in a definite causal structure, Bob cannot send a signal back to Alice, because she has already received her system and done her operation. Alice also can only receive information through the system that enters her laboratory, and that happens once, as pointed out in Fig.~\ref{fig:causalineq}(a). To illustrate, we can think that Alice and Bob enter their laboratories only in time to make their operations and leave right after that. In this scenario, by the time Bob makes his operation, Alice has already left her laboratory so he cannot communicate his operation to her.

\begin{figure}
\centering
\includegraphics[scale=0.3]{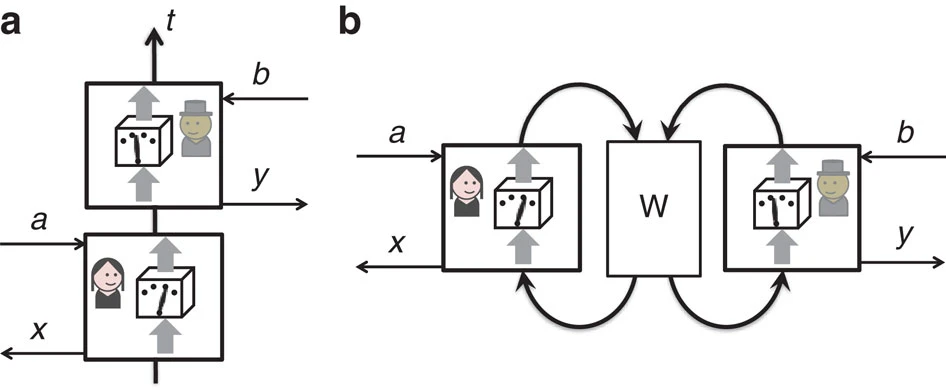}
\caption{(a) There exists a global background time according to which Alice's actions happen before Bob's. She sends her input $a$ to Bob, who can read it out at some later time. However, Bob cannot send his bit $b$ to Alice, since the system has already passed through her laboratory earlier. (b) Representation of generic situation where a causal structure is not assumed and could possibly violate a causal inequality, characterizing `indefinite order'. This will be discussed again in  section~\protect\ref{violation} of the next chapter. Figure from~\protect\cite{Oreshkov}.}
\label{fig:causalineq}
\end{figure}

The agents will be challenged with a communication task: every time one of them receives their system, that person will generate a random bit, which can assume the values 0 or 1. In each round, Alice will generate a bit $a$ and Bob will generate a bit $b$. Then Alice and Bob will try to guess each other's bit values. Alice's guess about $b$ will be called $x$ and Bob's guess about $a$ will be called $y$. Bob will also generate one extra random bit $b'$ and the game is that, if $b'$=0, we will check if Alice's guess $x$ is right and discard Bob's guess, while if $b'$=1, we will check if Bob guessed right and discard Alice's guess. The goal is to make as many right non-discarded guesses as possible, using the system as a resource if needed. That means maximizing the probability of success,
\begin{equation}
    P_{succ}:= \frac{1}{2}\left[P(x=b\mid b'=0) + P(y=a\mid b'=1) \right].
\end{equation}
For example, if signaling from Alice to Bob is possible and Alice is a good player, she will encode her bit value in the system, so that he ``guesses'' right if $b'$ happens to be $1$. The following inequality is always true if the events involved in this protocol are in a definite causal structure:
\begin{equation} 
       P_{succ} \leq \frac{3}{4}.
\end{equation}
This expression is a causal inequality. The idea of the proof is the following: if the events of Alice's and Bob's laboratories are in a definite causal structure, for each round it is true that either Bob cannot signal to Alice or Alice cannot signal to Bob\footnote{It can be both.}. Consider Bob cannot signal to Alice. The best case scenario for the guesses is if Alice can signal to Bob, otherwise it just means both guesses will be random with 1/2 chance to be right. So, consider Alice can signal to Bob. Then, if $b'=1$, only Bob's guess will be valid for this round and Alice can (and will, if trying to win the game) send him her bit value so he guesses it right, making $P(y=a\mid b'=1)=1$. While if $b'=0$, we will check Alice's guess, and the probability for her to get it right will be $P(x=b\mid b'=0)=1/2$, a random guess, since Bob cannot communicate with her. If Alice cannot signal to Bob, the situation is completely analogous exchanging the roles of Alice and Bob. In any situation, the best they could achieve is getting it always right for one value of $b'$ and getting it right with 50$\%$ probability for the other value of $b'$, resulting in the causal inequality:
\begin{equation} \label{Causal Inequal complete}
       P_{succ}= \frac{1}{2}\left[P(x=b\mid b'=0) + P(y=a\mid b'=1)\right] \leq \frac{1}{2}\left[1+ \frac{1}{2}\right]= \frac{3}{4}.
\end{equation}

Let us elaborate on the assumptions and formal proof of the inequality above. In the bipartite task described, assume:

\begin{itemize}
    \item[i.]\textbf{Causal structure}: The events A$_1$ and B$_1$ of the systems entering the laboratories of Alice and Bob, along with the events A, B, B', X and Y corresponding to obtaining the bits $a, b, b'$ and producing guesses $x$ and $y$ respectively, are in a causal structure. That is to say there exists a well-defined partial order in this 7 element set, as discussed earlier.
    
    \item[ii.]\textbf{Free-choice}: The bits $a, b$ and $b'$ can only be correlated with events in their causal future and each of them takes values $0$ or $1$ with probability $1/2$.
    
    \item[iii.] \textbf{Closed laboratories}: The guess $x$ can only be correlated with $b$ if B$ \prec$ A$_1 $ and the guess $y$ can only be correlated with $a$ if A$\prec$ B$_1$. 
\end{itemize}

The second assumption is a way to express that the bits generated are really independent and random. It implies all three bits are uncorrelated, since they can only correlate to events in their causal future and no pair of events can be in each other's causal future. The third assumption assures information can only arrive to the agents and affect their guesses through the entry of their laboratories. With these considerations, let us prove the inequality.
     
\begin{proof}

From assumption i, in each round of the experiment, only one of the following holds: A$_{1} \prec$ B$_{1},$ B$_{1} \prec $A$_{1},$ or neither:$A_{1}  \nprec \nsucc B_{1}$. Then, their probabilities obey
\begin{equation}\label{prob1}
    P\left(\text{A}_{1} \prec \text{B}_{1}\right)+P\left(B_{1} \prec \text{A}_{1}\right)+ P\left(\text{A}_{1}  \nprec \nsucc  \text{B}_{1}\right)=1
\end{equation}

a) From assumption ii, the bits $a, b, b'$ do not depend on which of the situations above holds.

Consider $b$, for instance. Since the bit is generated in Bob's laboratory after the system is received, we have B$_1\prec$ B. It is also true that the probability for A$_1\prec$ B$_1$ to happen cannot depend on $b$ by assumption ii, since it concerns two events in the causal past of B. Then,
\begin{equation}\label{A1precB1}
P\left(\text{A}_{1} \prec \text{B}_{1} \vert b\right)=P\left(\text{A}_{1} \prec \text{B}_{1}\right)
\end{equation}

If A$_1$ is outside the causal future of B$_1$, it is also outside the causal future of B, because of the transitivity property. Again by assumption ii, there cannot be a correlation between $b$ and the situation B$_1\nprec$A$_1$:
\begin{equation}\label{B1nprecA1}
P\left(\text{B}_{1} \nprec \text{A}_{1} \vert b\right)=P\left(\text{B}_{1} \nprec \text{A}_{1}\right).
\end{equation}
By definition we have
\begin{align*}
 P\left(\text{B}_{1} \nprec \text{A}_{1} \vert b\right)&= P\left(\text{A}_{1} \prec \text{B}_{1} \vert b\right) + P\left(\text{A}_{1} \nprec \nsucc \text{B}_{1} \vert b\right)
 \\
 &=P\left(\text{A}_{1} \prec \text{B}_{1} \right) + P\left(\text{A}_{1} \nprec \nsucc \text{B}_{1} \vert b\right),
\end{align*}
where we used Eq.~\eqref{A1precB1}. And the left hand side is not correlated with $b$ by  Eq.~\eqref{B1nprecA1}. So we can conclude that 
\begin{equation}\label{A1nprecnsuccb1}
P\left(\text{A}_{1} \nprec \nsucc \text{B}_{1} \vert b\right) = P\left(\text{A}_{1} \nprec \nsucc \text{B}_{1} \right)
\end{equation}
as well. Now, using Eq.~\eqref{A1precB1} and Eq.~\eqref{A1nprecnsuccb1}, we get
\begin{align*}
        1 &= P\left(\text{A}_{1} \prec \text{B}_{1}\vert b\right)+P\left(\text{B}_{1} \prec \text{A}_{1}\vert b\right)+ P\left(\text{A}_{1}  \nprec \nsucc  \text{B}_{1}\vert b\right)
        \\
        &=P\left(\text{A}_{1} \prec \text{B}_{1}\right)+P\left(\text{B}_{1} \prec \text{A}_{1}\vert b\right)+ P\left(\text{A}_{1}  \nprec \nsucc  \text{B}_{1}\right)
        \\
        &=P\left(\text{A}_{1} \prec \text{B}_{1}\right)+P\left(\text{B}_{1} \prec \text{A}_{1}\right)+ P\left(\text{A}_{1}  \nprec \nsucc  \text{B}_{1}\right)
        \\
        &\implies P\left(\text{B}_{1} \prec \text{A}_{1}\vert b\right) = P\left(\text{B}_{1} \prec \text{A}_{1}\right),
\end{align*}
where we used Eq.~\eqref{prob1} to go from the second to the third line. So we proved bit $b$ cannot depend on whether we have A$_1 \prec$ B$_1$, B$_1 \prec$ A$_1$ or A$_1 \nprec \nsucc$ B$_1$. The exact same argument can be made for bit $b'$ while for bit $a$ we only need to swap A$_1$ with B$_1$.

b) Let us use the fact above to calculate the probability of success summing over all possibilities:
\begin{align*}
P_{\text {succ }}&:=\frac{1}{2} \left[P\left(x=b \mid b^{\prime}=0\right)+P\left(y=a \mid b^{\prime}=1\right)\right]
\\
&=\frac{1}{2} P\left(x=b \mid b^{\prime}=0 ; \text{A}_{1} \prec \text{B}_{1}\right) P\left(\text{A}_{1} \prec \text{B}_{1}\right)
\\
&+\frac{1}{2} P\left(x=b \mid b^{\prime}=0 ; \text{B}_{1} \prec \text{A}_{1}\right) P\left(\text{B}_{1} \prec \text{A}_{1}\right)
\\
&+\frac{1}{2} P\left(x=b \mid b^{\prime}=0 ; \text{A}_{1} \not \npreceq \nsucceq \text{B}_{1}\right) P\left(\text{A}_{1} \npreceq \nsucceq \text{B}_{1}\right)
\\
&+\frac{1}{2} P\left(y=a \mid b^{\prime}=1 ; \text{A}_{1} \prec \text{B}_{1}\right) P\left(\text{A}_{1} \prec \text{B}_{1}\right)
\\
&+\frac{1}{2} P\left(y=a \mid b^{\prime}=1 ; \text{B}_{1} \prec \text{A}_{1}\right) P\left(\text{B}_{1} \prec \text{A}_{1}\right)
\\
&+\frac{1}{2} P\left(y=a \mid b^{\prime}=1 ; \text{A}_{1} \not \npreceq \nsucceq \text{B}_{1}\right) P\left(\text{\text{A}}_{1} \not \npreceq \nsucceq \text{B}_{1}\right)
\end{align*}
\begin{align}\label{psuccdecomp}
    &=\frac{1}{2} \left[ P\left(x=b \mid b^{\prime}=0 ; \text{A}_{1} \prec \text{B}_{1}\right) + P\left(y=a \mid b^{\prime}=1 ; \text{A}_{1} \prec \text{B}_{1}\right) \right] P\left(\text{A}_{1} \prec \text{B}_{1}\right)\nonumber
    \\
    &+\frac{1}{2} \left[ P\left(x=b \mid b^{\prime}=0 ; \text{B}_{1} \prec \text{A}_{1}\right) + P\left(y=a \mid b^{\prime}=1 ; \text{B}_{1} \prec \text{A}_{1}\right) \right] P\left(\text{B}_{1} \prec \text{A}_{1}\right)\nonumber
    \\
    &+\frac{1}{2} \left[ P\left(x=b \mid b^{\prime}=0 ; \text{A}_{1} \nprec \nsucc \text{B}_{1}\right) + P\left(y=a \mid b^{\prime}=1 ; \text{A}_{1} \nprec\nsucc \text{B}_{1}\right) \right] P\left(\text{A}_{1} \nprec \nsucc \text{B}_{1}\right).
\end{align}

Consider the case A$_1\prec$B$_1$. The first term of the first line above is  
\begin{align*}
     P\left(x=b \mid b^{\prime}=0 ; \text{A}_{1} \prec \text{B}_{1}\right) = 
     P\left(b=0; x=0 \mid b^{\prime}=0 ; \text{A}_{1} \prec \text{B}_{1}\right)+\\
     +P\left(b=1;x=1 \mid b^{\prime}=0 ; \text{A}_{1} \prec \text{B}_{1}\right).
\end{align*}
In this case, we have A$_1\prec$ B, because B$_1\prec$ B. Therefore, from assumption iii, $b$ cannot be correlated with $x$. We also have, from assumption ii, that $b$ is not correlated with $b'$ and that it does not depend on whether condition A$_1\prec$ B$_1$ is satisfied or not, as proven in a). The probability for b to assume values 0 or 1 is thus 1/2 when conditioned on these independent variables. The probability above becomes 
\begin{align}
    P\left(x=b \mid b^{\prime}=0 ; \text{A}_{1} \prec \text{B}_{1}\right) =& P\left(b=0\mid x=0;b^{\prime}=0 ; \text{A}_{1} \prec \text{B}_{1}\right)P\left(x=0\mid b^{\prime}=0 ; \text{A}_{1} \prec \text{B}_{1}\right)\nonumber
\\
&+P\left(b=1\mid x=1; b^{\prime}=0 ; \text{A}_{1} \prec \text{B}_{1}\right)P\left(x=1\mid b^{\prime}=0 ; \text{A}_{1} \prec \text{B}_{1}\right)\nonumber
\\
=&\frac{1}{2}P\left(x=0\mid b^{\prime}=0 ; \text{A}_{1} \prec \text{B}_{1}\right)+ \frac{1}{2}P\left(x=1\mid b^{\prime}=0 ; \text{A}_{1} \prec \text{B}_{1}\right)\nonumber
\\
=&1/2.
\end{align} 
The last equality comes from Alice's guess not depending on Bob's bit, since this is the case A$_1\prec$ B$_1$.

For B$_1\prec$ A$_1$, the analogous argument leads to
\begin{equation}
     P\left(y=a \mid b^{\prime}=1 ; \text{B}_{1} \prec \text{A}_{1}\right) = 1/2,
\end{equation}
while for A$_1\nprec \nsucc $B$_1$, the guesses are all independent of the bits, giving both
\begin{align}
     &P\left(x=b \mid b^{\prime}=0 ; \text{A}_{1} \nprec \nsucc \text{B}_{1}\right)= 1/2, 
\\
      &P\left(y=a \mid b^{\prime}=1 ; \text{A}_{1} \nprec\nsucc \text{B}_{1}\right)=1/2.
\end{align}

Substituting the four last results in Eq.~\eqref{psuccdecomp} we get
\begin{align*}
P_{s u c c}=&\left(\frac{1}{4}+\frac{1}{2} P\left(y=a \mid b^{\prime}=1 ; \text{A}_{1}\prec \text{B}_{1}\right)\right) P\left(\text{A}_{1} \prec \text{B}_{1}\right)+
\\
&+\left(\frac{1}{2} P\left(x=b \mid b^{\prime}=0 ; \text{B}_{1} \prec \text{A}_{1}\right)+\frac{1}{4}\right) P\left(\text{B}_{1} \prec \text{A}_{1}\right)
+\left(\frac{1}{4}+\frac{1}{4}\right) P\left(\text{A}_{1} \not \nprec \nsucc \text{B}_{1}\right)
\\
&\leq \frac{3}{4}P\left(\text{A}_{1} \prec \text{B}_{1}\right) + \frac{3}{4} P\left(\text{B}_{1} \prec \text{A}_{1}\right)+ \frac{2}{4} P\left(\text{A}_{1} \not \nprec \nsucc \text{B}_{1}\right) 
\end{align*}
\begin{equation}\label{Causal Inequal}
    \implies P_{s u c c} \leq \frac{3}{4}.
\end{equation}
\end{proof}
Now, we have found a property that is sufficient for showing that events in this specific setting are not in a causal structure: the probabilities violating the causal inequality~\eqref{Causal Inequal}. With this example, we can better appreciate the statement on probabilities made before~\eqref{classord}. Violating the inequality seems to require that
the sole interaction of Alice sends information to Bob in time for his sole interaction to be influenced, as well as the opposite, Bob's sole interaction influences Alice's, in the same round of experiment. In a global time, if Alice makes her operation at $t=1$s,
there is no way Bob could use his only interaction with the system at $t=3$s to send information that reaches her at $t=1$s. She, otherwise, could send him information. Even considering relativity of time, having at most one directional signaling well captures the idea of causality, implying impossibility to send information to one's past, and causal inequalities are a way to attest incompatibility of results with causal structure. We can say that if experiment results disobeyed the inequality, that would attest ``indefinite causal order''. Intuition says this is a bound never to be violated. Yet, in the next chapter we will discuss the Process Matrix formalism, which contains probabilities that do violate this inequality~(\ref{violation}) while obeying quantum mechanics locally. The motivations for considering such odd correlations have to do with how causality appears in QT, how that differs from GR and how to put the theories on an equal footing with respect to that. One step at a time, let us focus on Quantum Theory.

\section{Operational Quantum Dynamics}

Now that we formulated causality related notions based on experiment statistics, we should explore what Quantum Theory predicts for those statistics. When moving freely, quantum systems are described by state operators $\rho$ that evolve unitarily according to a wave equation. But how does quantum mechanics describe the situation where two agents realize one interaction each? In a quantum protocol we typically ask first: does Alice act on the system before Bob, is it Bob who acts first or do they act independently? The latter can be the case where they make measurements on an entangled pair of particles at spacelike separated events. Say Alice makes a measurement represented by $\mathcal{A}$ and Bob makes another measurement represented by $\mathcal{B}$. In the first case, the initial quantum state $\rho$ will be updated to $\mathcal{B}\mathcal{A}(\rho)$\footnote{Composition symbols will be suppressed between operators. This reads $\mathcal{B}\circ\mathcal{A}(\rho)=\mathcal{B}(\mathcal{A}(\rho))$.}, while if Bob acts first, we get $ \mathcal{A}\mathcal{B}(\rho)$. If they are spacelike separated, the description of the problem is changed, since the state space is identified with a tensor product of two spaces for which it makes sense to write $\mathcal{A}=\mathcal{A}\otimes\mathcal{I}$ and $\mathcal{B}=\mathcal{I}\otimes\mathcal{B}$. The update, then, does not depend on the order: $\mathcal{A}\otimes\mathcal{B}(\rho)$. Thus, there is an implicit assumption when we approach the dynamics of a quantum system: if it undergoes certain transformations, we have to say in advance in what order they occur, so that we know how the operators compose. As discussed before, this comes with the \emph{a priori} specification of the underlying spacetime or circuit and, therefore, causal structure.

But how natural is it, from the point of view of the theory alone, that the general Alice-Bob experiment is restricted to one of the 3 cases above? Is it a property of Quantum Theory to only output probabilities compatible with causal structure or is it an additional assumption? It all goes back to which we consider to be the most general transformations quantum systems can undergo according to the formulation of the theory. To answer this, we will make an overview on quantum operations~\cite{Nielsen, MilzPollock, Kraus} and comments on some generalizations~\cite{MilzPollock,Chiribellanetworks,ChirQuantumCircuitArc,Chiribella} which provide a useful way to understand the general form of quantum time evolutions.

Aside from highlighting the mathematical aspects of Quantum Theory that point to causality, quantum operations are important because, more basically, they represent the possible actions that agents like Alice and Bob can realize inside a laboratory according to quantum mechanics. They determine, for instance, the variables we named a and b in the last sections for the conditional probabilities $P(xy|ab)$ in the quantum formalism.

 \subsection{Basic concepts and notation for quantum mechanics}\label{notationOps}
Here we outline the basic notions of quantum mechanics needed to discuss quantum operations~\cite{Nielsen,schuller}. A Hilbert space is a complex vector space with inner product $(\mathcal{H},+,\cdot_\mathds{C},\braket{\cdot|\cdot})$ that is complete with respect to the induced norm. The adjoint of a densely defined operator $A:\mathcal{D}_A\to\mathcal{H}$ is the operator $A^\dagger$ satisfying $\braket{A^{\dagger}x|z}=\braket{x|Az} \forall z\in \mathcal{D}_A$ defined wherever such $y=A^\dagger x\in\mathcal{H}$ exists for $x$. The space of linear bounded operators from $\mathcal{H}_1$ to $\mathcal{H}_2$ will be denoted $\mathcal{L}(\mathcal{H}_1,\mathcal{H}_2)$ and the subspace of operators $A$ for which the trace $\operatorname{Tr}\left(\sqrt{A^{\dagger} A}\right)$ is finite, the trace-class operators, by $\mathcal{T}(\mathcal{H}_1,\mathcal{H}_2)$. Whenever $\mathcal{H}_1=\mathcal{H}_2=\mathcal{H}$, we will shorten the notation to $\mathcal{L}(\mathcal{H})$ and $\mathcal{T}(\mathcal{H})$. For finite dimensional Hilbert spaces, the space of trace-class linear bounded operators $\mathcal{T}(\mathcal{H})$ becomes the entire space of linear operators $\mathcal{L}(\mathcal{H}).$

\textbf{States. }A separable Hilbert space $\mathcal{H}$ can be associated to every isolated physical system and the possible states of the system are represented by operators $\rho\in \mathcal{T}(\mathcal{H})$ that are positive and have unit trace, the density operators. An operator $\eta \in \mathcal{L}(\mathcal{H})$ is positive if\footnote{From now on, we will use Dirac's braket notation. For reference, the symbol $\ket{k}\bra{j}$ represents the map from $\mathcal{H}$ to itself acting like $\braket{j,\cdot} k$ where $j,k$ are elements of a basis for $\mathcal{H}$.} $\bra{\psi} \eta \ket{\psi}\geq0$, $\forall \ket{\psi} \in \mathcal{H}$, and $\eta$ has unit trace if $\sum_n \bra{e_n}\eta\ket{e_n}=1$ for any basis $\{\ket{e_n}\}$ of $\mathcal{H}$. Let us call the set of density operators $\mathcal{Q(\mathcal{H})}$. Density operators provide a more comprehensive description of quantum states than unit vectors of $\mathcal{H}$, since they include general subsystems and statistical ensembles of pure states. Unit vectors can only describe pure states. When that is the case, the density operator associated to a state vector $\ket{\psi}\in \mathcal{H}$ is the projector $\ket{\psi}\bra{\psi}\in \mathcal{L}(\mathcal{H}).$

\textbf{Evolution. }The evolution of a closed quantum system is given by a unitary transformation, namely $U\in \mathcal{L}(\mathcal{H})$ such that $U^\dagger U=UU^\dagger=\mathds{1}_\mathcal{H}$, where $\mathds{1}_\mathcal{H}$ is the identity. It updates the state as $\rho \to U\rho U^{\dagger}$. For pure states, vectors transform as $\ket{\psi}\to U\ket{\psi}$. This is the discrete picture for evolution, which contrasts with the perhaps more familiar quantum wave equation for state vectors $i\hbar\partial{\ket{\psi}}/\partial{t}=H \ket{\psi}$ or the corresponding equation for density operators $i\hbar\partial\rho/\partial t = [H,\rho]$. However, the descriptions are equivalent, since every unitary can be written as $U=e^{-iH/\hbar}$ for some self-adjoint operator $H$ that can be interpreted as the hamiltonian of the system~\cite{Nielsen}.

\textbf{Measurements.} A generalized measurement is described by a family of operators $\{M_m\in \mathcal{L}(\mathcal{H})\},$ where the index m takes value in the set of all possible outcomes. The operators are required to obey the completeness relation, $\sum_m M_m^{\dagger} M_m=\mathds{1}$. For a initial state $\rho$, they give us the probability to obtain the result m, $p(m)=\operatorname{Tr}(M_m^{\dagger}M_m \rho)$, and also the state after the measurement, $\rho'=M_m \rho M_m^{\dagger}/\operatorname{Tr}(M_m^{\dagger}M_m \rho)$. The completeness relation assures the probability to measure any of the outcomes is 1. 

This description includes the special case of a special case of a projective  measurement: if the operators are orthogonal projectors $M_m=P_m$, we have $M_m^{\dagger}M_m=P_m^2=P_m$ and $p(m)=\operatorname{Tr}(P_m\rho)$. This measurement can be associated to the observable $A=\sum_m m P_m$ and it is called a PVM (Projection-Valued Measure) measurement. The more general POVM (Positive Operator-Valued Measure) measurements are also included above, with elements given by $E_m=M_m^{\dagger}M_m$~\cite{Nielsen}.

Although the sum symbols suggest a discrete set of allowed outcomes, this can be adapted for spectra with continuous parts~\cite{hall,schuller}. For instance, if a particle can be found anywhere in one dimension, one can define a PVM family $\{M_E:=Q(E)\}$ of position operators, one for each Borel-measurable set $E$ of the real line, such that $Q(\mathds{R})=\mathds{1}$ and the relation $E\to\operatorname{Tr}(Q(E) \rho)$ is well defined as a probability measure for every state $\rho$. The family is associated to a self-adjoint operator, the position observable, which can be written as $X = \int \lambda dQ(\lambda)$ according to the spectral theorem. In such infinite dimensional cases, attention might be required regarding domain and boundedness of operators. The linear operators $M_E^{\dagger} M_E$, however, are constructed to be positive and define a measure with $M_{\Omega}^\dagger M_{\Omega}=\mathds{1}$, where $\Omega$ is the union of all sets $E$. In particular, for unit vectors $\ket{\psi}$ we have $p_{\ket{\psi}}(E)=\operatorname{Tr}(M^{\dagger}_EM_E\ket{\psi}\bra{\psi})=||M_E\ket{\psi}||^2\leq p_{\ket{\psi}}(\Omega)= 1$, whenever this expression is defined. Thus, they are naturally defined as bounded operators. Nevertheless, we will not be too concerned about these matters throughout the text as most of the discussion happens around simple finite dimensional cases.

\textbf{Other remarks.}
 Given a finite dimensional Hilbert space $\mathcal{H}$, the space $\mathcal{L}(\mathcal{H})$ of linear operators is also a complex vector space with basic operations inherited from $\mathcal{H}$ and we can define a canonical inner product, the Hilbert-Schmidt product, as $\braket{\eta|\rho}:= \operatorname{Tr}(\eta^{\dagger}\rho)$. If $d$ is the dimension of $\mathcal{H}$, $d^2$ is the dimension of the space of linear operators $\mathcal{L}(\mathcal{H})$. Having chosen some basis  $\{\ket{k}\}_{k\leq d}$ for $\mathcal{H}$, we can use it to construct a basis for $\mathcal{L}(\mathcal{H})$ given by $\{\ket{i}\bra{k}\}_{i,k \leq d}$.

\subsection{Quantum Operations Axiomatic}\label{QuantumOpAxiomatic}
As summarized above, Quantum Theory for closed systems establishes there are two ways for a physical state to change: it can evolve unitarily or a measurement can be performed on it so that the state transforms into a final measurement state\footnote{An eigenstate, when measuring an observable with discrete non-degenerate spectrum.}. However, we also have open systems, that is, systems capable of interacting with other systems that do not enter the calculations explicitly. This is a practical issue because no system we prepare is completely isolated in reality, and one wishes to understand the noise caused by the environment. Furthermore, thinking of the more fundamental aspect, we would like to characterize how a quantum system can change in the most general physical scenarios. This includes non-isolated systems and even those strongly coupled to an environment and undergoing measurements not necessarily restricted to them.

Quantum operations intend to be the mathematical objects describing the most general changes quantum states can undergo. They are maps 
$\mathcal{E}$ transforming an input into an (unnormalized) output state. This makes them different from other formalisms used for open systems because they are well suited for describing discrete changes without reference to a continuous time parameter.

The evolutions of closed systems can be written as operations: a unitary transformation produces a state given by $\mathcal{E}(\rho)=U \rho U^\dagger$ and when a measurement described by a family $\{M_m\}$ returns the output $m$, the state transforms up to normalization into $\mathcal{E}(\rho)=M_m \rho M_m^\dagger$. We can conceive other operations as well, such as a measurement followed by unitary dynamics followed by another measurement, and things of the sorts. What is the form of the most general transformation $\rho \to \rho'$? One way to investigate that is to begin with general maps $\mathcal{E}:\mathcal{T}(\mathcal{H_I})\to \mathcal{T}(\mathcal{H_O})$ and establish physically reasonable axioms grounded in QT for the restrictions they should obey. Let us take a look at that approach as presented in references~\cite{Kraus,Nielsen}:

\begin{definition}
A \textbf{quantum operation} $\mathcal{E}$ is a map from the set of trace-class operators of an input Hilbert space $\mathcal{T}(\mathcal{H_I})$ to the trace-class operators of an output Hilbert space $\mathcal{T}(\mathcal{H_O})$ such that:

\begin{itemize}
    \item[A1.] Given a density operator $\rho\in \mathcal{Q}(\mathcal{H_I}) \subset\mathcal{T}(\mathcal{H_I})$, the trace $\operatorname{Tr}(\mathcal{E}(\rho))$ is the probability for the operation represented by $\mathcal{E}$ to occur when the initial state is given by $\rho$. Therefore, we must have $\operatorname{Tr}(\mathcal{E}(\rho))\leq 1$ for any density operator $\rho\in \mathcal{Q}(\mathcal{H_I}).$

    \item[A2.] $\mathcal{E}$ is a linear map. 
    
    \item[A3.] $\mathcal{E}$ is completely positive (CP). That is, $\mathcal{E}(A)$ must be positive for any positive operator $A\in \mathcal{T}(\mathcal{H}_I)$, and furthermore, if we introduce an extra system with Hilbert space $R$ of arbitrary dimensionality, it must be true that $(\mathcal{I}_R\otimes \mathcal{E})(A)$ is positive for any positive operator $A \in \mathcal{T}(R\otimes \mathcal{H}_I)$, where $\mathcal{I}_R$ denotes the identity supermap on operators of the system $R.$
\end{itemize}
\end{definition}
 We allow the input and output spaces of operations to be distinct to cover the possibility that part of the initial system ($\mathcal{H}_I$) is used and then discarded, making it different from the output space ($\mathcal{H}_O$). Let us understand why the 3 axioms above are considered to be the minimum physical requirements for a transformation:

\textbf{A1.} The first axiom is made to account for measurements. A standard operation that corresponds to a unitary evolution, $\mathcal{E}(\cdot) = U (\cdot) U^\dagger$, is such that $\operatorname{Tr}(\mathcal{E}(\rho))= 1$ for all density operators, therefore it is a trace-preserving transformation and takes states to states. Meanwhile, we will not demand that every operation obeys that property. For example, suppose a measurement is made and result $m$ is obtained. Then, we will define the transformation described by the measurement operator $M_m$ as $\mathcal{E}_m(\cdot) = M_m(\cdot)M_m^{\dagger}$. This transformation is not normalized and hence not guaranteed to return density operators for
input states. But, if we defined the operation to return the normalized state, it would not be defined for all possible states, namely the ones for which result m cannot be obtained inducing a division by 0. Besides, we know from the definition of these measurement operators that the trace $\operatorname{Tr}(M_{m}^{\dagger} M_m \rho)=\operatorname{Tr}(M_m \rho M_{m}^{\dagger})=\operatorname{Tr}(\mathcal{E}_m(\rho))$ is the probability of measuring value $m$ when the system is in the state $\rho$. Giving up the normalization turns that value, which is also the probability for the transformation to happen, into a property one can obtain from $\mathcal{E}_m$ itself.

Thus, to include measurements, the general form of an operation is a map that takes states $\rho$ to operators $\mathcal{E}(\rho)$ such that $\operatorname{Tr}(\mathcal{E}(\rho))$ is the probability for the operation to be applied. Of course, the final state has to be corrected by a normalization factor and will generally be given by
\begin{equation}\label{finalstate}
\rho' = \frac{\mathcal{E}(\rho)}{\operatorname{Tr}(\mathcal{E}(\rho))}.
\end{equation}
Under that interpretation, an operation is only `physical' if $\operatorname{Tr}(\mathcal{E}(\rho))\leq 1$, the maximum probability, when applied to density operators. This is the same as demanding the operation to be trace-nonincreasing on the entire domain. For trace-preserving transformations, we now interpret that $\mathcal{E}$ is applied with certainty, $\operatorname{Tr}(\mathcal{E}(\rho))=1$, and the expression above reduces to $\rho'=\mathcal{E}(\rho)$.

\textbf{A2.} The second axiom comes from the physical or statistical requirement that, if we initially have an ensemble, that is, a system that has probability $p_i$ of being in the quantum state $\rho_i$ for each $i$ in a finite set, then the final state after an operation should consist of a corresponding ensemble of final states. In quantum mechanics, the initial ensemble is represented by a mixed state of the form $\rho= \sum_i p_i \rho_i$. Thus, if the operation $\mathcal{E}$ is applied on $\rho$ with probability $p(\mathcal{E})=\operatorname{Tr}(\mathcal{E}(\rho))$, we ask that the final state is given by a mixture of the final states $\rho_i'$, each one related to $\rho_i$ as indicated in equation~(\ref{finalstate}), with probabilities $p(i|\mathcal{E})$. The latter is the probability that the initial state is $\rho_i$ given that operation $\mathcal{E}$ occurs. Mathematically, this requirement translates to
\begin{equation}\label{2axstate}
\rho'= \sum_i p(i|\mathcal{E})  \frac{\mathcal{E}(\rho_i)}{\operatorname{Tr}\left(\mathcal{E}(\rho_i)\right)}.
\end{equation}
Bayes rule for conditional probabilities asserts that
$$p(i|\mathcal{E})= p(\mathcal{E}|i)\frac{p_i}{p(\mathcal{E})}=\operatorname{Tr}(\mathcal{E}(\rho_i))\frac{p_i}{\operatorname{Tr}(\mathcal{E}(\rho))}.$$

Substituting this in equation~(\ref{2axstate}), comparing equations~(\ref{finalstate}) and~(\ref{2axstate}), we have
\begin{equation}
\mathcal{E}(\rho)=\mathcal{E}\left(\sum_i p_i \rho_i\right)=\sum_i p_i\mathcal{E}(\rho_i).
\end{equation}
This is a linearity condition for the case where the coefficients are real and sum up to 1 (convex-linearity). It does not arise specifically from the linear nature of QT, but from a general property of stochastic theories asserting that statistical experiments should behave as explained, in accordance with linearity of mixing~\cite{MilzPollock}.

Although we are mostly interested on how these operations act on quantum states, we have defined them over a bigger domain and going to a codomain also bigger than the state space. This is done because $\mathcal{Q}(\mathcal{H})$ is not a vector space and we need the usual notion of tensor products of operations over vector spaces to make sense of an operation acting on a composite system, as used in the axiom A3. Regardless, this does not create liberty to arbitrarily choose how an operation $\mathcal{E}$ acts outside of the set of states because, if a transformation with domain restricted to $\mathcal{Q}(\mathcal{H})$ obeys the axioms of an operation (with convex-linearity replacing A2), then it admits a unique linear extension to $\mathcal{T}(\mathcal{H})$, as demonstrated by Kraus in one of the main references about this subject~\cite{Kraus}. Therefore, convex-linearity on states and the possibility of defining tensor products of operations with no loss of generality justify the requirement for linearity on the domain $\mathcal{T}(\mathcal{H}_I)$.

\textbf{A3.} The third axiom is a requirement of consistency. Valid states have to be mapped to other valid states up to normalization, which is to say that positive operators in $\mathcal{T}(\mathcal{H}_I)$ are mapped to positive operators in $\mathcal{T}(\mathcal{H}_O).$ We also have to account for consistency of the quantum operation applied on a subsystem. Let us say that the input system, $\mathcal{H}_I$, is part of a bigger system described by the tensor product between $\mathcal{H}_I$ and a Hilbert space $R$ representing the rest of the system. If $\rho_{R,\mathcal{H}_I} \in \mathcal{Q}(R\otimes \mathcal{H}_I)$ is a valid density operator and the operation $\mathcal{E}$ acts on the degrees of freedom of $\mathcal{H}_I$ only, then the operator given by $(\mathcal{I}_{\mathcal{L}(R)}\otimes \mathcal{E})(\rho_{R \mathcal{H}_I})$, where $\mathcal{I}_{\mathcal{L}(R)}$ is the identity, should also be a valid density operator for the composite system up to normalization. Thus, we ask that $(\mathcal{I}_R\otimes \mathcal{E})$ maps positive operators to positive operators as well, characterizing complete positivity.

From now on, we will work with finite dimensional spaces unless stated otherwise. In that case, a quantum operation is a completely positive, trace-nonincreasing map belonging to $\mathcal{L}(\mathcal{L}(\mathcal{H}_I), \mathcal{L}(\mathcal{H}_O))$. The particular operations which are trace-preserving are called \textbf{quantum channels}.

Note that the set of axioms is fairly general. We are not asking for quantum operations to be always defined by unitary operators or measurements on a system. We are introducing a set of transformations that must consistently contain these ones, because we are already familiar with them in quantum mechanics, but could potentially contain much more. Operations satisfy minimum requirements to not be unphysical, but we have all the reasons to inquire whether all, or just some, of them are indeed realizable. We address this in the next section.
 
 It is good to emphasize that a quantum operation is determined unambiguously by how it acts on actual quantum states. To illustrate this, we can think of a two-level system~\cite{MilzPollock} for which we want to define operation maps. Then, instead of choosing Pauli matrices or other orthonormal basis for the 4-dimensional vector space $\mathcal{L}(\mathcal{H})\simeq \mathbb{C}^4$, we choose a basis whose elements are density operators, such as
$$\rho_1=\frac{1}{2}\begin{bmatrix}
1 & 1 \\
1 & 1 
\end{bmatrix} \quad 
\rho_2=\frac{1}{2}\begin{bmatrix}
1 & -i \\
i & 1 
\end{bmatrix}  \quad 
\rho_3=\begin{bmatrix}
1 & 0 \\
0 & 0 
\end{bmatrix} \quad 
\rho_4=\frac{1}{2}\begin{bmatrix}
1 & -1 \\
-1 & 1 
\end{bmatrix}.  $$ If we consider a ``ket'' representation of the space of operators $\mathcal{L}(\mathcal{H})$, an arbitrary operator is written in terms of the basis above as $A= \sum_i a_i \ket{\rho_i}$. Although the basis is not orthonormal, it is possible to find matrices $D_i$,
$$D_1=\frac{1}{2}\begin{bmatrix}
0 & 1+i \\
1-i & 2 
\end{bmatrix} \quad 
D_2=\begin{bmatrix}
0 & -i \\
i & 0 
\end{bmatrix}  \quad 
D_3=\begin{bmatrix}
1 & 0 \\
0 & -1 
\end{bmatrix} \quad 
D_4=\frac{1}{2}\begin{bmatrix}
0 & -1+i \\
-1-i & 2 
\end{bmatrix},
$$ such that $\braket{D_i|\rho_i}\equiv \operatorname{Tr}(D_i^{\dagger}\rho_j)=\delta_{ij}$. They are called dual matrices and assume the role of a ``bra'' for $\ket{\rho_i}$ in this context. The coefficients of the decomposition are $a_i:=\braket{D_i|A}= \operatorname{Tr}(D_i^\dagger A)$. Then, we can write the action of an operation $\mathcal{E}$ like
\begin{equation}\label{tomographicRep}
    \mathcal{E}(A)=\sum_i \operatorname{Tr}(D_i^\dagger A) \mathcal{E}(\rho_i).
\end{equation}
 This way to write the action of $\mathcal{E}$ using a basis of states is called tomographic representation. The name comes from the idea of quantum process tomography,
which is the method that uses this to reconstruct the transformation realized in a protocol by gathering experimental data from just a few input quantum states (the basis).

\subsection{Kraus decomposition and open system dynamics}

There are some ways to express the specific form of a quantum operation, one of them being the tomographic representation above, a type of linear decomposition. The Kraus form, a special case of operator-sum decomposition for linear maps, is another representation massively used in quantum information to model system-environment interactions~\cite{Nielsen,MilzPollock}. Before we get to it, consider the following.
\begin{lemma}[Choi-Jamiolkowski isomorphism]\label{ChoiJ}
Given two Hilbert spaces $\mathcal{H}_I$ and $R$ of same dimension and a choice of basis for each, $\{\ket{i}_{I}\}, \{\ket{j}_{R}\}$, define the vector
$\ket{\alpha}:= \sum_k \ket{k}_R\ket{k}_I \in R\otimes\mathcal{H}_I.$ We can think of it as the maximally entangled state, up to normalization, between a system associated to $\mathcal{H}_I$ and another system $R$. Given $\mathcal{E}\in \mathcal{L}(\mathcal{L}(\mathcal{H}_I),\mathcal{L}(\mathcal{H}_O))$, the operator $\sigma_{\mathcal{E}}\in \mathcal{L}(R\otimes \mathcal{H}_O)$, with action defined by
\begin{equation}
    \sigma_{\mathcal{E}} := \left(\mathcal{I}_{\mathcal{L}(R)} \otimes \mathcal{E}\right)(\ket{\alpha}\bra{\alpha}),
\end{equation} will be called the \textbf{Choi operator} or matrix of $\mathcal{E}$. When $R=\mathcal{H}_I$, the correspondence $\mathcal{E}\mapsto \sigma_\mathcal{E}$ is an isomorphism between the spaces $\mathcal{L}(\mathcal{L}(\mathcal{H}_I),\mathcal{L}(\mathcal{H}_O))$ and $\mathcal{L}(\mathcal{H}_I\otimes \mathcal{H}_O).$ The action of $\mathcal{E}$ is then completely characterized by $\sigma_{\mathcal{E}}$ through the inverse morphism:
\begin{equation}
    \mathcal{E}(A) = \operatorname{Tr}_{\mathcal{H}_I}[(A^T \otimes \mathds{1}_{\mathcal{H}_O})\sigma_{\mathcal{E}}]\quad \forall A\in \mathcal{L}(\mathcal{H}_I).
\end{equation}
\end{lemma}
We can check the last expression by using basis a $\{\ket{i}\}_{\mathcal{H}_I}$ to write $\sigma_{\mathcal{E}}=\sum_{ij}\ket{i}\bra{j}\otimes\mathcal{E}(\ket{i}\bra{j})$. Then, writing the action of the operation in a basis element of $\mathcal{L}(\mathcal{H}_I)$ as $\mathcal{E}(\ket{i}\bra{j})=\sum_{mn} e_{ijmn}\ket{m}\bra{n}$, we have
\begin{align}    \operatorname{Tr}_{\mathcal{H}_I}\left[\left( A^T\otimes\mathds{1}_{\mathcal{H}_O}\right)\sigma_{\mathcal{E}}\right] &= \operatorname{Tr}_{\mathcal{H}_I}\left[\left(A^T\otimes\mathds{1}_{\mathcal{H}_O}\right)\left(\sum_{ijmn}\ket{i}\bra{j}\otimes e_{ijmn}\ket{m}\bra{n} \right)\right] \nonumber
    \\
&=\sum_{ijmn}e_{ijmn}\operatorname{Tr}_{\mathcal{H}_I}\bigl[\left(A^T\otimes\mathds{1}_{\mathcal{H}_O}\right)\left(\ket{i}\bra{j}\otimes \ket{m}\bra{n} \right)\bigr] \nonumber
\\
&=\sum_{ijmn}e_{ijnm}\bra{j}A^T\ket{i} \ket{m}\bra{n} \nonumber
\\
&=\sum_{ij} A_{ij} \mathcal{E}(\ket{i}\bra{j})=\mathcal{E}(A). 
\end{align}

This isomorphism has a crucial role in the study of quantum evolutions. Essentially, it allows us to treat operations in $\mathcal{L}(\mathcal{L}(\mathcal{H}_I),\mathcal{L}(\mathcal{H}_O))$ as matrices in $\mathcal{L}(\mathcal{H}_I\otimes\mathcal{H}_O)$ at the same level of quantum states. Since it is an useful tool in linear algebra for manipulating general linear supermaps, we shall see it again later. Here, we will use it to prove that every quantum operation has a Kraus form, which is a decomposition in the shape $\mathcal{E}(\cdot)=\sum_iE_i(\cdot)E_i^\dagger$ for certain operators $E_i\in\mathcal{L}(\mathcal{H}_I,\mathcal{H}_O)$. Interestingly, a map can be written like that iff it is a quantum operation, as we state properly below for finite dimensions.

\begin{thm}[Choi-Kraus theorem]
 Let $\mathcal{H_I}$ and $\mathcal{H}_O$ be Hilbert spaces of dimension $m$ and $n$ respectively, and $\mathcal{E}:\mathcal{L}(\mathcal{H_I}) \to \mathcal{L}(\mathcal{H_O})$ be a quantum operation. Then, there exists a family of linear operators in $\mathcal{L}(\mathcal{H}_I,\mathcal{H}_O)$, $\left\{E_{i} \right\}_{1 \leq i \leq n m}$, that satisfies $\sum_{i} E_{i}^{\dagger} E_{i}\leq\mathds{1}_{\mathcal{H}_I}$ and such that $\mathcal{E}$ is described by the \textbf{Kraus form}:
$$\mathcal{E}(A)=\sum_{i} E_{i} A E_{i}^{\dagger}
$$
Conversely, any map $\mathcal{E}$ whose action is defined in this form for some family of $n\times m$ linear operators with $\sum_{i} E_{i}^{\dagger} E_{i}\leq\mathds{1}_{\mathcal{H}_I}$ is a quantum operation. The operators $E_i$ are called \textbf{Kraus operators} of $\mathcal{E}$~\cite{Kraus}.
\end{thm}
An inequality $A\leq B$ of operators in $\mathcal{L}(\mathcal{H})$ means that the operator $B-A$ is positive: $\braket{\psi|(B-A)|\psi}\geq 0$ for every $\ket{\psi} \in \mathcal{H}$. The theorem above is an important characterization because it allows more direct manipulation of operations and provides us with interpretations on how to achieve them. Let us prove it.
\begin{proof}($\impliedby$) If the action of an operator $\mathcal{E}:\mathcal{L}(\mathcal{H}_I)\to\mathcal{L}(\mathcal{H}_O)$ can be written as $\mathcal{E}(A)=\sum_{i} E_{i} A E_{i}^{\dagger}$ for Kraus operators $E_i$, then $\mathcal{E}$ is a quantum operation.

If each $E_i$ is linear, then $\mathcal{E}$ is automatically \emph{linear}. One can see this by writing $\mathcal{E}(c A +B)$ for arbitrary $A$ and $B$ in the domain and complex constant $c$. Using operators $E_i$ explicitly and  representing $A$ and $B$ by their expansions in a common basis for maps in $\mathcal{L}(\mathcal{H}_I)$, which will be of the form  $\{\ket{k}\bra{j}\}_{k,j\leq m}$, the statement follows from the linearity of the operators $E_i$. From the assumption that $\sum_{i} E_{i}^{\dagger} E_{i}\leq\mathds{1}_{\mathcal{H}_I}$ we have for any state $\rho$ that $$\operatorname{Tr}(\mathcal{E}(\rho))= \operatorname{Tr}\left(\sum_{i} E_{i} \rho E_{i}^{\dagger}\right) = \sum_{i}\operatorname{Tr}(E_{i}^{\dagger}E_{i} \rho )= \operatorname{Tr}\left(\sum_{i}E_{i}^{\dagger}E_{i} \rho \right)\leq \operatorname{Tr}(\mathds{1} \rho )=1,$$ guaranteeing that $\mathcal{E}$ is \emph{trace-nonincreasing} and the probability for it to be applied is well-defined. The step where the inequality appears can be verified by writing the state in an orthogonal basis in which it is diagonal, $\rho=\rho_{kk}\ket{e_k}\bra{e_k}$, and using $\bra{e_k}\sum_iE_i^{\dagger}E_i\ket{e_k}\leq 1$, a direct consequence of the operator inequality. We are only left to check complete positivity:

Consider a system composed of the input space and an extra space R. Let $A$ be a positive map in the space $\mathcal{L}( R\otimes \mathcal{H}_I)$ of linear operators over the composite Hilbert space. For a fixed $\ket{\psi} \in R\otimes \mathcal{H}_O$, define the vectors
$$\ket{\phi_i}:=(\mathds{1}_R\otimes E_i^{\dagger})\ket{\psi}$$
for each $E_i$. Suppose that the dimensions of the input and output spaces are equal. Then we can write
\begin{align}
\bra{\psi}(\mathcal{I}_R\otimes \mathcal{E})(A)\ket{\psi}&=\bra{\psi}\left(\mathcal{I}_R\otimes \sum_i E_i(\cdot)E_i^{\dagger}\right)(A)\ket{\psi}\nonumber
\\
&=\sum_i\bra{\psi}(\mathds{1}_R\otimes E_i)A(\mathds{1}_R\otimes E_i^{\dagger})\ket{\psi}\nonumber
\\
&=\sum_i \bra{\phi_i}A\ket{\phi_i}\geq 0,
\end{align}
where the last inequality is valid because each term of the sum is non-negative thanks to the positivity of $A$. We have proved that $\bra{\psi}(\mathcal{I}_R\otimes \mathcal{E})(A)\ket{\psi}\geq0$ for an arbitrary $\ket{\psi}\in R\otimes \mathcal{H}_O$, therefore the operator $(\mathcal{I}_R\otimes \mathcal{E})(A)$ is positive. Hence, $(\mathcal{I}_R\otimes \mathcal{E})$ maps positive operators to positive operators and $\mathcal{E}$ is \emph{completely positive}. If the input and output have different sizes, the proof is virtually the same, but we have to extend the smaller space to the size of the higher dimensional one, redefining the domain of the operators to cover the entire space acting trivially outside the original domain.

Since $\mathcal{E}$ satisfies the 3 properties, we have shown it is a quantum operation.

($\implies$) Conversely, let us prove that if $\mathcal{E}:\mathcal{L}(\mathcal{H}_I)\to\mathcal{L}(\mathcal{H}_O)$ is a quantum operation, then there exists a family of Kraus operators $E_i$ for it.

Let $\{\ket{k}_I\}_{k\leq m}$ be a basis for $\mathcal{H_I}$ and let us consider another Hilbert space $R$ such that dim($R$) = dim($\mathcal{H}_I$) and introduce a basis $\{\ket{k}_R\}_{k\leq m}$ for it. This auxiliary space will help us find the decomposition. Then, for each arbitrary $\ket{\psi} = \sum_k \psi_k\ket{k}_I \in \mathcal{H}_I$ define a corresponding vector $\ket{\Tilde{\psi}} \in R$ as\[\ket{\Tilde{\psi}}:= \sum_k \psi_k^{*} \ket{k}_R.\]
Using these objects, we can rewrite the general action of our operation $\mathcal{E}$:
\begin{align}
\mathcal{E}\left(\ket{\psi}\bra{\phi}\right)&=\sum\limits_{k,k'} \psi_k \phi^{*}_{k'}  \mathcal{E}\left(\ket{k}\bra{k'}_I \right) \nonumber
    \\
    &=\sum_{k, k'}\delta_{n k} \left[\sum\limits_{n,m} \psi_n \phi^{*}_m  \mathcal{E}\left(\ket{n}\bra{m}_I \right)\right] \delta_{m k'} \nonumber
    \\
       &=\left[\sum_k \psi_k \bra{k}_R\right]  \left[\sum\limits_{n,m} \ket{n}\bra{m}_R \otimes \mathcal{E}\left(\ket{n}\bra{m}_I \right)\right] \left[\sum_{k'} \phi^{*}_{k'} \ket{k'}_R\right] \nonumber
       \\
&=\bra{\Tilde{\psi}}\left[\mathcal{I}_{\mathcal{L}(R)} \otimes \mathcal{E} \left(\ket{\alpha}\bra{\alpha}\right)\right] \ket{\Tilde{\phi}}\nonumber
    \\
    &= \bra{\Tilde{\psi}} \sigma_{\mathcal{E}} \ket{\Tilde{\phi}}. \label{actionofE}
    \end{align}
Remember that $\sigma_{\mathcal{E}}$ is an element of $\mathcal{L}(R \otimes \mathcal{H}_O)$ while the vectors $\ket{\Tilde{\psi}},\ket{\Tilde{\phi}}$ belong to R. Thus\footnote{Technically, the symbol $\sigma_{\mathcal{E}} \ket{\Tilde{\psi}}$ represents the operator $\sigma_{\mathcal{E}}(\ket{\Tilde{\psi}}\otimes\cdot):\mathcal{H}_O \to R\otimes \mathcal{H}_O$. The bra $\bra{\Tilde{\psi}}$ acting on the space $R\otimes \mathcal{H}_I$ means $\braket{\Tilde{\psi},\cdot} \otimes (\cdot): R\otimes \mathcal{H}_O \xrightarrow[]{\bra{\Tilde{\psi}}\otimes \mathds{1}_I}\mathds{C}\otimes \mathcal{H}_O \xrightarrow{c \otimes \ket{\psi} \mapsto c \ket{\psi}} \mathcal{H}_O$.}, $\bra{\Tilde{\psi}}\sigma_{\mathcal{E}} \ket{\Tilde{\phi}}$ is not a number, but a map going from $\mathcal{H}_O$ to itself, as expected since it is in the image of $\mathcal{E}.$

As a consequence of $\mathcal{E}$ being CP, $\sigma_{\mathcal{E}}$ is positive and it is possible to decompose it with some family of vectors {$\ket{s_i}\in R\otimes \mathcal{H}_O$}, 
\begin{equation}\label{fixdecomp}
    \sigma_{\mathcal{E}} = \sum_i \ket{s_i}\bra{s_i}.
\end{equation} 
So, let us define the family of maps $E_i: \mathcal{H}_I \to \mathcal{H}_O$ with action given by
$$E_i \ket{\psi} := \braket{\Tilde{\psi} | s_i} \quad \in \mathcal{H}_O .$$ They are linear because $c \ket{\psi}+\ket{\phi} \mapsto  \braket{c^{*}
\Tilde{\psi}+\Tilde{\phi} |s_i} =  c \braket{\Tilde{\psi}|s_i}+\braket{\Tilde{\phi}|s_i}$. We can also verify that, according to \eqref{actionofE},
$$\sum_i E_i \ket{i}\bra{k} E_i^{\dagger} = \sum_i \braket{\Tilde{i}|s_i}\braket{s_i|\Tilde{k}}=\bra{\Tilde{i}} \sigma_{\mathcal{E}} \ket{\Tilde{k}}= \mathcal{E}(\ket{i}\bra{k}).$$
Since $\mathcal{E}$ is linear, its action on an arbitrary operator $A=\sum_{jk}A_{jk}\ket{j}\bra{k}\in \mathcal{L}(\mathcal{H}_I)$ is
\begin{equation}
    \mathcal{E}(A)=\sum_{jk}A_{jk}\mathcal{E}\left(\ket{j}\bra{k}\right)=\sum_{ijk}A_{jk}E_i\ket{j}\bra{k}E_i^\dagger=\sum_i E_i A E_i^{\dagger}.
\end{equation}
We have proved that the operators characterize the action of $\mathcal{E}$. The assumption that $\mathcal{E}$ does not increase trace means that, for $\ket{\psi}\in \mathcal{H}_I$, $\operatorname{Tr}\left(\mathcal{E}(\ket{\psi}\bra{\psi})\right)\leq\operatorname{Tr}(\ket{\psi}\bra{\psi})=\braket{\psi|\psi}$. Substituting the Kraus operators in that expression, we have \begin{equation}
    \sum_i\operatorname{Tr}\left(E_i^\dagger E_i\ket{\psi}\bra{\psi}\right)\leq \braket{\psi|\psi}\implies\bra{\psi}[\mathds{1}_{\mathcal{H}_I}-\sum_iE_i^\dagger E_i]\ket{\psi}\geq0,
\end{equation}
valid for all $\ket{\psi}\in\mathcal{H}_I$. Therefore, $\sum_i E_i^{\dagger}E_i \leq \mathds{1}_{\mathcal{H}_I} $, completing the proof.

The result is also guaranteed to hold in the infinite dimensional case (separable Hilbert spaces) by the Stinespring factorization theorem, with trace-class operator spaces and a  
sequence $\left\{E_{i}\right\}$ of bounded linear operators, with the sum in the Kraus representation turned into a convergent series~\cite{Kraus,Attal}.

\end{proof}

  We constructed the Kraus operators with a generic vector decomposition for $\sigma_{\mathcal{E}}$, indicating that the Kraus form is not unique. This is in agreement with the idea that a system can undergo the same transformation in different ways. In fact, there is a unitary freedom in the choice of Kraus operators because of different decompositions one might write for $\sigma_{\mathcal{E}}$~\cite{Nielsen}. There is always a factorization with the minimum possible number of Kraus operators for $\mathcal{E}$ which are mutually orthogonal, $\operatorname{Tr}(E_i E_{i'}^{\dagger})=\delta_{i i'}$. Such canonical form can be achieved by writing the diagonalized matrix $\sigma _\mathcal{E}=\sum_i\lambda_i\ket{i}\bra{i} $ and choosing the vectors in~(\ref{fixdecomp}) to be $\ket{s_i}:=\sqrt{\lambda_i}\ket{i}$~\cite{Landi,MilzPollock}. The number of Kraus operators then coincides with the rank of the Choi matrix and is no bigger than the multiplication of the dimensions $d_{\mathcal{H}_I}\times d_{\mathcal{H}_O}$.

The fact that quantum operations are the maps which have a Kraus decomposition helps us understand them better. In special, it is possible to show that there always exists a model involving evolutions of a quantum system + environment which realizes any given quantum operation. The formal mathematical result is known as Stinespring dilation~\cite{Stinespring_1955}. The idea behind it is to use the Kraus decomposition of $\mathcal{E}$ to explicitly construct a sequence of unitary and measurement operators on an extended Hilbert space which reproduces its application. The simplest case is for a trace-preserving operation. Suppose we have 
\begin{equation}
    \mathcal{E}=\sum_{i=1}^{d}E_i (\cdot) E_i^{\dagger}:\mathcal{L}(\mathcal{H}_I)\to \mathcal{L}(\mathcal{H}_O), \quad \quad \sum_{i=1}^{d} E_i^{\dagger}E_i = \mathds{1}_{\mathcal{H}_I}.
\end{equation}
Let us define the environment as a Hilbert space $E$ of dimension $d$, the same as the number of Kraus operators, or Kraus rank, and let $\ket{0}$ be a normalized fixed vector of $E$ representing the initial state of the environment. Although we could choose the initial state to be mixed, it turns out we will not need that. We could always arrive at the same operation by purifying the environment. Let $\{\ket{i}\}_{1\leq i\leq d}$ be a basis for $E$. Then, we can define a map $U$ that acts on vectors of the form $\ket{\psi}\ket{0} \in \mathcal{H}_I\otimes E$ like
$$U\left(\ket{\psi}\ket{0}\right)=\sum_{i=1}^{d}E_i \ket{\psi} \ket{i}. $$ 
Now, since
$$\bra{\psi}\bra{0}U^{\dagger}U\ket{\phi}\ket{0}=\sum_{i=1}^{d}\bra{\psi} E_i^{\dagger}E_i\ket{\phi}=\braket{\psi|\phi},$$ this map preserves the inner product when acting on the subspace $\mathcal{H}_I\otimes \ket{0}$, mapping it to a vector space in the codomain. As a consequence, there exists~\cite{Nielsen} a unitary extension of it on the entire domain $\mathcal{H}_I\otimes E$, which will also be called $U$. We constructed this definition for $U$ so that we can write $\bra{i}U\ket{0}$\footnote{This is an operator because the product is only with respect to $E$, similarly to a partial trace.}$=E_i$. Then,
\begin{align}
    \mathcal{E}(\rho) = \sum_{i=1}^d E_i\rho E_i^\dagger&= \sum_{i=1}^{d} \bra{i}U\ket{0}\rho \bra{0}U^{\dagger}\ket{i} \nonumber
    \\
&=\operatorname{Tr}_E\left[\sum_{mn}\rho_{mn}U\ket{m}\ket{0}\bra{n}\bra{0}U^\dagger\right] \nonumber
    \\
    &=\operatorname{Tr}_E \left[ U\left(\rho\otimes \ket{0}\bra{0}\right)U^{\dagger}\right],
\end{align}
where $\operatorname{Tr}_E$ is the partial trace with respect to space $E$. Therefore, the application of any trace-preserving $\mathcal{E}$ can be interpreted as if the system $\mathcal{H}_I$ evolved unitarily together with an environment $E$ whose initial state was $\ket{0}$. The environment was subsequently ignored or ``traced-out'' after the joint evolution, leaving us with the correct output state.

For the trace decreasing case, we can add one extra Kraus operator $E_{d+1}:=\sqrt{\mathds{1}-\sum_{i=1}^d E_i^\dagger E_i}$  so that we have $\sum_{i=1}^{d+1} E_i^{\dagger}E_i=\mathds{1}$ again. Then, we can follow the steps above leading to a unitary operation with an augmented environment E'. At the end, we have to project the unitarily evolved system back to the space $\mathcal{H}_O\otimes E$ and only after that trace-out the environment E. The operation has the form
$$\mathcal{E}(\rho) = \sum_{i=1}^{d}E_i \rho E_i^{\dagger}=\operatorname{Tr}_E \left[P_{\mathcal{H}_OE} U_{\mathcal{H}_IE'}\left(\rho\otimes \ket{0}\bra{0}\right)U_{\mathcal{H}_IE'}^{\dagger}P_{\mathcal{H}_OE}\right],$$ which can be visualized in Fig.~\ref{fig:QuantumOp}. Thus, unitary evolutions and projections on composite systems cover all quantum operations. Conversely, any evolution of this form has a Kraus decomposition, as we will discuss shortly in an example.
  \begin{figure}
    \centering
    \includegraphics{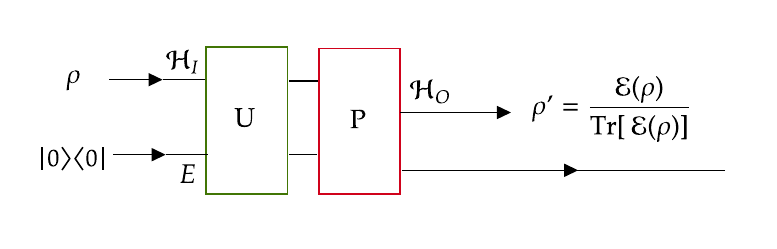}
    \caption{Schematics of Stinespring dilation. Any map obeying the axioms of quantum operations can be rewriten in the form of a unitary evolution $U$ followed by a projective measurement $P$ on the system of interest together with an environment. Part of the system can then be discarded, leaving a system in state $\rho'$ in the output Hilbert space $\mathcal{H}_O$. Therefore, quantum operations are always described by open quantum system dynamics.}
    \label{fig:QuantumOp}
\end{figure}
  
  In quantum as much as classical mechanics, we postulate how closed systems evolve and expect that any system can be seen as part of a closed system, ultimately the entire universe. This assumption is not made for quantum operations, however, and they end up satisfying it independently. Any quantum operation can be reduced to a simple evolution of an open quantum system, and the converse is true. Then, quantum operations are a general description of quantum dynamics, as good as systems+environment models or master equations~\cite{LIDAR200135,Nielsen}. This reinforces that
  fundamental aspects of evolutions in QT are abstractly characterized by the mathematical properties of these maps. We draw special attention to the fact that they can always be depicted as quantum circuits, as in Fig.~\ref{fig:QuantumOp}.

\subsection{Elementary examples of quantum operations}

With the framework settled, we can talk about some of the simplest examples of quantum operations. We have mentioned unitary evolutions and projective measurements because they are known possible operations coming from modeling open systems. We can also define operations by simply stating their action on a basis, since the map has to be linear, or by providing Kraus operators that represent it. Those three approaches are equivalent ways of defining specific quantum operations and the form chosen is a matter of convenience.

\noindent a) The $\mathcal{H}_I\simeq\mathbb{C}$ case

Consider an operation (a linear CP trace non-increasing map) from $\mathcal{L}(\mathcal{\mathbb{C}})$ to $\mathcal{L}(\mathcal{H}_O)$ for some Hilbert space $\mathcal{H}_O$. How can we interpret an operation of this type?

Each functional $f \in \mathcal{L}(\mathbb{C})$ is completely defined by the action of $f$ on the number $1\in \mathbb{C}$, because, by linearity, $f(a+bi)=(a+bi)f(1)$ for every complex number $a+bi$. The action of $f$ may be written as $f=f(1) \mathds{1}$, where $\mathds{1}:\mathbb{C}\to\mathbb{C}$ is the identity. Analogously, an operation $\mathcal{E}:\mathcal{L}(\mathbb{C})\to \mathcal{L}(\mathcal{H}_O)$ is characterized by its action on the identity $\mathds{1}$ because, by linearity again, $\mathcal{E}(f)=\mathcal{E}(f(1) \mathds{1})=f(1)\mathcal{E}(\mathds{1}).$ The information encoded in this map is just one state $\mathcal{E}(\mathds{1})/\operatorname{Tr}(\mathcal{E}(\mathds{1}))\equiv\rho'\in \mathcal{H}_O$. So one can see this map as a representation of the state $\rho'$ itself, or its preparation. To find a Kraus decomposition, we observe that
$$\mathcal{E}(f)=f(1)\mathcal{E}(\mathds{1})\text{ and  }\mathcal{E}(f)=\sum_k E_k f(1) \mathds{1} E_k^{\dagger}.$$ If the map represents a density operator $\rho'=\mathcal{E}(\mathds{1})/{\operatorname{Tr}(\mathcal{E}(\mathds{1}))}$, and $\rho' =\sum_i p_i \ket{\psi_i}\bra{\psi_i}$, 
$$\mathcal{E}(\mathds{1})=\sum_k E_k \mathds{1} E_k^{\dagger}=\rho'\operatorname{Tr}(\mathcal{E}(\mathds{1})) \implies E_k = \sqrt{\operatorname{Tr}(\mathcal{E}(\mathds{1})) p_k}\ket{\psi_k},$$
where the ket $\ket{\psi_k}$ represents the map $\ket{\psi_k}: \mathbb{C}\to\mathcal{H}_O$ taking every complex number to correspondent multiples of $\ket{\psi_k}.$ As expected, the Kraus operators are maps from the input to the output Hilbert spaces. The adjoint of these maps are the bras $\bra{\psi_k}$.

\noindent b) The $\mathcal{H}_O \simeq \mathbb{C}$ case and the trace 

A general operation from $\mathcal{L}(\mathcal{H}_I)$ to $\mathcal{L}(\mathcal{H}_O)$ has Kraus operators going from $\mathcal{H}_I$ to $\mathcal{H}_O$. So, if $\mathcal{H}_O\simeq \mathbb{C}$, each Kraus operator $E_k$ is a linear functional on $\mathcal{H}_I$. Since they are elements of the dual space, each $E_k$ can be represented by a bra $\bra{v_k}$ because of Riesz representation theorem. Therefore, the operation acts like
$$\mathcal{E}(\rho)=\sum_k E_k\rho E_k^{\dagger}=\sum_k \bra{v_k} \rho \ket{v_k}.$$
If we define an operator $P:=\sum_k\ket{v_k}\bra{v_k}$, then the action can be rewritten as $\operatorname{Tr}(P\rho)$. Operations of this kind can be interpreted as quantum effects, that is, operations describing the change in a yes-no measurement apparatus when it interacts with a system. If $P$ is a projection describing the evolution when outcome $m$ is obtained, the corresponding operation takes the state as input and returns the probability for the apparatus to indicate the value $m$, as opposed to ``not $m$''. State preparations and effects are used as primitive concepts in approaches to axiomatically reconstruct quantum mechanics from more basic assumptions on probabilities~\cite{Kraus,Ludwig1983}. In particular, if the vectors $\ket{v_k}$ form an orthonormal basis, this operation is simply the trace operation $\mathcal{E}(\rho)=\operatorname{Tr}(\rho).$

\noindent c)The partial trace 

The partial trace on states is also a quantum operation. One can see that by defining its Kraus operators. Let $S\otimes R$ be the Hilbert space of the composite system, with $R$ being the space we want to trace out. Define operators
$$E_k\left(\sum_j \lambda_j \ket{\psi_j}\ket{j}\right):=\lambda_k\ket{\psi_k},$$
for some basis $\{\ket{j}\}$ of $R$ and arbitrary vectors $\ket{\psi_j}$. Then, define the operation $\mathcal{E}:\mathcal{L}(S\otimes R)\to S$ acting like $\mathcal{E}(\rho)=\sum_k E_k \rho E_k^{\dagger}$. By definition, this is a quantum operation, since it is defined in terms of a Kraus decomposition. Moreover, one can verify
$$\mathcal{E}(\rho^S\otimes\rho^R)= \sum_k E_k\left(\sum_{ii'}\rho^S_{ii'}\ket{s_i}\bra{s_i'} \otimes \sum_{jj'} \rho^R_{jj'}\ket{j}\bra{j'}\right)E_k^{\dagger} $$
$$= \sum_k E_k\left(\sum_{i i' j j'}\rho^S_{ii'}\rho^R_{j j'}\ket{s_i}\ket{j}\bra{s_i'}\bra{j'}\right)E_k^{\dagger} = \sum_k \rho^R_{kk} \rho^S =\operatorname{Tr}_R(\rho^S\otimes\rho^R).$$

\noindent d) General evolution of open systems

We showed that if a map is a quantum operation it has a Kraus decomposition, and the converse. We also showed that for any given Kraus decomposition, there exists a system+environment model realizing the operation. The converse of this statement is simpler: if we define an operation by describing the evolution of a system+environment, then we can directly find a Kraus decomposition for it.

An operation describing unitary evolution given by operator $U$ followed by a projective measurement $P_m$ on a system and environment acts like
\begin{equation}\label{sys+environ}
    \operatorname{Tr}_E \left(P_mU(\rho \otimes \sigma) U^{\dagger}P_m\right),
\end{equation}where $\sigma$ is the initial state of the environment.  If $\sigma=\sum_n \sigma_n \ket{n}\bra{n}_E$, then we have
$$\mathcal{E}(\rho)= \sum_{k,n} \sigma_n \bra{k}P_m U(\rho \otimes \ket{n}\bra{n}_E) U^{\dagger}P_m \ket{k}= \sum_{k,n} E_{k n} \rho E_{k n}^{\dagger},$$
with \begin{equation}\label{KrausopOpen}
    E_{k n}:=\sqrt{\sigma_n}\bra{k}P_m U\ket{n}_E.
\end{equation} 
 From the expression above, one can obtain the special cases when either $P_m=\mathds{1}$ or $U=\mathds{1}$.

\noindent d)The CNOT gate

Let us discuss one explicit operation to clarify the representations in this overview. The CNOT, or controlled-NOT, gate is a computational transformation which takes 2 quantum bits as inputs and returns 2 quantum bits. Its action can be described as a conditional flipping on the second bit depending on the state of the first bit. That is, if the first input bit is the control $\ket{c}$ and the second is the target $\ket{t}$, then the map takes $\ket{00}\mapsto\ket{00}$, $\ket{01}\mapsto\ket{01}$, $\ket{10}\mapsto\ket{11}$ and $\ket{11}\mapsto\ket{10}$.
The associated quantum operation $\mathcal{E}_U:\mathcal{L}(\mathcal{H}^{\otimes 2})\to \mathcal{L}(\mathcal{H}^{\otimes 2})$ is given by $\mathcal{E}(\rho)=U \rho U^{\dagger}$, with U being the matrix corresponding to the CNOT gate:
   \begin{center}
    \includegraphics{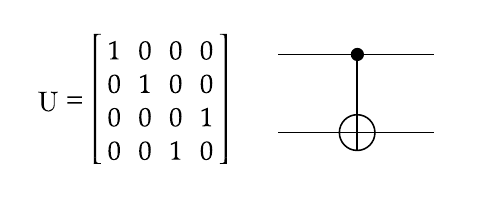}
    \end{center} with respect to the computational basis $\{\ket{0},\ket{1}\}^{\otimes 2}$. As we can see from the form of the matrix, U is unitary. As a result, the operation is trace-preserving and it is already written in the trivial Kraus decomposition $U \rho U^{\dagger}$.

We can also consider the situation in which the control $\ket{c}$ is measured in the computational basis. This new operation $\mathcal{E}_{UP}:\mathcal{L}(\mathcal{H}^{\otimes 2})\to \mathcal{L}(\mathcal{H}^{\otimes 2})$ acts like $\mathcal{E}_{UP}(\rho)=PU\rho U^\dagger P$ for the observable
\begin{equation}
    P:= \left(\ket{0}\bra{0}-\ket{1}\bra{1}\right)\otimes \mathds{1}_{\mathcal{H}}
\end{equation}
From calculating $PU$ we find the Kraus operators:  
$$E_1 =\ket{00}\bra{00},\quad E_2 =\ket{01}\bra{01},\quad E_3 =-\ket{10}\bra{11},\quad E_4 =-\ket{11}\bra{10},$$
and we can check that $\sum_k E_k^{\dagger}E_k= \mathds{1}_{\mathcal{H}}$. 

 The above was the situation in which 2 bits are taken as inputs, they pass through a CNOT gate and a measurement is made on the control right after that. If we further select only the cases for which the control returns the output -1 and then discard the control giving only the state of the target as output, this is an example of non-trace preserving operation.

In fact, the domain of the operation is still $\mathcal{L}(\mathcal{H}^{\otimes 2})$, with a 2-dimensional $\mathcal{H}$, but the codomain is now $\mathcal{L}(\mathcal{H})$, since only the target bit is returned. The description given is similar to what was done above. The unitary evolution is followed by the projection $P_1=-\ket{1}\bra{1}_c$, and then a partial trace with respect to the control space is performed. Therefore, the operation acts like $\mathcal{E}(\rho)=\on{Tr}_c\left( P_1 U \rho U^{\dagger}P_1\right)=\bra{1}_c U \rho U^{\dagger}\ket{1}_c$. There are now only two Kraus operators, which can be found by writing $\bra{1}_c U$ explicitly:
\begin{align}
    &\bra{1}_c U= \bra{1}_c \left[\ket{00}\bra{00}+ \ket{01}\bra{01}+ \ket{10}\bra{11}+  \ket{11}\bra{10}\right]=\ket{0}\bra{11} +\ket{1}\bra{10}
\end{align}
The Kraus operators fitting the expression above are $E_1=\ket{0}\bra{11}$ and $E_2=\ket{1}\bra{10}$. We can check that indeed
\begin{equation}
    \sum_k E_k^{\dagger}E_k = \ket{10}\bra{10}+\ket{11}\bra{11}<  \mathds{1}_{\mathcal{H}^{\otimes 2}},
\end{equation}
showing that this is a case of operation that does not preserve trace. The expression $\rho'= \mathcal{E}(\rho)/\operatorname{Tr}(\mathcal{E}(\rho))$ will then give us the normalized transformed state.

\subsection{The initial correlation problem and quantum supermaps}

Quantum operations are linear and, as a consequence, probing them with a few different states is sufficient for reconstructing their action. But, for that to make sense, we need to assume one can independently prepare the system $\mathcal{H}_I$ in different initial states. One example where this would not be the case is the evolution of a system and environment starting in a correlated state. Preparing the input system in different states would change the environment state modifying the operation itself. This leads us to the fact that operations can \emph{only} describe situations for which that independence exists. In fact, a map taking states of the system $\mathcal{H}_I$ as arguments is CP and linear if, and only if, the initial state in its system+environment models is a \emph{product state} $\rho_{\mathcal{H}_I}\otimes \eta_E$~\cite{Sudarshan}.

This is known as the initial correlation problem, and experimentalists began to see it in practice in the late 90s when trying to reproduce quantum gates in laboratory. The transformations with environment noise which should be described by quantum operations were not even CP maps~\cite{MilzPollock}. Considering non-complete positive maps to describe them represents an alternative, although an artificial one, since this property is closely related to a consistency requirement that probabilities remain well-defined. At the same time, nobody wants to drop linearity of the maps, which guarantees quantum process tomography can be done. Fortunately, there is a natural resolution to this problem: quantum supermaps.

The initial correlation problem has to do with not knowing how the input quantum state that suffers the operation was prepared. Starting with a possibly correlated system+environment state, the experimenter could choose to prepare a certain initial state $\rho$ for the system in different ways, which could even be conditioned on measurement results, and each of them would leave the environment in a distinct state. In fact, the act of preparation can involve unitary, projective evolutions and every possible way to achieve the desired state before it undergoes the transformation $\mathcal{E}$. And this is precisely what a quantum operation is. A preparation $\mathcal{A}\in\mathcal{L}(\mathcal{L}(\mathcal{H}_I))$ of an initial state for a system $\mathcal{H}_I$ will be a quantum operation whose form is defined by the procedures adopted by the experimenter. If we wish to write the final state of the system after preparation and evolution, we get something of the form
\begin{equation}\label{supermap}
\rho' = \operatorname{Tr}_E \left[P U \left(\mathcal{A}\otimes\mathcal{I}_{\mathcal{L}(E)}(\rho_{IE})\right)U^{\dagger}P\right]   
\end{equation}
which is the regular system-environment model defined in the last sections, only this time the initial state is further specified as the result of preparation $\mathcal{A}\otimes\mathcal{I}_{\mathcal{L}(E)}$ on some fixed $\rho_{IE}\in\mathcal{L}(\mathcal{H}_I\otimes\mathcal{H}_E)$ that need not be a product state. We can then define a map $\mathcal{M}: \mathcal{L}(\mathcal{L}(\mathcal{H}_I)) \to \mathcal{L}(\mathcal{H}_O),$
\begin{equation}\label{supermapaction}
    \mathcal{A} \mapsto \mathcal{M}(\mathcal{A}) = \rho'.
\end{equation}
This map is linear and completely positive on its domain, as we can expect due to the form~(\ref{supermap}). Those properties guarantee a valid output quantum state for every preparation. It is virtually the same kind of map as a quantum operation, but while quantum operations take states to states, a supermap takes preparations to states. With the Choi-Jamiolkowski isomorphism~(\ref{ChoiJ}), we can associate a preparation $\mathcal{A}$ to a matrix $\sigma_{\mathcal{A}}\in\mathcal{L}(\mathcal{H}_I\otimes \mathcal{H}_I)$. The corresponding map $M:\mathcal{L}(\mathcal{H}_I\otimes \mathcal{H}_I)\to \mathcal{L}(\mathcal{H}_O)$ taking Choi matrices to states is a linear CP trace non-increasing map, being therefore mathematically identical to a quantum operation.

The advantage of considering such maps is that now we include the case where the initial state of a composite system is not a product state. Instead of considering input states, we consider input operations which reveal how the environment is changed after preparation, and the initial correlation problem goes away. Furthermore, they are operationally adequate because they take as inputs the objects the experimentalist will actually control, preparation settings, and outputs the final states, which are the objects that will be measured. Because supermaps are linear and CP, we have again that they can be Kraus decomposed, i.e. there are maps $\mu_k$ such their action can be written as
\begin{equation}
    \mathcal{M}(\mathcal{A}) = \sum_k \mu_k \mathcal{A} \mu_k^{\dagger}.
\end{equation}
A quantum operation is a special case of quantum supermap where the state $\rho_{IE}$ is a product state $\rho_I\otimes\rho_E$. The prepared state is $\rho=\mathcal{A}(\rho_I)$ resulting in $\mathcal{M}(\mathcal{A})=\mathcal{E}(\mathcal{\rho})$. 
There are a lot of properties of quantum operations and further discussions about supermaps, but a comprehensive study is not the goal of this chapter. This is intended to give a general picture of how evolutions are treated in Quantum Theory operationally and introduce mathematical objects we will use in the following chapter. 

\section{Wrapping up: where is causality after all?}\label{causalityCombs}
So, this was the overview and we promised causality was hidden somewhere around here. In fact, it can all be summarized in a simple statement: all quantum operations are in principle implementable in a quantum circuit~\cite{Nielsen}, namely the one shown in Fig.~\ref{fig:QuantumOp}. If those objects really are all there is in terms of dynamics, the evolution always has an associated causal structure and therefore it is mathematically compatible with causality. But it is still not clear how this relates to the probabilities measured in the problem with Alice and Bob and the causal inequality setting.

Operations are transformations taking an input to an output state while, in the bipartite problem, both Alice and Bob act on a system. We can think of using the idea of a supermap, as defined in \eqref{supermapaction}, to describe this situation. This time, the map takes two operations as inputs, $\mathcal{M}(\mathcal{A},\mathcal{B})=\rho'$, instead of just one. This is exactly what is done in the study of general supermaps~\cite{ChirSupermap} and quantum combs/networks~\cite{Chiribellanetworks,MilzPollock} for an arbitrary number of input maps. These formalisms use the Choi-Jamiolkowski isomorphism \eqref{ChoiJ} to simplify the analysis of such maps and even maps between them. 

It could be that new situations appeared in such cases involving multiple steps. For example, if $\mathcal{A}_\rho$ is the preparation of an input state $\rho$ and $\mathcal{B},\mathcal{C}$ are operations to be applied on $\rho$, a supermap could give us the final state $\mathcal{M}(\mathcal{A}_\rho,\mathcal{B},\mathcal{C})=\rho'$. The transformation $\rho\to\rho'$ has to be a quantum operation, and thus it is compatible with a circuit, but it is not clear whether it could be written as a circuit containing the specified input operations. In reference~\cite{Chiribellanetworks}, the authors develop an axiomatic approach for $N$-supermaps analogous to the axiomatic for quantum operations reproduced here. It is possible to show that the properties of those maps induce a generalized Stinespring dilation~\cite{Chiribellanetworks,MilzPollock,PollokStinesp}. See Fig.~\ref{fig:GeneralStinesp} for a representation of the Stinespring dilation of a supermap with $k$ inputs. Any such supermap is implementable with a circuit structure with open slots and therefore the final evolution is compatible with causality involving the specified input operations.

But one of the properties these maps have to obey so that the above is true is the condition that the output of the $N$-th entry, $\mathcal{A}_N$, cannot statistically influence the input of the entry to its left, $\mathcal{A}_{N-1}$. This is usually postulated, but can also be derived by asking for maps to obey CP and linearity in a recursive construction~\cite{Chiribellanetworks}. Either way, this is the mathematical requirement of causal ordering and it is taken to be true when defining quantum dynamics through supermaps.
The main reason is because, with this assumption, we can always understand evolutions as coming from a generalized system+environment model, which should be true if QT is globally valid. In particular, if the operations of Alice and Bob are taken as inputs of a quantum supermap, they are guaranteed to be applied in a definite order, and their probabilities will obey the causal inequality.

\begin{figure}
   \includegraphics[scale=0.48]{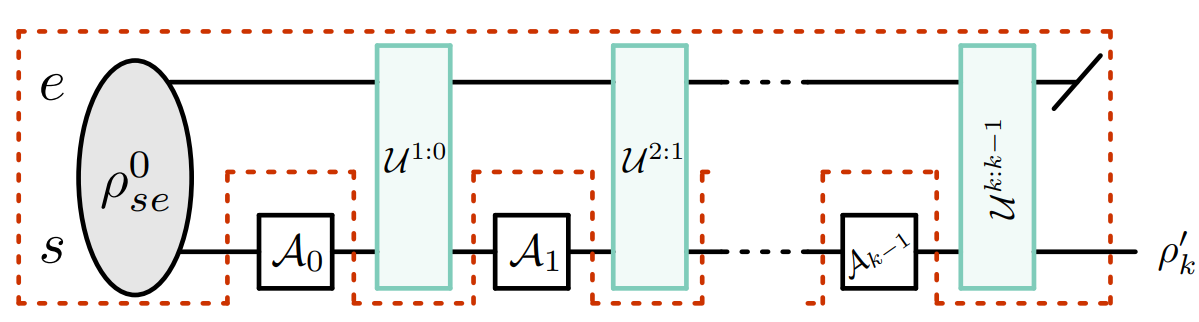}
    \caption{Generalized Stinespring dilation. For a trace-preserving supermap that takes k inputs there exists a set of unitary maps $\{\mathcal{U}^{1:0},\dots, \mathcal{U}^{k:k−1}\}$ and an initial state $\rho^0_{se}$ describing system+environment such that the output state of the system is $\rho'_k=\on{Tr}_e\left[\mathcal{U}^{k:k-1}\left(\mathcal{A}_{k-1}\otimes \mathcal{I}_e\right)\dots\mathcal{U}^{1:0}\left(\mathcal{A}_{0}\otimes \mathcal{I}_e\right)(\rho^0_{se})\right]$. Figure from~\protect\cite{MilzPollock}.}
    \label{fig:GeneralStinesp}
\end{figure}

 Despite the naturality of this assumption, a reasonable question to ask now is: would non-compatibility with circuit implementation cause any sort of paradox? Can we conceive evolutions with indefinite causal structure that nevertheless agree with quantum mechanics for the observers involved? That is what we explore in the next chapter.

\chapter{The process matrix formalism}\label{Chap Process Matrix}
In the present chapter, we will develop the process matrix formalism, which consists on a generalization of Quantum Theory achieved by letting aside the demand for causal structure in its evolutions~\cite{Oreshkov,Chiribella}. As we have seen in the last chapter, Quantum Theory requires that each operation is either applied before, after or ``with spacelike separation of'' any other. The process matrix formalism removes that demand in a global sense and analyzes what properties should be satisfied so that quantum mechanics is still valid locally. The notions of closed laboratories and definite causal order of section~(\ref{Operationalview}) are fundamental in the formalism. It can be revealed in advance that the resulting framework includes processes describing usual physical situations, like a quantum system shared by two agents frequently called Alice and Bob in quantum information problems, as well as new types of structures.

The motivation behind this formalism is being a first step in building a framework for a Probabilistic Theory of Quantum Gravity, as proposed in~\cite{Hardy}. The author argues that General Relativity and Quantum Theory are incompatible, one of the reasons being the way causality is treated in each theory. Then, a first step in trying to make a theory of Quantum Gravity would be to first remove these issues of incompatibility. 

In General Relativity, causal relations between events are a property of the metric in spacetime. And the metric is, in turn, obtained uniquely from solving Einstein’s equations with boundary conditions. Therefore, causal relations come out as a result, an output, of General Relativity, and they can behave dynamically.

Quantum Theory, on the other hand, presumes operations and measurements fit an order structure, which is predefined and has to be given in advance. One can say, for example, that an operation $\mathcal{A}$ is applied before, after or with spacelike separation from operation $\mathcal{B}$ by looking at the spacetime points of their applications. Or, without making reference to a spacetime, one should establish the order \emph{a priori} by specifying time steps or a quantum circuit the system will go through. Only then it is possible to compose operations and extract results from the theory, as argued in the last chapter. Since causal relations determine what are the orders for operations made at definite locations and, therefore, the final probabilities for the experiments, they are an input required by Quantum Theory. 

Then, how are we supposed to understand causal relations in a quantum spacetime? Of course, it could be the case that spacetime cannot even be regarded as a totally quantum entity, as argued and modeled in~\cite{Penrose,Diosi,Scully2018}. Even then, the fact that the theories both aim to explain physical phenomena in certain scales where their effects are dominant poses a problem, for instance, in small enough spatio-temporal scales and for big enough masses~\cite{Bischoff2018}. The leap from the quantum to the classical world is a regime of interest which leads to a compromise for one or both theories regarding their domain of validity and completeness.

 The assumption that spacetime and gravity should be described as quantum entities or that they should at least arise as an emergent effect from a microscopic quantum world leads us to a quest for unification of Physics. As noted above, the two theories, in the way they are formulated now, are incompatible in this sense: to describe a system and its evolution in Quantum Theory, one needs predefined order. Then, spacetime could not be described by it since order will only be established after Einstein`s equations or eventual modified equations for spacetime are solved. In addition, General Relativity is also incompatible with Quantum Theory because it is deterministic\footnote{Here and in the majority of the quantum information literature, saying a theory is deterministic means just the opposite of probabilistic: its predictions for direct experimental outcomes have no uncertainty, not even a fundamental one. They predict a unique value for an outcome when the system meets fixed initial conditions. Thus, classical mechanics is deterministic while both statistical and quantum mechanics are not. Not to be mistaken with notions such as the existence of a differential equation governing the evolution of a physical entity, which garantees a uniquely \emph{determined} evolution for it. This holds true for the quantum wave function, uniquely determined by its equation, but we still cannot in general arrive at unique values for results of single quantum measurements.}, while Quantum Theory is irreducibly probabilistic. The mathematical structures are not ready to merge.

%

 In reference~\cite{Hardy}, the author introduces a program for constructing a Theory of Quantum Gravity using an operational methodology. According to it, the features we discussed are conservative and radical features of General Relativity and Quantum Theory when compared to Classical Theory, as illustrated in the diagram below. The argument is that if we are able to construct a framework which displays both radical features, then it could contain both General Relativity and Quantum Theory. Hopefully it could contain a Theory of Quantum Gravity as well.
 \begin{figure}[ht]
     \centering
\includegraphics[scale=0.95]{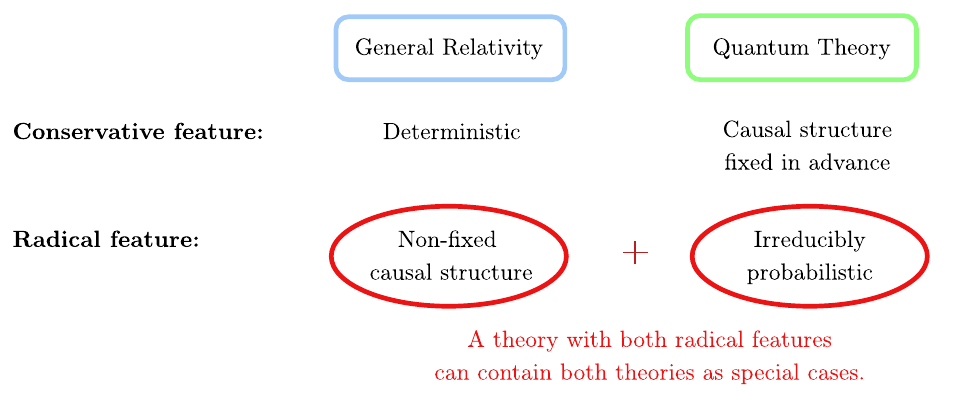}
     \caption{Reasoning behind a program for a probabilistic Theory of Quantum Gravity~\protect\cite{Hardy}.}
     \label{fig:TQG}
 \end{figure}

One attempt to construct a theory with both features is to generalize Quantum Theory, which is already irreducibly probabilistic, in a way that it does not require fixed causal structure. This would leave causal structure free to be determined by something else, i.e. solutions (possibly quantum) for the metric which could affect order in new ways. 

For example, in Quantum Gravity, if one considers that the spacetime metric can be found in a superposition state, this is expected to generate a superposition of causal structures as well, possibly resulting, depending on the configuration, in a superposition of distinct order relations of operations, $\mathcal{A}$ before $\mathcal{B},$ $\mathcal{B}$ before $\mathcal{A}$ and $\mathcal{A}$ spacelike separated from $\mathcal{B}$. However, it is necessary to not assume General Relativity \emph{a priori}, because one wishes to construct a general quantum framework where General Relativity will only fit later. Causal relations have to be formulated from the point of view of Quantum Theory alone, not making any reference to an underlying spacetime, so that it is possible to make clear assumptions about causality or remove them. That is why we introduced in the last chapter a way to treat causal relations relying on general outcome probabilities. From there, we can set up a theory which agrees with quantum mechanics locally leaving causal structure to be determined by a formulation of gravitation consistent with the program. The process matrix formalism was constructed to be that theory. 

The formalism has attracted interest, aside from the foundational point of view, in the context of quantum information processing, because it goes beyond the quantum circuit model. It has been suggested that indefinite order appears within regular quantum mechanics scenarios, without gravity~\cite{Araujo2015}. The main protocol included in these discussions is the quantum switch~\cite{Chiribella}, one of the simplest cases of indefinite order, which is claimed to have been realized experimentally~\cite{ReviewExp}. A lot of the work done in the area does not necessarily focus on causal-modeling, and processes with indefinite order are explored in quantum computation, quantum information, quantum metrology and thermodynamics~\cite{Araujo,Guerin,Wei,ChirDisc,Goswami2020,Guo2020,Colnaghi_2012,Taddei,Refrigerator_Vedral,RefrigExp,Zhao2020}. Questions have been raised about whether the experiments can be regarded as realizations of indefinite order, as opposed to simulations, and in what sense there is incompatibility with definite causality~\cite{Voji,Oreshkov2019timedelocalized}. On the other hand, gravitational scenarios involving indefinite orders have also been considered~\cite{tbell,Rindler,QSonEarth,shells}, generating motivation to search for adequate anda meaningful operational notions for events and causal relations in scenarios involving curvature of spacetime, classical or quantum.

For all of those reasons, the field of indefinite orders has been gaining relevance over the years. In this chapter, we are going to construct the process matrix formalism and discuss indefinite order through the definition of causal non-separability according to references~\cite{Oreshkov} and~\cite{Costa}. Next, we will talk about the quantum switch as a causally non-separable process~\cite{Chiribella}. Finally, we will quickly comment on experimental implementations of the switch in optical tables~\cite{ReviewExp} and discuss the description of these experiments from the operational and the spacetime background point of view.

\section{Local quantum mechanics: quantum instruments}

Let us construct the process matrix formalism for the bipartite setting, as first introduced in reference~\cite{Oreshkov}. This will be done by requiring that the conditional probabilities in experiments are consistent with Quantum Theory locally, namely inside each of Alice and Bob's closed laboratories. The possible operations in a laboratory will be the ones allowed by quantum operational dynamics, therefore causal structure is definite inside them. Besides that, no other assumption is made about global structure. We will only demand probabilities to be well defined. It could be that asking for the validity of Quantum Theory inside laboratories indirectly implied that
at most one-way signaling experiments are allowed. Then, a probabilistic theory compatible with local quantum mechanics would be inevitably causal, with probabilities satisfying causal inequalities like \eqref{classord}.  Interestingly, it turns out this is not the case.

Let us define the objects that represent possible procedures inside local laboratories. In \ref{QuantumOpAxiomatic}, we talked about how a projective measurement transformation can be written as a trace decreasing quantum operation. Thus, for a fixed measurement outcome, an operation describes how the state transforms. But we wish to describe the full measurement procedure, with the possibility of the quantum state suffering distinct operations depending on the outcome of a measurement and to store the measured value. In a measurement procedure, there is a probability value for obtaining each outcome $m$, which also determines the operation to be applied, $\mathcal{M}_m$. The input state transforms like $\rho \to \mathcal{M}_m (\rho)/\operatorname{Tr}[\mathcal{M}_m (\rho)]$. We are interested in complete families including the operations corresponding to each possible m of a measurement, and those are called quantum instruments~\cite{Davies&Lewis}. Instruments model the most general transformations on quantum systems done in laboratory and keep the classical information gathered. Let us state the definition for finite dimensional spaces.

\begin{definition}

 A \textbf{quantum instrument} is a family of linear, completely positive and trace non-increasing maps (quantum operations) from an input Hilbert space of operators to an output Hilbert space of operators, 
\begin{equation}
    \{\mathcal{M}^A_j\}_{j=1}^n, \quad \quad \mathcal{M}_{j}^{A}: \mathcal{L}\left(\mathcal{H}^{A_{1}}\right) \rightarrow \mathcal{L}\left(\mathcal{H}^{A_{2}}\right),
\end{equation} such that $\sum_{j=1}^{n} \mathcal{M}_{j}^{A}$ is completely positive and trace-preserving - CPTP. 

Equivalently, in terms of Kraus operators, a quantum instrument is a family of maps between the spaces above with action

\begin{equation}
    \mathcal{M}^A_j(\rho)= \sum_{k=1}^m E_{jk} \rho E_{jk}^{\dagger}
\end{equation}
for a family of linear operators $E_{kj}:\mathcal{H}^{A_{1}}\to \mathcal{H}^{A_{2}}$ such that $\sum_{k=1}^m E_{jk}^{\dagger}E_{jk} \leq \mathds{1}^{A1}$ for each fixed $j$, and moreover such that
$$\sum_{j=1}^{n}\sum_{k=1}^m E_{jk}^{\dagger}E_{jk} = \mathds{1}^{A1}, $$
where $m=d_{A1} \times d_{A2}$ is the multiplication of the the dimensions of $\mathcal{H}^{A_{1}}$ and $\mathcal{H}^{A_{2}}$ respectively (the optimal Kraus decomposition) and $\mathds{1}^{A1}$ is the identity on $\mathcal{H}^{A1}$.
\end{definition}
The indices $j$ correspond to the possible results of measurement, when there is a measurement. A regular quantum unitary operation on a system is a special case where the input and output spaces are the same and there is only one possibility for $j$, therefore only one operator $\mathcal{M}^{A}$, which is CPTP for the second condition to hold.

The simplest examples of quantum instruments can be derived from the simplest examples of quantum operations. An instrument from an input Hilbert space $\mathcal{H}_I \simeq \mathbb{C}$ to an output Hilbert space corresponds to a family of unnormalized state operations $\{\rho_i\}$ associated to probabilities $p_i=\operatorname{Tr}(\rho_i)\leq 1$ such that $\sum_i p_i=1$. Hence, the instrument can be interpreted as a probabilistic source of quantum states. An instrument from a input Hilbert space to an output Hilbert space $\mathcal{H}_O \simeq \mathbb{C}$ corresponds to a family of quantum effect operations $\{\mathcal{E}_m=\operatorname{Tr}(P_m(\cdot))\}$ such that $\sum_m P_m = \mathds{1}_{\mathcal{H}_I}$. Thus, the instrument describes a POVM: it takes a quantum state as input and returns the probabilities $p(m |\rho )=\operatorname{Tr}(P_m\rho)$ of getting each possible measurement outcome $m$ in the interaction. The probability of getting any outcome is, of course, $\sum_m \operatorname{Tr}(P_m\rho)= \operatorname{Tr}(\rho)=1.$

\section{Conditions on processes}

We will consider the general bipartite case, as the one in the communication task for the causal inequality~(\ref{causalineqsec}), where there are two agents who perform an instrument in a system inside their respective closed laboratories. If one agent, say Alice, performs a quantum instrument $\{\mathcal{M}^A_i
\}_{i=1}^n$ in her laboratory, she can get probabilities corresponding to the effects of measurement. That is, a function which gives information on what is the probability of measuring outcome $i$ for the state $\rho$. The form of this probability function is given by the quantum operations framework: it will always be the trace of a completely positive trace non-increasing map acting on a density operator. In this case, it is given by $P(\mathcal{M}^A_i)= \operatorname{Tr}\left[\mathcal{M}^A_i(\rho)\right]$. In the same way,  Bob, who performs its own quantum instrument $\{\mathcal{N}^B_j
\}_{i=1}^m$, will get the probability function $P(\mathcal{N}^B_j)= \operatorname{Tr}\left[\mathcal{N}^B_j(\rho)\right]$. The question now is: what is the general form allowed for the probability function in the joint description of these two laboratories?
\begin{equation}
    P(\mathcal{M}^A_i,\mathcal{N}^B_j)=?
\end{equation}
A complete list of the probabilities $P(\mathcal{M}^A_i,\mathcal{N}^B_j)$ for every possible outcomes will be called a \textbf{process}. One can recognize this as a proposal to construct maps similar to what we discussed in~(\ref{causalityCombs}). We will find out what are the general properties that the validity of quantum mechanics inside individual local laboratories imposes on processes, without assuming the causality requirement. The case for $N$ parties can be considered~\cite{Araujo2015}, but as did the authors of~\cite{Oreshkov}, we are only going to construct explicitly the bipartite case. 

\subsection{Probabilities behave linearly}
The first thing processes have to obey so that they consistently predict quantum mechanics results inside each laboratory is linearity in each of the entries of $P$. This requirement comes from the linear behavior of quantum instruments in quantum mechanics. Consider a situation in which a quantum instrument $\{\mathcal{M}_i\}_{i=1}^n$ is applied to a system with probability $p$, and another instrument $\{\mathcal{N}_i\}_{i=1}^n$ is applied with probability $(1-p)$. The probability to observe outcome $i$ is then $p P(\mathcal{M}_i)+ (1-p)P(\mathcal{N}_i)$. At the same time, the quantum mechanical description of instrument says that this situation is described by the new randomized instrument $\{\Tilde{\mathcal{M}}_i=p \mathcal{M}_i+(1-p)\mathcal{N}_i\}_{i=1}^n$. Thus,
$$P\left(p\mathcal{M}_i+(1-p)\mathcal{N}_i \right)= p P\left(\mathcal{M}_i\right)+ (1-p)P\left(\mathcal{N}_i\right).$$ This shows that the map $P$ must be convex-linear to agree with quantum mechanics locally. Consider another situation in which we have a quantum instrument $\{\mathcal{M}_i\}_{i=1}^n$ and we wish to treat two measurement outcomes as one, for instance, outcome $n$ can be relabeled to $n-1$. Then, the probability of measuring the outcome $n-1$ in this new setting will be $P(\mathcal{M}_{n-1})+P(\mathcal{M}_n)$. Quantum mechanics represents this situation with a new instrument, the coarse-grained instrument, $\{\Tilde{\mathcal{M}}_i\}_{i=1}^{n-1}$ such that $\Tilde{\mathcal{M}}_i=\mathcal{M}_i$ for $i< n-1$ and $\Tilde{\mathcal{M}}_{n-1}=\mathcal{M}_{n-1}+\mathcal{M}_{n}$. Therefore,
$$P\left(\mathcal{M}_{n-1}+\mathcal{M}_n\right)=P(\mathcal{M}_{n-1})+P(\mathcal{M}_n).$$
These two properties of convex-linearity and additivity impose linearity on the map $P$. 

Indeed, additivity implies that $P(0)=0$, that $P(-\mathcal{M})=-P(\mathcal{M})$ and that $P(n\mathcal{M})=nP(\mathcal{M})$ for any integer $n$. So, consider an element $C\mathcal{M}+\mathcal{N}$ for an arbitrary real scalar $C$. If we decompose it like $C=n+c$, where $n\in \mathds{Z}$ and $0\leq c\leq 1$, we can write \begin{align} P(C\mathcal{M}+\mathcal{N})&= P(n\mathcal{M}+c\mathcal{M}+\mathcal{N})
\\\nonumber
&= nP(\mathcal{M}) + P(c\mathcal{M}+(1-c)0) + P(\mathcal{N}) =CP(\mathcal{M})+P(\mathcal{N}).
\end{align} This reasoning can be made separately for each entry of a process. Therefore, $P(\cdot,\cdot)$ has to be a bilinear function.

Since each instrument element is a quantum operation, we can work with their Choi operators instead. For convenience, we are going to use a slightly modified version of the Choi-Jamiolkowki isomorphism~(\ref{ChoiJ}):
\begin{equation}
    \mathcal{E}\mapsto E=\left[\left(\mathcal{I}_{\mathcal{H}_I} \otimes \mathcal{E}\right)(\ket{\mathds{1}}\rangle\langle\bra{\mathds{1}})\right]^T, \quad \ket{\mathds{1}}\rangle := \sum_k \ket{k}\ket{ k}_{\mathcal{H_I}\otimes\mathcal{H}_I}.
\end{equation}Besides notation, the difference from the mapping we were working with before is just the transposition taken with respect to the basis in which $\ket{\mathds{1}}\rangle$ is written. Note that the inverse morphism is also modified. The action of the map in a state $\rho$ is given by
\begin{equation}\label{ChoiJNew}
    \mathcal{E}(\rho)=\left[\operatorname{Tr}_{\mathcal{H}_I}\left[\left(\rho\otimes\mathds{1}_{\mathcal{H}_O}\right)E\right]\right]^T.
\end{equation}
We will denote the Choi operator of $\mathcal{M}^A_i:\mathcal{L}(\mathcal{H}^{A_1}) \to \mathcal{L}(\mathcal{H}^{A_2})$ by $M_i^{A_1 A_2}\in \mathcal{L}(\mathcal{H}^{A_1}\otimes \mathcal{H}^{A_2})$ and analogously for other instrument elements. Then, quantum processes can be seen as \textbf{bilinear} maps from Choi operators in $ \mathcal{L}(\mathcal{H}^{A_1}\otimes \mathcal{H}^{A_2}) \times  \mathcal{L}(\mathcal{H}^{B_1}\otimes \mathcal{H}^{B_2})$ to $[0,1]\subset \mathbb{R}$. The natural property of the tensor product between vector spaces is that bilinear functionals on $V_1 \times V_2$ are isomorphic to linear functionals on $V_1\otimes V_2$. Thus, we can say that a bipartite process can further be represented by a \textbf{linear} function $f$ from space $\mathcal{L}(\mathcal{H}^{A_1}\otimes\mathcal{H}^{A_2}\otimes\mathcal{H}^{B_1}\otimes\mathcal{H}^{B_2})$ to $[0,1]\subset\mathbb{R}$. Any such map is given by a ``bra'' of that Hilbert space, which is equipped with the Hilbert-Schmidt product we talked about at the ending of section~\ref{notationOps}: $\braket{\eta | \rho}=\operatorname{Tr}(\eta^{\dagger}\rho)$. Therefore, corresponding to every bipartite process there exists an element $W^{A_1 A_2 B_1 B_2}\in\mathcal{L}(\mathcal{H}^{A_1}\otimes\mathcal{H}^{A_2}\otimes\mathcal{H}^{B_1}\otimes\mathcal{H}^{B_2})$ such that
\begin{align}
    P(\mathcal{M}^A_i, \mathcal{N}^B_j)&=f(M^{A_1A_2}_i\otimes N^{B_1B_2}_j)=\bigl<W^{A_1 A_2 B_1 B_2}|M^{A_1A_2}_i\otimes N^{B_1B_2}_j\bigr>\nonumber
    \\
&
=\operatorname{Tr}\left[W^{A_1 A_2 B_1 B_2} \left(M^{A_1A_2}_i\otimes N^{B_1B_2}_j \right)\right] \label{processMatrixProb}
\end{align}

The elements $W^{A_1 A_2 B_1 B_2}$ that in principle correspond to processes consistent with local quantum mechanics are going to be called \emph{process matrices}. Here, we are implicitly assuming that the probability function is non-contextual, that is, the probability to measure values $i,j$ does not depend on variables concerning the implementation of the instruments. For instance, there is no explicit dependence on the particular sets  $\{\mathcal{M}^{A}_{m}\},\{\mathcal{N}^B_{n}\}$ of instrument elements used by the agents.

\subsection{Probabilities are non-negative and sum up to 1}
 The above representation comes from demanding linearity, but not all matrices in $\mathcal{L}(\mathcal{H}^{A_1}\otimes\mathcal{H}^{A_2}\otimes\mathcal{H}^{B_1}\otimes\mathcal{H}^{B_2})$ represent situations compatible with local quantum mechanics. We also need to demand the basic assumption that probabilities are positive: $P(\mathcal{M}^A_i,\mathcal{N}^B_j)\geq0$, which corresponds to the condition
\begin{align}
    &\operatorname{Tr}\left[W^{A_1 A_2 B_1 B_2} \left(M^{A_1 A_2}\otimes N^{B_1 B_2} \right)\right]\geq 0 \, ,
    \\
    \forall\quad M^{A_1 A_2}&\geq 0 \in \mathcal{L}(\mathcal{H}^{A_1}\otimes\mathcal{H}^{A_2}),  N^{B_1 B_2} \geq0 \in \mathcal{L}(\mathcal{H}^{B_1}\otimes\mathcal{H}^{B_2}),\nonumber
\end{align} where $M^{A_1 A_2}$ and $N^{B_1 B_2}$
 are arbitrary positive matrices, since those are the Choi maps of quantum operations. One can say equivalently that the linear map $f$ must be positive on pure tensors (POPT) with respect to the partition $A_1A_2 - B_1B_2$.
 
 Furthermore, we also demand that probabilities are well-defined if Alice's and Bob's input spaces are extended to describe auxiliary systems that are independent of the process. That means Alice and Bob can use other systems, called ancillas, to operate on their main systems. Those ancillas can even be entangled systems shared between the laboratories. For arbitrary operations on the extended spaces $\mathcal{M}:\mathcal{L}(\mathcal{H}^{A'_1}\otimes\mathcal{H}^{A_1})\to \mathcal{L}(\mathcal{H}^{A_2})$ and $\mathcal{N}:\mathcal{L}(\mathcal{H}^{B'_1}\otimes\mathcal{H}^{B_1})\to \mathcal{L}(\mathcal{H}^{B_2})$ and their Choi maps $M^{A'_1 A_1 A_2}$ and $N^{B'_1 B_1 B_2}$, we demand 
 \begin{equation}\label{CPprocesses}
     \operatorname{Tr}\left[\rho\otimes W^{A_1 A_2 B_1 B_2}\left(M^{A'_1 A_1 A_2}\otimes N^{B'_1 B_1 B_2}\right)\right] \geq 0 \, , \quad \forall \rho\geq0 \in \mathcal{L}(\mathcal{H}^{A'_1}\otimes\mathcal{H}^{B'_1}).
 \end{equation} where the form of the extended process matrix $ W^{A'_1 A_1 A_2 B'_1 B_1 B_2}$ is $\rho\otimes W^{A_1 A_2 B_1 B_2}$ indicates the independence between auxiliary systems and the process. It can be shown~\cite{POPTpositivity} that condition~\eqref{CPprocesses} for maps is equivalent to positivity, $W^{A_1 A_2 B_1 B_2}\geq0$. Since positivity is more restrictive than POPT, we stick with the former.
 With positivity of probabilities guaranteed, we are left to demand normalization. We need that $P(\mathcal{M}^A,\mathcal{N}^B)=1$ for all complete instruments $\mathcal{M}^A=\sum_i\mathcal{M}^A_i$ and $\mathcal{N}^B=\sum_j\mathcal{N}^B_j$, because it represents the probability of detecting any possible outcome. The possible operations of this type are all CPTP maps (quantum channels). For Choi operators, this translates as
 \begin{align}
     &\operatorname{Tr}\left[W^{A_1 A_2 B_1 B_2} \left( M^{A_1 A_2}\otimes N^{B_1 B_2}\right)\right]=1  
     \\
      \forall & M^{A_1 A_2}, N^{B_1 B_2} \text{ Choi maps of quantum channels.}\nonumber
 \end{align}
 
 \section{Process matrices}
 Summing it all up, we have constructed the general maps describing probabilities compatible with Quantum Theory in the terms outlined above for the bipartite case. 
 \begin{definition}\label{processmatrix}
 A bipartite process matrix acting on input-output spaces $(\mathcal{H}^{A_1}\to \mathcal{H}^{A_2})$ and $(\mathcal{H}^{B_1}\to \mathcal{H}^{B_2})$ is a map $ W^{A_1 A_2 B_1 B_2}\in \mathcal{L}(\mathcal{H}^{A_1}\otimes\mathcal{H}^{A_2}\otimes\mathcal{H}^{B_1}\otimes\mathcal{H}^{B_2})$ such that
 \begin{equation}
      W^{A_1 A_2 B_1 B_2}\geq 0 
     \end{equation}
and, given any maps $M^{A_1 A_2}\in \mathcal{L}(\mathcal{H}^{A_1}\otimes\mathcal{H}^{A_2})$ and $N^{B_1 B_2} \in \mathcal{L}(\mathcal{H}^{B_1}\otimes\mathcal{H}^{B_2})$ satisfying $M^{A_1 A_2},N^{B_1 B_2}\geq 0$ and $\operatorname{Tr}_{A_2}M^{A_1 A_2}=\mathds{1}_{A_1},$ $\operatorname{Tr}_{B_2}M^{B_1 B_2}=\mathds{1}_{B_1},$
\begin{equation}\label{normalizationProcess}
          \operatorname{Tr}\left[W^{A_1 A_2 B_1 B_2} \left( M^{A_1 A_2}\otimes N^{B_1 B_2}\right)\right]=1.
          \end{equation}
 \end{definition}
 
 Processes encompass all the basic quantum setups descriptions, from a mere specificatiom of a quantum state to a situation with agents making measurements on systems and sending information to each other. For instance, in the bipartite case, the situation in which Alice and Bob share a quantum state and then realize measurements on it with no output state is represented by a process matrix with $\mathcal{H}^{A_2},\mathcal{H}^{B_2}\simeq\mathbb{C}$, since the input maps are POVM elements. Thus, their Choi operators $M^{A_1A_2}_i,N^{B_1B_2}_j$  are in the spaces $\mathcal{L}(\mathcal{H}^{A_1})$ and $\mathcal{L}(\mathcal{H}^{B_1})$. In this case, the expression for the probabilities $P(\mathcal{M}^{A}_i,\mathcal{N}^B_j)=\operatorname{Tr}\left[W^{A_1 B_1 } \left( M^{A_1 }\otimes N^{B_1}\right)\right]$ is simply Born's rule in its usual form for a state given by a density matrix $W^{A_1 B_1}$ and a measurement operator $M^{A_1}\otimes N^{B_1}$. A quantum state shared by Alice and Bob is then generally represented by a bipartite process matrix $W^{A_1A_2 B_1B_2}=\rho^{A_1 B_1}\otimes\mathds{1}^{A_2B_2}.$
 
 Another simple example is the case where Bob receives a state $\rho^{B_1}\in \mathcal{L}(\mathcal{H}^{B_1})$, performs a measurement and sends the output to Alice through a quantum channel $\mathcal{C}$. The probability for Bob to measure outcome $j$ and Alice to measure outcome $i$ is $\operatorname{Tr}[\mathcal{M}^A_i \circ \mathcal{C}\circ\mathcal{N}^B_j(\rho^{B_1})].$ The process matrix for this case is given by 
 \begin{equation}\label{ChannelMatrix}
W^{A_1A_2 B_1B_2}=\mathds{1}^{A_2}\otimes (C^{B_2A_1})^T \otimes \rho^{B_1} \, ,
 \end{equation}
 where $C$ is the Choi map of the channel. We can check this is correct by manipulating the expression for the probability~\ref{processMatrixProb} and using the inverse CJ morphism, $\mathcal{M}(X)=\left[\operatorname{Tr}_{1}[(X\otimes \mathds{1}_{2})M]\right]^T$:
 \begin{align*}
 P(\mathcal{M}^A_i,\mathcal{N}^B_j) &= \operatorname{Tr}\left[\mathds{1}^{A_2}\otimes \left(C^{B_2A_1}\right)^T \otimes \rho^{B_1}\circ\left(M_i^{A_1A_2}\otimes N_j^{B_1B_2}\right)\right] \\
 &=\operatorname{Tr}\left[\mathds{1}^{A_2}\otimes \left(\sum_{ij}\ket{i}\bra{j}^{B_2}\otimes \mathcal{C}(\ket{i}\bra{j})^{A_1}\right) \otimes \rho^{B_1}\circ\left(M_i^{A_1A_2}\otimes N_j^{B_1B_2}\right)\right] \\
 &=\operatorname{Tr}_{A_2B_2}\Biggl[\sum_{ij}\operatorname{Tr}_{A_1}\left[\left(\mathcal{C}(\ket{i}\bra{j})^{A_1}\otimes\mathds{1}^{A_2}\right) M_i^{A_1A_2}\right]  \\
    & \qquad \qquad \qquad  \otimes \operatorname{Tr}_{B_1} \left[\left(\mathds{1}^{B_1}\otimes\ket{i}\bra{j}^{B_2}\right) \left(\rho^{B_1}\otimes\mathds{1}^{B_2}\right)N^{B_1B_2}_j\right]\Biggr] \\
&= \operatorname{Tr}_{A_2B_2}\left[\sum_{ij} \left[\mathcal{M}^A_i\left(\mathcal{C}(\ket{i}\bra{j}\right)^{A_1}\right]^T  \otimes \left[\ket{i}\bra{j}^{B_2} \circ \left[\mathcal{N}^B_j\left(\rho^{B_1}\right)\right]^T\right]\right] \\
&=\sum_{ij}\operatorname{Tr}_{A_2}\left[ \left[\mathcal{M}^A_i\left(\mathcal{C}(\ket{i}\bra{j}\right)^{A_1}\right]^T  \bra{j}\left[\mathcal{N}^B_j\left(\rho^{B_1}\right)\right]^T\ket{i}\right] \\
&= 
\sum_{ij}\bra{i}\mathcal{N}^B_j\left(\rho^{B_1}\right)\ket{j}\operatorname{Tr}_{A_2}\left[\mathcal{M}^A_i\left(\mathcal{C}(\ket{i}\bra{j})^{A_1}\right)\right] = \operatorname{Tr}\left[\mathcal{M}_i^A\circ \mathcal{C}\circ \mathcal{N}_j^B(\rho^{B_1})\right].
 \end{align*}
 
In reference~\cite{Oreshkov}, the authors derive a characterization of process matrices in terms of their expansion in a basis. For each space $\mathcal{L}(\mathcal{H}^X)$, $X=A_1,A_2,B_1,B_2$, a Hilbert-Schmidt basis is a set of matrices $\{\sigma_i^X\}_{i=0}^{d^2-1}$, where $d$ is the dimension of $\mathcal{H}^X$, such that $\sigma_0^X = \mathds{1}_X$, $\operatorname{Tr}(\sigma^X_i \otimes \sigma^X_j)= d \delta_{ij}$ and $\operatorname{Tr}(\sigma^X_j=0)$ for $j>0$. An example of such a basis for a 2-dimensional space is the set of Pauli matrices. Any element of $\mathcal{L}(\mathcal{H}^{A_1}\otimes \mathcal{H}^{A_2}\otimes \mathcal{H}^{B_1}\otimes \mathcal{H}^{B_2})$ is decomposable as $\sum_{abcd}w_{abcd} \,  \sigma^{A_1}_a\otimes\sigma^{A_2}_b\otimes\sigma^{B_1}_c\otimes\sigma^{B_2}_d$, but the conditions for a process are only met when the decomposition has certain types of terms. A term is of type A$_1$ if it has the form $\sigma^{A_1}_i\otimes \mathds{1}^{A_2B_1B_2}$, $i\neq0$. A term of type A$_2$B$_1$ is written as $\mathds{1}^{A_1}\otimes \sigma^{A_2}_i\otimes \sigma^{B1}_j\otimes\mathds{1}^{B_2}$, $i,j\neq0$, and so on. Fig.~\ref{fig:allowedProcesses} illustrates the interpretation of allowed terms.

\begin{figure}[ht]
    \centering
    \hspace{1.8cm}\includegraphics[scale=0.18]{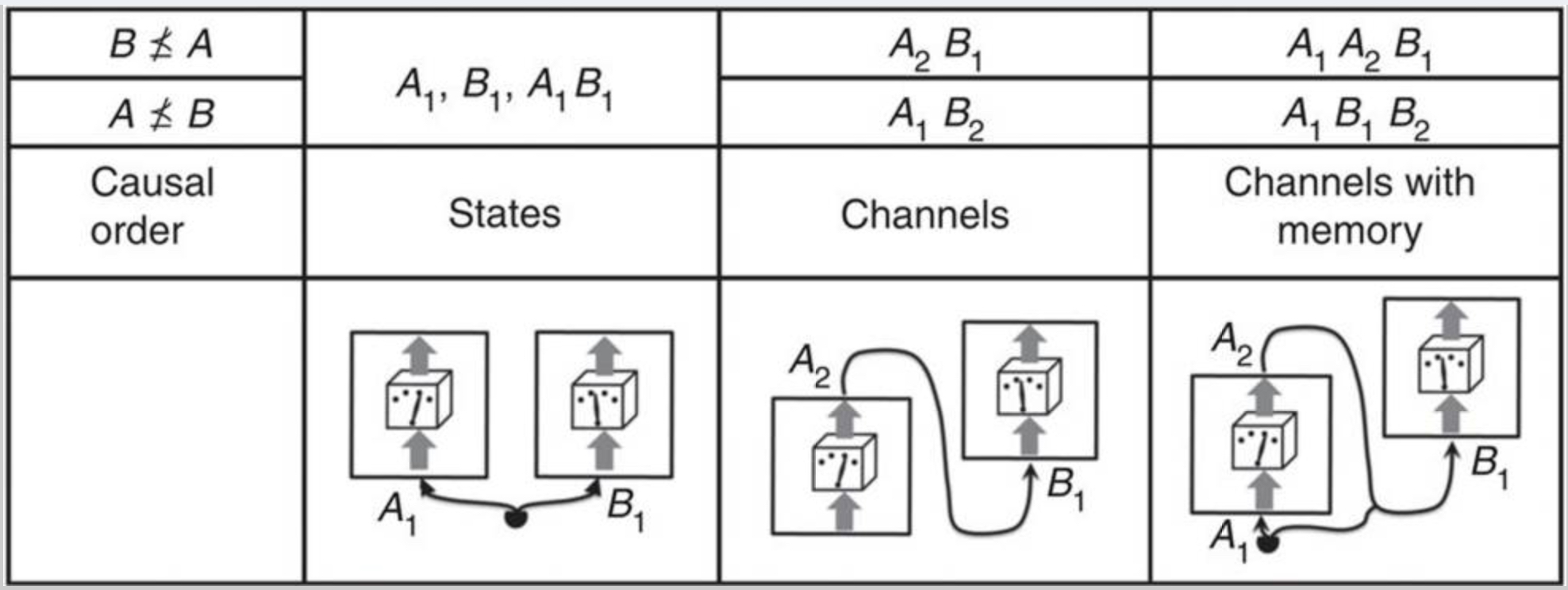} 
    \caption{Reference \cite{Oreshkov} proves that a matrix satisfies the conditions of a process only if it contains the terms listed in this table when written in terms of a Hilbert-Schmidt basis. Although the individual terms are compatible with signaling in at most one direction, a general process matrix can contain terms from both rows and may not be decomposable into a mixture of quantum channels from Alice to Bob and from Bob to Alice. Figure from~\protect\cite{Oreshkov}.}
    \label{fig:allowedProcesses}
\end{figure}  

When the process matrix has a non trivial part on two or more spaces, it means a correlation is being established between them. A process matrix of the type $\mathds{1}^{A_2}\otimes W^{A_1B_1B_2}$ describes a situation in which, for instance, Bob operates on part of a system and both his output and the the other part of the system are sent to Alice through a channel (see Fig.~\ref{fig:allowedProcesses}). In this case Bob can signal to Alice and she cannot signal to Bob. Matrices of this type will be denoted $W^{A\nprec B}$. The opposite situation is given by matrices $W^{B\nprec A}$ of the type $W^{A_1A_2B_1}\otimes \mathds{1}^{B_2}$. If the causal structure is definite but unknown we can have a process given by a classical probabilistic mixture of these matrices. So, we can make a definition characterizing causality in this framework. In analogy with the previous discussion about signaling probabilities, we define:\begin{definition}
A bipartite process matrix is called \textbf{causally separable} if it can be written as
 \begin{equation}\label{causallyseparable}
    W^{A_1A_2B_1B_2}= q W^{A\nprec B}+(1-q) W^{B \nprec A}.
    \end{equation}
    for $0\leq q \leq 1.$ \end{definition} Causally separable process matrices generate probabilities with a definite causal structure. The objects representing a bipartite situation compatible with a quantum circuit all fall into this category, including entangled states shared by Alice and Bob. No process of this type can violate a causal inequality, since every causally separable process is also causal, i.e. it produces causal conditional probabilities as described by condition~\eqref{classord}. Indeed, this is a weaker notion of causality, in the sense that it is theory-dependent, while the condition for probabilities used in the causal inequality is not. In reference~\cite{Giarmatzipaper}, the differences between causal and causally separable processes are further analyzed.
 We also refer the reader to~\cite{Araujo2015} for an alternative characterization of process matrices and the generalization for the multipartite case.

 In the classification represented in Fig.~\ref{fig:allowedProcesses} it is possible to infer that the process formalism contains causally non-separable processes. Explicit examples are given in the next sections. It is an interesting result that the formalism comprises more than usual quantum evolutions by assuming the theory is valid only locally. This means, in particular, that there are processes which cannot be represented by a circuit model. See Fig.~\ref{fig:combsXprocesses}.

 \begin{figure}
\centering
    \includegraphics[scale=0.8]{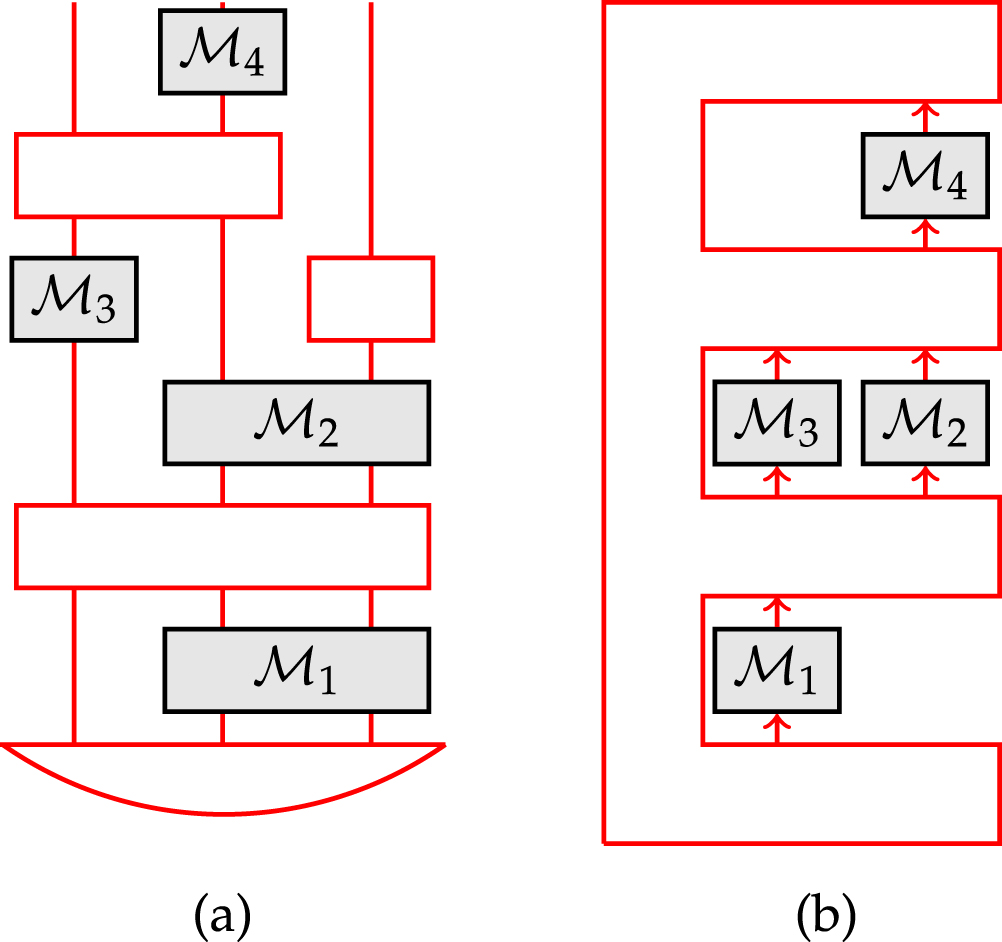}
    \caption{ (a) When a set of operations $\mathcal{M}_i$ made by multiple local agents is compatible with a causal structure, the total transformation can be described by inserting them in a circuit. (b) Process matrices, on the other hand, cannot always be decomposed in this way, since the formalism allows processes not compatible with any causal sequence . Figure from~\protect\cite{Araujo2015}.}
    \label{fig:combsXprocesses}
\end{figure}

 \section{Violation of the causal inequality}\label{violation}

 Surprisingly, there is at least one bipartite process which not only is causally non-separable, but also can be used to violate the causal inequality:
 \begin{equation}\label{processviolates}
     W^{A_1A_2B_1B_2}=\frac{1}{4}\left[\mathds{1}^{A_1 A_2 B_1 B_2}+ \frac{1}{\sqrt{2}}\left(\sigma^{A_2}_z\sigma^{B_1}_z+\sigma^{A_1}_z\sigma^{B_1}_x\sigma^{B_2}_z\right)\right],
 \end{equation}
where $\sigma_{z,x}$ are the Pauli matrices, and the tensor products between matrices and identities are omitted. Recalling the task for the inequality, Alice and Bob each produce a bit in their laboratories, bits $a$ and $b$. They also register a guess for each other's bits $x$ and $y$. Bob produces an extra bit $b'$ and the task is, when $b'=0$, maximizing the probability of Alice guessing Bob's bit right ($x=b$), while when $b'=1$, maximizing the probability of Bob making a right guess ($y=a$). As we have seen, the probability of success has an upper bound, $P_{succ}\leq 3/4$, when we assume the protocol happens in a definite causal structure. With the process above, if Bob measures the received system, assumed to be a qbit, in the $z$ basis, the probabilities calculated using~(\ref{processMatrixProb}) reduce to those of a channel from Alice to Bob. If he otherwise decides to measure the system in the $x$ basis, the probabilities are those of a channel from Bob to Alice. The process thus gives Bob the power of activating a channel in the convenient direction for each case $b'=0$ or $1$. With this resource and the right procedures for Alice and Bob, they can achieve $P_{succ}=(2+\sqrt{2})/4$.

More specifically, in order to violate the inequality Alice has to do the following. Once the qbit arrives at her laboratory, she measures it in the $z$ basis and makes her guess x according to the measurement $\ket{\uparrow}\rightarrow x=0$ and $\ket{\downarrow}\rightarrow x=1$. Then, she generates bit a and prepares the system in a new state, encoding the value of bit a according to the same rules. For example, the Choi operator that corresponds to this action for the special case of result $\ket{\uparrow}$ and $a=1$ is $E=\ket{\uparrow}\bra{\uparrow}\otimes \ket{\downarrow}\bra{\downarrow}$, since it generates Alice's correct output state when inserted in the formula $\mathcal{E}(\rho)=\left[\operatorname{Tr}_{\mathcal{H}_I}\left[\left(\rho\otimes\mathds{1}_{\mathcal{H}_O}\right)E\right]\right]^T$. One can check that the Choi operator for the general case can be written as
$$
M^{A_{1} A_{2}}(x, a)=\frac{1}{4}\left[\mathds{1}+(-1)^{x} \sigma_{z}\right]^{A_{1}} \otimes\left[\mathds{1}+(-1)^{a} \sigma_{z}\right]^{A_{2}},
$$
where $\sigma_z$ is the z Pauli matrix, given by $\ket{0}\bra{0}-\ket{1}\bra{1}$.

Bob must proceed as follows. First, he evaluates bit $b'$, the one that sets which guess will be considered in the current round. If $b'=1$, Bob will try to read Alice's bit, so he will measure the qbit that enters his laboratory in the $z$ basis and make his guess $y=0/1$ for state $\ket{\uparrow}/\ket{\downarrow}$, because that is how Alice encodes her bit in every round. He can prepare the system in an arbitrary state $\rho^{B_2}$, which will not matter for the task since Alice's guess will be discarded. If $b'=0$, Alice will have to guess Bob's bit, so he must try to send it to her. In that case, when he receives the qbit, he will measure it in the $x$ basis and make a random guess $y$, which will not be taken into account in this round. Then, he evaluates bit $b$. The result of the measurement on the system will determine how to send the value of bit $b$ to Alice.  If the measurement result was $\ket{x_+}$, he will prepare the system encoding $b$ in the $z$ basis like $0\rightarrow \ket{\uparrow},1\rightarrow\ket{\downarrow}$. But, if the outcome is $\ket{x_-}$, the encoding will be reversed, $0\rightarrow\ket{\downarrow},1\rightarrow\ket{\uparrow}$. The Choi map of this procedure is given by
$$
\begin{aligned}
 N^{B_{1} B_{2}}\left(y, b, b^{\prime}\right) &=b^{\prime} N_{1}^{B_{1} B_{2}}(y, b)+\left(b^{\prime} \oplus 1\right) N_{2}^{B_{1} B_{2}}(y, b), \\
N_{1}^{B_{1} B_{2}}(y, b) &=\frac{1}{2}\left[\mathds{1}+(-1)^{y} \sigma_{z}\right]^{B_{1}} \otimes \rho^{B_{2}} , \\
N_{2}^{B_{1} B_{2}}(y, b) &=\frac{1}{4}\left[\mathds{1}+(-1)^{y} \sigma_{x}\right]^{B_{1}} \otimes\left[\mathds{1}+(-1)^{b+y} \sigma_{z}\right]^{B_{2}},
\end{aligned}
$$
where $\oplus$ is the sum modulo $2$, so that the term vanishes for $b'=1.$ The general probability function for the process $P(xy|abb')$ is given by~(\ref{processMatrixProb}). But to evaluate the inequality we need $P(x|ab, b'=0)$ and $P(y|ab, b'=1)$. We can get those by averaging the general probability, $P(y|ab, b'=1)=\sum_x P(xy|ab,b'=1)$. The expression can be written as 
\begin{align}
P(y|abb')&=\sum_{x} \operatorname{Tr}\left[W^{A_{1} A_{2} B_{1} B_{2}}\left(M^{A_{1} A_{2}}(x, a) \otimes N^{B_{1} B_{2}}\left(y, b, b^{\prime}\right)\right)\right]\nonumber \\
&=\operatorname{Tr}_{B_{1} B_{2}}\left\{ N^{B_{1} B_{2}}\left(y, b, b^{\prime}\right) \operatorname{Tr}_{A_{1} A_{2}}\left[W^{A_{1} A_{2} B_{1} B_{2}}\left(\sum_{x} M^{A_{1} A_{2}}(x, a)\right)\right]\right\}\label{marginalproby},
\end{align}and the second factor in the trace can be seen as reduced matrix
$$
\overline{W}^{B_{1} B_{2}}(a):=\operatorname{Tr}_{A_{1} A_{2}}\left[W^{A_{1} A_{2} B_{1} B_{2}}\left(\sum_{x}M^{A_{1} A_{2}}(x, a)\right)\right],
$$
 where we have omitted the identity $\mathds{1}^{B_1B_2}$ inside the trace. We have that $$\sum_{x} M^{A_{1} A_{2}}(x, a)= \sum_x\frac{1}{4}\left[\mathds{1}+(-1)^x \sigma_z\right]^{A_1}\otimes\left[\mathds{1}+(-1)^a\sigma_z\right]^{A_2}= \frac{1}{2}\left[\mathds{1}\right]^{A_1}\otimes\left[\mathds{1}+(-1)^a \sigma_z\right]^{A_2},$$ and W given by~(\ref{processviolates}). We can therefore calculate the reduced matrix:
 \begin{align*}
      \overline{W}^{B_{1} B_{2}}(a)&= \operatorname{Tr}_{A_1A_2}\left[\left(\frac{1}{4}\mathds{1}^{A_1A_2B_1B_2}+\frac{1}{4\sqrt{2}}\left(\sigma_z^{A_2}\sigma_z^{B_1}+\sigma_z^{A_1}\sigma_x^{B_1}\sigma_z^{B_2}\right)\right)\right. 
      \\
      &\hspace{8cm}\left.
      \times\left( \frac{1}{2}\left[\mathds{1}\right]^{A_1}\otimes\left[\mathds{1}+(-1)^a \sigma_z\right]^{A_2}\right)\right]
 \\
&= \frac{1}{4}\operatorname{Tr}_{A_1A_2}\Bigg[\frac{\mathds{1}^{A_1}}{2}\otimes\left[\mathds{1}+(-1)^a\sigma_z\right]^{A_2}\otimes \mathds{1}^{B_1B_2} +
    \\
    &+ \frac{\mathds{1}^{A_1}}{2\sqrt{2}}\otimes\left[\sigma_z+(-1)^a \mathds{1}\right]^{A_2}\otimes \sigma_z^{B_1}\otimes \mathds{1}^{B_2}
 +\frac{\sigma_z^{A_1}}{2\sqrt{2}}\otimes \left[\mathds{1}+(-1)^a\sigma_z\right]^{A_2}\otimes\sigma_x^{B_1}\otimes\sigma_z^{B_2}\Bigg],
\end{align*}
  where we have used that $\sigma_z^2= \mathds{1}$. When we make the partial trace with respect to $A_1$, the last term vanishes because the trace of $\sigma_z$ is 0. Taking the trace for the other terms, using $\operatorname{Tr}_X(\mathds{1^{X}})=2$, yields
 \begin{align*}
     \overline{W}^{B_{1} B_{2}}(a)&=\frac{1}{4}\left[2\mathds{1}^{B_1B_2}+\frac{1}{\sqrt{2}}(-1)^a 2\sigma_z^{B_1}\otimes \mathds{1}^{B_2} \right]
     \\
     &=\frac{1}{2}\left[\mathds{1}+(-1)^{a} \frac{1}{\sqrt{2}} \sigma_{z}\right]^{B_{1}} \otimes \mathds{1}^{B_{2}}.
 \end{align*}
Now, we can use this to find the probability $P(y|ab, b'=1)$. We just need to use the formula~\eqref{marginalproby} with $N^{B_1B_2}(y,b,1)$: 
\begin{align*}
P\left(y | a b, b^{\prime}=1\right) &=\operatorname{Tr}_{B_1B_2}\left[\left[\frac{1}{2}\left(\mathds{1}+(-1)^y\sigma_z\right)^{B_1}\otimes \rho^{B_2}\right] \circ\left[ \frac{1}{2}\left(\mathds{1}+\frac{(-1)^a}{\sqrt{2}}\sigma_z\right)^{B_1}\otimes \mathds{1}^{B_2}\right]\right] \\
&=\frac{1}{4}\operatorname{Tr}_{B_1B_2}\left[\left[\mathds{1}+(-1)^a\frac{1}{\sqrt{2}}\sigma_z+ (-1)^y\sigma_z+\frac{(-1)^{y+a}}{\sqrt{2}}\mathds{1}\right]^{B_1}\otimes \rho^{B_2}\right] \\
&=\frac{1}{4}\left[2+0+0+\frac{(-1)^{y+a}}{\sqrt{2}}2\right] 1 \\
&=\frac{1}{2} \left[1+\frac{(-1)^{y+a}}{\sqrt{2}}\right].
\end{align*}
Thus, the first probability in the causal inequality for this case is
$$P(y=a|b'=1)=\frac{2+\sqrt{2}}{4}.$$
The other probability, namely $P(x=b|b'=0),$ can be similarly calculated. First we have the reduced matrix for $P(x|abb')$. Since we are interested in $b'=0$, we can go directly to the evaluation of the reduced matrix:
\begin{align*}
    \overline{W}^{A_1A_2}(b,b'=0) &= \operatorname{Tr}_{B_1B_2}\left[W^{A_1A_2B_1B_2}\left(\sum_y N^{B_1B_2}(y,b)\right)\right]\\
    &= \operatorname{Tr}_{B_1B_2}\Bigl[\frac{1}{4}\left[\mathds{1}^{A_1A_1B_1B_2}+\frac{1}{\sqrt{2}}\left(\sigma_z^{A_2}\sigma_z^{B_1}+\sigma_z^{A_1}\sigma_x^{B_1}\sigma_z^{B_2}\right)\right]
    \\
    &\hspace{7cm}\times\frac{1}{2}\left[\mathds{1}^{B_1B_2}+(-1)^b\sigma_x^{B_1}\sigma_z^{B_2}\right]\Bigr],
\end{align*}
Making the distributive in the expression above, one can identify some terms which contain sigma matrices over $B_1$ or $B_2$. For example, there is the term $(-1)^b \mathds{1}^{A_1A_2}\otimes\sigma_x^{B_1}\otimes\sigma_z^{B_2}$. Since they are inside the trace, and the trace of such matrices is $0$, we can cancel them. We are left only with
\begin{align*}
\overline{W}^{A_1A_2}(b,b'=0)&=\frac{1}{8}\operatorname{Tr}_{B_1B_2}\left[\mathds{1}^{A_1A_2B_1B_2}+(-1)^b \sigma_z^{A_1}\otimes \mathds{1}^{A_2}\otimes \left(\sigma^2_x\right)^{B_1}\otimes \left(\sigma^2_z\right)^{B_2}\right]\\
&=\frac{1}{8}\left[4 \mathds{1}^{A_1A_2}+(-1)^b\frac{4}{\sqrt{2}}\sigma_z^{A_1}\otimes\mathds{1}^{A_2}\right]=\frac{1}{2}\left[\mathds{1}+\frac{(-1)^b}{\sqrt{2}}\sigma_z \right]^{A_1}\otimes \mathds{1}^{A_2}.
\end{align*}
With this, we get, with similar calculation steps,
\begin{align*}
    P(x|ab,b'=0)&=\operatorname{Tr}_{A_1A_2}\left[M^{A_1A_2}(x,a)\overline{W}^{A_1A_2}(b,b'=0)\right]\\
    &=\frac{1}{2}\left[1+ \frac{(-1)^{x+b}}{\sqrt{2}}\right] \xrightarrow{x=a} \frac{2+\sqrt{2}}{4}.
\end{align*}
Hence, the probability of success for this protocol with process $W^{A_1A_2B_1B_2}$ is
\begin{equation}
    P_{succ}= \frac{1}{2}\left[P(y=a|b'=1)+P(x=b|b'=0)\right] = \frac{2+\sqrt{2}}{4} >\frac{3}{4},
\end{equation}
and thus the process is not causal and not causally separable.

 The process above is an explicit example of how the more general formalism allows something with indefinite order. Operationally, we have motivation for demanding that quantum mechanics is obeyed only locally, since the laboratories are what agents have access to.  Process matrices like the one above are, to date, merely interesting mathematical objects. It is not known as of the time this thesis is being written whether there are matrices with a physical interpretation which violate the causal inequality and its variations~\cite{multip_abbot,Branciard_2015}, although some can be simulated~\cite{Ara_jo_2017}. We can also study processes which are causally non-separable~(\ref{causallyseparable}). Even in cases they cannot violate a causal inequality~\cite{Araujo2015}, they are still more general than quantum circuits.

At first, one would expect that the physics we already know is contained within the set of process matrices with definite order, while new processes could emerge from quantum gravity effects or be studied abstractly, but the issue is a little more subtle than that. It has been suggested that the circuit model should be generalized even for describing protocols in regular quantum mechanics~\cite{Chiribella,Araujo2015}. Depending on the context where the formalism is used, which has to do with the notion of laboratories and the Hilbert spaces associated to them\cite{Oreshkov2019timedelocalized}, processes with indefinite order may provide a description for quantum protocols realizable on top of a classical Minkowski spacetime, in table-top experiments. We will discuss this more clearly in the next section.

\section{The quantum switch}\label{quantum switch}
The quantum switch is one of the simplest examples of process with indefinite order. It was introduced in~\cite{Chiribella}, before the formulation we reproduced above was made. In reference~\cite{Araujo2015}, the switch was then recognized as a protocol represented by a process with indefinite order in the process matrix formalism.

The protocol involves two quantum systems: the target and the control system. Usually the control is taken to be a qbit, that is, a quantum system represented by a 2-dimensional Hilbert space with basis denoted by $\{\ket{0},\ket{1}\}$. In the switch, the control is entangled with the order in which two operations $\mathcal{A}$ and $\mathcal{B}$ will be applied on the target system: if the state of the control is $\ket{\phi}^c = \ket{0}$, $\mathcal{A}$ is applied before $\mathcal{B}$, while if the state of the control is $\ket{1}$, the operations are applied in the reversed order. It was originally proposed as a supermap taking two operations as inputs and returning a new one. When $\mathcal{A}$ and $\mathcal{B}$ are unitary operators on state vectors\footnote{That is, the actual operations being applied are A$(\rho)=\mathcal{A}\rho \mathcal{A}^\dagger$ and B$(\rho)=\mathcal{B}\rho \mathcal{B}^\dagger$.}, we can write the resulting map as \begin{equation}\label{QSsupermap}
\mathcal{V}(\mathcal{A},\mathcal{B}) = \ket{0}\bra{0}\otimes \mathcal{B}\circ\mathcal{A}+\ket{1}\bra{1}\otimes\mathcal{A}\circ\mathcal{B}.
\end{equation} Thus, if the control is initially in a state $1/\sqrt{2}(\ket{0}+\ket{1})$ and the target in $\ket{\psi}^t$, we are left with the pure state $1/\sqrt{2}(\ket{0}\otimes \mathcal{B}\circ\mathcal{A} \ket{\psi}^t+\ket{1}\otimes\mathcal{A}\circ\mathcal{B}\ket{\psi}^t)$. The definition can be extended for general quantum operations using an expression analogous to~(\ref{QSsupermap}) for pairs of Kraus operators of the inputs~\cite{ReviewExp,Vilasini:2022ist}. Either way, when the control is prepared in a superposition, the resulting state displays a quantum superposition of orders. However, if the control system is discarded, we get the traced-out final state \begin{equation}  \rho = \frac{1}{2}    \left(\mathcal{B}\mathcal{A}\ket{\psi}\bra{\psi}^t\mathcal{A}^{\dagger} \mathcal{B}^{\dagger}+\mathcal{A}\mathcal{B}\ket{\psi}\bra{\psi}^t\mathcal{B}^{\dagger}\mathcal{A}^{\dagger}\right),
\end{equation}which is a classical mixture of the operations applied in distinct orders. We are led to introduce a third party, Charlie, who can later make a measurement on the target+control systems to attest the quantumness of the order. For instance, he can measure the control in the basis $\ket{\pm}:=1/\sqrt{2}(\ket{0}\pm \ket{1})$ and get a disentangled superposition state of the target.
\begin{figure}[ht]
\centering
    \includegraphics[scale=0.25]{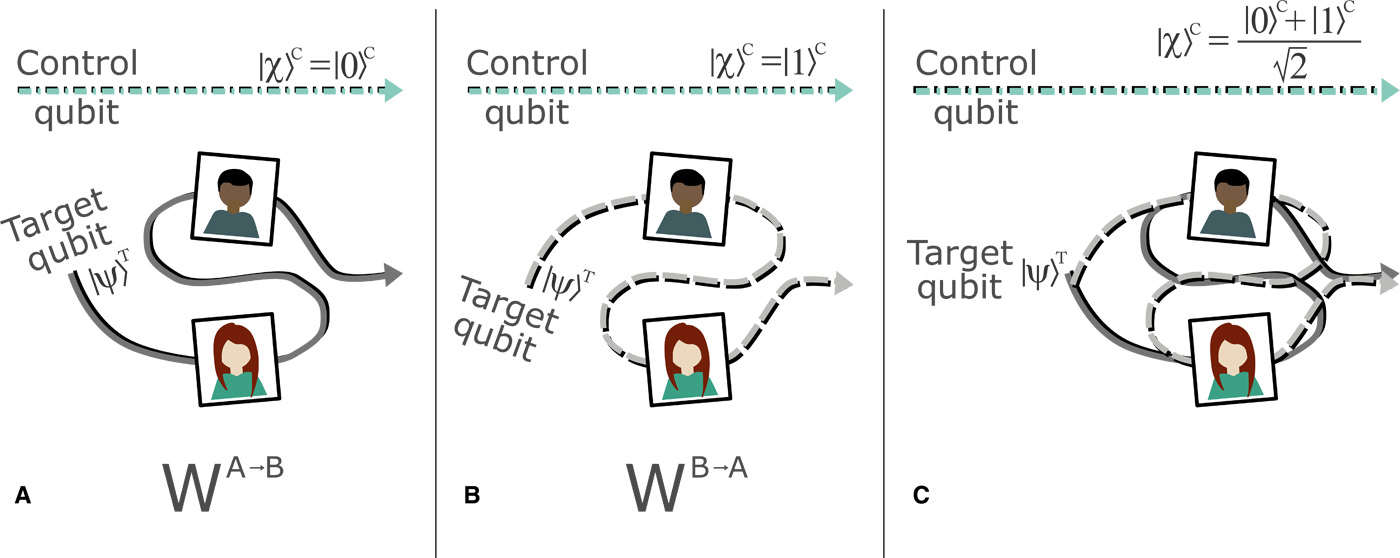}
    \caption{ The quantum switch is a protocol in which two parties, Alice and Bob, act once on a
target qbit in initial state $\ket{\psi}^t$. The order in which the operations are applied depends on the state of another qbit, the control system.  (A) If the control is in the state $\ket{0}^c$, Alice's operation is applied before Bob's. (B) If otherwise the control is in state $\ket{1}^c$, Bob's operation comes before Alice's. (C) Then, if the control is prepared in a superposition state $1/\sqrt{2}(\ket{0}+\ket{1})^c$, we have a superposition of the two situations. To attest the superposition, one has to introduce a third party, Charlie, to measure the systems in the end (see the text for further explanation). Figure from~\protect\cite{Rubino}.}
    \label{fig:QSfigure}
\end{figure}
In the formalism of processes, the quantum switch can be described by a tripartite quantum process. For simplicity, let us consider the target to be also a qbit. The dimensions of input and output spaces of the laboratories are $d_{A_I}=d_{A_O}=d_{B_I}=d_{B_O}=2,$ $d_{C_I}=4,$ $d_{C_O}=1$\footnote{We will now use letters I and O to denote input and output spaces instead of numbers 1 and 2.}, meaning that Alice and Bob can perform their operations on the target, while Charlie can perform a measurement on the 4-dimensional system target+control $\mathcal{H}^{C_I}:=\mathcal{H}^{C^c_I} \otimes \mathcal{H}^{C^t_I}$, with no outgoing system.

Sometimes it is useful to work at the level of process vectors, instead of matrices, when the processes are pure. A pure process is one that can be written as a projector onto some vector $W=\ket{w}\bra{w}$. This is analogous to the correspondence between pure density matrices and state vectors $\rho=\ket{\psi}\bra{\psi}.$ The Choi operators for some of the agents' operations, present in the expression for the probability~(\ref{processMatrixProb}), can also be written in this way. If the operation $\mathcal{M}_U$ is of the form $\mathcal{M}_U(\rho)=U \rho U^{\dagger}$ where $U$ is a linear operator $U:\mathcal{H}_I\to\mathcal{H}_O$ acting on a vector state like $U\ket{\psi},$ its Choi map is 
\begin{align*}
M^{\mathcal{H}_I\mathcal{H}_O}_U &=\left[\sum_{ik}\ket{i}\bra{k}\otimes U \ket{i}\bra{k} U^{\dagger}\right]^T \\
&= \left[\left(\mathds{1}_{\mathcal{H}_I}\otimes U \right)\ket{\mathds{1}}\rangle \langle\bra{\mathds{1}}\left(\mathds{1}_{\mathcal{H}_I}\otimes  U^{\dagger}\right)\right]^T \\
&= \ket{U^*}\rangle\langle\bra{U^*},
\end{align*}
where $\ket{U^*}\rangle :=(\mathds{1}_{\mathcal{H}_I}\otimes U^*)\ket{\mathds{1}}\rangle$ is the Choi vector associated to $U$, and $*$ denotes complex conjugation with respect to the basis used to write $\ket{\mathds{1}}\rangle:=\sum_k\ket{k}\ket{k}$.

So, in the case that the local operations are represented by these Choi vectors as $M^{A_1A_2}_i=\ket{M^{A*}_i}\rangle \langle \bra{M^{A*}_i}$ and the process matrix is pure, we can work with probability amplitudes
\begin{equation}
    \langle\bra{M^{A*}_i}\otimes\langle\bra{N^{B*}_j} ... \ket{w},
\end{equation} from which the probability is obtained through the squared modulus, agreeing with~\eqref{processMatrixProb}.

The process matrix representing a channel $\mathcal{C}$ from Bob to Alice is given by the transpose of its Choi map $(C^{B_OA_I})^T$ according to equation \eqref{ChannelMatrix}. If the channel is the identity, this turns out to be simply $\ket{\mathds{1}}\rangle \langle \bra{\mathds{1}}^{B_OA_I},$ by definition~\eqref{ChoiJNew}. We can construct the process matrix for the quantum switch based on this. The situation in which Alice receives a system in state $\ket{\psi}$ and, after realizing her operation, sends it to Bob, who then sends it to Charlie is represented by the process vector $\ket{\psi}^{A_I}\otimes \ket{\mathds{1}}\rangle^{A_O B_I}\otimes\ket{\mathds{1}}\rangle^{B_O C_I^t}$. Thus, the process matrix of the quantum switch, when the control system is in state $1/\sqrt{2}(\ket{0}+\ket{1})$ is given by $\ket{w}\bra{w}$, with
\begin{equation}\label{QSprocess}
\ket{w}= \frac{1}{\sqrt{2}}\left(\ket{\psi}^{A_I}\ket{\mathds{1}}\rangle^{A_O B_I}\ket{\mathds{1}}\rangle^{B_O C_I^t}\ket{0}^{C_I^c}+ \ket{\psi}^{B_I}\ket{\mathds{1}}\rangle^{B_O A_I}\ket{\mathds{1}}\rangle^{A_O C_I^t}\ket{1}^{C_I^c}\right).
\end{equation} In fact, if we act on $\ket{w}$ with some operations $\mathcal{A}$ and $\mathcal{B}$ represented by operation vectors, we recover the entangled state which characterizes the switch:
\begin{align}        
\langle\bra{\mathcal{A}^*}^{A_IA_O}\otimes \langle\bra{\mathcal{B}^*}^{B_IB_O}\ket{w} &=\sum_j \bra{j}^{A_I}\otimes\bra{j}^{A_O}\mathcal{A}^{*\dagger}\otimes\sum_k \bra{k}^{B_I}\otimes\bra{k}^{B_O}\mathcal{B}^{*\dagger} \nonumber
\\ 
&\qquad \times \frac{1}{\sqrt{2}} \left[\ket{\psi}^{A_I}\sum_m \ket{m}^{A_O}\ket{m}^{B_I}\sum_n\ket{n}^{B_O}\ket{n}^{C^t_I}\ket{0}^{C_I^c} + \cdots \right] \nonumber
\\
&=\frac{1}{\sqrt{2}}\sum_{jmn}\braket{j|\psi}^{A_I}\bra{j}^{A_O}\mathcal{A}^{T}\ket{m}^{A_O}\bra{m}^{B_O}\mathcal{B}^{T}\ket{n}^{B_O}\ket{n}^{C_I^t}\ket{0}^{C_I^c} + \cdots \nonumber
\\
&= \frac{1}{\sqrt{2}}\sum_{jn}\bra{j}^{A_O}\left(\mathcal{B}\mathcal{A}\right)^T\ket{n}^{B_O}\braket{j|\psi}^{A_I}\ket{n}^{C_I^t}\ket{0}^{C_I^c}+ \cdots \nonumber
\\
&= \frac{1}{\sqrt{2}}\left[\mathcal{B}\mathcal{A}\ket{\psi}^{C_I^t}\ket{0}^{C_I^c} + \mathcal{A}\mathcal{B}\ket{\psi}^{C_I^t}\ket{1}^{C_I^c}\right].  \label{QSstate}
\end{align} 
The quantum switch is sometimes treated as a 4-partite process, where an agent is introduced to prepare the initial state~\cite{Oreshkov2019timedelocalized,Vilasini:2022ist,Rubino}. In our case, the initial state information is contained in the process, and only the operations of A, B and C are taken as inputs. 
\begin{figure}[ht]
\centering
    \includegraphics[scale=0.9]{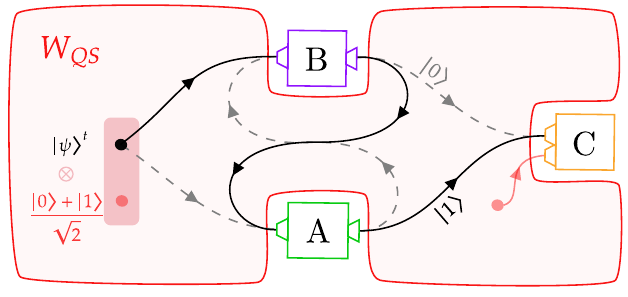}
    \caption{Sketch of the tripartite process for the quantum switch~(\protect\ref{QSprocess}). Alice's and Bob's laboratories are represented by boxes A and B, which have input and output Hilbert spaces. Charlie's lab, represented by box C, can only make measurements with no output systems.}
    \label{fig:QSprocess}
\end{figure}

Because the QS matrix is a projector, it cannot be written as a non-trivial probabilistic mixture of other different processes and, still, it allows distinct orders for the actions of the agents. Then, as argued in~\cite{Giarmatzipaper}, the switch cannot satisfy the extended tripartite notion of causal separability for this case, $W^{sep}=q W^{A\prec B \prec C}+ (1-q)W^{B\prec A \prec C}$. In reference\cite{Araujo2015}, the causal non-separability of the matrix is also demonstrated through the use of a causal witness, an operator $S$ such that $\operatorname{Tr}(SW)$ is non-negative for all causally separable matrices $W$. They show it is possible to construct a causal witness such that $\operatorname{Tr}(SW_{switch})$ is negative\footnote{It is possible to construct a causal witness for any causally non-separable process, as also showed in the reference.}. The authors demonstrate that, despite being causally non-separable, the quantum switch cannot violate any causal inequality.

This is a concrete description of a process with indefinite order. We can imagine that, if causally ordered process matrices are represented by a circuit diagram, this one could be drawn as a quantum superposition of two of them. Fig.~\ref{fig:QSprocess} illustrates that, mimicking diagrams made for usual quantum networks with definite order. 

\subsection{Implementations and discussion}

The quantum switch is claimed to have been realized experimentally\cite{ReviewExp,Rubino,Procopio,Goswami,Wei,Taddei,RubinoAgain}. The implementations follow the general idea represented in Fig.~\ref{fig:QSprocess}: they use the fact that we can put systems in quantum superposition of paths, for example make photons pass through a beam splitter and suffer operations in distinct orders depending on the path. 
\begin{figure}[ht]
\centering
\includegraphics[scale=0.37]{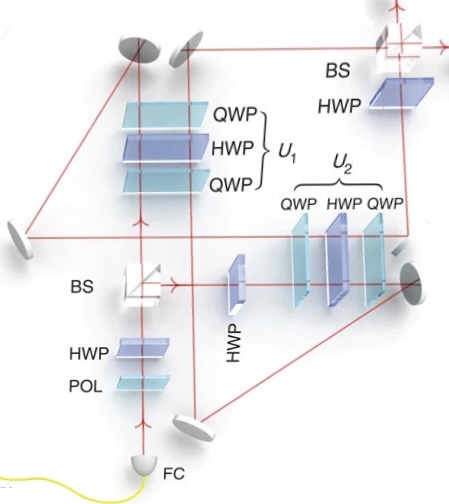}
\caption{Quantum switch with path as control: In this optical implementation of a quantum switch, the control is the spatial degree of freedom of the photon and the target is its polarization. A single photon is sent to a 50/50 beamsplitter which prepares the spatial qbit in a superposition of states $\ket{0}:$ transmitted and $\ket{1}:$ reflected. Operations $U_1,U_2$ are each implemented with three waveplates acting on the polarization. In path $\ket{0}$ the photon encounters operation $U_1$ before $U_2$ in the interferometer and the opposite happens for the other path. Note that the system always passes through the same waveplates, but hits them at different spots depending on the path. Figure from~\protect\cite{Procopio}.}
\label{fig:ExpQSProcopio}
\end{figure}

In the photonic implementations, the control and target qbits are encoded in two different degrees of freedom of the same photon. In reference \cite{Procopio}, for instance, the target system is the polarization of the photon. Therefore, Alice's and Bob's operations are made on that space, while the control is the path. See Fig.~\ref{fig:ExpQSProcopio}.

In reference \cite{Goswami}, the control qbit is the polarization while the target is the transverse spatial mode of the photon. In this case, polarizing beam splitters are used to orientate the photon along distinct paths depending on polarization, as showed in Fig.~\ref{fig:ExpQSGoswami}. In this implementation, the authors were able to make each operation be applied always in the same location in space, differently from the one in Fig.~\ref{fig:ExpQSProcopio}, where the system hits each plate in a different region of the plates depending on the path.

\begin{figure}[ht]
\centering
\includegraphics[scale=0.3]{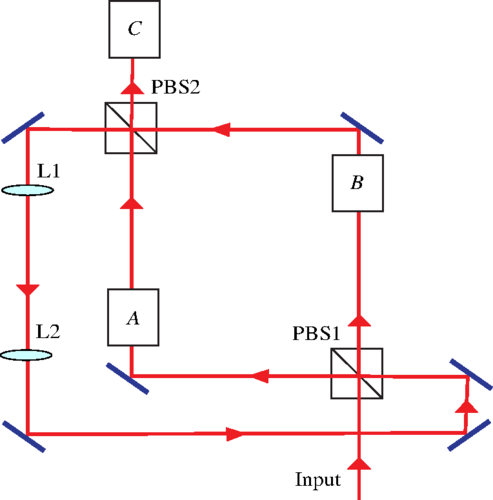} \hspace{1.5cm}
\includegraphics[scale=0.8]{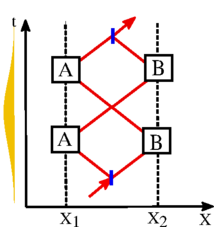}
\caption{Quantum switch with polarization as control: schematic of the protocol in reference~\protect\cite{Goswami}. A polarizing beamsplitter (PBS1) is used in order for the photon to travel through each path depending on the control state. The setup is arranged so that, when the beams are first combined in PBS2, they all take the path back to PBS1. Operations done at boxes A and B act on the spatial mode of the photon and are realized in distinct orders depending on the path. A polarization measurement in the diagonal basis is realized at C after the photon has gone through the interferometer twice. The figure in the right is a spacetime diagram with the events in which operations are applied in the protocol. Figures from~\protect\cite{Goswami}.}\label{fig:ExpQSGoswami}
\end{figure}

Still, as Fig.~\ref{fig:ExpQSGoswami} shows, we can draw the spacetime diagram for the protocol. It involves 4 spacetime events, as opposed to 2 operational events being considered in the mathematical description of the quantum switch, one for Alice's and another for Bob's operation. In reference~\cite{Procopio}, it is argued that these experiments should be considered genuine implementations of indefinite causal order because \emph{``any attempt to physically distinguish the two times at which a photon can pass through a gate would reveal which-way information and thus destroy the
interference. The results of the experiment confirm that such information is not available anywhere and that the interpretation of the experiment in terms of four, causally-ordered events cannot be given any operational meaning''}. This argument is strongly opposed in reference~\cite{Voji}, where the authors devise a non-destructive observable that could make the distinction in principle without revealing which-way information. The authors also show it is possible to describe these protocols in the language of the process matrix formalism as a process with 5 inputs, A, A', B, B' and C, one for each spacetime event, by means of introducing the vacuum as a possible input for the laboratories. The process matrix constructed in this way is causally separable and hence can be written as a circuit. The authors argue this description agrees more with the experimental setups and that there is an operational difference from a genuine quantum switch with 2 events.

The discussion of what constitutes an actual laboratory implementation of the quantum switch or, more broadly, indefinite order has been repeatedly brought up in the literature, for instance in \cite{Voji,Oreshkov2019timedelocalized,Vilasini:2022ist}. The fact that quantum mechanics allows for superpositions of time evolutions, which can be achieved through space or time delocalization of systems, makes it possible for the statistics of the quantum switch process to be reproduced in laboratory. The quantum switch and its variations seem to provide advantages when compared to standard quantum computation~\cite{Araujo,Taddei,Colnaghi_2012,Guerin,Wei,ChirDisc,EnhancedChiribella}, although some results might need revision~\cite{desbancou}. Thus, the experiments appear to be significant for quantum information processing. On the other hand, we clearly cannot conclude they are totally incompatible with causal structure, at least not in the sense of incompatibility with the existence of underlying spacetime, since we obviously can describe the experiments correctly using quantum mechanics on top of a definite Minkowski spacetime. This is where the information-theoretic and spacetime notions of causal structure collide. In a sense, if an event A is characterized by the application of an operation inside Alice's laboratory whose input space is an unusual Hilbert space, i.e. that of a temporally delocalized system, the order between the ‘events’ in the realizations of the switch really may be considered indefinite~\cite{Oreshkov2019timedelocalized}. That notion deviates from the initial closed laboratories proposal, with spacetime regions that cannot communicate, in favor of the informational approach. From the spacetime structure point of view, each localized spot is a different event and we are able to see the branches of the superposition as distinct pairs of localized operations (A,B') and (B,A'), with all four of them being part of a causal structure.

The first approach, which considers that the experimental switches are protocols with indefinite order, has gained traction in the quantum information community. A lot of theoretical work has been produced on further classification of processes, conjectures on which of them can have a physical realization in the same sense of the photonic quantum switch, and whether it is possible to violate causal inequalities in that context~\cite{Giarmatzipaper,Purves_2021,pagewooters,nogo_Costa_2022,Vilasini:2022ist,Barrett_2021}. 

If the goal is to study structures due to gravity and quantum spacetime, the second approach seems more appropriate. That is, we would like to affirm that a protocol is causally ordered whenever its probabilities are compatible with quantum mechanics on a classical spacetime. However, doing this is not that simple. In an optical-table experiment, we know \emph{a priori} that the Minkowski spacetime description accounts for the predictions consistently and that it is not an artificial solution, since it comes from a theory that is well-tested in other contexts. With that in mind, we may consider that a genuine explanation of the protocol involves 5 spacetime events without need to break the paradigm of relativity, as done in~\cite{Voji}, instead of 3 operational ones. But if we are trying to gather information from an unknown and maybe indefinite spacetime, we have no grounds to make a `genuine' assignment of the events of interest. For all we know, even the more concrete concept of laboratories could become nebulous if a quantum spacetime was considered. We fatally need an operational way to define the events/laboratories, in addition to the already operational way we analyze the set of probabilities that comes from them. That should clarify the confusion generated by the quantum switch and provide a way to understand protocols in terms of processes without a lot of ambiguity. It is indeed reasonable to search for an operational and relational approach to defining events, since that is the spirit in which Einstein first formulated relativity -- events are defined by the intersection of the world-lines of pairs of (classical) particles.
A more detailed attention to these matters~\cite{quantumevents,Esteban2020quantum, pagewooters,Vilasini:2022ist} may enrich the process matrix formalism when it comes to foundational aspects.

\chapter{Clocks on curved spacetimes}\label{Chap Clocks}
Despite the fact that we successfully constructed a formalism for local quantum mechanics that does not demand definite causality, extracting meaningful information about causal structure from statistics of experiments can still be challenging. That is because the correspondence between process matrices and associated physical protocols, when they exist, is not straightforward. Further specification on the nature of the
events of application of the operations can be useful for investigating in which sense causality is being violated in a process with indefinite order. For instance, in table-top experiments, we could establish that something is characterized as a single application of an instrument only if it happens at a definite location and time, as measured by a ruler and a clock available at the laboratory. Violating definite order under this requirement would appear meaningful, indicating conflict with the usual description with a background space and global time. One can think of doing the same sort of specification in curved spacetimes, and ultimately also try to define events/laboratories on quantum spacetimes.

Here, we will present the formalism for classical and quantum ideal clocks on curved spacetimes according to reference~\cite{Zych}. This formalism can be employed to operationally specify the events for processes when gravity is involved, and will be used to formulate the gravitational quantum switch~\cite{tbell} in the next chapter. Studying them is also a basic introduction to classical and quantum mechanics on curved spacetimes. The formalism of quantum clocks explains the behavior of systems such as the atomic and optical clocks with which experiments were realized to test time dilations predicted by special and general relativity~\cite{Hafele1972AroundtheWorldAC,OpticalClocksExp}. The formalism also includes the possibility for a clock to be delocalized, as a particle following a quantum superposition of multiple paths. In the last chapter, we use this to describe a test of quantum mechanics on a curved classical spacetime.

\section{Free relativistic particle}
Let us first summarize the general relativistic lagrangian and hamiltonian treatment of a particle on a curved spacetime~\cite{gr3,Carroll}, since it will be used in the description of clocks.

A free point-particle with mass $m$ in a spacetime with metric $g_{\mu\nu}$ moves along a curve $x^\mu$ according to the action
\begin{equation}\label{particleAction}
S=-mc^2\int d\tau=-mc\int \sqrt{-g_{\mu\nu}\dot{x}^\mu\dot{x}^\nu} d\lambda=\int L d\lambda,
\end{equation}
where $\dot{x}^\mu=dx^\mu/d\lambda$, $c$ is the speed of light, $\tau$ is the proper time, with $d\tau=1/c\sqrt{-g_{\mu\nu}dx^\mu dx^\nu}$, and $\lambda$ is a generic parametrization for the curve $x^\mu(\lambda)$. Applying the variational principle to this action, which is simply the relativistic interval, returns geodesic curves for the particle.

In fact, defining $f:= g_{\mu\nu}\dot{x}^\mu\dot{x}^\nu=\dot{x}^\mu\dot{x}_\mu$, the extremization of the action functional is expressed by
$$\delta (-mc\int\sqrt{-f}d\lambda)=\frac{1}{2} mc \int (-f)^{-1/2}\delta(f) d\lambda=0.$$ 
If we fix the parametrization $\lambda=\tau$ after the calculation of the variation, we may substitute $f|_{\lambda=\tau}$ outside the $\delta$ symbol~\cite{Carroll}. This happens to be the norm of the 4-velocity $f|_{\lambda=\tau}=U^\mu U_\mu=-c^2$, simplifying the expression to 
$$\frac{m}{2} \int \delta(f) d\tau=0 \implies \delta\left(\frac{m}{2} \int \dot{x}^\mu\dot{x}_\mu d\tau\right)=0.$$ Thus, for curves parametrized by their proper times, we can use the lagrangian 
\begin{equation}\label{L'}
    L'= \frac{m}{2}\frac{dx^\mu}{d\tau}\frac{dx_\mu}{d\tau},
    \end{equation}with $S'=\int L'd\tau$, to solve the original $L=-mc\sqrt{-f}$.

The lagrangian in \eqref{L'} is easier to deal with, has the same canonical momentum as the former and also leads to the geodesic equation. For the original lagrangian \eqref{particleAction}, we can readily verify that the expression for the hamiltonian
\[
H= p_\mu \dot{x}^\mu - L = m \frac{d{x}_\mu}{d\tau}\dot{x}^\mu -L
\]
identically vanishes. This is not an unfortunate accident, but a direct consequence of its reparametrization invariance. One can generally redefine this type of lagrangian by adding constraints to get relational information about the system, as explored in~\cite{Kiefer}. The expression for $L'$, on the other hand, was already obtained from fixing a parametrization. We can verify the associated hamiltonian does not vanish:
$$H' =m \frac{dx_\mu}{d\tau} \frac{dx^\mu}{d\tau}-L'=m \frac{dx_\mu}{d\tau} \frac{dx^\mu}{d\tau}-\frac{m}{2}\frac{dx_{\mu}}{d\tau}\frac{dx^{\mu}} {d\tau}=\frac{m}{2}\frac{dx_\mu}{d\tau} \frac{dx^\mu}{d\tau}=L'.$$ Since the parametrization is $\tau$, we conclude that $H'=m U^\mu U_\mu/2=-mc^2/2$. Another way to write that is $p_\mu p^\mu=-m^2c^2$. Indeed, fixing $\tau$ as the parameter turns the conjugate momenta of the problem into the actual 4-momentum components and its norm becomes a conserved quantity of the problem. The Hamilton equations, when combined, return again the geodesic equation~\cite{gr3}.

The formulation above is covariant and includes the time coordinate. It all works fine, but it is interesting to formulate a 3D hamiltonian system as well. It makes the dynamics description meaningful from the point of view of a fixed observer, who perceives motion with respect to coordinate time $x^0=ct$. For this, we can consider the invariant action $\int p_\mu dx^\mu$ decomposed in time and space parts:
\begin{equation}\label{S3D}
\int p_\mu \frac{dx^\mu}{dt} dt=\int \left(c p_0+g_{ij}p^i\frac{dx^j}{dt}\right)dt=\int L_{3D}dt.
\end{equation}
This action is simply $2S'$, but we are now using $t=x^0/c$ instead of $\tau$ as parameter and our coordinates are only the spatial $x^i$, with conjugate momenta $p_i$, as in classical mechanics. The form of this lagrangian is explicit, returning the hamiltonian $H_{3D}=g_{ij}p^i\dot{x}^j-L_{3D}=-c p_0$. This is suggestive, since it is the expression for energy in flat coordinates, $E=c p^0=-cp_0$. However, that interpretation does not hold in a general context, nor does it need to for a valid hamiltonian.

We can now find a formula for $H_{3D}=-c p_0$. Since we do not use $\tau$ as parameter, the norm of the 4-momentum is not necessarily fixed as before. So we can solve for \eqref{S3D} adding the extra constraint $H'\equiv p_\mu p^\mu/(2m)=-mc^2/2$ for the solutions, therefore using information from the 4D formulation in the 3D description. The reader is referred to~\cite{gr3} for formal justification of this procedure. In the end, to find an expression for $H_{3D}$, we only need to manipulate the expression
\begin{align*}
   p^\mu p_\mu + m^2c^2=0 \quad \implies \quad &g^{00}(p_0)^2 +g^{0i}p_0p_i+ g^{i0}p_ip_0 + g^{ij}p_ip_j + m^2c^2=0
    \\
    &p_0= \frac{- 2 g^{0i}p_i - \sqrt{4(g^{0i}p_i)^2-4g^{00}(m^2c^2+g^{ij}p_ip_j)}}{2 g^{00}},
\end{align*}
where we used that the metric is symmetric, $g^{0i}=g^{i0}$. Then, we generally have
\begin{equation}\label{H3D}
       H_{3D}=-c p_0=+c \frac{g^{0i} p_i}{g^{00}} +\sqrt{c^2\left(\frac{g^{0i}p_i}{g^{00}}\right)^2-\frac{(m^2c^4+c^2 g^{ij}p_ip_j)}{g^{00}}}.
\end{equation}

\section{Classical clocks}
A clock is a system that in some way displays the passage of time. We might look at a mechanical clock and say that an event happened at 9 o'clock if it coincides with a certain 90° configuration of the clock's pointers, which are constantly moving. Anything with dynamics really could be used as a clock in that context. When special and general relativity come to play, simultaneity is not an absolute notion anymore and time passes differently depending on the velocity of the reference frame and on whether the clock is in a region of gravitational influence~\cite{Carroll}. We can think of an ideal clock system in relativity as a point-particle with internal degrees of freedom. It is punctual so that time cannot pass differently in different parts of it and it witnesses the proper time elapsing through the change of its internal state. 

This serves as an effective theory to model non-punctual dynamical systems, like a bunch of interacting particles, that are sufficiently localized in space so that time does not perceptively run differently in different regions of the system. We will later discuss how the formalism meets this approximation.

\subsection{Hamiltonian} \label{subsecHamiltonian}
All through this chapter we will consider for simplicity a static, symmetric metric $g_{\mu \nu}$ such that $g^{i0}=0$, $i=1,2,3$ unless stated otherwise. Clocks are relativistic particles, but the presence of internal degrees of freedom requires a modification of the description above. Let us go to the particle's rest frame, represented by primed coordinates. We can also choose the time coordinate to be $c\tau$, which means we are fixing $g'_{00}=-1$. Then, we can write $p_\mu p^\mu =g'_{00}p'^0 p'^0=-(p'^0)^2$. In the formulation above, this equals $-m^2c^2$ and thus we can understand the quantity $cp'^0= mc^2$ as the energy of the particle in the rest frame, agreeing with special relativity, as it locally should. So, if we want a particle with internal dynamics, we should add an energy term to this, namely the internal hamiltonian:\begin{equation}\label{Hrest} cp'^0=mc^2+H_{int}\equiv H_{rest}.\end{equation}
Indeed, this equals the 3D hamiltonian of the last section in these coordinates, because $-cp'_0=-cg'_{00}p'^0=cp'^0$. We can use it to find the expression for generic coordinates. Note that the the modification of the hamiltonian also modifies the momentum, which now has the norm $p_\mu p^\mu= -(p'^0)^2=-H_{rest}^2/c^2$. To arrive at the new general form of the hamiltonian, we only need to replace the old momentum norm $-m^2c^2$ by this new norm in \eqref{H3D}, resulting in \begin{equation*} H=c \frac{g^{0i} p_i}{g^{00}} +\sqrt{c^2\left(\frac{g^{0i}p_i}{g^{00}}\right)^2-\frac{(H_{rest}^2+c^2 g^{ij}p_ip_j)}{g^{00}}}.\end{equation*} Since the metric is diagonal, $g^{i0}=0$, and static, $(g^{00})^{-1}=g_{00}$, this simplifies to
\begin{equation}\label{ClockHamiltonian}
      H= \sqrt{-g_{00}\left(c^2 p_ip^i+ H_{rest}^2\right)}.
\end{equation}
The squared energy term was replaced by the squared $H_{rest}$, expressing the mass-energy equivalence. If an opaque box full of moving particles is seen by an observer, it will appear to have a mass bigger than the sum of all masses. The kinetic energy $E$ of the particles will add an effective mass of $E/c^2$ which can be really perceived in its inertia. Any box comprising any kind of energetic system will behave like this. In the same way, we are adding to this system the extra effective mass of $H_{int}/c^2$ in order to account for the internal dynamics.

\subsection{Lagrangian}
The system described by the hamiltonian \eqref{ClockHamiltonian} has external coordinates $(x^i,p^i)$ as well as internal ones of which $H_{int}$ is a function. We will consider that the internal dynamics has $N$ degrees of freedom, and denote by $(q_k,\omega_k)$, $k=1,\dots,N$ the generalized coordinates and conjugate momenta defined with respect to the rest frame. The expression in \eqref{ClockHamiltonian} reduces to $H_{rest}$ in the rest frame and the lagrangian can then be written as
\begin{equation}
    L_{rest} \equiv \sum_k \dot{q}'_k\omega_k - H_{rest}.
\end{equation}
The Hamilton equation for $\dot{q}_k$ is
\begin{equation}\label{qdot}
    \dot{q}_k=\frac{\partial H}{\partial \omega_k}=-\frac{g_{00} H_{rest}}{H}\frac{\partial H_{rest}}{\partial \omega_k},
\end{equation}
which reduces to $\dot{q}'_k=\frac{\partial H_{rest}}{\partial \omega_k}$ in the rest frame.

The lagrangian associated to the hamiltonian in \eqref{ClockHamiltonian}, includes both internal and external coordinates and momenta:
\begin{equation}\label{preClockLagrangian}
    L=\dot{x}^ip_i + \sum_k \dot{q}_k\omega_k - H.
\end{equation}
From the Hamilton equation for $\dot{x}^i$, we have
\begin{equation}\label{xdot}
\dot{x}^i=\frac{\partial H}{\partial p_i}=\frac{(-1)}{2H}\frac{2 c^2 p^i}{g^{00}} = -\frac{g_{00} c^2 p^i}{H}.
\end{equation}
Substituting \eqref{qdot} and \eqref{xdot} into \eqref{preClockLagrangian}, we get
\begin{align}
     L&=-g_{00}\left(c^2 p_ip^i+ \sum_k H_{rest}\frac{\partial H_{rest}}{\partial \omega_k} \omega_k\right) H^{-1} - H \nonumber
     \\
     &=\frac{-g_{00}\left(c^2 p_ip^i+ H_{rest}\sum_k \dot{q}'_k \omega_k\right)+g_{00}\left(c^2 p_ip^i+ H_{rest}^2\right)}{H} \nonumber
     \\
     &=-\frac{g_{00} H_{rest}}{H} L_{rest},
\end{align}
where $\dot{q}_k^{\prime}$ is the rest frame version of \eqref{qdot}. The term multiplying $L_{rest}$ takes a more familiar form if we get $H_{rest}$ by isolating it in the formula \eqref{ClockHamiltonian} and get $H$ from \eqref{xdot}:
\begin{align*}
   -\frac{g_{00} H_{rest}}{H}& = -\frac{g_{00} \left[-g^{00}H^2 - c^2p_ip^i\right]^{1/2}}{H}
   \\
   &=\left[(-g_{00})^2\left(-g^{00}-\frac{c^2p_ip^i}{(g_{00}c^2)^2p_ip^i/(\dot{x}_i\dot{x}^i)}\right)\right]^{1/2}
   \\
   &=\frac{1}{c}\left[-g_{00}c^2-\dot{x}_i\dot{x}^i\right]^{1/2} =\frac{1}{c} \sqrt{-g_{\mu \nu}\dot{x}^\mu\dot{x}^\nu}.
\end{align*}
Given the parameter is $t$, we used a choice of coordinates such that $x^0=ct$ to get to the last equality. The result is simply the derivative of the proper time $\dot{\tau}$, which corresponds to the first lagrangian we introduced for a relativistic particle \eqref{particleAction}, except for the multiplying constant $-mc^2$. The final lagrangian is then
\begin{equation}\label{ClockLagrangian}
    L=\frac{L_{rest}}{c}\sqrt{-g_{\mu \nu}\dot{x}^\mu\dot{x}^\nu} = L_{rest} \dot{\tau}.
\end{equation}
In fact, one can verify that $L_{rest}=-mc^2 + L_{int}$ by doing the legendre transformation on \eqref{Hrest}. Therefore, the generalization is also clear in the lagrangian description, with the internal degrees of freedom being added to the energy terms. 

Since all derivatives are taken with respect to a parameter $t$, the action for this lagrangian is
\begin{equation}
    S=\int L dt=\int L_{rest}\frac{d\tau}{dt}dt = \int L_{rest} d\tau,
\end{equation}
which suggests that the proper time works as the time for internal dynamics regardless of the parameter $t$ chosen for constructing the formalism.

\subsection{Routhian and the non-relativistic limit}\label{subsecRouthian}
It is convenient that we introduce a third picture for the dynamics of a clock. The routhian is a function between the lagrangian and the hamiltonian. It consists of a partial legendre tranformation of the lagrangian with respect to only some degrees of freedom~\cite{goldsteinClassical,Zych}. In our case, those will be the internal degrees. We define the routhian of our system as
\begin{equation}
    R:= \sum_k \omega_k \dot{q}_k - L.
\end{equation}
Now, due to the form we found for the lagrangian in \eqref{ClockLagrangian}, we can write
\[
\omega_k \equiv \frac{\partial L}{\partial \dot{q}_k }= \frac{\partial L_{rest}}{\partial(dq_k/d\tau)} ,
\]
and the routhian becomes
\begin{equation}\label{ClockRouthian}
    R=\sum_k \frac{\partial L_{rest}}{\partial(\frac{dq_k}{d\tau})} \frac{dq_k}{d\tau} \dot{\tau} - L_{rest}\dot{\tau} = H_{rest} \dot{\tau}.
\end{equation}

By definition, the routhian acts as a lagrangian for external degrees of freedom and as a hamiltonian for the internal ones. In particular, given an \emph{internal} observable $A_{int}$, we may write its evolution with respect to $t$ as a Poisson bracket with the routhian:
\begin{equation}\label{ObsEvol}
    \dot{A}_{int}=\left\{A_{int},R\right\} = \left\{A_{int}, H_{rest}\right\}\dot{\tau}.
\end{equation}
We can then conclude immediately that
\[
\frac{d A_{int}}{d \tau}= \{A_{int}, H_{rest}\} ,
\]
which confirms that the internal degrees of freedom of a clock traveling along a worldline evolve according to the proper time elapsed from the initial to the final point. For instance, if $dA_{int}/d\tau=:v_{rest}$ is constant, we get
\begin{equation*}
    A_{int}(t_f)-A_{int}(t_i)=\int \dot{A}_{int} dt=\int \frac{dA_{int}}{d\tau}\dot{\tau}dt= v_{rest} \tau,
\end{equation*}
where $\tau$ is the proper time elapsed between the initial and final points of the worldline we are integrating over. Similarly, if two clocks for which we can measure this observable travel on different worldlines parametrized from $t_i$ to $t_f$ and $v_{rest}$ is constant for both, the difference of proper time between the paths can be measured by comparing how much $A_{int}$ changed in each of them: $ \Delta A_{int (2)}-\Delta A_{int (1)} =v_{rest} (\tau_2-\tau_1).$

 For an observer with coordinate time $t$, the internal evolution of the system is perceived with a time dilation of $\dot{\tau}$. Thus, the internal Hamiltonian from their perspective, which is actually the routhian, gets ``redshifted'' and looks like $H_{rest}\dot{\tau}$. We can see that all this was developed without need to specify the actual form of the internal dynamics. This shows that there is an universality of time dilation in this formalism. These systems can be thought of as ideal clocks because the time dilation exhibited when they follow a trajectory does not depend on their nature, but only on the trajectory and the underlying spacetime. 

In particular, when we go to the non-relativistic low-energy limit, these systems represent ideal clocks for classical mechanics. One can verify this by considering the metric for a spherical mass such as the Earth, the Schwarzchild metric to the lowest order\cite{Carroll}: $g_{00}\approx -(1+2\Phi(x)/c^2)$, $g_{ij}\approx \delta_{ij}$, where $\Phi(x)=-GM_{Earth}/x$ is the newtonian gravitational potential. The routhian \eqref{ClockRouthian} for this metric to the order $1/c^2$ reads
\begin{equation}\label{leRouthian}
    H_{rest}\dot{\tau} \approx H_{rest}\left(1-\frac{\vec{v}^{\,2}}{2c^2}+\frac{\Phi(x)}{c^2}\right)=:R_{le}.
\end{equation}
This is the low-energy limit, but it is not yet entirely non-relativistic, since the terms $H_{rest}[1-\vec{v}^{\,2}/(2c^2)]$ and $H_{rest}[1-\Phi(x)/c^2]$ give rise to special relativistic and gravitational time dilations respectively, as we can attest by the presence of $\dot{\tau}$ in the evolution of internal observables~\eqref{ObsEvol}. And so much that is true that the associated hamiltonian, $H = R +\dot{x}_i p^i$, is
\begin{equation}\label{leHamiltonian}
   H_{le}:= H_{rest} + \frac{1}{2}\frac{c^2 \vec{p}^{\,2}}{H_{rest}} + H_{rest} \frac{\Phi(x)}{c^2} = mc^2+H_{int}+ \frac{1}{2}\frac{ \vec{p}^{\,2}}{m+\frac{H_{int}}{c^2}} +\left(m+\frac{H_{int}}{c^2}\right) \Phi(x),
\end{equation}
which still does not fit the classical newtonian description. Here we used that $\dot{x}^i=p^i/M$ where $M=H_{rest}/c^2$, as well as \eqref{Hrest}.

We are missing one more limit that should be taken due to the internal dynamics. If we decompose the routhian \eqref{leRouthian} as $$R_{le}=H_{rest}+ mc^2\left(-\frac{\vec{v}^{\,2}}{2c^2}+\frac{\Phi(x)}{c^2}\right) + (H_{rest}-mc^2)\left(-\frac{\vec{v}^{\,2}}{2c^2}+\frac{\Phi(x)}{c^2}\right),$$ we can see that only the last term contributes to time dilation felt by the internal degrees of freedom. Note that $H_{rest} -mc^2\equiv H_{int}$. Therefore, to get to the non-relativistic limit, we should also demand for $H_{int}$ to be much smaller than the non-dynamical energy $mc^2$. We can interpret this as demanding slow internal evolution. If we neglect terms of order $H_{int}/c^2$, the routhian looses the last term and becomes 
\begin{equation}
    R_{le,nr}=mc^2+H_{int}- m\frac{\vec{v}^{\,2}}{2}+m\Phi(x).
\end{equation} In this case, the evolution of an internal observable will, for any coordinates, be given by the bracket $\{A_{int},H_{int}\}$, as expected in a global time theory. We also have that the hamiltonian in~\eqref{leHamiltonian} becomes
\begin{equation}
H_{le,nr}=mc^2+H_{int}+\frac{\Vec{p}^{\,2}}{2m}+ m\Phi(x),
\end{equation}
the classical newtonian expression. Therefore, the non-relativistic limit for a clock comprises the low energy limit as well as slow internal dynamics. In fact, the existence of interesting effects due to gravity in the regime of low energies that cannot be explained by newtonian physics is part of the reason we are studying clocks.

\subsection{Effective dynamics of a localized N-particle system}
As said earlier, the formalism of ideal clocks can serve as a model for a localized system of interacting particles traveling through spacetime, a more realistic approach to physical clocks. Here, we describe an example of such a system and delimit for which approximations the clock formalism is a good model for it. This can be regarded as a non-formal derivation of the point-particle clock framework.

Let us consider a system of $N$ particles that interact electromagnetically through the 4-potential $A_{\mu}$ on a spacetime with metric $g_{\mu \nu}$. For $n$ going from $1$ to $N$, their masses will be denoted by $m_n$ and their charges by $e_n$. The action for the system is given by
\begin{equation}
S_N=\int \sum_{n=1}^N \left( -m_n c^2 \frac{d\tau_n}{d\lambda} +e_n A_{\mu}(x_n)\frac{d x^\mu_n (\lambda)}{d\lambda}\right) d\lambda \equiv \int L_N d\lambda,
\end{equation}
where $\tau_n$ is the proper time of particle $n$. Dots will indicate derivatives with respect to $\lambda$. This lagrangian is the sum of the free particle lagrangians and the interaction terms $e_n A_\mu \dot{x}_n^\mu$. When we extremize the action, these terms generate  minimally coupled generalized Lorentz force expressions (see, for instance, references~\cite{jackson,Landau1,Landau2}). In flat spacetime, the evaluation of the spatial components returns the Lorentz force and the time component gives out the expression for electromagnetic work.

Let us fix the arbitrary parameter $\lambda$ as $ct=:x^0$, with $t$ being the coordinate time measured by a clock positioned at the origin with respect to which the coordinates $x_n ^\mu$ were defined. If we consider a worldline $Q^\mu(t)$, we can write $c\dot{\tau}=\sqrt{-g_{\mu\nu}(Q)\dot{Q}^\mu \dot{Q}^\nu}$. If we choose a coordinate system (indicated by primed symbols) such that $\dot{Q}'^i=0$, $i= 1, 2, 3$, and $Q'^0$ is the proper time along such a worldline, the lagrangian can be rewritten as:
\begin{equation}\label{NLagragian}
L_N=\sum_{n=1}^N \left(- m_n c \sqrt{-g'_{\mu \nu}\frac{d x'^{\mu}_n}{d\tau}\frac{d x'^{\nu}_n}{d\tau}} +e_n A'_{\mu}\frac{d x'^\mu_n}{d\tau}\right)\frac{1}{c}\sqrt{-g_{\mu\nu}(Q)\dot{Q}^\mu \dot{Q}^\nu}.
\end{equation}
This looks like \eqref{ClockLagrangian} because it has the factor $\dot{\tau}$. But, for this to be a clock lagrangian, $\tau$ has to be the proper time felt by the system of particles. That is, the curve $Q^\mu$ must be the worldline of the system's center of mass and the primed coordinates $Q'^\mu$ must correspond to the frame in which the total momentum of the system is zero. 

The canonical momentum for the nth particle is
\begin{equation}
    p_{ni}:=\frac{\partial L_N}{\partial \dot{x}_n^i}=\frac{-m_n c(-g_{i\nu}\dot{x}_n^\nu-g_{\nu i}\dot{x}_n^\nu)}{2 c \dot{\tau}}+e_nA_i =m_n\frac{\dot{x}_{n i}}{\dot{\tau}}+e_nA_i.
\end{equation}
The \emph{linear momenta} correspond to the terms $m_n\dot{x}_{ni}/\dot{\tau}=:P_{ni}$. To find the frame of zero total momentum, we should equate the vectorial sum of the linear momenta of all particles to zero, but this is not possible to do in general. Each $\vec{P}_n=(P_n^1,P_n^2,P_n^3)$ is part of a 1-form that lives in the cotangent space of the point the $n$-th particle is at, $x_n$. It does not make sense to sum vectors or 1-forms belonging to different spaces, unless we can somehow identify them with each other. For this, we consider that there exists a region $\Omega$ in the spacetime which contains $x_n$, $\forall n =1,\dots,N$, and such that we can make the approximation
\begin{equation} \label{gmunuApprox}
g_{\mu \nu}(x) \approx g_{\mu \nu}(y), \quad \forall x,y \in \Omega.    
\end{equation}
This guarantees that we can cover the region occupied by the particles with one common coordinate chart in which the metric takes the Minkowski form. From there, it is possible to define total momentum and arrive at a generally covariant notion of center of mass~\cite{Dixon}. We can go through an easier path as follows. Assuming a smooth spacetime, the condition above is always satisfied if the particles are sufficiently close to each other, leading to the alternative approximation:
\begin{equation}\label{xm=xnApprox}
    x_n^\mu \approx x_m^\mu \quad \forall n,m = 1,\dots,N.
\end{equation}
We define the total linear momentum as
\begin{equation}\label{totalMomentumN}
P_i := \sum_n P_{ni}= \sum_n \left[p_{ni}(x_n)-e_n A_i(x_n)\right],\end{equation} which is a well-defined sum because we are assuming \eqref{xm=xnApprox} that implies \eqref{gmunuApprox}. The center of mass coordinates are then defined by the condition $P'_i=\sum_n P'_{ni}= \sum_n \frac{\partial x^\mu}{\partial x'^i}P_{n\mu}=0$\footnote{Here, the $0$-th component of $P_n$ is defined analogously to the formula \eqref{totalMomentumN} with $i=0$, and we are applying the coordinate change rule for this 1-form.}. Validity of \eqref{gmunuApprox} guarantees the condition $P_i'=0$ is generally covariant. That is, if $P^i=\sum_n g^{ij}(x_n)P_{nj}$, 
\begin{align*}
P'^i=\sum_n P'^i_n &=\sum_n  \frac{\partial x'^i}{\partial x^\mu} g^{\mu\nu}(x_n) P_{n\nu}
\\
&=\sum_n  \frac{\partial x'^i}{\partial x^\mu} g^{\mu\nu}\frac{\partial x^\alpha}{\partial x^\nu} P_{n\alpha}=\sum_n  \left(\frac{\partial x'^i}{\partial x^\mu} g^{\mu\nu}\frac{\partial x'^j}{\partial x^\nu}\right)\left(\frac{\partial x^\alpha}{\partial x'^j} P_{n\alpha}\right) =\sum_n g'^{ij}(x_n) P'_{nj},
\end{align*}
and because of \eqref{gmunuApprox} we can take $g^{ij}(x_n)\approx g^{ij}(x_N)$, $\forall n$, out of the sum. Thus, we have
\begin{equation}\label{outofthesum}
    P'^i \approx g'^{ij}(x_N) \sum P'_{nj}=g'^{ij} P'_{j}=0.
\end{equation}

Since we are working with the stronger condition \eqref{xm=xnApprox}, we can take the worldline $Q^\mu$ to be simply the path of one of the particles (instead of using a more complicated definition for center of mass). We can choose, for instance, $Q^\mu=x_N^\mu$. Then, the primed coordinates for which $\dot{Q}^i=0$ can be interpreted as the center of momentum frame, the frame that accompanies the center of mass. An observer in this frame can only see the evolution of the interaction between the particles, while the system as a hole is not moving. Thus, it makes sense to interpret the term in \eqref{NLagragian},
\begin{equation}
    L_{rest}:= \sum_{n=1}^N \left(- m_n c \sqrt{-g'_{\mu \nu}\frac{d x'^{\mu}_n}{d\tau}\frac{d x'^{\nu}_n}{d\tau}} +e_n A'_{\mu}\frac{d x'^\mu_n}{d\tau}\right),
\end{equation}
as the ``rest frame'' lagrangian. Under the approximations discussed, the total lagrangian \eqref{NLagragian} reads \begin{equation} L_N \approx L_{rest} \dot{\tau},\end{equation} coinciding with the lagrangian for a clock. One can also verify that the hamiltonian description returns 
 \eqref{ClockHamiltonian} for the momenta $P^i$, which we leave for reference~\cite{Zych}.

One way to analyze the approximations quantitatively is to verify how much $P^\mu$ differs from $g^{\mu\nu}(x_N)P_\mu$ when the metric varies in the region of the particles. That is, how much sense does it make to take $g^{\mu \nu}$ out of the sum in \eqref{outofthesum}. The difference is given by
$$P^\mu -g^{\mu \nu}(x_N)P_\nu=\sum_n \left(g^{\mu\nu}(x_n)-g^{\mu\nu}(x_N)\right)P_{n\nu}.$$
We can quantify the approximation \eqref{gmunuApprox} by assuming that $$\exists K>0 \text{ such that } \left|g^{\mu\nu}(x_m)-g^{\mu\nu}(x_n)\right|\leq\frac{K}{4N}, \quad \forall \mu, \nu, m, n.$$ Thus, if we define $E_{max}:= \text{max}\{\left|P_{n\mu}(x_n)\right|; n=1,\dots,N;\mu=0,\dots, 3\}$, which quantifies the energy scale of the system, we have
\begin{equation}
\left|P^\mu -g^{\mu \nu}(x_N)P_\nu \right|=\left|\sum_n \left(g^{\mu\nu}(x_n)-g^{\mu\nu}(x_N)\right)P_{n\nu}\right| \leq K E_{max}. \end{equation}

Therefore, the formalism for clocks as point particles with internal degrees of freedom is a good model for this $N$-particle system, provided it describes a narrow worldtube in spacetime and has an energy bound. For smooth spacetimes, we will always be able to find a small enough volume in which to confine the particles so that $K$ is sufficiently small for the approximation to be good. 

\section{Quantum Clocks}
It is customary one turns directly to Quantum Field Theory (QFT) to treat quantum systems on curved spacetimes. General relativity encompasses effects due to velocities and curvature. If one disregards curvature, near-light velocities are the ones generating new effects compared to classical physics. So,
it is natural to think we might be dealing with a high-energy system, making the QFT treatment necessary. However, we can be interested in the regime for quantum systems on a curved spacetime that neither have high velocities nor are in a region of high gravitational energy. In that case, we still need to account for curvature, but it all should work somehow like usual quantum mechanics. Take, for example, quantum clocks that are able to feel time dilation at different heights on Earth~\cite{Hafele1972AroundtheWorldAC, OpticalClocksExp}. How are they described? In this section we discuss, based on references~\cite{LAMMERZAHL,Zych}, the process of first quantization of the formalism for clocks and arrive at a wave equation that includes time dilation effects in weak gravity.

As argued in the beginning of the chapter, the hamiltonian for a free relativistic particle comes directly from the norm of the 4-momentum $p_\mu p^\mu-m^2c^2=0$ and considering that $H=-cp_0$. The clock hamiltonian \eqref{ClockHamiltonian} is not so different. It is obtained from $p_\mu p^\mu-M^2c^2=0$, where $M=H_{rest}/c^2= m+H_{int}/c^2$ is the effective mass of the system. Let us first think of the case with no internal degrees of freedom. In a flat spacetime described by the Minkowski metric, we may use an approach analogous to that used for obtaining the quantum wave equation for a particle: we take the momenta $p_\mu$ into operators $-i\hbar\partial_{\mu}$\footnote{Here, and from now on, we use $\partial_\mu$ to denote $\frac{\partial}{\partial x^\mu}$.} acting on a field  $\varphi_{KG}$ and get the Klein-Gordon equation $\left(-\eta^{\mu\nu}\partial_\mu\partial_\nu+\frac{m^2c^2}{\hbar^2}\right)\varphi_{KG}=0$. For the general relativistic case, we need to consider the minimally coupled Klein-Gordon (KG) equation, that is, we need to substitute the standard partial derivative $\partial_\mu$ by the covariant derivative $D_\mu$, resulting in
\begin{equation}\label{KGeq}
  \hbar^2g^{\mu\nu}D_\mu D_\nu\varphi_{KG}-m^2c^2\varphi_{KG}=0.
\end{equation}
In the low-energy limit, this field can be regarded as a particle in first quantization. The valid inner product for the functions $\varphi_{KG}$ is~\cite{LAMMERZAHL}:
\begin{equation}\label{InnerProductKG}
\braket{\psi|\varphi}=\int\left(\psi^*\varphi-\frac{\hbar}{2m^2c^2}\psi^*\nabla^2\varphi\right)d^3V,
\end{equation}
where $d^3V=\sqrt{\on{det} g_{ij}}d^3x$.

\subsection{A Hamilton operator for a quantum particle in weak gravity}

In parallel to the classical case, we can try to extract a hamiltonian operator for a particle from \eqref{KGeq}. Since isolating $-c p_0$ does not make sense anymore, the trick is to manipulate this equation until a term in the form $i\hbar \partial_t $ appears, isolate it, and identify the hamiltonian operator with the other side of the equality~\cite{LAMMERZAHL}. To continue the calculations, we need to specify a metric we can work with.

Since we are mainly interested in situations of low energies and weak gravity, we may proceed with a simple version of the parametrized post-newtonian metric. It consists of an expansion of a general metric theory, $g_{\mu\nu}$, which reduces to newtonian gravity when its parameters are set to $0$ and it may be used to approximate solutions to Einstein's field equations in the aforementioned regimes~\cite{will_1993}. The non-zero components for that metric are 
\begin{equation}\label{PostNewtMetric}
g_{00}=-\left(
		  1+2\frac{\Phi}{c^2}+2\beta\frac{\Phi^2}{c^4}
        \right) \qquad g_{ij}=\left(
		 1-2\gamma\frac{\Phi}{c^2}
       \right)\delta_{ij},
\end{equation}
and its inverse is given by
\begin{equation}
	g^{00}=-\left(
	1-2\frac{\Phi}{c^2}+2(2-\beta)\frac{\Phi^2}{c^4}
	\right) \qquad	g^{ij}=\left(
	1+2\gamma\frac{\Phi}{c^2}+4\gamma^2\frac{\Phi^2}{c^4}
	\right)\delta^{ij},
\end{equation}
where $\gamma,\beta$ are the parameters and $\Phi=-Gm_{\operatorname{curv}}/x$ is the newtonian potential of the total mass generating the gravitational field.

Now, we can evaluate the Klein-Gordon equation \eqref{KGeq}. 
The following are open calculations for the steps outlined in reference \cite{LAMMERZAHL} until equation (8), for a vanishing electromagnetic field.

The covariant derivative acts on a vector as
\begin{equation*}
D_\mu T^\nu=\partial_\mu T^\nu+\Gamma^\nu_{\mu\sigma}T^\sigma,
\end{equation*}
where
\begin{equation*}
\Gamma^\nu_{\mu\sigma}=
\frac{1}{2}g^{\nu\rho}
(\partial_\mu g_{\sigma\rho}+\partial_\sigma g_{\mu\rho}-\partial_\rho g_{\mu\sigma})
\end{equation*}
are the Christoffel symbols. Its action on a scalar is just a partial derivative, $D_\mu\varphi=\partial_\mu\varphi$. When applied to a 1-form, it returns
\begin{equation*}
D_\mu T_\nu=\partial_\mu T_\nu-\Gamma^\sigma_{\mu\nu}T_\sigma.
\end{equation*}
Thus, the operator in the first term of the KG equation~(\ref{KGeq}) is
\begin{equation*}
D_\mu D_\nu=D_\mu \partial_\nu=\partial_\mu \partial_\nu-\Gamma^\sigma_{\mu\nu}\partial_\sigma,
\end{equation*}
and the equation can be rewritten as
\begin{equation}\label{KGeqnew}
\hbar^2g^{\mu\nu}\partial_\mu \partial_\nu\varphi_{KG}
-\hbar^2g^{\mu\nu}\Gamma^\sigma_{\mu\nu}\partial_\sigma\varphi_{KG}
-m^2c^2\varphi_{KG}=0.
\end{equation}

With the metric we are considering, most of the Christoffel symbols will not be needed in the expression above. Here, we compute those that we will indeed be used:
\begin{align}\label{ChrisSymb}
\Gamma_{00}^0&=g^{00}
             \left(
               -\frac{\partial_t\Phi}{c^3}-\frac{2\beta \Phi\partial_t \Phi}{c^5}
             \right),\nonumber
\\
\Gamma_{ij}^0&=\delta_{ij}g^{00}\frac{\gamma \partial_t\Phi}{c^3},\nonumber
\\
\Gamma_{00}^i&=g^{ij}
              \left(
                \frac{\partial_j \Phi}{c^2}+\frac{2\beta \Phi\partial_j \Phi}{c^4}
              \right),
\\
\Gamma^k_{ij}&=\delta_{ij}g^{kl}\frac{\gamma}{c^2}\partial_l\Phi \ , \ \ \ \ \ \text{for} \ k=i,j,\nonumber
\\
\Gamma^k_{ij}&=-\delta_{ij}g^{kl}\frac{\gamma}{c^2}\partial_l\Phi \ , \ \ \ \ \ \ \ \ \text{for} \ k\neq i,j.\nonumber
\end{align}
Note that
\begin{equation*}
g^{ij}\Gamma_{ij}^k=-
  \left(
    1+2\gamma\frac{\Phi}{c^2}+4\gamma^2\frac{\Phi^2}{c^4}
  \right)
g^{kl}\frac{\gamma}{c^2}\partial_l\Phi.
\end{equation*}

We draw attention to the fact that $\partial_0=(1/c)\partial_t$, while $\partial_i$ adds no explicit dependence on $c$. Expanding the sums in equation~(\ref{KGeqnew}) we have 
\begin{multline*}
\frac{\hbar^2g^{00}}{c^2}\partial_t\partial_t\varphi_{KG}
+\hbar^2g^{ij}\partial_i\partial_j\varphi_{KG}
-\frac{\hbar^2g^{00}}{c}\Gamma_{00}^0\partial_t\varphi_{KG}
-\frac{\hbar^2g^{ij}}{c}\Gamma_{ij}^0\partial_t\varphi_{KG}
\\
-\hbar^2g^{00}\Gamma_{00}^i\partial_i\varphi_{KG}
-\hbar^2g^{ij}\Gamma_{ij}^k\partial_k\varphi_{KG}
-m^2c^2\varphi_{KG}=0. \nonumber
\end{multline*}
Inserting the Christoffel symbols~(\ref{ChrisSymb}) and going up to order $1/c^4$, we find that 
\begin{align}\label{KGeqc2}
&\frac{\hbar^2g^{00}}{c^2}\partial_t\partial_t\varphi_{KG}
+\hbar^2g^{ij}\partial_i\partial_j\varphi_{KG}
+\hbar^2\left(\left(g^{00}\right)^2
-g^{ij}\delta_{ij}g^{00}\gamma 
 \right)\frac{\partial_t\Phi}{c^4}\partial_t\varphi_{KG} 
\\
-\hbar^2&g^{00}g^{ij}\left(\frac{\partial_j\Phi}{c^2}+\frac{2\beta \Phi\partial_j \Phi}{c^4}\right)\partial_i\varphi_{KG}-\hbar^2\left(
\frac{\gamma}{c^2}+2\gamma^2\frac{\Phi}{c^4}
\right)
g^{kl}(\partial_l\Phi)\partial_k\varphi_{KG}
-m^2c^2\varphi_{KG}=0. \nonumber
\end{align}

We can see that this equation is in powers of $c^2$, meaning we can try to solve it perturbatively. The following ansatz from~\cite{Kiefer1991} is presented in reference~\cite{LAMMERZAHL}:
\begin{equation}\label{ansatz}
\varphi_{KG}(x)=\exp
  \left(
    \frac{i}{\hbar}
      \left[
        c^2S_0(x)+S_1(x)+c^{-2}S_2(x)+...
      \right]
  \right),
\end{equation}
where $S_n(x)$ are scalar functions.

\hspace{2.5cm}\textcolor[RGB]{220,180,200}{\rule{8cm}{0.2pt}}

\noindent $\mathbf{c^4}:$ The lowest order term that emerges when we substitute the ansatz in \eqref{KGeqc2} is of order $c^4$. The equation for it reads:
\begin{align*}
\hbar^2g^{ij}\partial_i\partial_j e^{\frac{i}{\hbar}c^2S_0(x)}
&=\hbar^2g^{ij}
  \left(
    \frac{ic^2\partial_i\partial_j S_0(x)}{\hbar}
    -\frac{c^4(\partial_iS_0(x))(\partial_jS_0(x))}{\hbar^2}
  \right)
e^{\frac{i}{\hbar}c^2S_0(x)}
\\
&\approx -\frac{c^4}{\hbar^2}\sum_{i=1}^{3}(\partial_iS_0(x))^2
\exp \left(\frac{i}{\hbar}c^2S_0(x)\right)=0,
\end{align*}
where we considered only the first order approximation of $g^{ij}\approx 1$ to get from the first to the second line. From this expression, we obtain 
\begin{equation}
\sum_{i=1}^{3}(\partial_iS_0(x))^2
=
\nabla S_0(x) \cdot \nabla S_0(x)=0 \implies 
\nabla S_0(x)=0,
\end{equation}
where $\nabla$ is the 3D gradient. We can thus conclude that $S_0=S_0(t)$ depends only on $t$.

\noindent\textcolor[RGB]{220,180,200}{\rule{\linewidth}{0.2pt}}

\noindent $\mathbf{c^2:}$ The next order, $c^2$, gives us 
\begin{equation*}
\frac{\hbar^2g^{00}}{c^2}\partial_t\partial_t \exp \left(\frac{i}{\hbar}c^2S_0(t)\right)
-
m^2c^2\exp \left(\frac{i}{\hbar}c^2S_0(t)\right)=0.
\end{equation*}
Evaluating the derivative and considering only terms of order $c^2$, with $g^{00}\approx -1$, 
\begin{equation}
\left[c^2(\partial_tS_0(t))^2
-
m^2c^2\right]\exp \left(\frac{i}{\hbar}c^2S_0(t)\right)=0 \implies
\partial_tS_0(t)=\pm m.
\end{equation}
Let us choose the solution $S_0=-mt$. With $S_0$ determined, the ansatz~(\ref{ansatz}) becomes
\begin{equation}
\varphi_{KG}=\exp\left(\frac{i}{\hbar}\left[-mc^2t+S_1(x)+c^{-2}S_2(x)+...\right]\right)
=:
\varphi_0\varphi_1\varphi_2...
\end{equation}
where we defined $\varphi_0:=e^{-\frac{i}{\hbar}mc^2t}$, $\varphi_1:=e^{\frac{i}{\hbar}S_1}$, $\varphi_2:=e^{\frac{i}{\hbar c^2} S_2}$, and so on.

\noindent\textcolor[RGB]{220,180,200}{\rule{\linewidth}{0.2pt}}

\noindent $\mathbf{c^0:}$ Now we will look into the terms of order up to $1$ in equation~(\ref{KGeqc2}). They are
\begin{align*}
\frac{\hbar^2g^{00}}{c^2}\partial_t\partial_t\varphi_{KG}
&=
(m^2c^2-2m^2\Phi-2m\partial_t S_1)\varphi_{KG}
\\
&=
(m^2c^2-2m^2\Phi)\varphi_{1}\frac{\varphi_{KG}}{\varphi_1}+2i\hbar m(\partial_t\varphi_1)\frac{\varphi_{KG}}{\varphi_1},
\\
\hbar^2g^{ij}\partial_i\partial_j\varphi_{KG}
&=\hbar^2(\bigtriangleup\varphi_{1})\frac{\varphi_{KG}}{\varphi_1},
\\
-m^2c^2\varphi_{KG}
&=-m^2c^2\varphi_{1}\frac{\varphi_{KG}}{\varphi_1},
\end{align*}
where $\bigtriangleup$ is the laplacian operator in 3D. All the other terms would lead to order $c^{-2}$ or higher. Summing all of these and setting their sum to zero, we obtain
\begin{equation*}
i\hbar\partial_t\varphi_1=-\frac{\hbar^2}{2m}\bigtriangleup\varphi_1+m\Phi\varphi_1.
\end{equation*}
Defining $p=- i \hbar\nabla$, we have that $p^2=-\hbar^2\bigtriangleup$ and the above equation becomes
\begin{equation}
i\hbar\partial_t\varphi_1=\frac{p^2}{2m}\varphi_1+m\Phi\varphi_1,
\end{equation}
which is equation (7) from reference~\cite{LAMMERZAHL}. This is just the particle quantum wave equation with newtonian potential for $\varphi_1$, from which we get a first approximation with no relativistic terms for the hamiltonian operator\footnote{We will almost always omit identities to keep the calculations clean in this section. It is assumed we are working with operators. For instance, p symbolizes the momentum operator and $\Phi$ means $\Phi \mathds{1}$.}, $H\approx p^2/(2m)+m\Phi\mathds{1}$.

Before proceeding to the next order, we derive the equality below for later:
\begin{align}
-\hbar^2\partial_t\partial_t\varphi_1
&=
i\hbar\partial_t(i\hbar\partial_t\varphi_1)
=
i\hbar\partial_t
\left[
\left(
  \frac{p^2}{2m}+m\Phi
\right)\varphi_1
\right]\nonumber
\\
&=\left(
\frac{p^2}{2m}+m\Phi
\right)\left(
\frac{p^2}{2m}+m\Phi
\right)\varphi_1
+i\hbar m(\partial_t \Phi)\varphi_1 \nonumber
\\
&=
\left(
\frac{p^4}{4m^ 2}+\Phi p^2+\frac{p^2\Phi}{2}-i\hbar\nabla \Phi\cdot p
+m^2\Phi^2+i\hbar m\partial_t\Phi
\right)\varphi_1.
\label{delttphi1}
\end{align}

\noindent\textcolor[RGB]{220,180,200}{\rule{\linewidth}{0.2pt}}

\noindent $\mathbf{c^{-2}:}$ In the following, we denote 
\begin{equation}
    \varphi:=\varphi_1 \exp{\left(\frac{iS_2}{\hbar c^2}\right)} = \varphi_1 \varphi_2,
\end{equation}
which, up to order $c^{-2}$, is also equal to  $\varphi_{KG}/\varphi_0$.
The first term up to order $c^{-2}$ in equation~(\ref{KGeqc2}) is
\begin{align*}
\frac{\hbar^2g^{00}}{c^2}\partial_t\partial_t\varphi_{KG}
&=
\frac{\hbar^2g^{00}}{c^2}[(\partial_t\partial_t\varphi_0)\varphi
+2(\partial_t\varphi_0)(\partial_t\varphi)
+\varphi_0(\partial_t\partial_t\varphi)]
\\
&=\frac{\hbar^2g^{00}}{c^2}
\left(
  \frac{-m^2c^4}{\hbar^2}\varphi_{KG}
  -\frac{2imc^2}{\hbar}\varphi_0\partial_t\varphi
  +\varphi_0(\partial_t\partial_t\varphi_1)\varphi_2
\right)
\\
&=
g^{00}
\left(
  -m^2c^2\varphi_{KG}-2im\hbar\varphi_0(\partial_t\varphi)
+\frac{1}{c^2}\varphi_0(\hbar^2\partial_t\partial_t\varphi_1)\varphi_2
\right),
\end{align*}
where from the first to the second line we neglected the time derivatives of $\varphi_2$, since they lead to higher order terms. The second term of~(\ref{KGeqc2}) up to $c^{-2}$ is
\begin{equation}
\hbar^2g^ij\partial_i\partial_j\varphi_{KG}
=
\left(
  1+2\gamma\frac{\Phi}{c^2}
\right)
\varphi_0\hbar^2\bigtriangleup\varphi
=
-\left(
1+2\gamma\frac{\Phi}{c^2}
\right)
\varphi_0p^2\varphi.
\end{equation} 

\noindent The third term becomes
\begin{equation}
\hbar^2\left(\left(g^{00}\right)^2
-
g^{ij}\delta_{ij}g^{00}
\right)\frac{\partial_t\Phi}{c^4}\partial_t\varphi_{KG}
=
-\frac{im\hbar}{c^2}
(1+3\gamma)(\partial_t\Phi)\varphi_{KG},
\end{equation}
where we considered $g^{00}=-1$ and $g^{ij}\delta_{ij}=1\cdot\delta^{ij}\delta_{ij}=3.$ The fourth and fifth terms are
\begin{equation}
-\hbar^2g^{00}g^{ij}\left(\frac{\partial_j\Phi}{c^2}+\frac{2\beta \Phi\partial_j \Phi}{c^4}\right)\partial_i\varphi_{KG}
=\frac{i\hbar\varphi_0}{c^2}\nabla \Phi\cdot p \varphi,
\end{equation}

\begin{equation}
-\hbar^2\left(
\frac{\gamma}{c^2}+2\gamma^2\frac{\Phi}{c^4}
\right)
g^{kl}(\partial_l\Phi)\partial_k\varphi_{KG}
=-
\frac{i\hbar\gamma\varphi_0}{c^2}\nabla \Phi\cdot p \varphi.
\end{equation}

\noindent And, finally, the last term remains $-m^2c^2\varphi_{KG}$.

 Summing all of these terms and multiplying the equation by $g_{00}/\varphi_0$ we have 
\begin{align}
-m^2c^2\varphi-&2im\hbar(\partial_t\varphi)
+\frac{1}{c^2}(\hbar^2\partial_t\partial_t\varphi_1)\varphi_2
-g_{00}\left(
1+2\gamma\frac{\Phi}{c^2}
\right)
p^2\varphi \nonumber
\\ \nonumber
\\
&\hspace{0.5cm}-g_{00}\frac{im\hbar}{c^2}
(1+3\gamma)(\partial_t\Phi)\varphi
+ g_{00}i\hbar\frac{(1-\gamma)}{c^2}\nabla \Phi\cdot p\varphi
-g_{00}m^2c^2\varphi=0.
\end{align} 
Manipulating this equation further, one gets
\begin{align*}
&-m^2c^2\varphi-2im\hbar(\partial_t\varphi)
+\frac{1}{c^2}(\hbar^2\partial_t\partial_t\varphi_1)\varphi_2
+\left(
1+2\frac{\Phi}{c^2}
\right)\left(
1+2\gamma\frac{\Phi}{c^2}
\right)
p^2\varphi
\\
&-(-1)\frac{i m\hbar}{c^2}
(1+3\gamma)(\partial_t\Phi)\varphi
+(-1) i\hbar\frac{(1-\gamma)}{c^2}\nabla \Phi\cdot p\varphi
+\left(
1+2\frac{\Phi}{c^2}+2\beta\frac{\Phi^2}{c^4}
\right)m^2c^2\varphi=0;
\nonumber
\\
\\
&-i\hbar(\partial_t\varphi)
+\frac{1}{2mc^2}(\hbar^2\partial_t\partial_t\varphi_1)\varphi_2
+\left(
1+2(1+\gamma)\frac{\Phi}{c^2}
\right)
\frac{p^2}{2m}\varphi
\\
&+\frac{i\hbar}{2c^2}
(1+3\gamma)(\partial_t\Phi)\varphi
-\frac{i\hbar(1-\gamma)}{2mc^2} \nabla \Phi \cdot p \varphi
+\left(
+m\Phi+\frac{\beta m\Phi^2}{c^2}
\right)\varphi=0;\nonumber
\\
\\
&i\hbar\partial_t\varphi=
\frac{1}{2mc^2}(\hbar^2\partial_t\partial_t\varphi_1)\varphi_2
\\
&+\left(
\frac{p^2}{2m}+m\Phi+(1+\gamma)\frac{\Phi}{mc^2}p^2
+\frac{\beta m\Phi^2}{c^2}
+\frac{i\hbar(1+3\gamma)}{2c^2}\partial_t\Phi
-\frac{i\hbar(1-\gamma)}{2mc^2}\nabla \Phi\cdot p
\right)\varphi.
\nonumber
\end{align*}

Inserting~(\ref{delttphi1}) in the equation above, we obtain
\begin{align*}
&i\hbar\partial_t\varphi=
\left(
-\frac{p^4}{8m^3c^2}-\frac{\Phi p^2}{2mc^2}-\frac{p^2\Phi}{4mc^2}+\frac{i\hbar}{2mc^2}\nabla \Phi\cdot p
-\frac{m\Phi^2}{2c^2}-\frac{i\hbar \partial_t\Phi}{2c^2}
\right)\varphi
\\
&+\left(
\frac{p^2}{2m}+m\Phi+(1+\gamma)\frac{\Phi}{mc^2}p^2
+\frac{\beta m\Phi^2}{c^2}
+\frac{i\hbar(1+3\gamma)}{2c^2}\partial_t\Phi
-\frac{i\hbar(1-\gamma)}{2mc^2}\nabla \Phi\cdot p
\right)\varphi,
\nonumber
\end{align*}
giving us the expression
\begin{align}
i\hbar\partial_t\varphi=
\Bigg(\frac{p^2}{2m}-\frac{p^4}{8m^3c^2}&+m\Phi
+\left(\gamma+\frac{1}{2}\right)\frac{\Phi}{mc^2}p^2 
\\
&-\left(\frac{1}{2}-\beta\right)\frac{m\Phi^2}{c^2}+\frac{3 i\hbar\gamma}{2c^2} \partial_t \Phi +\frac{i\hbar\gamma}{2mc^2}\nabla \Phi\cdot p
+\frac{\hbar^2\bigtriangleup \Phi}{4mc^2}
\Bigg)\varphi, \nonumber
\end{align}
which corresponds to equation (8) from~\cite{LAMMERZAHL}. 
The expression in parenthesis acts as a Hamiltonian $H$ for $\varphi$. However, the probabilities of this theory are acquired using the inner product in \eqref{InnerProductKG}. We can make a transformation such that the product turns into $\braket{\varphi|\psi}=\int(\varphi')^*\psi'd^3x$, the familiar form of quantum particle mechanics. The transformations for $\varphi$ and $H$ are given by~\cite{LAMMERZAHL}\footnote{The transformation for $\varphi$ here is generally time-dependent because $\Phi$ can be time-dependent. This causes the transformation of the Hamiltonian to have an extra term proportional to $\partial_t \Phi$, because we demand that that $H' \varphi'=i\hbar\partial_t\varphi'$. Writing $\varphi'=\operatorname{T}(\varphi)$, we can manipulate the expression $H \varphi= i\hbar \partial_t \varphi$ to get $H'= \operatorname{T}H\operatorname{T}^{-1}+i \hbar(\partial_t \operatorname{T})\operatorname{T}^{-1}$, getting the extra term. We point out the sign of the exponent $1/4$ of this last term in equation~\eqref{Hprimef} is switched in the original reference~\cite{LAMMERZAHL}. We believe that to be a typo, since the reasoning above as well as the expression that follows in the paper itself~\eqref{HamiltonianOperator} agree with what we used for $H'$.}:
\begin{align}
\varphi'&=\left(1+\frac{p^2}{m^2c^2}\right)^{1/4}(\on{det}g_{ij})^{1/4} \varphi
    \\
    H'&= \left(1+\frac{p^2}{m^2c^2}\right)^{1/4} (\on{det}g_{ij})^{1/4} H \left(1+\frac{p^2}{m^2c^2}\right)^{-1/4}(\on{det}g_{ij})^{-1/4} + i\hbar\partial_t (\on{det}g_{ij})^{1/4}.\label{Hprimef}
\end{align}
Calculating the expression above up to order $c^{-2}$, we get to the final expression for $H'$:
\begin{equation}\label{HoperatorLamm}
    H'=\frac{p^2}{2m}+m\Phi-\frac{p^4}{8m^3c^2}
+\frac{2\gamma+1}{2mc^2}\left(\Phi p^2-i\hbar \nabla \Phi \cdot p\right) -\left(\frac{1}{2}-\beta\right)\frac{m\Phi^2}{c^2} - 3\gamma\frac{\hbar^2\bigtriangleup\Phi}{4c^2m}.
\end{equation}

\hspace{2.5cm}\textcolor[RGB]{220,180,200}{\rule{8cm}{0.2pt}}

If $\varphi$ obeys the quantum wave equation with relativistic corrections with this hamiltonian, then $\varphi_0 \varphi = e^{-\frac{i}{\hbar}m c^2 t} \varphi \approx \varphi_{KG}$ up to this order obeys $i\hbar \partial_t\varphi_{KG}\approx H_{fp} \varphi_{KG}$ with $H_{fp}$ being the operator $H'$ in \eqref{HoperatorLamm} plus the term $mc^2 \mathds{1}$. Moreover, the post-newtonian metric for general relativity is characterized by choosing $\beta=\gamma=1$. The explicit form of the metric, which is given in~\ref{PostNewtMetric}, when we consider these parameters turns into the Schwarzschild metric in isotropic coordinates in the approximation of low gravitational potential~\cite{will_1993}. We can simplify the expression by also identifying when $-i\hbar \nabla= p $ acts on $\Phi$. Given these observations, the hamiltonian operator found for a general relativistic free particle in the regime of weak gravity is
\begin{equation}\label{HamiltonianOperator}
    H_{fp} = mc^2 + \frac{p^2}{2m}  +m\Phi -\frac{p^4}{8m^3c^2} + \frac{1}{2}\frac{m\Phi^2}{c^2}  +\frac{3}{2mc^2}\left(\Phi p^2-\left[p\Phi\right]\cdot p+\frac{1}{2}\left[p^2 \Phi \right]\right),
\end{equation}
where we omitted the identity operator $\mathds{1}$ next to the scalar terms and the brackets indicate the terms $\left[p\Phi\right]$ and $\left[p^2\Phi\right]$ result from the application of $p$ on $\Phi$ which is then multiplied by the input function, rather than a composition of $p$ with $\Phi \mathds{1}$.

\subsection{First quantization for clocks} \label{subsecFirstQClocks}
As in the classical case we saw in \ref{subsecHamiltonian}, the hamiltonian for a particle can be used for a clock, a particle with internal degrees of freedom, if we consider that the system has mass $M=H_{rest}/c^2$~\cite{Zych}. Now, in the quantum case, we can consider the corresponding rest hamiltonian operator $$H_{rest}=mc^2 \mathds{1}_{\mathcal{H}_{int}} +H_{int},$$ with $H_{int}$ being the operator describing the internal clock dynamics aside from the mass term. Let us call the Hilbert space on which $H_{rest}$ acts $\mathcal{H}_{int}$. In particular, if the internal system is in a pure state  $\ket{E_i}\in\mathcal{H}_{int}$ that is an eigenstate of $H_{rest}$, $H_{rest}\ket{E_i}=E_i\ket{E_i}$, the total system's hamiltonian operator is that of a free particle, but with with a modified mass. The complete system is described in a space $\mathcal{H}_{int}\otimes \mathcal{H}_{ext}$, with $\mathcal{H}_{ext}$ comprising external degrees of freedom. If we rewrite the hamiltonian operator \eqref{HamiltonianOperator} with $$M:=\left(m\mathds{1}_{\mathcal{H}_{int}}+H_{int}/c^2\right)\otimes \mathds{1}_{\mathcal{H}_{ext}}$$ in place of $m$, we get, to first order in $H_{int}/(mc^2)$,
\begin{equation}\label{ClockHamiltonianOperator}
H_{clock}\approx H_{cm} + H_{int} \otimes \left(  \mathds{1}_{\mathcal{H}_{ext}}+\frac{\Phi}{c^2}\mathds{1}_{\mathcal{H}_{ext}}-\frac{p^2}{2m^2c^2}\right),\end{equation}
where the dynamics of the center of mass is given by
\begin{align*}
H_{cm}=mc^2 &\left(\mathds{1}_{\mathcal{H}_{int}}\otimes \mathds{1}_{\mathcal{H}_{ext}}\right) 
\\ + &\mathds{1}_{\mathcal{H}_{int}} \otimes \left[\frac{p^2}{2m} +m\Phi -\frac{p^4}{8m^3c^2}  + \frac{1}{2}\frac{m\Phi^2}{c^2}  +\frac{3}{2mc^2}\left(\Phi p^2-\left[p\Phi\right]\cdot p+\frac{1}{2}\left[p^2 \Phi \right]\right)\right],
\end{align*}
which is basically \eqref{HamiltonianOperator} except for the considerations on the extended space. So, this is the final hamiltonian operator for a clock in a low energy, weak gravity limit that is nonetheless relativistic.

In fact, we encountered the classical version of this hamiltonian in the way to find the newtonian limit for a clock in \eqref{leHamiltonian}. If we rewrite that classical expression to first order in $H_{int}/(mc^2)$, we arrive at 
\begin{equation}\label{HleComparison}
    H_{le}\approx H_{cm} + H_{int}\left(1+\frac{\Phi}{c^2}-\frac{p^2}{2m^2c^2}\right)\qquad H_{cm}=mc^2+\vec{p}^2/2m+m\Phi.
\end{equation} The term in parenthesis is simply $\dot{\tau}$ up to order $1/c^2$, indicating that we can write a routhian \eqref{leRouthian} as $H_{rest}\dot{\tau}$. The interaction terms between $H_{int}$ and variables associated to external degrees of freedom describe time dilation, a fact we can check by evaluating internal observables in the routhian formalism. 

    The quantum treatment is pretty much analogous. Since the hamiltonian operator~\eqref{ClockHamiltonianOperator} has virtually the same form of the classical one \eqref{HleComparison} for this approximation, we can describe the system with the routhian, now an operator, $H_{rest}\dot{\tau}$, which acts like a hamiltonian for the internal portion of the system. It generates a quantum wave equation for an observer looking at the system with respect to their coordinate time, $i\hbar \partial_t=H_{rest}\dot{\tau}$. And, in the rest frame, $i\hbar \partial_{\tau}=H_{rest}$. If we evaluate the commutator of an internal observable with the routhian, it will give us time dilation in the Heisenberg picture, in analogy with \eqref{ObsEvol}. Thus, the internal state of a quantum clock in this regime evolves according to proper time and it is time dilated depending on region of spacetime or velocity, in the same proportion as a classical system. In particular, from an observer's perspective, the time dilation that applies to a quantum clock's internal evolution at low energies and weak gravity is the same as that which applies to a classical clock, provided they describe the same path on spacetime. That is why we can use quantum clocks to perform classical tests of relativistic time dilation. However, note that the external degrees of freedom of a quantum clock are also quantized, and thus they can be indefinite, for instance the clock could be in a superposition of distinct velocities or in a superposition of being at locations that have distinct curvatures, and produce quantum superpositions of distinctly dilated time evolutions. 

\chapter{Quantum switch in a quantum spacetime}\label{Chap:QGravitySwitch}
  If we go back to the motivations behind the process matrix formalism, the situations of interest where indefinite causal structure is originally expected to appear should lie in the regime in which both General Relativity and Quantum Theory are relevant. That is because even though GR and QT both have definite causality, the causal relation between a pair of events in GR is determined by the distribution of matter and energy in their past lightcones. So, if matter presents significant quantum behavior and both theories are valid, this should affect causal relations. We expect that the framework of Quantum Theory was sufficiently generalized by the process formalism so that we can now explore how quantum behavior of gravity could potentially affect causal structure. But, concretely, how could a quantum spacetime induce the application of a process with indefinite order? We can start to address this by examining simple examples of processes that are already studied in the field of indefinite orders, like the quantum switch. As suggested in reference~\cite{quantumCausality}, we can think of the quantum switch as a toy-model to describe a spacetime that is itself in a quantum superposition. This kind of exploration might provide insight for causal models, as well as for generally understanding how order works in the interface between quantum mechanics and gravity. In this chapter, we present the first theoretical proposal of indefinite order induced by a superposition state of a massive object's gravity, the gravitational quantum switch of reference~\cite{tbell}. 
 
 Since we have no theory of Quantum Gravity, the description of such a protocol has its obstacles. For instance, some frameworks assume that a spatial superposition state of a sufficiently massive body would not live long enough to present any considerable effects~\cite{Diosi,Penrose,Scully2018}, while several approaches indeed consider the metric can present quantum features~\cite{Kiefer}. Even so, it is a challenge to define states and their evolution in the absence of a classical spacetime manifold in which we can separate time from space with a foliation. Instead of choosing some specific quantum gravity framework, the authors approach this by making only a few minimal assumptions, which concern the joint validity of the quantum superposition principle and of time dilation. Therefore, their results should hold in any theory that agrees with their way to join these features. In their approach, at low energy and for a weak gravitational field, the quantum formalism of chapter \ref{Chap Clocks} is used to characterize events in a quantum spacetime.
 
Finally, the authors also formulate a Bell's theorem for temporal order, which shows that if Bell's inequalities are violated in a specific task under certain theory-independent\footnote{The theory independence might be under analysis, as we comment in the end of the chapter.} assumptions, then temporal order is not described by classical variables. Afterwards, they demonstrate that two copies of the gravitational quantum switch could be used to violate a Bell inequality in that setting, basically showing that order must be non-classical, like a quantum variable, in theories conforming to their basic treatment.

\section{Operational events in spacetime}\label{operationalEvents}

As discussed at the end of chapter \ref{Chap Process Matrix}, specifying how we treat events/laboratories might be useful to discern in what sense causality is challenged by the realization of a process with indefinite order in each context. Specially in the case of a quantum spacetime, the lack of a well-defined underlying manifold makes it difficult not only to talk about causality and processes, but even to generally describe a protocol on it. To overcome this, the authors of~\cite{tbell} adopt a ``physical'' interpretation of events, meaning that they are specified operationally relative to physical systems used as clocks. An example of operational event A is the location in which a chosen clock measures some specific proper time $\tau_{a}$. The causal relations are defined as usual, A$\prec$B if A can send a signal to B, which in GR corresponds to sending a system that travels no faster than light. 

For classical clocks, we showed how velocity and curvature induce universal time dilation in subsection \ref{subsecRouthian}. Therefore, the presence of a mass near a clock delays operational events (in comparison with the virtual situation in which the mass is not present). Furthermore, the relations between events can also differ for distinct mass configurations. Let us consider an elementary example to illustrate this.

    Let $\on{a}$ and $\on{b}$ represent two agents at fixed positions with respect to a coordinate system, each one in possession of a clock, and let the clocks be initially synchronized. Suppose that a third agent is able to place a point-like body of mass $M$ in one of two locations, causing distinct time dilations on the clocks as shown in Fig.~\ref{fig:ZychLightcones}. If we define event A as the clock of agent $\on{a}$ attaining proper time $\tau_a=\tau^*$ and event B as the clock of $\on{b}$ showing $\tau_b=\tau^*$, then the causal relation between A and B can be different depending on where the mass is placed. Let us write an approximation of the metric around the mass M in the form:\begin{equation}\label{metricPostNc2}
g_{00}=-\left(1+2\frac{\Phi(r)}{c^2}\right), \qquad g_{ij}=\delta_{ij}\left(1+2\frac{\Phi(r)}{c^2}\right)^{-1}.
\end{equation}Similarly to the one we used in chapter~\ref{Chap Clocks}, this metric represents a weak field approximation for the Schwarzschild metric in isotropic coordinates~\cite{will_1993}, where $\Phi= -GM/r$ is the newtonian potential and $r$ is the distance from the clock to the mass $|r-r_M|$. Although the placing of the massive object generally changes the coordinate description, the authors assume that these spatial quantities can be defined operationally with respect to the coordinates of a distant observer, who does not feel the difference between one and the other mass configuration. Thus, we assume we are working with the coordinates of an observer at infinity. For simplicity, we also assume the clocks and the mass are aligned so that we can work with only one spatial dimension. Let us call the two mass configurations $\on{K_{A\prec B}}$ and $\on{K_{B\prec A}}$, as labeled in Fig.~\ref{fig:ZychLightcones}. We can see that in $\on{K_{A\prec B}}$ time runs slower for agent $\on{b}$. Therefore, there must exist some pair of events A and B defined in terms of the clocks measuring $\tau^*$ such that A is in the past lightcone of B. For instance, Fig.~\ref{fig:ZychLightcones} represents the ticks of the clocks with little dots. If A and B are defined as the third dot on $\on{a}$ and $\on{b}$'s clocks, that is, $\tau^* = 3$ ticks, then we have A$\prec$B.

\begin{figure}
    \centering
    \includegraphics[scale=0.35]{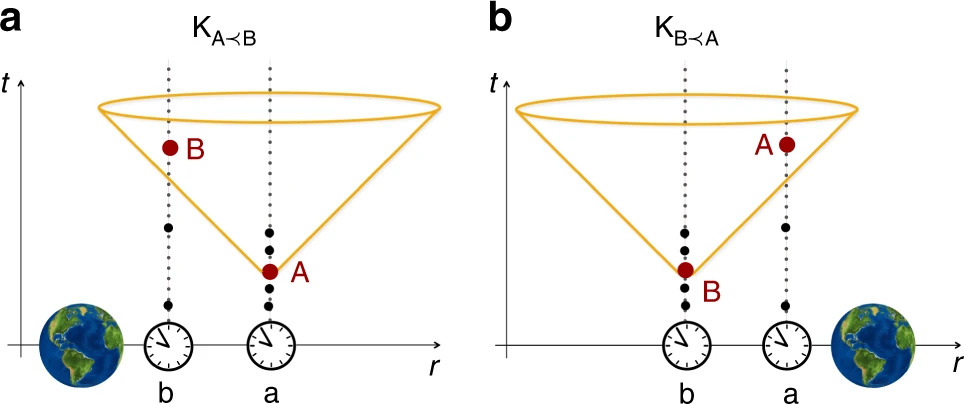} 
    \caption{An example of how mass distribution influences causal order between operational events in GR. \textbf{a} - Configuration $\on{K_{A\prec B}}$ consists of a massive body closer to agent $\on{b}$, therefore $\on{b}$ experiences stronger time dilation than $\on{a}$. Let us say time is measured in clock ticks, and both clocks have the same internal mechanism. The ticks are then represented by the black dots in the worldlines of the agents. If we define events A and B as the clocks $\on{a}$ and $\on{b}$ measuring proper time $\tau^*=3$ ticks, we have A$\prec$B. \textbf{b} -  If the agents were in configuration $\on{K_{B\prec A}}$ we would have B$\prec$A for the same definition. Thus, the causal relation between A and B generally depends on the placement of the mass done in their causal past. Figure from~\protect\cite{tbell}.}
    \label{fig:ZychLightcones}
\end{figure}

We can calculate the general condition for which event A enters the past lightcone of event B in configuration $\on{K_{A\prec B}}$. If agent $\on{a}$ sends a light signal to agent $\on{b}$ at event A, namely when $\tau_a=\tau^*$, the time taken for it to travel between the agents reads
\begin{equation}\label{lightTravelt}
t_c=\int dt =\int\frac{dt}{dr'}dr'=\frac{1}{c}\int_{r_a}^{r_b}\sqrt{\frac{-g_{rr}(r')}{g_{00}(r')}}dr', 
\end{equation}
where the last equality comes from the fact that, for a photon, $$ds^2\equiv g_{00}c^2dt^2+g_{rr}dr'^2=0.$$ The quantity $t_{c}$ is a coordinate time interval with respect to the coordinate system we chose for writing the metric. The agents will perceive this with a time dilation given by $d\tau/dt=\sqrt{-g_{00}}$. Thus, for agent a, the photon arrives at $\on{b}$'s position at proper time $\tau_{af}=\tau^*+\sqrt{-g_{00}(r_a)} t_c$. Setting $\tau_a=\tau_b=t=0$ in the beginning of the experiment, b perceives the photon's arrival at 
\begin{equation}\label{tbFinal}
    \tau_{bf}=\sqrt{-g_{00}(r_b)}\Delta t = \sqrt{-g_{00}(r_b)} \frac{(\tau_a^{f}-0)}{\sqrt{-g_{00}(r_a)}}= \sqrt{-g_{00}(r_b)} \left(\frac{\tau^*}{\sqrt{-g_{00}(r_a)}}+t_c\right).
\end{equation}
Then, event A being in the past lightcone of B translates to $\tau_{bf}<\tau^*$, and that condition is satisfied if
\begin{equation}
    \tau^*>t_c \frac{-g_{00}(r_b)}{1-\sqrt{\frac{g_{00}(r_b)}{g_{00}(r_a)}}}.
\end{equation}
This reasoning can be applied analogously for configuration $\on{K_{B\prec A}}$, with B entering the past lightcone of A. Since the events are timelike separated  for both classical cases, every observer of each spacetime agrees on the respective time orderings: A happens before B for $\on{K_{A\prec B}}$ and B happens before A for $\on{K_{B\prec A}}$. This illustrates the dynamical structure allowed in GR: the causal relation between two events can be a consequence of how a third agent prepared the mass configuration in their causal past.

If we consider that operational events can also be defined with respect to clocks that are quantum, this generates other consequences. For instance, they can be in quantum superposition of paths on a definite spacetime feeling distinct time dilations \cite{Zych2011,Zych_2012,Pikovski2015,Roura,Esteban2020quantum}. One can imagine this interpretation of events could possibly generate indefinite causality even when a classical spacetime description is available. Indeed, the next chapter is dedicated to an example of this. However, assuming this operational notion for events is also what allows us to understand agents on a quantum spacetime using an equivalence between some scenarios with delocalized clocks and scenarios with delocalized masses, as we will see in a moment.

\section{Gravitational quantum switch} \label{gravQSwitch}
In the configuration $\on{K_{A\prec B}}$ of Fig.~\ref{fig:ZychLightcones}, if one sends a quantum system in a pure state $\ket{\psi}$ to be manipulated at event A and subsequently sent to event B to suffer another operation, then $\on{a}$'s operation $\on{U_{A}}$ will be applied before $\on{b}$'s operation $\on{U_B}$, and analogously for $\on{K_{B\prec A}}$, leading to one of the final states: 
\begin{align}
    \ket{\tilde{\psi}_1}^{S} &= \on{U_B U_A}\ket{\psi}^{S},  \label{psi_1}\\
     \ket{\tilde{\psi}_2}^{S} &= \on{U_A U_B}\ket{\psi}^{S}. \label{psi_2}
\end{align} 
 Thus, if we consider the possibility that a massive configuration can be in quantum superposition, we expect that it could be used as quantum control of order in a quantum switch. To make this discussion meaningful, we will proceed with the following assumptions:

\begin{itemize}
    \item[i.] Macroscopically  distinguishable states of physical systems can be assigned orthonormal quantum states. 
    
    \item[ii.] Gravitational time dilation reduces to that predicted by GR in the classical limit.
    
    \item[iii.] The quantum superposition principle holds for all systems, independently of mass or nature.
\end{itemize}

From the perspective of a distant observer, we can use the same coordinates to describe situations happening in configurations $\on{K_{A\prec B}}$ and $\on{K_{B\prec A}}$, as we did in the last section. These coordinates define a foliation of spacetime in equal time slices such that, for each fixed time t, we have a spacelike hypersurface to which we can associate a Hilbert space and specify the states of quantum systems living on top of it. Time t can be operationally defined as the time measured by the clock of the distant observer. The mass itself is among the systems being described and, from assumption i, we can assign two orthogonal quantum states $\ket{\on{K_{A \prec B}}}^M$ and $\ket{\on{K_{B \prec A}}}^M$ to the two mass configurations discussed. Since each state corresponds to a classical configuration, by assumption ii, the evolution of $S$ for each case happens according to usual time dilation resulting in \eqref{psi_1} for $\ket{\on{K_{A \prec B}}}^M$ and \eqref{psi_2} for $\ket{\on{K_{B \prec A}}}^M$. Finally, by assumption iii, the state $\frac{1}{\sqrt{2}}\ket{\on{K_{A \prec B}}}^M+\ket{\on{K_{B \prec A}}}^M$ is physically allowed.

\begin{figure}[ht]
    \centering
\includegraphics{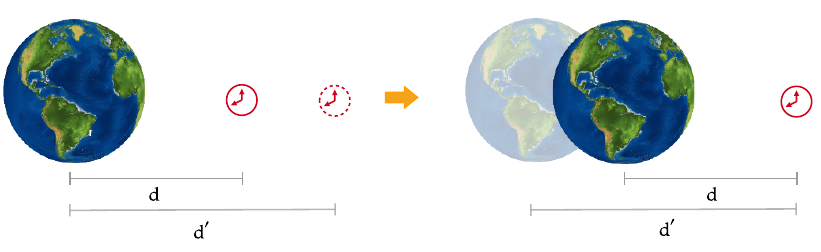}
    \caption{Portrayal of the principle of relativity of superpositions \protect\cite{ZychRelQuantSup} for a massive system. If gravitational influence only depends on relative distances, when a clock is in a spatial superposition of distances d and d' from a mass, this should be equivalent to the situation (if it exists) in which the mass is in the superposition state in the right side of the figure. Since a theory of Quantum Gravity is not available, the statement could be false, but the authors of~\protect\cite{tbell} make this assumption, among a few others, to understand its consequences for time order.}
    \label{fig:RelativitySuperpositions} 
\end{figure}

At this stage, there is one more consideration we should make. Suppose that a system 1 is interacting with a system 2 from a distance. If we move system 1 to the left and leave system 2 still, the situation is physically equivalent to equally moving system 2 to the right, leaving system 1 still. That means the interaction is invariant under translations, a standard physical requirement. Physical laws are usually invariant with respect to such symmetries, indicating that they depend only on relative quantities. Particularly, if  system 1 is a mass and 2 is a clock, the gravitational time dilation induced on system 2 depends only on relative distance. The authors of~\cite{tbell} argue, based on reference~\cite{ZychRelQuantSup}, that when a principle like this is valid classically, linearity of Quantum Theory implies it is also valid for each branch of a quantum superposition. In the following, this is assumed to be true, gravitational systems included. In analogy to the example above, this means that the situation in which a clock is in a spatial superposition of two distances from a localized mass is equivalent to the situation in which the mass is the one in a spatial superposition while the clock has a fixed location (see Fig.~\ref{fig:RelativitySuperpositions}). The equivalence statement refers to experiment probabilities, since both scenarios present the same probability amplitudes for the same relative distances.

 We can now calculate the joint evolution for the mass, the system $S$ and the clocks of agents $\on{a}$ and $\on{b}$ when the mass configuration is in a superposition state. Consider, in the following, that the systems are aligned so that the problem reduces to one spatial dimension. Denote the coordinate position of agent $\on{a}$ (and their clock) by $R_{a}$ and the mass position by $R_M$. The hamiltonian operator for this clock is the one we calculated for the post-newtonian metric in \eqref{ClockHamiltonianOperator}. We set $p=0$, because the clock is stationary, and take the low energy limit with $r= R_a - R_M >0$:
 \begin{equation}
     H_a\approx m_a c^2 + m_a\Phi(r) + H^a_{int}\left(1+\frac{\Phi(r)}{c^2}\right) = H_{rest}^a \left(1+\frac{\Phi(R_a - R_M)}{c^2}\right),
\end{equation}
 where $m_a$ is the clock's mass, $H_{int}$ is the hamiltonian of the dynamical part of internal evolution, while $H^a_{rest}:= m_ac^2 + H^a_{int}$ fully describes internal evolution. This is defined analogously for the clock of agent $\on{b}:$
 \begin{equation}
     H_b\approx H_{rest}^b \left(1+\frac{\Phi(R_b - R_M)}{c^2}\right).
 \end{equation}The general evolution for $\on{a}$'s clock with respect to time t of the distant observer can be written as
 \begin{equation}\label{ClockEvolutionGQS}
     e^{-\frac{i}{\hbar} H_{rest}^a t\left[1+\frac{\Phi\left(R_a - R_M\right)}{c^2}\right]} \ket{R_a}\ket{s_a(\tau_0)}= \ket{R_a} \ket{s_a(\tau_0 +\tau(R_a -R_M,t))},
 \end{equation}
where $$\tau(R_a-R_M,t):= t \left(1+\frac{\Phi(R_a - R_M)}{c^2}\right).$$ As discussed in subsection \ref{subsecFirstQClocks}, the Hilbert space of a quantum clock is a tensor product between an internal and an external space. The vector $\ket{R_a}$ of the composite state above is represented in the position basis and belongs to the space $\mathcal{H}^a_{ext}$, while $\ket{s_a(\tau)}\in \mathcal{H}^a_{int}$ is represented in an internal clock basis that indicates the clock is reading proper time $\tau$.
 
 Assume that the clocks of $\on{a}$ and $\on{b}$ are initially synchronized with the clock of the distant agent so that at $t=0$ they are in states $\ket{s_a(\tau=0)}$ and $\ket{s_b(\tau=0)}$. Also consider the initial state of system $S$, the target, to be $\ket{\psi}^S$ and that agents $\on{a}$ and $\on{b}$ will apply their operations
$\on{U_{A/B}}$ at events A$=(R_a,\tau_a=\tau^*)$ and B$=(R_b,\tau_b=\tau^*)$ respectively. At last, let the initial state of the mass be the superposition state $\frac{1}{\sqrt{2}}\left(\ket{\on{K_{A \prec B}}}+\ket{\on{K_{B \prec A}}}\right)^M$ of the two configurations in Fig.~\ref{fig:ZychLightcones}.

Therefore the initial state of the system is 
\begin{equation}\label{initialStateGQS}
    \ket{\psi(0)}=\frac{1}{\sqrt{2}}\ket{R_a} \ket{R_b}\ket{s_a(\tau=0)} \ket{s_b(\tau=0)} \ket{\psi}^S \left(\ket{\on{K_{A \prec B}}}+\ket{\on{K_{B \prec A}}}\right)^M,
\end{equation}
and the total hamiltonian for this scenario is
\begin{equation}
    H_{tot}= H_a+H_b+\on{\mathcal{O}_{A}}+\on{\mathcal{O}_{B}},
\end{equation}
with operators $\on{\mathcal{O}_{A/B}}$ defined as \begin{equation}
\on{\mathcal{O}_{A/B}}:= \delta(\tau-\tau^*,r-R_{a/b}) \on{O_{A/B}}, \qquad  e^{-i\on{O_{A/B}}} := \on{U_{A/B}},
\end{equation}
where the symbol $\delta$ represents a Dirac delta. Let us call $r_a=R_a-R_M$ the distance between the mass and clock $\on{a}$ in configuration $\on{K_{A\prec B}}$, then $r_b=r_a-h$ is the distance between the mass and clock $\on{b}$. We can use equation \eqref{ClockEvolutionGQS} to evolve the first term of the sum in \eqref{initialStateGQS}. Looking only at the clocks we get: $\ket{s_a(\tau(r_a,t))} \ket{s_b(\tau(r_a-h,t))}$. The same reasoning can be applied to the other configuration by making $r_a \to r_a-h$ and $r_a-h\to r_a$. If we evolve the total initial state until a time t such that both clocks are reading proper times greater than $\tau^*$ we have
\begin{multline}\label{preSyncEvolution}
     \ket{\psi(t)}=\frac{1}{\sqrt{2}}\ket{R_a} \ket{R_b}\Bigl[\ket{s_a(\tau(r_a,t))} \ket{s_b(\tau(r_a-h,t))}e^{-i\on{O_B}} e^{-i\on{O_A}} \ket{\psi}^S \ket{\on{K_{A \prec B}}}^M 
     \\
   +  \ket{s_a(\tau(r_a-h,t))} \ket{s_b(\tau(r_a,t))}e^{-i\on{O_A}} e^{-i\on{O_B}} \ket{\psi}^S\ket{\on{K_{B \prec A}}}^M\Bigr].
\end{multline}
Here we used relativity of superpositions. That is, we assumed that the form of the entanglement generated by time dilation between the mass configuration state and the internal states of the clocks is analogous to the entanglement between position and internal states of the clocks if they were the ones in a delocalized state $\ket{R_a}^a\ket{R_b}^b\to\ket{R_a}^a\ket{R_b}^b+\ket{R_b}^a\ket{R_a}^b$, and the mass was in a fixed position $\ket{\on{K_{A \prec B}}}^M+\ket{\on{K_{B \prec A}}}^M\to \ket{R_M}^M$. This is why the quantum clocks formalism was used from the beginning. The above only makes sense if the degrees of freedom of the clocks, at least the internal ones, are quantum. 

The internal clock states in \eqref{preSyncEvolution} can be disentangled from the joint state if the mass distributions are swapped right after the operations are applied: $\ket{\on{K_{B \prec A}}}^M \to \ket{\on{K_{A \prec B}}}^M$, $\ket{\on{K_{A \prec B}}}^M\to\ket{\on{K_{B \prec A}}}^M$. This way, the clocks that were delayed start to run faster, the ones that were faster get delayed and, after a time t, they get synchronized again, leaving us with a pure state for the rest of the systems:
\begin{equation}
    \ket{\psi(t)}=\ket{R_a} \ket{R_b}\ket{s_a(\tau_f)} \ket{s_b(\tau_f)}\left( \on{U_B} \on{U_A} \ket{\psi}^S \ket{\on{K_{B \prec A}}}^M+ \on{U_A} \on{U_B}\ket{\psi}^S\ket{\on{K_{A \prec B}}}^M\right),
\end{equation}
where $\tau_f:=\tau(r_a,t)+\tau(r_a-h,t)$. So, we  have the state of a quantum switch, like that in \eqref{QSstate}. The control system is the spacetime in the region. As in any other quantum switch (see section \ref{quantum switch}), the control has to be measured in a diagonal basis so that the final state results in a superposition of orders instead of a classical mixture. Thus, after a projective measurement of the mass in basis $\left(\ket{\on{K_{A \prec B}}}\pm  \ket{\on{K_{B \prec A}}}\right)^M$, the final state is
\begin{equation}
   \on{U_B} \on{U_A} \ket{\psi}^S \pm \on{U_A} \on{U_B}\ket{\psi}^S=\ket{\psi_1}\pm \ket{\psi_2}. 
\end{equation}

Note that the clocks of $\on{a}$ and $\on{b}$ used to define the events of interest are ideal, meaning we do not specify the form of their internal dynamics. In fact, this should not matter for the description. We can always simulate time dilation effects for specific clock systems by artificially interacting with their internal dynamics, for instance using electric and magnetic fields. But we can only assert time dilation is happening if any possible internal dynamics is time dilated in the same form (see chapter~\ref{Chap Clocks}). Thus, the assignment of events and surrounding causal structure is universal in this sense.

The difference between this quantum switch and the experimental implementations in~\ref{quantum switch} is that operations $\on{U_A}$ and $\on{U_B}$ are set to be realized at fixed operational events, in this case $(R_a,\tau^*)$ and $(R_b,\tau^*)$ respectively. So, the knowledge of the time of application of the operations does not create a problem but is indeed assumed, while in current implementations of the quantum switch, acquiring information about when the operations are applied destroys the superposition state, like a which-path indicator. Since there is no restriction here on how the agents measure time, it is actually crucial for the protocol of this chapter that the relative ordering of operational events is universally altered. Gravity enables this by providing genuine time dilation. The final indefinite order of operations in the gravitational switch therefore comes as a consequence of the more general indefinite operational order structure that emerges in the quantum spacetime.

Let us present a variation of this protocol that uses another pair of configurations to achieve the switch. We still consider two configurations in one spatial dimension. In one configuration, the mass is at a distance $r$ from agent $\on{a}$ and at a distance $r+h$ from $\on{b}$, while in the other the mass is located further to the left by an amount of $L$, as shown in Fig.~\ref{fig:GQSvariation} below. 
\begin{figure}[ht]
    \centering
    \includegraphics[scale=0.95]{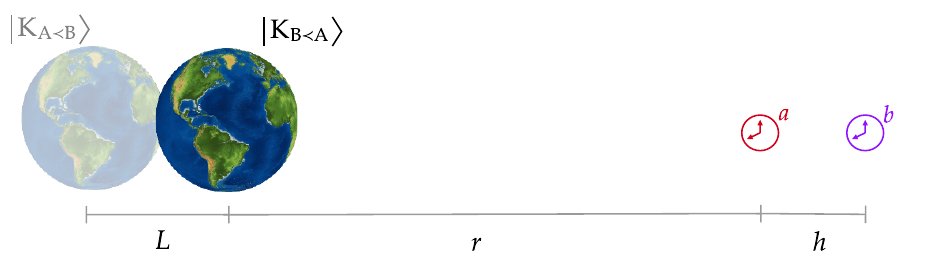}
    \caption{Alternative scenario for the realization of a gravitational quantum switch. The mass is on the left of the agents in both configurations of the superposition as seen by the distant observer. Even then, we can find operational events A and B such that A(B) is in the past lightcone of B(A) for $\on{K_{A\prec B}}\left(\on{K_{B\prec A}}\right).$}
    \label{fig:GQSvariation}
\end{figure}

Note that the mass is on the left of agents $\on{a}$ and $\on{b}$ for both configurations. In the last section we defined events A and B as the clocks measuring the same proper time $\tau^*$ only for simplicity. Generally, if an event A is defined as the agent $\on{a}$ measuring a proper time $\tau_a^*$ and event B as the agent $\on{b}$ measuring a proper time $\tau_b^*$ (not necessarily equal to $\tau_a^*$), then we can find values for $\tau_a^*$ and $\tau_b^*$ such that A$\prec$B in one configuration and B$\prec$A in the other.

In order to have A$\prec$B in configuration $\ket{\on{K_{A\prec B}}}$ we demand that a photon departing from event A arrives at agent $\on{b}$'s position before proper time $\tau_b^*$. This condition reads:
\begin{equation}\label{taubfCondition}
    \tau_{bf}=\sqrt{-g_{00}(r+L+h)}\left(\frac{\tau_a^*}{\sqrt{-g_{00}(r+L)}}+t_c(r+L,h)\right) \leq \tau_b^*,
\end{equation}
where $t_c(x,y)=\int_x^{x+y}dt$ is the coordinate time interval for light to travel between $x$ and $x+y$ given by equation \eqref{lightTravelt}. On the other hand, we also need that the light signal sent from event B arrives at agent $\on{a}$ before $\tau_a^*$ for configuration $\ket{\on{K_{B\prec A}}}$. Thus,  
\begin{equation}\label{tauafCondition}
    \tau_{af}=\sqrt{-g_{00}(r)}\left(\frac{\tau_b^*}{\sqrt{-g_{00}(r+h)}}+t_c(r,h)\right) \leq \tau_a^*.
\end{equation}
Let us fix $\tau_b^*=\tau_{bf}$, meaning that a photon would arrive just in time to suffer the operation at event B. Then, we can substitute \eqref{taubfCondition} in \eqref{tauafCondition}, which results in
\begin{equation}\label{tauafCondition2}
    \tau^*_{a}\geq \sqrt{-g_{00}(r)}\frac{\sqrt{\frac{g_{00}(r+L+h)}{g_{00}(r+h)}} t_c(r+L,h) + t_c(r,h)}{1 -\sqrt{\frac{g_{00}(r+L+h)g_{00}(r)}{g_{00}(r+h)g_{00}(r+L)}}}.
\end{equation}
Therefore, choosing some $\tau_a^*$ satisfying this and $\tau_b^*=\tau_{bf}$, we can use the pair of configurations to realize gravitational quantum control of orders just like in the previous case.

\section{Bell's theorem for temporal order}

 Although the assumptions in the last section seem natural enough for quantum gravity scenarios and lead to the formulation of a quantum switch, we are in new territory. As briefly commented in section \ref{quantum switch}, the quantum switch cannot violate any causal inequality. The only way we know of certifying indefiniteness of order from the probabilities of a quantum switch is to attest that it is a causally non-separable process (see section \ref{quantum switch}). This method, however, assumes the validity of Quantum Theory when local operations are applied on the target. So, even if we predict correct probabilities in the gravitational switch, one could question whether indefinite order has a physical meaning or there are unfulfilled hypotheses and the effect is just apparent, specially in this case where the required description of the protocol is coordinate/observer dependent. We had to make it that way, because the general coordinate-independent meaning of quantum states and operations is not clear in a quantum spacetime. All indicates that we should not rely on a test that depends on the validity of the quantum formalism. We should look for a theory-independent certification of indefinite order to account for such scenarios.
 
 Next, we formulate a specific task that would be impossible to realize if all events involved in it had definite order relations~\cite{tbell}. This resembles the causal inequality task of section~\ref{causalineqsec}, however, we note that some extra assumptions are made to attain such a result. While causal inequalities depend mainly on assuming causal structure operationally~\cite{Oreshkov,Branciard_2015}, to state the Bell's theorem for temporal order, we further need the notion of state and physical transformation. Therefore, this result will be a little more restrictive than obedience to causal inequalities, since it only makes sense for a class of general probabilistic theories in which these notions can be defined. 
 
 \subsection{Definitions}
Here, we make the definitions required for the formulation and proof of Bell's theorem for temporal order. In the context of generalized probabilistic theories, we define a \textbf{state} $\omega$ as the specification in an experiment of the conditional probabilities in the form $P(o|i,\omega)$, which denotes the probability of measuring outcome $o$ given a measurement with setting $i$ is performed on the system in state $\omega$. A \textbf{transformation} is given by a function taking states to states, $\omega\mapsto\on{T}(\omega)$.

\begin{definition}
The state $\omega$ of a system composed by 2 subsystems $\on{S}_1$ e $\on{S}_2$ is called a \textbf{product state}, denoted by $\omega_1 \otimes \omega_2$, if the joint probabilities for measurements $i_1$ and $i_2$ on each of them factorize as
\begin{equation}
    P\left(o_{1}, o_{2} | i_{1}, i_{2}, \omega\right)= P\left(o_{1} | i_{1}, \omega_{1}\right) P\left(o_{2} | i_{2}, \omega_{2}\right),
\end{equation}where $\omega_i$ is a state for system S$_i$, $i=1,2.$
\end{definition}

\begin{definition}\label{SeparableStateDef}
     A composite bipartite state $\omega$ is called \textbf{separable} if its joint probabilities can be written as
\begin{equation}
 P(o_{1}, o_{2} | i_{1}, i_{2}, \omega)=\int \mathrm{d} f \rho(f) P(o_{1} | i_{1}, \omega_{1}^{f}) P(o_{2} | i_{2}, \omega_{2}^{f}),
\end{equation}for some variable $f$, a probability distribution $\rho(f)$ and families of states $\omega_{1}^f$ of  S$_1$ and $\omega_{2}^f$ of S$_2$. To indicate this is the case, we denote the state as $\omega=\int \mathrm{d} f \rho(f) \omega_{1}^{f} \otimes \omega_{2}^{f}.$
\end{definition}

\begin{definition}
    A transformation on a composite system is \textbf{local} if, for every product state taken as input, it acts non-trivially on just one subsystem. For instance, if $\on{T}$ acts locally on S$_1$, we have
\begin{equation}
    \on{T}(\omega_1 \otimes \omega_2) = \on{T}_1(\omega_1) \otimes \omega_2,
\end{equation} where $\on{T}_1$ is a transformation on system S$_1$. Therefore, the joint probabilities for the transformed state when the input is a product state are in the form
\begin{equation}
    P\left(o_{1}, o_2| i_{1}, i_2, T(\omega_{1}\otimes \omega_2)\right) =P\left(o_{1} | i_{1}, T_1(\omega_{1})\right) P\left(o_{2} | i_{2}, \omega_{2}\right).
\end{equation} 
\end{definition}

  Local transformations $\on{T}_1$ and $\on{T}_2$ on systems $\on{S}_1$ and $\on{S}_2$ combine depending on their spatio-temporal locations: if they are spacelike separated, the resulting tranformation is $\on{T}_1 \otimes \on{T}_2$ that acts like $\omega_1 \otimes \omega_2 \mapsto \on{T}_1(\omega_1) \otimes \on{T}_2(\omega_2)$ on a product state. If $T_1$ is in the future of $\on{T}_2$ the transformation becomes $\on{T}_1 \circ \on{T}_2(\omega)$ and analogously for $\on{T_2}$ in the future of $\on{T}_1$. 
  
 \subsection{Formulation}

 Consider a system S composed of two subsystems S$_1$ and S$_2$, on which operations will be made, and a third system denoted by M. Also consider that the system S$_1$ suffers a transformation $\on{T_{A_{1}}}$ at an event $\on{A}_{1}$ and a transformation $\on{T_{B_{1}}}$ at an event $\on{B}_{1}$. Similarly, the system S$_2$ undergoes transformations $\on{T_{A_{2}}}$ and $\on{T_{B_{2}}}$ at events $\on{A}_{2}$ and $\on{B}_{2}$. Moreover, each system S$_j$ is measured at event C$_j$ with a setting $i_j$ producing and outcome $o_j$, $j=1,2$. The system M is also measured at an event D producing outcome $z$. The setup for the task can be visualized in the diagram of Fig.~\ref{fig:hipTBellZych}.

\begin{figure}[ht]
    \centering
    \includegraphics[scale=0.3]{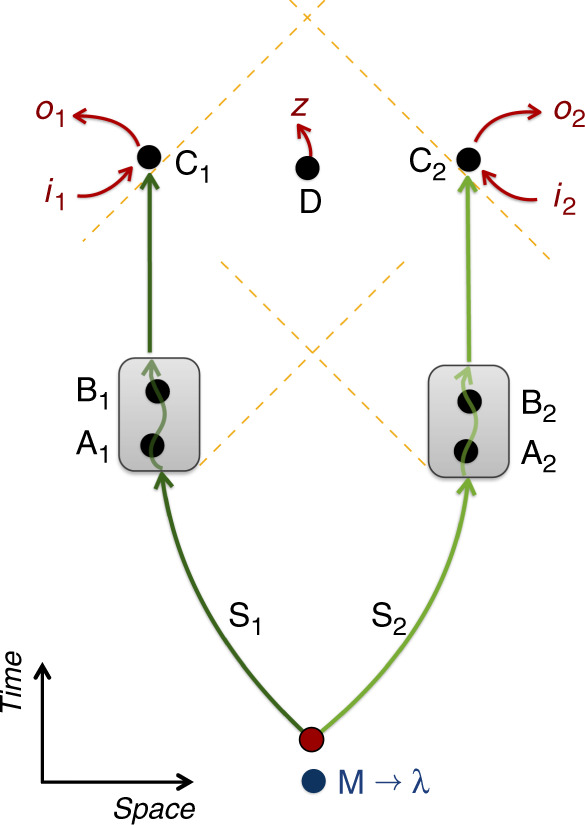} 
    \caption{Setup corresponding to the hypotheses of the Bell's theorem for temporal order. Local transformations $\on{T}_{\on{A}_1},\on{T}_{\on{B}_1}, \on{T}_{\on{A}_2}, \on{T}_{\on{B}_2}$ are applied on systems S$_1$ and S$_2$ at the indicated events. At events $\on{C}_1$ and $\on{C}_2$, measurements $i_1,i_2$ are made on S$_1$ and S$_2$, respectively, returning outcomes $o_1$ and $o_2$. Finally, at event D, system M is also measured, returning outcome $z$. The yellow lines represent lightcones. Figure from~\protect\cite{tbell}.} 
    \label{fig:hipTBellZych}
\end{figure}

Then, we can state the following:
\begin{thm}[Bell's theorem for temporal order]\label{BellsThm}
If an experiment in the setup above obeys the following assumptions, the final measurements on system S cannot result in a violation of any bipartite Bell inequality:
\begin{itemize}
\item[i] \textbf{Local initial state}: The total initial state $\omega$ describing the three systems, $\on{S}_1$, $\on{S}_2$ and $\on{M}$, is separable.
\item[ii] \textbf{Local operations}: All the transformations $\mathds{T} = \{\on{T}_{\on{A}_1},\on{T}_{\on{B}_1}, \on{T}_{\on{A}_2}, \on{T}_{\on{B}_2}\}$ and measurements are local.
    \item[iii] \textbf{Classical order}: The events in which the transformations are applied are classically ordered according to the definition in reference \cite{tbell}:
    A set of events is classically ordered if, for each pair of events $\on{A}$ and $\on{B}$,  there exists a spacelike surface and a classical variable $\lambda$ defined on it which settles the causal relation between the two events in a deterministic or probabilistic fashion (causal relations are determined by classical hidden variables). 

    \item[iv]\textbf{Spacelike separation}: The events $(\on{A}_1$,$\on{B}_1)$ are spacelike separated from events $(\on{A}_2$,$\on{B}_2)$ and the events $\on{C}_1$, $\on{C}_2$ and D are pairwise spacelike separated.
    \item[v]\textbf{Free choice}: The choices of measurement settings are statistically independent of the rest of the experiment. 
\end{itemize}
\end{thm}
 
\begin{proof}
Assumption i with definition \eqref{SeparableStateDef} implies that there exists a variable $f$ that determines the states $\omega^f_1$ and $\omega^f_2$ of subsystems $\on{S}_1$ and $\on{S}_2$. By assumption iii, there also exists a classical variable $\lambda$ determining the order of events. 

    Let us define a variable $\sigma_i$ that can assume, for each $i=1,2$, one of two values: (A$_i\prec $B$_i$) or (B$_i\prec $A$_i$), the possible permutations of the events A$_i$ and B$_i$. For example, we can write $P(\sigma_1|\lambda)$ to denote the probability that the ordering between the events A$_1$ and B$_1$ is $\sigma_1:=($A$_1 \prec$B$_1)$ for a given $\lambda$. 

Then, we define $\on{T}^{\sigma_i}$ as the composition of transformations associated to $\sigma_i$. For example,\begin{equation}
  \sigma_1=(\on{A}_1\prec \on{B}_1) \implies  \on{T}^{\sigma_1} := \on{T}_{\on{B}_1} \circ \on{T}_{\on{A}_1}.
\end{equation}The transformation applied on system S before measurements can be generally expressed as $ \on{T}^{\sigma_1}\otimes \on{T}^{\sigma_2}$, since all transformations are local by assumption ii.

From assumptions i and ii, we know that the state $\omega$ of system S is separable ($\int df P(f)\omega_1^f\otimes\omega_2^f$) and that a pair of operations is applied locally for each subsystem S$_1$ and S$_2$. Thus, the probabilities for results of measurements $i_1,i_2$ made at C$_1$ and C$_2$, knowing the operations and their order of application deterministically $(\mathds{T}, \sigma_1, \sigma_2)$, should look like\begin{align*}
    P(o_{1}, o_{2}, z | i_{1}, &i_{2}, \mathds{T}, \omega,\sigma_1, \sigma_2)= 
    \\
    &\int df P(f)P\left(o_{1}| i_{1}, \operatorname{T}^{\sigma_1}\left(\omega_1^f\right)\right) P\left(o_{2}| i_{2}, \operatorname{T}^{\sigma_2}\left(\omega_2^f\right)\right) P\left(z| f,\sigma_1,\sigma_2\right).
\end{align*}Note that, by assumption iv, the probability of getting outcome z at event D has no dependence on measurements or results obtained at C$_1$ and C$_2$. 

Now, we wish to compute the probability of obtaining results $o_1$, $o_2$ and $z$ at the
end of the experiment without conditioning on information about the order. By assumption iii, the order is given by the hidden variable $\lambda$. The dependence of the protocol statistics on the two variables $f$ and $\lambda$ can be combined in a joint probability distribution $P(f,\lambda)$. We can then write the probability above for the case in which order is not determined, which has explicit dependence on $\lambda$. After that, we integrate over $\lambda$ and sum over the possible permutations $\sigma_1,\sigma_2$ to get the result\begin{align}
    P(o_{1}, o_{2}, z | i_{1}, i_{2}, \mathds{T}, \omega) = &  \sum_{\sigma_{1} \sigma_{2}} \int d \lambda d f  P(\lambda, f)P\left(o_{1} |i_{1},T^{\sigma_{1}}\left(\omega_{1}^{f}\right)\right) \nonumber
    \\
    &\times P\left(o_{2} | i_{2}, T^{\sigma_{2}}\left(\omega_{2}^{f}\right)\right) P\left(\sigma_{1} | \lambda\right) P\left(\sigma_{2} | \lambda\right) P\left(z | \lambda, f, \sigma_{1}, \sigma_{2}\right), \label{decompProbBell}
\end{align}

The quantity $P:=P\left(\sigma_{1} | \lambda\right) P\left(\sigma_{2} | \lambda\right) P\left(z | \lambda, f, \sigma_{1}, \sigma_{2}\right) P(\lambda, f)$ in \eqref{decompProbBell} can be simplified. Since order is determined by $\lambda$ independently of $f$, we can use that $P(\sigma_{i}| \lambda)= P(\sigma_{i}| \lambda, f )$ to get
\begin{align}
 P&=P\left(\sigma_{1} | \lambda, f \right) P(\lambda, f) P\left(\sigma_{2} | \lambda, f \right) P\left(z | \lambda, f, \sigma_{1}, \sigma_{2}\right) \nonumber
     \\\nonumber \\
     &= P\left(\lambda, f,\sigma_{1}  \right) P\left(\sigma_{2} | \lambda, f \right) P\left(z | \lambda, f, \sigma_{1}, \sigma_{2}\right) \nonumber
     \\\nonumber \\
     &= P\left(\lambda, f, \sigma_{1} \right) P\left(\sigma_{2} | \lambda, f,\sigma_1 \right) \frac{P\left( \lambda, f, \sigma_{1}, \sigma_{2}|z\right)P(z)}{P( \lambda, f, \sigma_{1}, \sigma_{2})} \nonumber
     \\\nonumber \\
     &= P( \lambda, f, \sigma_{1}, \sigma_{2})\frac{P\left( \lambda, f, \sigma_{1}, \sigma_{2}|z\right)P(z)}{P( \lambda, f, \sigma_{1}, \sigma_{2})}\nonumber
     \\ \nonumber     \\
     &= P\left( \lambda, f, \sigma_{1}, \sigma_{2}|z\right)P(z),
\end{align}
where we used Bayes rule in all steps: in the first two terms from line 1 to 2, in the third from line 2 to 3 and in the first two terms again to get to the last line. From line 2 to line 3 we also used $P\left(\sigma_{2} | \lambda, f \right)=P\left(\sigma_{2} | \lambda, f, \sigma_1\right)$ since, again, the regions associated to indices 1 and 2 are spacelike separated preventing that dependence. Expression \eqref{decompProbBell} then becomes
\begin{align*} 
  P(o_{1},  o_{2}, z | &i_{1}, i_{2}, \mathds{T},\omega) = 
 \\ &\sum_{\sigma_{1} \sigma_{2}} \int d \lambda d f \,P\left(o_{1} | i_{1}, T^{\sigma_{1}}\left(\omega_{1}^{f}\right)\right)  P\left(o_{2} | i_{2}, T^{\sigma_{2}}\left(\omega_{2}^{f}\right)\right) P\left(\lambda, f, \sigma_{1}, \sigma_{2} | z\right) P(z).
\end{align*}
Dividing both sides by $P(z)$, we get the probability for results $o_1$ and $o_2$ conditioned on result $z$ :
\begin{align}
 P\left(o_{1}, o_{2} | i_{1}, i_{2}, z, \mathds{T}, \omega \right)=\hspace{2cm}&\nonumber
 \\
 \sum_{\sigma_{1} \sigma_{2}} \int d \lambda d f \, P\left(o_{1} | i_{1},  T^{\sigma_{1}}\left(\omega_{1}^{f}\right)\right)  &P\left(o_{2} | i_{2}, T^{\sigma_{2}}\left(\omega_{2}^{f}\right) \right) P\left(\lambda, f, \sigma_{1}, \sigma_{2} | z\right) \nonumber
 \\
&\hspace{-0.2cm}= \int d \bar{f} P\left(o_{1} | i_{1}, T^{\sigma_{1}}\right) P\left(o_{2} | i_{2}, T^{\sigma_{2}}\right) P(\bar{f} | z), \label{provaBellZych}
\end{align}
where we denoted by $\bar{f}$ the set of variables $\sigma_1$, $\sigma_2$, $\lambda$ and $f$. One can show that the possibility of expressing the joint probabilities in a factorized form like the one above is a necessary and sufficient condition for the existence of a deterministic hidden-variable model for an experiment~\cite{Fine}. And then, by Bell's theorem~\cite{Bell,CHSH}, Bell inequalities must be obeyed by the expected values and conditional probabilities from the measurements made at events C$_1$ and C$_2$. With the expression above, that must be true even if the outcomes are conditioned on the measurement result $z$ obtained at event D.
\end{proof}

\subsection{Violation}
Let us show how two copies of the gravitational quantum switch can be used to violate Bell's theorem for temporal order. The events of interest are defined here with respect to the clocks of the agents, as explained in section~\ref{operationalEvents}. 

Consider that a bipartite system S starts in a product quantum state $\ket{\psi_1}^{S_1}\ket{\psi_2}^{S_2}$. Each part of the system is sent to the events where the operations will be applied. For $i=1,2$, the agents $\on{a}_i,\on{b}_i,\on{c}_i$ define events A$_i$, B$_i$ and C$_i$. The agents $\on{a}_1,\on{b}_1$ are assumed to be sufficiently distant from the agents $\on{a}_2,\on{b}_2$ so that the events in which they apply operations according to their clocks, A$_1$ and B$_1$, are spacelike separated from A$_2$ and B$_2$. The same is assumed for the agents $\on{c_1,c_2,d}$ who determine the events for the measurements in Fig.~\ref{fig:hipTBellZych}. Thus, $\on{a}_i,\on{b}_i,\on{c}_i$ only interact with the system S$_i$. Now, consider that one can prepare classical configurations for a massive system associated to states $\ket{\on{K}}$ and $\ket{\on{K'}}$ such that, in configuration $\ket{\on{K}}$, the relations between events are A$_1\prec$ B$_1\prec$ C$_1$ and B$_2\prec$ A$_2\prec$ C$_2$, while in configuration $\ket{\on{K'}}$ they are B$_1\prec$ A$_1\prec$ C$_1$ and A$_2\prec$ B$_2\prec$ C$_2$. An example of such pair of configurations is given in Fig.~\ref{fig:BellTempOrderViolation}, where a massive system is acting as control of two switches of the type presented in Fig.~\ref{fig:GQSvariation}.

\begin{figure}
    \centering
    \includegraphics[scale=0.3]{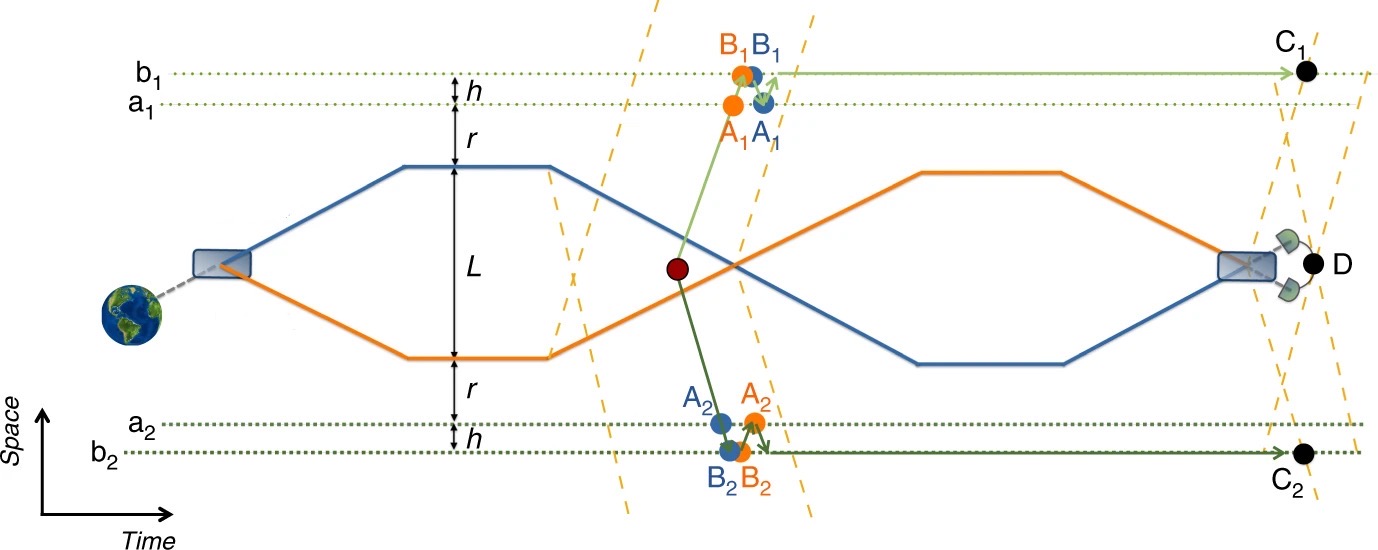}
       \caption{Protocol violating Bell's inequalities for temporal order using gravitational quantum control of order. The orange worldline represents configuration $\ket{\on{K}}^M$ of the massive system and the blue worldline represents $\ket{\on{K'}}^M$. The mass is prepared in state $1/\sqrt{2}\left(\ket{\on{K}}+\ket{\on{K'}}\right)^M$. The system S is prepared in the product state $\ket{\psi_1}\ket{\psi_2}$, and the subsystems S$_1$ and S$_2$ (dark and light green paths) are sent to spacelike separated regions to suffer two operations $\{\on{U_{A_i}},\on{U_{B_i}}\}$ for $i=1,2$ respectively. The relations between operational events are different depending on the configuration. For instance, the events A$_1$ and B$_1$ of configuration $\ket{\on{K}}$ (upper orange dots), obey A$_1\prec$B$_1$, while in configuration $\ket{\on{K'}}$ we have B$_1\prec$A$_1$ (upper blue dots). As a result, the quantum state of M enables entanglement of the total composite state. In particular, the system S$_1\otimes$S$_2\otimes$M can be put in a maximally entangled state for the right choice of operations, as long as one can disentangle the states of the agents' clocks from it. This is done in the protocol above by swapping the mass configurations in a symmetric fashion (the crossing of the blue and orange trajectories) and waiting for the clocks to resynchronize, as described in section \ref{gravQSwitch}. The subsystems of S are then measured at events C$_1$ and C$_2$, while the mass M is measured at D. Considering the state measured is maximally entangled, the results of measurements on S$_1\otimes$S$_2$ conditioned on the outcome of M can maximally violate a Bell inequality. Figure from~\protect\cite{tbell}.}
    \label{fig:BellTempOrderViolation}
\end{figure}

Therefore, if the mass is prepared in the superposition state $\frac{1}{\sqrt{2}}\left(\ket{\on{K}}+\ket{\on{K'}}\right)^M,$ we get an entangled state of orders. Indeed, the total state of systems S, M and the clocks of the agents after the operations is of the form
\begin{multline}     \frac{1}{\sqrt{2}}\ket{R_{a_1}} \ket{R_{b_1}}\ket{R_{a_2}} \ket{R_{b_2}}\Bigl[\ket{\tau_{a_1},\tau_{b_1},\tau_{a_1},
\tau_{b_2};\on{K},t} \on{U_{A_2}}\on{U_{B_2}} \ket{\psi_2}^{S_2}  \on{U_{B_1}}\on{U_{A_1}} \ket{\psi_1}^{S_1} \ket{\on{K}}^M 
     \\
+\ket{\tau_{a_1},\tau_{b_1},\tau_{a_1},
\tau_{b_2};\on{K'},t} \on{U_{B_2}}\on{U_{A_2}} \ket{\psi_2}^{S_2}  \on{U_{A_1}}\on{U_{B_1}} \ket{\psi_1}^{S_1} \ket{\on{K'}}^M \Bigr],
\end{multline}
where the states $\ket{R_{x}}$ are position state of the agents, who remain stationary, while the ket $\ket{\tau_{a_1},\tau_{b_1},\tau_{a_1},
\tau_{b_2};\on{K},t}$ represents the internal states of all clocks, which depend on the mass configuration because of time dilation. In order to get a pure state for the system S$\otimes$M, one can swap the mass configuration after the operations are done. This is depicted in Fig.~\ref{fig:BellTempOrderViolation}. After the same amount of coordinate time t has passed in the swapped scenario, the clocks will disentangle from the state above because $\ket{\tau_{a_1},\tau_{b_1},\tau_{a_1},
\tau_{b_2};\on{K},2t}= \ket{\tau_{a_1},\tau_{b_1},\tau_{a_1},
\tau_{b_2};\on{K'},2t}.$
The pure final state of the main system is thus given by\begin{align}
\ket{\psi}^{SM}=\frac{1}{\sqrt{2}} \on{U_{A_2}}\on{U_{B_2}} \ket{\psi_2}^{S_2}  \on{U_{B_1}}\on{U_{A_1}} \ket{\psi_1}^{S_1} \ket{\on{K}}^M + \on{U_{B_2}}\on{U_{A_2}} \ket{\psi_2}^{S_2}  \on{U_{A_1}}\on{U_{B_1}} \ket{\psi_1}^{S_1} \ket{\on{K'}}^M.
\end{align} And if agent d makes a measurement in the basis $\ket{\pm}=\frac{1}{\sqrt{2}}\left(\ket{\on{K}}\pm\ket{\on{K'}}\right)^M$ at event D, the joint state of S$=$S$_1\otimes$S$_2$ conditioned on this result reads
\begin{equation}\label{finalConditionedState}
 \ket{\psi}^{S}_{\pm}=\frac{1}{\sqrt{2}} \on{U_{A_2}}\on{U_{B_2}} \ket{\psi_2}^{S_2}  \on{U_{B_1}}\on{U_{A_1}} \ket{\psi_1}^{S_1}  \pm \on{U_{B_2}}\on{U_{A_2}} \ket{\psi_2}^{S_2}  \on{U_{A_1}}\on{U_{B_1}} \ket{\psi_1}^{S_1}.
\end{equation}
Both of these states can violate a Bell inequality for the right choices of unitary operations made by a$_i$ and b$_i$ and measurements for agents c$_i$. 

For example, consider that S$_1$ and S$_2$ are two-level systems starting at initial state $\ket{\psi_1}^{S_1}\ket{\psi_2}^{S_2}=\ket{\uparrow}\ket{\uparrow}$, written in the z basis. We can choose the operations
\begin{equation}
    \on{U_{A_1}}=\on{U_{B_2}} =\frac{\mathds{1}+\sigma_x}{\sqrt{2}}  \qquad \on{U_{A_2}}=\on{U_{B_1}}=\sigma_{z},
\end{equation}
where $\sigma_i$ are the Pauli matrices. One can quickly check that 
$$\frac{\mathds{1}+\sigma_x}{\sqrt{2}} \sigma_z = \frac{\sigma_z+\sigma_y}{\sqrt{2}},\qquad  \sigma_z\frac{\mathds{1}+\sigma_x}{\sqrt{2}}=  \frac{\sigma_z-\sigma_y}{\sqrt{2}}$$
Then, the state \ref{finalConditionedState} for this case becomes
\begin{align}
     \ket{\psi}^{S}_{\pm}&=\frac{1}{\sqrt{2}}\left[\left(\frac{\ket{\uparrow}+\ket{\downarrow}}{\sqrt{2}}\right)^{S_1} \left(\frac{\ket{\uparrow}+\ket{\downarrow}}{\sqrt{2}}\right)^{S_2}  \pm \left(\frac{\ket{\uparrow}-\ket{\downarrow}}{\sqrt{2}}\right)^{S_1} \left(\frac{\ket{\uparrow}-\ket{\downarrow}}{\sqrt{2}}\right)^{S_2} \right] \nonumber
     \\
     &=:\frac{1}{\sqrt{2}}\left[\ket{+}\ket{+}\pm \ket{-}\ket{-} \right].\label{entangledState}
\end{align}
A Bell inequality known as CHSH\cite{CHSH,Fine} can be evaluated in this setting. It is a statement on the expectation values of measurements on a bipartite system. It also considers that two agents c$_1$ and c$_2$ realize measurements on subsystems S$_1$ and S$_2$ of a bigger system S, while maintaining spacelike separation. For the task, agent c$_1$ can choose between two measurement settings $i_1=\{0,1\}$, and analogously c$_2$ can choose between a different pair $i_2=\{0,1\}$. Considering that the measurements chosen can only have outcomes $\pm 1$, the inequality below must be obeyed in a classical setting:
\begin{equation}
    E(i_1=0,i_2=0) +   E(i_1=0,i_2=1)  + E(i_1=1,i_2=0) -  E(i_1=1,i_2=1)\leq2,
\end{equation}
where $E(i_1,i_2)$ denotes the expectation value for the product of the observables being measured with settings $i_1$ and $i_2$. If we choose the possible settings to be 
\begin{align}
    i_1&=\left\{0:\text{Measurement of }\frac{\sigma_y-\sigma_z}{\sqrt{2}}, 1:\text{Measurement of }\frac{\sigma_y+\sigma_z}{\sqrt{2}}\right\}
    \\
    i_2&=\left\{0:\text{Measurement of }\sigma_y, 1:\text{Measurement of }\sigma_z\right\},
\end{align}
it is possible to calculate each expectation value ($E=\braket{\psi|A|\psi}$ for an observable $A$) and verify that the states in \eqref{entangledState} maximally violate the inequality:
\begin{equation}
    \bigl[E(0,0) + E(0,1)  + E(1,0) -  E(1,1)\bigr]_{\ket{\psi}^{S}_{\pm}} = \mp 2\sqrt{2}.
\end{equation}
Hence, the gravitational quantum switch in Fig.~\ref{fig:BellTempOrderViolation} does not obey the statement of Bell's theorem for temporal order in a situation in which all the assumptions in \ref{BellsThm} are believed to be satisfied, except for classical order of the events in the experiment.

Although we are using a picture of Earth to represent the massive system in figures, one might wonder how far are we from the realization of these protocols using microscopic or mesoscopic systems as the mass. Would we in principle need to go to the Planck scale to achieve indefinite order with the gravitational influence of these systems? Despite the complete infeasibility of the experiment in Fig.~\ref{fig:BellTempOrderViolation} in laboratory, the authors of~\cite{tbell}  analyze approximate  conditions on clocks and masses for small systems and tiny distances to answer this negatively: effects due to quantum gravity in the low-energy regime like this switch of orders are expected to show up for scales much bigger than Planck's, although still not very close to experimental realization. They also consider models in which a superposition of mass is postulated to decohere~\cite{Scully2018,Diosi,Penrose} and show that the duration of the protocol can in principle be made smaller than the decoherence time scale, so the violation could be verified even for these theories.

\section{Remarks}
The reader may have noticed a slight change of vocabulary in this chapter. For instance, the terms ``indefinite \emph{causal} order/structure'', ``indefinite \emph{temporal} order'', or even just ``indefinite order'' can appear interchangeably, although we tried to stick to the original works. By the definitions, we can verify that the term ``temporal order'', used in reference \cite{tbell}, refers to the same concept we were calling ``causal order'' in the first chapters, and the authors even compare their results to causal inequalities and causal witnesses. 
The former term was used in theoretical works involving universal time dilation as well as in experimental works about table-top implementations\cite{RubinoAgain}, while the latter is the most common, appearing in most works about order related to abstract computational structures, quantum information in general and also in works on the intersection of gravity and quantum theory~\cite{Voji}. Indeed, different authors hold different opinions on when to attribute the word `(non-)causal' to elements of their formalisms and there are a lot of subtleties to consider, such as the existence of structure to even start investigating causality, the nature of events, whether something can only be done using a close timelike curve, whether causal paradoxes emerge and can be probed by Alice and Bob, etc. This is a relatively new area of research, specially regarding explicit descriptions of order on quantum spacetimes as done here. Either way, we opt for mostly using the term ``indefinite order'' and rely on the definitions throughout the thesis. 

The idea of order remains fairly similar in this chapter, but the events related by that order are described rather differently. The authors of~\cite{tbell} specify their notion of operational event and the events in each protocol, which is hardly done in other works on indefinite orders. A clear notion of events, which does not necessarily need to be this one, provides better grounds to argue that a protocol should be interpreted as the realization of a specific process. The definition can also be useful for understanding what order can (and cannot) inform us about gravity and its quantum nature. For instance, velocities can also cause universal time dilation on clocks. In particular, some gravitational situations are analogous to accelerated frames, due to the equivalence principle. Some works explored the idea of indefinite order using superpositions of these frames~\cite{Rindler,RindlerDaZych}. We also have the fact that quantum clocks in superposition of paths on top of classical spacetimes can produce analogous results~\cite{Zych,QSonEarth} (indeed, the quantum gravity version is based on that relation). Despite all of these possible analogies in classical spacetimes, protocols with specified operational events are still interesting because they all seem to require genuine time dilation to produce indefinite orders, at least in the formulation that uses ideal clocks. Although the probabilities of the switch are always the same, the setup of the gravitational quantum switch in the quantum gravity case still has differences from its analogous classical spacetime counterparts. For instance, a quantum superposition of the mass would cause indefiniteness of order relations for several events defined by clocks surrounding it, while in the counterpart where a quantum clock is near a classical mass (see Fig.~\ref{fig:RelativitySuperpositions}), only that one clock attests indefiniteness of order. For more comments on this, the reader is referred to the supplemental material of~\cite{tbell}.

The result of Bell's theorem for temporal order in the last section is proposed as a theory independent but device-dependent test, meaning it relies on the notions of states and transformations and not just pure probabilities. In the hierarchy of tests of indefinite order, its violation should lie in between the detection of a causal witness~\cite{Araujo2015}, which is theory and device-dependent, and the violation of a causal inequality~\cite{Oreshkov}, which is independent of both. However, the validity of the theorem as a theory independent result was very recently questioned in reference~\cite{RindlerDaZych} by some of the same authors who proposed the theorem. They point out that one assumption is made but not stated: that the evolution of system S is trivial except for the transformations applied at events A$_i$, B$_i$ and C$_i$. Adding that to the list of assumptions makes the theorem work fine, but the authors argue that theory dependent notions would be required to properly write such an assumption on evolutions. If this is confirmed, it weakens the interpretation of the violation of the inequalities as a signature of indefinite order, putting it closer to a causal witness test. 
 
 \chapter{Quantum switch in a classical curved spacetime}\label{Chap QSonEarth}
 In this chapter, we describe the original work~\cite{QSonEarth}, which includes results obtained during the development of this Master's project. It consists of a formulation of a quantum switch in a classical curved spacetime. This switch uses quantum agents on a spacetime produced by a spherical mass to generate indefinite order of operations on a target system. The indefiniteness of order comes as a result of the entanglement between proper time rate and relative distance from the mass for different positions of the agents. The case for which the massive body is the Earth is analyzed. This protocol shares similarities with the gravitational quantum switch of the previous chapter, even though the gravitational field is not in a quantum superposition. The point in studying this specific implementation is that the switch can be used as a resource probe whether delocalized quantum systems feel distinct gravitational time dilations as expected from them. Its realization would probe the physical regime described by quantum mechanics on curved spacetimes, which has
not yet been explored experimentally. 

\section{Introduction}
As we have seen, the quantum switch is a class of protocols, involving a \emph{target} quantum system, two \emph{operations} $\mathcal{A},\mathcal{B}$ that can be applied on it and a way to associate the order of their application to the state of another quantum system, the \emph{control}. If the control is prepared in a superposition, the final state of the target indicates that there was a superposition of the orders of application: $\mathcal{A}$ before $\mathcal{B}$ and $\mathcal{B}$ before $\mathcal{A}$. The quantum switch is an example of process with indefinite order that has been reproduced in optical tables~\cite{Procopio,Rubino,Goswami,Taddei}, where the control and target are usually two degrees of freedom (dofs) of the same system, such as the path and polarization of a photon. Indefiniteness of order in a quantum switch can be testified using a causal witness~\cite{Araujo2015,Rubino,Goswami} and via the Bell's theorem for temporal order~\cite{tbell,RubinoAgain}.

We have discussed the gravitational quantum switch, a thought experiment in which the control system is spacetime itself. It is just one of the many proposals  in quantum gravity phenomenology making use of the hypothesis that the gravitational field can exist in a superposition of classical configurations~\cite{tbell,Ford,Anastopoulos,Bose, Vedral,Belenchia,Rovelli,Howl}. The relation between the gravitational quantum switch and optical implementations in classical spacetime was discussed in~\cite{Voji}. A proposal for simulating the gravitational quantum switch using accelerated agents on Minkowski spacetime was described in~\cite{Rindler}. 

With the current technology, it is still a challenge to put a body of sufficient mass in a superposition for enough time to realize the gravitational quantum switch. In the present work, we introduce a strategy for implementing a quantum switch that also happens because of the gravitational time dilation in an analogous way, but in a classical background. For this, we can borrow some ideas from the gravitational protocols we have seen in chapter~\ref{Chap:QGravitySwitch}. In the gravitational quantum switch, the relevant feature for indefinite order to emerge is the relative distance from the massive spherical body to the agents. In the simplest case presented in section~\ref{gravQSwitch}, the mass is in a superposition of being on the left or on the right side of two centralized agents, see Fig.~\ref{fig:ZychLightcones}. The control of order can also be achieved in a less symmetric situation as in the figure below.
\begin{figure}[ht]
\centering
   \includegraphics[scale=0.22]{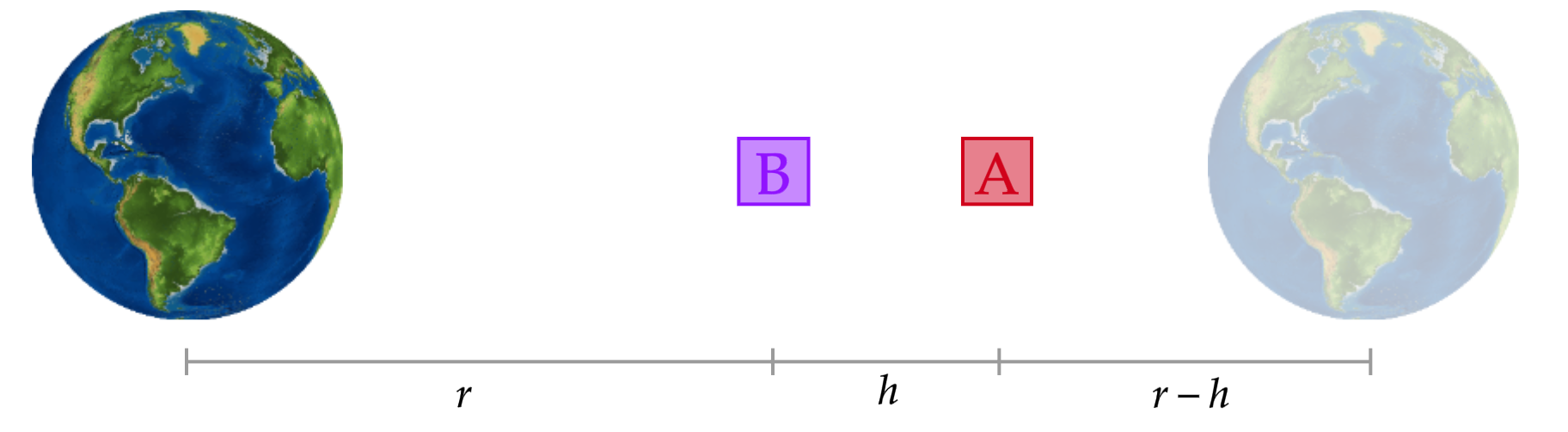}
    \caption{A mass configuration for a gravitational quantum switch.}
    \label{fig:Earthsup}
\end{figure}   

The boxes A and B represent the laboratories where the agents realize operations $\mathcal{A}$ and $\mathcal{B}$. The superposition manifests itself in the proper times perceived by the agents' clocks due to gravitational time dilation. The calculations for choosing the events A and B for the switch are analogous to what was made in \eqref{tauafCondition2}, using the distances of the problem. The target system then travels back and forth a number of times between the laboratories, but the agents can only apply their operations once, when their clocks show a predetermined proper time $\tau_{a/b}^*$. This procedure could also be accomplished, for example, if agent A was a quantum system capable of interacting with the target and if we prepared it in a superposition of positions as depicted in Fig.~\ref{fig:agentssup}.
\begin{figure}[ht]
\centering
\includegraphics[scale=0.22]{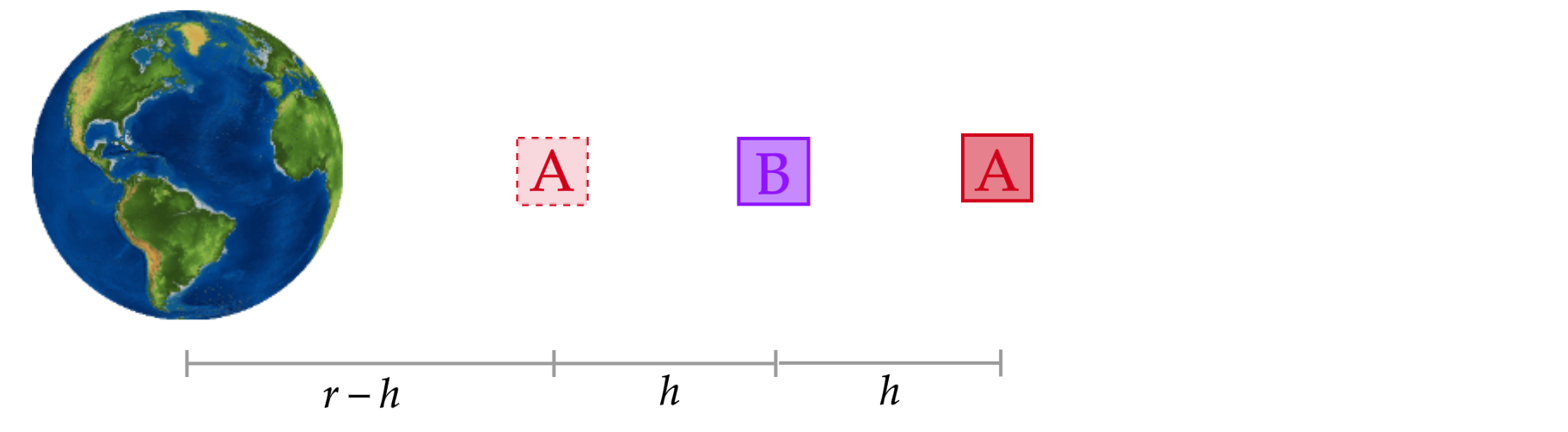}
    \caption{Configuration of quantum agents for a quantum switch in a classical metric.}
    \label{fig:agentssup}
\end{figure}

 The protocol presented here is based on this idea, also explored in~\cite{Rindler}. Since we are not putting the mass in superposition anymore, the main reason why this alternative experiment could be unrealistic on Earth is its duration for quantum agents at reasonable positions (heights). That time scale is determined by how long it takes for the agents to accumulate the necessary relative proper time delay so that their events can be used to make a switch. Here, we further allow the distances from the alternative positions of the agents to the central mass to change over time. That is, the boxes in the picture above will not be static. We will show that the freedom to move the agents during the protocol can be used to considerably decrease the total time of realization, provided we make a careful choice of paths. When the massive body is the Earth, the duration of the experiment, which in the static case of Fig.~\ref{fig:agentssup} would be so large as to prevent any possibility of experimental realization, can be brought from the time scale of a year to that of seconds.

We first describe our general protocol and then compute its minimum duration on Earth's surface. Next we illustrate how the protocol could be implemented using few-level systems
as agents. Then we discuss our results and their relation to previous protocols for the quantum switch. The content of the next sections is mostly a transcript of~\cite{QSonEarth}.

\section{Quantum switch with entangled agents}

Consider a spherical body with mass $M$ and radius $R$. The gravitational field outside the body is described by the Schwarzschild metric,
\begin{equation}
ds^2 = -\left(1-\frac{R_S}{r}\right) c^2 dt^2 + \left(1-\frac{R_S}{r}\right)^{-1} dr^2 + r^2 d\Omega^2 \, ,
\label{eq:metric}
\end{equation}
where $d\Omega^2$ is the metric of the unit sphere and $R_S=2GM/c^2$ is the Schwarzschild radius. Our protocol for the quantum switch involves three quantum systems that can be manipulated in the vicinity of its surface $r=R$, which we call the agents A and B and the target system. By agents we mean systems that are able to interact with the target system, and thereby operate on its state. The three systems have nontrivial internal structures, which we will discuss in detail later. In this section, we first introduce the relevant features of quantum mechanics on curved spacetimes required for the description of the quantum agents in the Schwarzschild metric, and then present our proposed protocol for general noncommuting operations $\mathcal{A}$ and $\mathcal{B}$. A concrete implementation with a specific choice of operations is described in the next section.

\subsection{Quantum systems with internal dofs in a weak gravitational field}

The dynamics of quantum systems with internal degrees of freedom in a curved spacetime has been analyzed in several works in the regime of weak gravitational fields and motions that are slow in comparison with the speed of light (see \cite{Zych} and references therein). In particular, the case of a system with internal degrees of freedom in a weak gravitational field produced by a central mass is discussed in \cite{Zych2011,tbell}. The metric is then given by equation~\eqref{eq:metric}, with $R_S/r\ll1$, and is completely characterized by the gravitational potential
\begin{equation}
\Phi = - \frac{GM}{r} \, .
\label{eq:potential}
\end{equation}

Consider a system with internal degrees of freedom described in the absence of the gravitational potential by a Hilbert space $\mathcal{H}^{int}$. The same Hilbert space describes the internal degrees of system in the presence of the gravitational potential $\Phi$. The spatial configuration of the system is represented by a wavefunction $\psi(x) \in \mathcal{H}^{ext} \simeq L^2(\mathbb{R}^3;\mathbb{C})$. The full Hilbert space of the system is then $\mathcal{H} = \mathcal{H}^{int} \otimes \mathcal{H}^{ext}$. The Hamiltonian in the curved geometry can be written in terms of the Hamiltonian $H_{int}$ describing the system in the absence of a gravitational potential and the gravitational potential $\Phi$, as discussed in \cite{Zych,Zych2011,tbell}. It describes a modified Schr\"odinger equation that includes corrections due to the effect of gravitational time dilation in the evolution of the wavefunction and of the internal states.

The evolution of the internal degrees of freedom has a simple description when the wavefunction of the system is well localized at all times. Let  the support of $\psi(t,x)$ be restricted, for each $t$, to a finite region $V_{x(t)}$ around a point $x(t)$, within which the variation of the potential is negligible, $(\Phi(t,x')- \Phi(t,x))/ \Phi(x(t)) \ll 1$, $\forall x'\in V_{x(t)}$. In this case, the system has a well defined position at each time and a well defined proper time along its evolution, which for a slow motion is simply determined by
\begin{equation}
\frac{d \tau}{dt} = \sqrt{1 - \frac{R_S}{r}}  \, .
\end{equation}
We define a path $\rm{P}$ as a worldline $x(t)$ in spacetime. A path state $\ket{\rm{P}}$ is a wavefunction $\psi(t,x) \in \mathcal{H}^{ext}$ that is well localized at the event $x(t)$ of the path $\rm{P}$ for each instant of time $t$. The state at a given time $t$ is represented as $\ket{{\rm P};t}$. For localized states, the explicit form of the wavepacket is not relevant for the evolution of the internal state $\ket{\phi} \in \mathcal{H}^{int}$, which is determined by
\begin{equation}
i \frac{d}{d\tau_P} \ket{\phi}_{\rm P} = H_{int} \ket{\phi}_{\rm P} \, ,
\label{eq:internal-evolution-general}
\end{equation}
where $\tau_P$ is the proper time along the path $\rm{P}$, and $H_{int}$ is the Hamiltonian for the internal degrees of freedom in the absence of a gravitational potential. The state evolves with respect to the proper time as in the absence of the gravitational potential, but the proper time $\tau_P$ and the coordinate time $t$ are now distinct, being related by a factor describing the gravitational time dilation. Equivalently, the evolution can be written in terms of the coordinate time $t$ by
\begin{equation}
i \frac{d}{dt} \ket{\phi}_{\rm P} = H \ket{\phi}_{\rm P} \, ,
\end{equation}
with the Hamiltonian
\begin{equation}
H = \left. \frac{d\tau}{dt}\right|_{\rm P} H_{int} \simeq \left( 1+\frac{\Phi|_{\rm P}}{c^2} \right) H_{int} \, ,
\end{equation}
as presented in \cite{tbell}.

Consider a superposition of localized states of the form
\begin{equation}
\ket{\Psi} = \sum_i \ket{\rm{P}_i}\ket{\phi_i} \, .
\end{equation}
For brevity, we call such states path superposition states. For each path $\rm{P}_i$ in the superposition, the corresponding internal state $\ket{\phi_i} $ evolves according to equation~\eqref{eq:internal-evolution-general} with respect to the proper time along $\rm{P}_i$. Suppose that the initial state is separable, with $\ket{\phi_i}=\ket{\phi}$ at $t=t_0$. Time evolution will in general produce entanglement with respect to the bipartition  $\mathcal{H} = \mathcal{H}^{int} \otimes \mathcal{H}^{ext}$, as the internal state will evolve by distinct amounts of proper time along the distinct paths. In this sense, the proper times along the paths become entangled with the paths. On the other hand, if the proper times along the paths are the same at some instant of time $t$, then the time-evolved state will be separable again at such a time $t$, if the initial state is separable.

The dynamics of localized states in a curved spacetime provides a simple context for the analysis of gravitational effects in quantum systems. It is natural to expect the evolution of internal degrees of freedom to take place with respect to the proper time along the path followed by the localized state, as follows from the approach developed in \cite{Zych,Zych2011}. For path superposition states, the entanglement of the proper times with the spatial localization of the paths leads to new quantum effects that depend on the curvature of spacetime, as for instance a drop of visibility in quantum interferometric experiments \cite{Zych2011,Zych_2012} and a universal mechanism for decoherence in the position of composite particles \cite{Pikovski2015}.

In addition, models explored in recent works for the formulation of a phenomenology of the low-energy limit of quantum gravity \cite{Bose,Vedral,Rovelli,tbell} also rely on the validity of the quantum mechanics of nonrelativistic systems in curved spaces. In order to describe the evolution of a quantum system in a superposition of geometries, one must first be able to describe its evolution in each of the geometries in the superposition. In such works, it is assumed that an entanglement between proper times and paths takes place for each geometry. In the context of quantum gravity, the superposition of proper times in each geometry is combined with effects due to the superposition of geometries. The interferometric experiment proposed in \cite{Bose,Vedral} constitutes an example of a model involving superpositions of proper times in a superposition of geometries, as discussed in \cite{Rovelli}, as well as the gravitational quantum switch proposed in \cite{tbell}.


\subsection{Protocol for the quantum switch}
\label{sec:general-protocol}

Let us now describe our protocol for the quantum switch. The agents A and B and the target are quantum systems with internal degrees of freedom in the Schwarzschild metric \eqref{eq:metric}. We restrict to the weak-field regime $R_S/r \ll 1$. According to the discussion in the previous subsection, the Hilbert space of the agent A has the form $\mathcal{H}^{ext}_A \otimes \mathcal{H}^{int}_A$, where $\mathcal{H}^{int}_A $ describes its internal state and $\mathcal{H}^{ext}_A$ describes its position, and similarly for the agent B and target.

Our protocol involves a path superposition state for the agent A that includes two path states $|{\rm P}_{{\rm A}\prec {\rm B}}\rangle,  |{\rm P}_{{\rm B}\prec {\rm A}}\rangle \in \mathcal{H}^{ext}_A$, while B remains at a constant position at a height $h$ above the surface $r=R$ of the massive body that produces the gravitational field. The paths for A are represented in Fig.~\ref{LabsTraveling}. Both start from a common departure point at $r=R$ with the same angular position as that of agent B. Next, they separate horizontally in a symmetric manner up to a distance $d$. For the path ${\rm P}_{{\rm A}\prec {\rm B}}$, A starts traveling up at the instant $t_0$ until it reaches a point $X_{{\rm A}\prec {\rm B}}$ at $r=R+h$ at time $t_1$. Put $\Delta t_v=t_1-t_0$. A remains at this position afterwards. For ${\rm P}_{{\rm B}\prec {\rm A}}$, A also travels up to a point $X_{{\rm B}\prec {\rm A}}$ at $r=R+h$ in an interval $\Delta t_v$, but starting at a later time $t_2=t_1+\Delta t_s$. The target system, traveling horizontally, meets the point $X_{{\rm A}\prec {\rm B}}$ at $t_3$ and then travels towards $X_{{\rm B}\prec {\rm A}}$ in a time interval $\Delta t_c$. After that, the agent A is measured in a diagonal basis, as we will discuss in more detail later.

\begin{figure}
\hspace{-0.9cm}\includegraphics[scale=0.066]{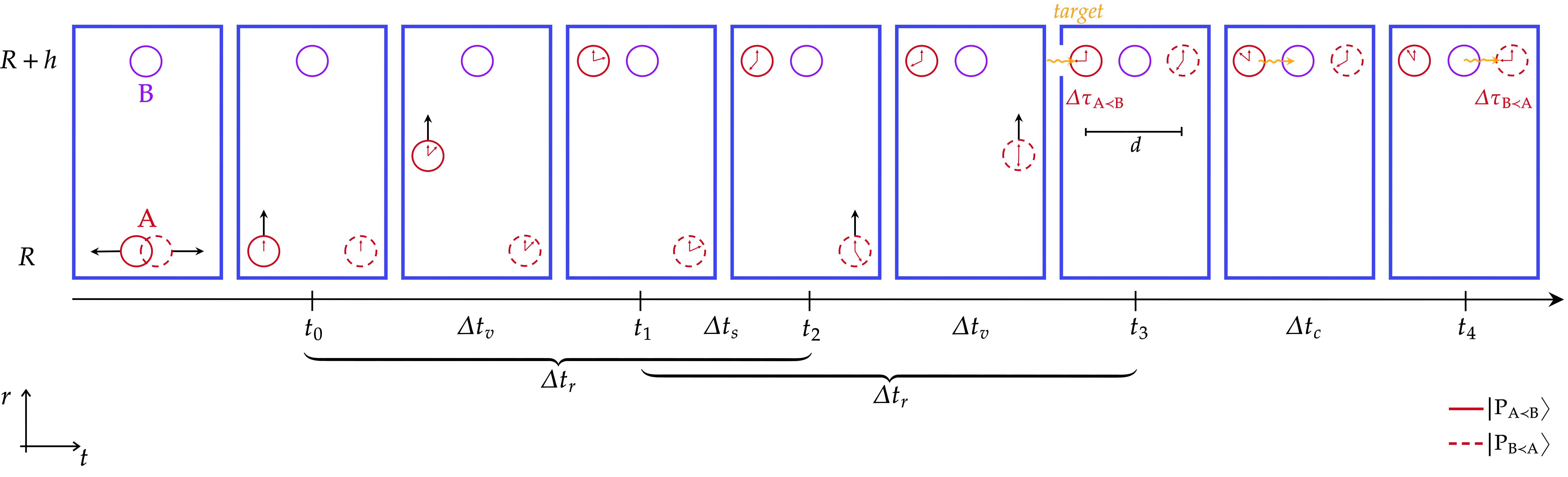}
\caption{Superposition of paths. The vertical axis represents the radius $r$ and the horizontal axis represents time. At $t=t_0$, the agent A is prepared in a superposition state $(|{\rm P}_{{\rm A}\prec {\rm B}}\rangle+|{\rm P}_{{\rm B}\prec {\rm A}}\rangle)/\sqrt{2}$. For the path $|{\rm P}_{{\rm A}\prec {\rm B}}\rangle$, A starts traveling up at $t_0$ to a height $h$ above the surface $r=R$. For the path $|{\rm P}_{{\rm B}\prec {\rm A}}\rangle$, it travels up in the same manner starting at $t_2$. The target system, traveling horizontally at height $h$, crosses $|{\rm P}_{{\rm A}\prec {\rm B}}\rangle$, meets agent B, and then crosses $|{\rm P}_{{\rm B}\prec {\rm A}}\rangle$.}
\label{LabsTraveling}
\end{figure}

We consider a path superposition of the form
\begin{equation}
\ket{\Psi_{\rm A}} = \frac{1}{\sqrt{2}} \left(\ket{{\rm P}_{{\rm A}\prec {\rm B}}} \ket{\phi_{\rm A}}_{{\rm P}_{{\rm A}\prec {\rm B}}} + \ket{{\rm P}_{{\rm B}\prec {\rm A}}} \ket{\phi_{\rm A}}_{{\rm P}_{{\rm B}\prec {\rm A}}} \right) \, ,
\end{equation}
with
\begin{equation}
\ket{\phi_{\rm A};t_0}_{{\rm P}_{{\rm A}\prec {\rm B}}} = \ket{\phi_{\rm A};t_0}_{{\rm P}_{{\rm B}\prec {\rm A}}} \, ,
\end{equation}
i.e., we assume that the internal state of A is the same for both paths before they separate, in which case it remains the same while the paths remain at a common height, and that both paths are equally probable.

The agent A is configured to operate on the target at a specific instant $\tau^*$ in its proper time as indicated by an internal clock. This means that it is prepared in a state for which the probability of interacting with the target is considerable at $\tau^*$, but not at other times. The proper time of A must then be equal to $\tau^*$ when the target meets it for the interaction to take place. This must happen for both ${\rm P}_{{\rm A}\prec {\rm B}}$ and ${\rm P}_{{\rm B}\prec {\rm A}}$ for the interaction to occur regardless of the path taken. The agent B can interact with the target when their worldlines intersect. Under these conditions, the operations $\mathcal{A}$ and $\mathcal{B}$ are applied in distinct orders for each component $|{\rm P}_{{\rm A}\prec {\rm B}}\rangle$ or $|{\rm P}_{{\rm B}\prec {\rm A}}\rangle$ in the superposition of paths.

Let $\Delta\tau_{{\rm A}\prec {\rm B}}$ be the proper time along the path ${\rm P}_{{\rm A}\prec {\rm B}}$ from $t_0$ to the moment the target reaches it at $t_3$, and $\Delta\tau_{{\rm B}\prec {\rm A}}$ be the proper time along the path ${\rm P}_{{\rm B}\prec {\rm A}}$ from $t_0$ to the moment the target reaches it at $t_4$. Then the quantum switch will happen only if
\begin{equation}\label{1stCond}
\Delta\tau_{{\rm A}\prec {\rm B}}=\Delta\tau_{{\rm B}\prec {\rm A}} = \tau^* \, .
\end{equation}
Put $\Delta t_r\equiv\Delta t_v+\Delta t_s$. The interval $\Delta \tau_{{\rm A}\prec {\rm B}}$ has contributions from the time elapsed for A while it travels up and while it remains at radius $R + h$,
\begin{equation}\label{taua}
\Delta\tau_{{\rm A}\prec {\rm B}} 
=\Delta\tau_v +\sqrt{\left(1-\frac{R_S}{R+h}\right)} \ \Delta t_r \, ,
\end{equation}
where $\Delta \tau_v$ is the proper time elapsed for A while it travels up, corresponding to the coordinate time $\Delta t_v$. The interval $\Delta \tau_{{\rm B}\prec {\rm A}}$ includes contributions from the time elapsed for A while it stays at $r=R$, while it travels up, and while it waits the arrival of the target at $R+h$,
\begin{equation}\label{taub}
\Delta\tau_{{\rm B}\prec {\rm A}}
=\Delta\tau_v +\sqrt{\left(1-\frac{R_S}{R}\right)} \ \Delta t_r +\Delta\tau_{c} \, ,
\end{equation}
where $\Delta\tau_{c} = \sqrt{1-R_S/(R+h)} \ \Delta t_c$ and the proper time $\Delta \tau_v$ elapsed for the agent while it travels up is the same as in the other path. Substituting Eqs.~(\ref{taua}) and (\ref{taub}) into equation~(\ref{1stCond}), we find
\begin{equation} \label{main}
\left(
\sqrt{1-\frac{R_S}{R+h}}-\sqrt{1-\frac{R_S}{R}} \
\right)
\frac{\Delta t_r}{\Delta t_c} =\sqrt{1-\frac{R_S}{R+h}} \, .
\end{equation}
In the weak field regime, characterized by $R_S \ll R$, this equation reduces to
\begin{align}
\Delta t_r &= \frac{R}{R_S} \left( \frac{2R}{h} + 2  \right) \Delta t_c \nonumber \\
	& = \left( \frac{c^2}{g h} - \frac{c^2}{2} \frac{R_{0101}}{g^2} \right) \Delta t_c \, ,
	\label{eq:weak-field}
\end{align}
where $g=GM/R^2$ and $R_{0101}=-c^2 R_S/R^3$ is a component of the curvature tensor of the Schwarzschild metric (\ref{eq:metric}). The subleading term in the weak-field approximation is independent of $R_S/R$ and the following terms are proportional to higher powers of $R_S/R$. The first term in equation~(\ref{eq:weak-field}) depends on the acceleration of gravity at the radius where the experiment is performed. The second term describes the effect of curvature.

The parameter $\Delta t_r$ sets a time scale for the duration of the experiment. The total time of the protocol is $\Delta t_{\rm exp}\equiv t_4 - t_0$. For small $d$, such that $\Delta t_c \ll t_3-t_0$, this is well approximated by $t_3-t_0=\Delta t_r+\Delta t_v$. If the paths remain at distinct heights for a large amount of time, $\Delta t_v \ll \Delta t_s$, then we have $\Delta t_{\rm exp} \simeq \Delta t_r$. On the other hand, if $\Delta t_s= 0$, then $\Delta t_{\rm exp} \simeq 2 \Delta t_r$. In general, $\Delta t_{\rm exp} \sim \Delta t_r $.

Near the surface of the spherical mass, we can take $h\ll R$. The first term in equation~(\ref{eq:weak-field}) is then dominant. Considering the target to be a photon, we have $\Delta t_c\simeq d/c$. Under these approximations,
\begin{equation}\label{simple}
\Delta t_r \simeq \frac{cR^{2}d}{GMh} \, .
\end{equation}
We see that $\Delta t_r$ depends on two fundamental constants, $c$ and $G$; two properties of the massive body, $M$ and $R$; and two variables $d$ and $h$ that can be adjusted in the experiment. The duration of the experiment is minimized for the smallest possible distance $d$ between the paths and the largest possible height $h$. The distance $d$ in any implementation of the protocol will be limited by possible interactions between the agents and the precision of the clock. If the distance between the agents is so small that they can interact, their operations on the target will not be independent, as assumed. In addition, the clock must be sufficiently precise to resolve the time of flight $d/c$ of the photon between the paths of A. The height $h$ will be limited by the experimental capability of transporting A along its path-superposition state without decoherence.

Substituting the numerical values of $c,G$ and the radius $R_\odot$ and mass $M_\odot$ of the Earth in equation~(\ref{main}), we can estimate the duration of the experiment near the surface of Earth,
\begin{equation}\label{factorEarth}
\Delta t_{\rm exp} \sim 3 \times 10^7 \, \frac{d}{h} {\rm  s}\, .
\end{equation}
For an atomic clock with a precision of $10^{15} {\rm  Hz}$, for instance, the time of flight of the photon can be resolved for a distance of $0.3 \, \mu{\rm m}$. Setting $d = 0.3 \, \mu{\rm m}$ and $h=1 {\rm  m}$, we find $\Delta t_{\rm exp} \sim 9 {\rm  s}$.

We can also consider the case of a small mass. As in~\cite{tbell}, this example can be used to show that the effect does not require any physical quantity to be at the Planck scale to be observed. In this setting, it is natural to bring the departure point for the paths of A as close as possible to the mass, and we can take $h \gg R \gg R_S$. The second term in equation~(\ref{eq:weak-field}) is then dominant.  In this regime,
\begin{equation} \label{small-mass}
\Delta t_r \simeq \frac{cRd}{G M} \, .
\end{equation}

A well-known protocol for a gravitational quantum switch was previously formulated in the context of quantum gravity in \cite{tbell}. In that case, the agents and target move in a quantum state of the gravitational field produced by a mass in a superposition of positions.  In \cite{ZychRelQuantSup}, however, it is argued on general grounds that outcomes of a process in which a localized system interacts with a system in a superposition of positions can be reproduced by a process in which the first system is delocalized while the second system is localized. This suggests the possibility of simulating the gravitational quantum switch proposed in \cite{tbell} with delocalized quantum agents in the classical gravitational field of a central mass at a definite position. Our work was motivated by this correspondence, but it was not our purpose to exactly simulate the protocol of \cite{tbell}. Instead, we aimed to reproduce its relevant features using quantum agents in the Schwarzschild metric in order to obtain an efficient implementation of the quantum switch in this context, as a possible test of quantum mechanics on curved spacetimes. We can now compare the duration of our protocol with that expected for the protocol of \cite{tbell} under the correspondence proposed in \cite{ZychRelQuantSup}, as a means of testing its efficiency.

In the protocol of \cite{tbell}, for a small distance $d\ll R$ between the agents, the minimum proper time elapsed for the agents for the quantum switch to occur is $\tau^*= 2 r_b^2 c/GM$, where $r_b$ is the distance from the agents to the mass. Setting $r_b = R_\odot$, we find that $\tau^*$ is of the order of a year. This result can be compared with equation~\eqref{factorEarth}. The duration of the experiment is suppressed by a factor of $d/h$ using dynamical agents as described here.

For the small mass limit, a mass of $M = 0.1 \, \mu{\rm g}$ was considered in \cite{tbell}, with one agent at a distance of $1 {\rm  fm}$ and the other at a distance of $0.1\, \mu{\rm m}$ from the mass. The protocol for the Bell test using static agents explored in~\cite{tbell} would then take around $10 \, {\rm  h}$. Setting $R=1 \, {\rm  fm}$ in equation~\eqref{small-mass} and assuming $d=R$, we obtain $\Delta t_{\rm exp} \sim 5 \times 10^{-2} \, {\rm  s}$. In general, the duration of our protocol is of the order of one second if $Rd \sim 10^{-28} \, {\rm  m}^2$, and grows linearly with $Rd$.

\section{A model for the operations} \label{sec:Modelfortheops}

We discussed the spacetime features required for the realization of the quantum switch in a Schwarzschild spacetime. We now explore possibilities for the operations performed by the agents. For concreteness, we consider a model involving a particular choice of quantum systems as agents and target.  The relevant features of the model are not restricted to this specific quantum system, however, which provides an illustration of a procedure that can be adapted to other systems of interest.

\subsection{Operations with indefinite order}

The internal Hilbert space of agent A has a subsystem $\mathcal{H}_\mathsmaller{\tiny \VarClock}$ which we call the trigger. It also includes a subsystem $\mathcal{H}_A$ with six energy levels $\ket{A_i}$, $i=0,\dots,5$. That is, $\mathcal{H}^{int}_A = \mathcal{H}_\mathsmaller{\tiny \VarClock} \otimes \mathcal{H}_A$. The trigger system will play the role of an internal clock for the agent A. The agent B is a system with internal degrees of freedom described by a Hilbert space $\mathcal{H}^{int}_B = \mathcal{H}_B$ with five energy levels $\ket{B_i}$, $i=1,\dots,5$. The energy level diagrams for $\mathcal{H}_A$ and $\mathcal{H}_B$ are represented in Fig.~\ref{model}. The labels $e_0$, $e_1$, etc, are energy differences between pairs of levels for the allowed transitions. We assume that the transitions are induced by the absorption and emission of photons.

\begin{figure}[ht]
	\begin{center}
\includegraphics[scale=.3]{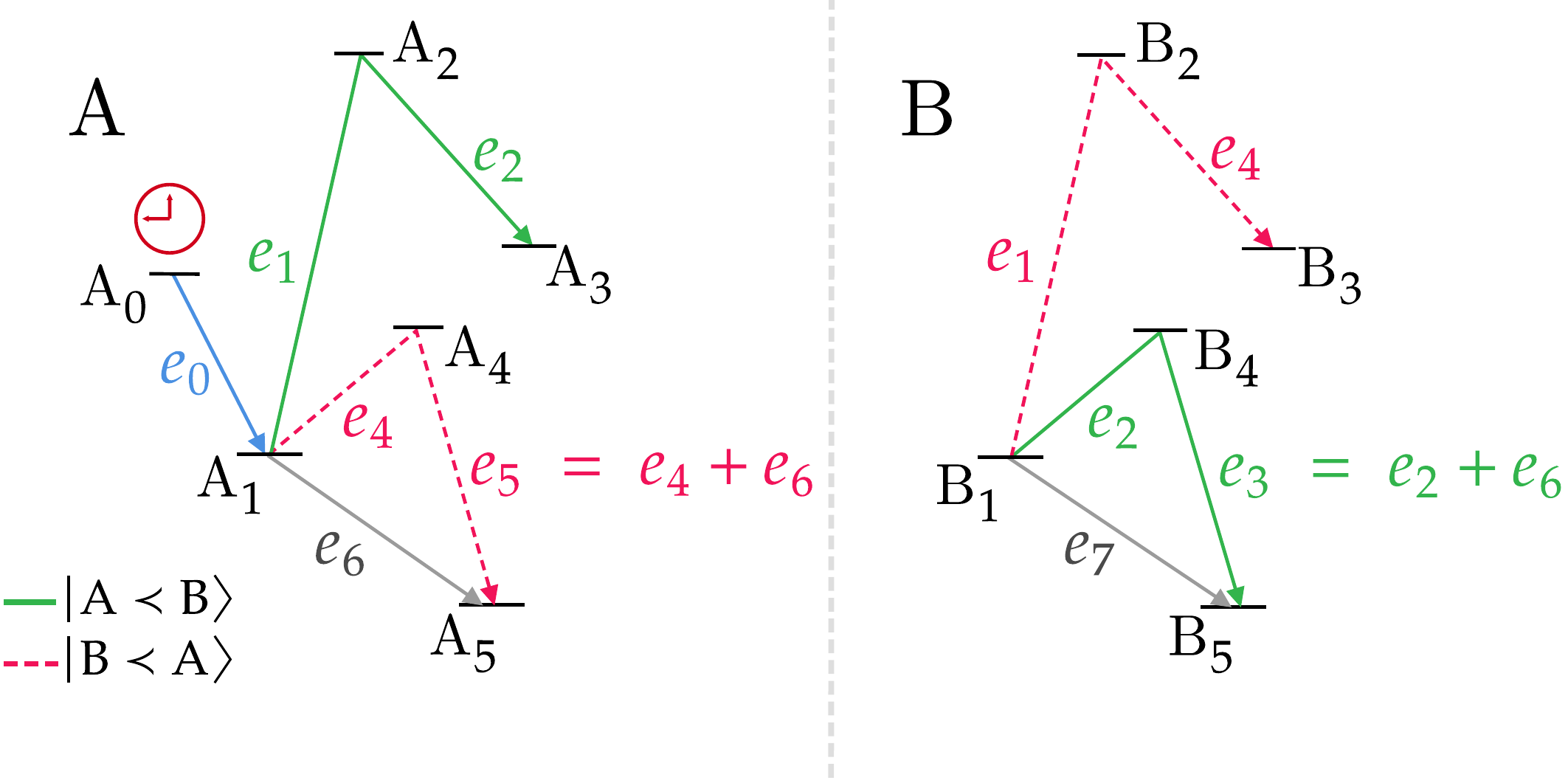}
		\caption{Energy levels of the agents A and B.}
        \label{model}
        \end{center}
\end{figure}

The trigger is coupled to the six-level system $\mathcal{H}_A$. The time-evolution of internal degrees of freedom of a quantum system following a path $P$ in a curved spacetime is generated by an internal Hamiltonian $H_{int}$ evolving with respect to the proper time $\tau_P$ along $P$ \cite{tbell,Pikovski2015,Zych}, as described by equation~\eqref{eq:internal-evolution-general}. For the agent A, in particular,
\begin{equation}
i \frac{d}{d\tau_P} \ket{\phi_{\rm A}}_P = H_{int} \ket{\phi_{\rm A}}_P \, ,
\label{eq:internal-evolution}
\end{equation}
where $\ket{\phi_{\rm A}}_P \in \mathcal{H}_{\tiny \VarClock} \otimes \mathcal{H}_A$. For a path superposition state, the evolution of the internal state is described by equation~(\ref{eq:internal-evolution}) for each path in the superposition.  A complete description of the state $\ket{\Psi_{\rm A}}$ of A also includes its spatial location,
\begin{equation}
\ket{\Psi_{\rm A}} = \frac{1}{\sqrt{2}} \left(\ket{{\rm P}_{{\rm A}\prec {\rm B}}} \ket{\phi_{\rm A}}_{{\rm P}_{{\rm A}\prec {\rm B}}} + \ket{{\rm P}_{{\rm B}\prec {\rm A}}} \ket{\phi_{\rm A}}_{{\rm P}_{{\rm B}\prec {\rm A}}} \right) \, .
\label{eq:internal-evolution-superposition}
\end{equation}

At the beginning of the experiment, $\mathcal{H}_A$ is prepared in the state $\ket{A_0}$, which is stable in the absence of the trigger. The trigger is prepared in a state $\ket{\raisebox{-2pt}{\footnotesize \VarClock} \, ; \tau=0}$. We assume that a sharp transition from $\ket{A_0}$ to $\ket{A_1}$ is induced by the trigger after a proper time $\tau^*$ has elapsed for A since $t_0$, with $\tau^*$ given by equation~(\ref{1stCond}). Denoting the unitary evolution under $H_{int}$ by $U(\tau^*,0)$ and putting
\begin{equation}
\ket{\Psi_A;0} \equiv \ket{\raisebox{-2pt}{\footnotesize \VarClock}\, ;\tau=0} \ket{A_0} \, ,
\end{equation}
we require that
\begin{equation}
U(\tau,0) \ket{\Psi_{\rm A};0} \simeq \left\{ \begin{array}{ll}
	\ket{\raisebox{-2pt}{\footnotesize \VarClock}\, ;\tau} \ket{A_0}   \, , & {\rm for } \, \tau < \tau^* - \epsilon \, ,\\
	\ket{\raisebox{-2pt}{\footnotesize \VarClock}\, ;\tau^*} \ket{A_1}  \, , & {\rm for } \, \tau = \tau^* \, .
\end{array}
\right. \label{eq:trigger-condition}
\end{equation}
We provide a concrete example of a unitary evolution satisfying the above properties in Appendix~\ref{App:trigger}, as an illustration of how such a trigger could be implemented. 

The trigger plays the role of a clock that at $\tau^*$ changes the state of A into a new state that can interact with the target. In other words, the instrument of A, represented by $\mathcal{H}_A$, is switched on by the trigger at $\tau^*$. It is in fact sufficient that the trigger generates a nonzero projection on $\ket{A_1}$. The agent B is prepared in the state $\ket{B_1}$ at the time $t_3+d/2c$. As B remains at a fixed position, an external clock can be used to prepare it in the required state at the scheduled time.

We require the levels $\ket{A_1}$ and $\ket{B_1}$ to have a small decay time satisfying $\Delta \tau_1 \ll d/c$ and $\epsilon \ll \Delta \tau_1$. A can then absorb a photon of energy $e_1$ or $e_4$ and get excited to the level $\ket{A_2}$ or $\ket{A_4}$ only if the photon arrives at $\tau^*$, within a time-window of approximately $\Delta \tau_1$. If no photon reaches the system at this time, it decays to the level $\ket{A_5}$ by emitting a photon of energy $e_6$, which testifies that A has not absorbed an incoming photon during the process. Similarly, the system $\mathcal{H}_B$ can be excited only at the coordinate time $t_3+d/2c$. If no photon reaches the system at this time, it decays to the level $\ket{A_5}$ by emitting a photon of energy $e_6$.

The experiment is designed so that, for $|{\rm P}_{{\rm A}\prec {\rm B}}\rangle$, a photon of energy $e_1$ meets A at $t_3$. If this photon is absorbed, A is excited to the level $\ket{A_2}$ and then rapidly decays to $\ket{A_3}$, emitting a photon of energy $e_2$. We represent such interaction as $\mathcal{A} \left(\ket{A_1} \ket{e_1} \right)=\ket{A_3} \ket{e_2}$. This step is understood as an operation $\mathcal{A}_\text{targ}$ on the photon state, i.e., $\mathcal{A}_\text{targ} |e_1\rangle=|e_2\rangle$. The emitted photon can be absorbed by B. When this does not happen, B decays to its ground state $\ket{B_5}$, emitting a photon of energy $e_7$ that testifies that the experiment was not completed. If the photon is absorbed, B is excited to the level $\ket{B_4}$ and quickly decays to $\ket{B_5}$ by emitting a photon of energy $e_3$. We represent such interaction as $\mathcal{B} \left(\ket{B_1} \ket{e_2} \right)=\ket{B_5} \ket{e_3}$. The operation $\mathcal{B}_\text{targ}$ on the photon state is $\mathcal{B}_\text{targ}|e_2\rangle=|e_3\rangle$. The operations are performed in the order $\mathcal{B} \mathcal{A}$. Defining
\begin{equation}
\ket{{\rm A}\prec{\rm B}} = \left|{\rm P}_{{\rm A}\prec{\rm B}}\right \rangle \ket{\raisebox{-2pt}{\footnotesize \VarClock}}_{{\rm A}\prec{\rm B}} \, ,
\end{equation}
where $\ket{\raisebox{-2pt}{\footnotesize \VarClock}}_{{\rm A}\prec{\rm B}}$ is the final state of the trigger, and introducing
\begin{equation}
\ket{\psi_1} = \ket{A_1} \ket{B_1} \ket{e_1} \, ,
\end{equation}
the final joint state of the agents and photon is 
\begin{equation}
\ket{{\rm A}\prec{\rm B}} \mathcal{B} \mathcal{A} \ket{\psi_1} = \ket{{\rm A}\prec{\rm B}} \ket{A_3}\ket{B_5} \ket{e_3} \, .
\end{equation}

For the path $|{\rm P}_{{\rm B}\prec {\rm A}}\rangle$, the sequence of events proceeds analogously, implementing operations $\mathcal{B} \left( \ket{B_1} \ket{e_1} \right)=\ket{B_3} \ket{e_4}$, with an action $\mathcal{B}_\text{targ} |e_1\rangle=|e_4\rangle$ on the target, and $\mathcal{A} \left( \ket{A_1} \ket{e_4} \right)=\ket{A_5} \ket{e_5}$, with an action $\mathcal{A}_\text{targ} |e_4\rangle=|e_5\rangle$ on the target, performed now in the switched order $\mathcal{A} \mathcal{B}$. In this case, the final joint state is
\begin{equation}
\ket{{\rm B}\prec{\rm A}} \mathcal{A} \mathcal{B} \ket{\psi_1} = \ket{{\rm B}\prec{\rm A}} \ket{A_5}\ket{B_3} \ket{e_5} \, ,
\end{equation}
where 
\begin{equation}
\ket{{\rm B}\prec{\rm A}} = \left|{\rm P}_{{\rm B}\prec{\rm A}}\right \rangle \ket{\raisebox{-2pt}{\footnotesize \VarClock}}_{{\rm B}\prec{\rm A}} \, .
\end{equation}

We assume the amplitude transitions for the processes in which the target interacts with both agents to be the same in the two paths. Then the final state of the system is
\begin{equation}
\frac{ \ket{{\rm A}\prec{\rm B}} \mathcal{B} \mathcal{A} \ket{\psi_1} + \ket{{\rm B}\prec{\rm A}} \mathcal{A} \mathcal{B} \ket{\psi_1} }{\sqrt{2}} \, .
\label{eq:switch-composite}
\end{equation}
A quantum switch is thus implemented, with the two alternative paths of the agent A playing the role of the control of the switch. As the agents are quantum systems, their interactions with the target are described by operators acting on $\mathcal{H}_A \otimes \mathcal{H}_B \otimes \mathcal{H}_\text{targ}$, where $\mathcal{H}_\text{targ}$ is the Hilbert space of the target, spanned by states $\ket{e_i}$. When the agents are classical, their operations are represented by operators on $\mathcal{H}_\text{targ}$.

Measuring the agents in a diagonal basis, the superposition of orders can be encoded in a superposition of target states. We define
\begin{align}
|{\rm F}_{{\rm A}\prec {\rm B}}\rangle &= \ket{{\rm A}\prec{\rm B}} \ket{A_3}\ket{B_5} \, , \\
|{\rm F}_{{\rm B}\prec {\rm A}}\rangle &= \ket{{\rm B}\prec{\rm A}} \ket{A_5}\ket{B_3}  \, .
\end{align}
Measuring the agents in the basis $|{\rm F}_{{\rm A}\prec {\rm B}}\rangle \pm|{\rm F}_{{\rm B}\prec {\rm A}}\rangle$ takes the photon to the state 
\begin{equation} \label{atomsfinal}
\frac{\mathcal{B}_\text{targ} \mathcal{A}_\text{targ} |e_1\rangle \pm \mathcal{A}_\text{targ} \mathcal{B}_\text{targ} |e_1\rangle}{\sqrt{2}}=\frac{|e_3\rangle \pm |e_5\rangle}{\sqrt{2}} \, ,
\end{equation}
and we obtain a superposition of the orders of the operations $\mathcal{A}_\text{targ}$ and $\mathcal{B}_\text{targ}$ on the target. The superposition of orders can then be verified by performing observations on the target system.

The measurement on the basis $|{\rm F}_{{\rm A}\prec {\rm B}}\rangle \pm|{\rm F}_{{\rm B}\prec {\rm A}}\rangle$ includes the measurement of the clock. This can be avoided by resynchronizing the clock states after the application of the operations, which would disentangle the clock from the rest of the system, allowing the measurement on the basis $|{\rm F}_{{\rm A}\prec {\rm B}}\rangle \pm|{\rm F}_{{\rm B}\prec {\rm A}}\rangle$ to be performed only on the path and few-level systems. This could be done by making A follow the paths of the protocol in a reversed way, similarly as done in~\cite{tbell}. Another possibility is to artificially synchronize the clock states by directly manipulating them, as done for instance in~\cite{Margalit}. 

We described the result of the operations of the agents A and B on an incoming photon of energy $e_1$, selecting runs of the experiment in which both agents absorbed some photon in the process. In this case, the operations $\mathcal{A} \mathcal{B}$ and $\mathcal{B} \mathcal{A}$ produce outgoing photon states of different energies, and the superposition of orders in the quantum switch leads to a superposition of final energies for the photon after a measurement in a diagonal basis, as described by equation~\eqref{atomsfinal}. As discussed in subsection~\ref{sec:general-protocol}, this can happen only if there is a superposition of distinct proper times along the alternative paths of A, allowing the photon to cross the path ${\rm P}_{{\rm A}\prec {\rm B}}$ at $t_3$ and the path ${\rm P}_{{\rm B}\prec {\rm A}}$ at $t_4$ at the same proper time $\tau^*$ of A, as required for the operation $\mathcal{A}$ to be applied for both paths. Hence, the verification of the operation of the quantum switch for an incoming photon of energy $e_1$, as described by equation~\eqref{atomsfinal}, testifies to the superposition of proper times along the alternative paths.

The case of incoming photons of definite energies $e \neq e_1$ can be analyzed similarly. If the incoming photon has an energy $e \neq e_1$, then at most one agent can operate on it nontrivially. As a result, the final state of the system formed by the target and the few-level systems is the same for both paths after the application of the operations, making the switch of the order of operations trivial. For instance, an incoming photon with energy $e_4$ can interact nontrivially only with A, in which case the operation on the target is given by $\mathcal{A}_\text{targ} \ket{e_4}= \ket{e_5}$. For the interaction with B, we have $\mathcal{B}_\text{targ} \ket{e_4}=\ket{e_4}$ and $\mathcal{B}_\text{targ} \ket{e_5}=\ket{e_5}$. If A does not absorb the incoming photon, then it decays to its ground state, emitting a photon of energy $e_6$ that testifies that the interaction did not take place, and we can discard this run of the experiment.  The final state for the path ${\rm P}_{{\rm A}\prec {\rm B}}$ is
\begin{equation}
\ket{{\rm A}\prec{\rm B}}\ket{A_5}\ket{B_5} \ket{e_5}\, ,
\end{equation}
while the final state for ${\rm P}_{{\rm A}\prec {\rm B}}$ is obtained by replacing ${\rm A}\prec{\rm B}$ with ${\rm B}\prec{\rm A}$ in the expression above. The final state of the system is
\begin{equation}
\frac{\left( \ket{{\rm A}\prec{\rm B}} + \ket{{\rm B}\prec{\rm A}} \right)}{\sqrt{2}}\ket{A_5}\ket{B_5} \ket{e_5}
\end{equation}
The target is already disentangled from the rest of the system, and its final state is simply
\begin{equation}
\frac{\mathcal{B}_\text{targ} \mathcal{A}_\text{targ} |e_4\rangle + \mathcal{A}_\text{targ} \mathcal{B}_\text{targ} |e_4\rangle}{\sqrt{2}}= \ket{e_5} \, .
\label{eq:final-input-4}
\end{equation}
In this case, as for any state with definite energy $e \neq e_1$, it is not necessary to perform a measurement of the states of the agents before the measurement of the target, as the target disentangles from the rest of the system. The quantum switch is trivial, with both orders of operations producing the same result.

\subsection{Quantum switch for arbitrary input states}

In order to conceptually clarify the nature of the quantum switch implemented by the protocol, let us now discuss the case of a generic target state. As the quantum switch is nontrivial only for an incoming photon with initial energy $e_1$, our main purpose in analyzing the case of an arbitrary input target space is to discuss the general features of our proposal and clarify its relation to other implementations of the quantum switch.

In the examples previously considered, runs of the experiment were selected according to whether or not photons of energies $e_6$ or $e_7$ were produced. Accordingly, the formulation of the quantum switch for a generic input state involves a postselection of runs of the experiment, referring to the presence or not of the photons $e_6,e_7$. In addition, as the agents are quantum systems, each operation corresponds to an interaction between the target and the agent. In constrast, in implementations of the quantum switch with classical agents, the operations are represented directly on $\mathcal{H}_\text{targ}$.

Let the target system be described by the Hilbert space $\mathcal{H}_\text{targ}$ spanned by the states $\{\ket{e_i}; \, i=1,\dots,5\}$. This is sufficient for our purposes, since it includes all states involved in nontrivial interactions with the agents. The photons of energies $\ket{e_6}$ and $\ket{e_7}$ indicate whether or not each agent has absorbed some photon. Let $\mathcal{H}^A_d$ be the Hilbert space with basis $\{ \ket{0}_A,\ket{1}_A\}$, where the states $\ket{0}_A,\ket{1}_A$ describe configurations in which the photon $\ket{e_6}$ is absent or present, respectively. The Hilbert space $\mathcal{H}^B_d$ with basis $\{ \ket{0}_B,\ket{1}_B\}$ is defined analogously. Such photons are used to select runs of the experiment. There are four possible postselections, which we denote by $\zeta=0,1,2,3$ and which correspond to the situations where both photons, only $e_6$, only $e_7$, or none was emitted, respectively. We call the systems $\mathcal{H}^A_d$ and $\mathcal{H}^B_d$ the detectors of A and B.

Let us first consider the case of the path ${\rm P}_{{\rm A}\prec {\rm B}}$, for which the target first interacts with agent A and then with B. An incoming photon $\ket{e_i}$ can be absorbed or not by A. Let the amplitudes for these processes be
\begin{align}
& c_{iA} \, ,  & \text{if absorbed by A} \, , \nonumber \\
& d_{iA} = e^{i \delta_{iA}} \sqrt{|1-c_{iA}|^2} \, , & \text{if not absorbed by A} \, .
\end{align} 
The amplitudes $c_{iA}$ are nonzero only for $i=1,4$. The incoming photon may not be absorbed by A, but then be absorbed by B. The interaction of the incoming photon with B is described similarly in terms of amplitudes $c_{iB}$, which are nonzero for $i=1,2$. The photon can also interact nontrivially with both agents. This is possible only for an incoming photon of energy $e_1$. Let $f_{BA}$ be the amplitude for an incoming photon with energy $e_1$ that was scattered by A with energy $e_2$ to also be scattered by B, and the amplitude for the second scattering not to occur be $g_{BA}=e^{i \gamma_{BA}} \sqrt{1-|f_{BA}|^2}$.  The amplitude for double scattering is then $f_{BA} c_{1A}$. This exhausts all possible processes for ${\rm P}_{{\rm A}\prec {\rm B}}$.

By computing the  state that results from the interactions of A and B successively with a generic input state of the form 
\begin{equation}
\ket{\psi} = \ket{A_1} \ket{B_1} \sum_{i=1}^5 \alpha_i \ket{e_i} \, ,
\end{equation}
we find
\begin{multline}
U_B U_A \left( \ket{0}_A \ket{0}_B \ket{\psi} \right) 
= \ket{0}_A \ket{0}_B \ket{\psi_{BA}} + \ket{0}_A \ket{1}_B\ket{\psi_{0A}} \\
+ \ket{1}_A \ket{0}_B \ket{\psi_{B0}} + \ket{1}_A \ket{1}_B \ket{\psi_{00}}  \, ,
\label{eq:BA-operations-1}
\end{multline}
where
\begin{align}
\ket{\psi_{BA}} &=  \alpha_1 c_{1A} f_{BA} \ket{A_3} \ket{B_5} \ket{e_3} \, ,\nonumber \\
\ket{\psi_{0A}} &=  \alpha_1 c_{1A} g_{BA} \ket{A_3} \ket{B_5} \ket{e_2} + \alpha_4 c_{4A} \ket{A_5} \ket{B_5} \ket{e_5} \, , \nonumber \\
\ket{\psi_{B0}} &= \alpha_1 d_{1A} c_{1B} \ket{A_5} \ket{B_3} \ket{e_4} + \alpha_2 c_{2B} \ket{A_5} \ket{B_5} \ket{e_3} \, , \nonumber \\
\ket{\psi_{00}} &= \sum_i \alpha_i d_{iA} d_{iB} \ket{A_5} \ket{B_5} \ket{e_i} \, .
\label{eq:BA-operations-2}
\end{align}

Let us now consider the path ${\rm P}_{{\rm B}\prec {\rm A}}$. Let $f_{AB}$ be the amplitude for an incoming photon with energy $e_1$ that was scattered by B with energy $e_4$ to be also scattered by A, and $g_{AB}= e^{i \gamma_{AB}} \sqrt{1-|f_{AB}|^2}$ be the amplitude for the second scattering not to occur. The amplitude for double scattering is then $f_{AB} c_{1B}$. After the application of both operations,
\begin{multline}
\hspace{-0.3cm}U_A U_B \left( \ket{0}_A \ket{0}_B \ket{\psi} \right)
= \ket{0}_A \ket{0}_B \ket{\psi_{AB}} 
\\ + \ket{0}_A \ket{1}_B\ket{\psi_{A0}} 
+ \ket{1}_A \ket{0}_B \ket{\psi_{0B}} + \ket{1}_A \ket{1}_B \ket{\psi_{00}}  \, ,
\label{eq:AB-operations-1}
\end{multline} where
\begin{align}
\ket{\psi_{BA}} &=  \alpha_1 c_{1B} f_{AB} \ket{A_5} \ket{B_3} \ket{e_5} \, ,\nonumber \\
\ket{\psi_{0B}} &=  \alpha_1 c_{1B} g_{AB} \ket{A_5} \ket{B_3} \ket{e_4} + \alpha_2 c_{2B} \ket{A_5} \ket{B_5} \ket{e_3} \, , \nonumber \\
\ket{\psi_{A0}} &= \alpha_1 d_{1B} c_{1A} \ket{A_3} \ket{B_5} \ket{e_2} + \alpha_4 c_{4A} \ket{A_5} \ket{B_5} \ket{e_5} \, ,
\label{eq:AB-operations-2}
\end{align}
and $\ket{\psi_{00}}$ is given in equation~\eqref{eq:BA-operations-2}.

When the agent A is in the path superposition state \eqref{eq:internal-evolution-superposition}, the final state is given by
\begin{equation}
\frac{ \ket{{\rm A}\prec{\rm B}} U_B U_A  + \ket{{\rm B}\prec{\rm A}} U_A U_B }{\sqrt{2}} \left( \ket{0}_A \ket{0}_B \ket{\psi} \right) \, .
\end{equation}
The final state is a superposition of the states resulting from the interactions of agents A and B with the target in switched orders. This describes a quantum switch in an extended target space that includes the states of the detectors and of the few level systems. As the agents are quantum systems in our protocol, it is natural that their actions are described in a Hilbert space that includes the few-level systems. On the other hand, the detectors play a distinct role, allowing us to distinguish runs of the experiment where the target was scattered or not by the five- and six-level systems. We are interested in the case where the state of the detectors is measured after the interactions between the target and the agents. 

It turns out that for each possible outcome for the measurement of the detectors, the final state in $\mathcal{H}_A \otimes \mathcal{H}_B \otimes \mathcal{H}_\text{targ}$ is a superposition of states obtained by the application of the operations of A and B in switched orders, projected into the subspace associated with such an outcome. Concretely, let $\ket{\zeta}$ be the state of the detectors associated with the postselection $\zeta$ and $P^{(\zeta)}$ be the orthogonal projection on $\ket{\zeta}$. For instance, $\ket{\zeta=0} = \ket{1}_A \ket{1}_B$ and $P^{(0)} = \ket{1}_A \bra{1}_A \otimes \ket{1}_B \bra{1}_B$, and similarly for the other postselections. Let us introduce
\begin{align}
P^{(\zeta)} \left[ U_B U_A \left( \ket{0}_A \ket{0}_B \ket{\psi} \right) \right] &\equiv  \ket{\zeta} \mathcal{B}^{(\zeta)} \mathcal{A}^{(\zeta)} \ket{\psi} \, ,\nonumber \\
P^{(\zeta)} \left[ U_A U_B \left( \ket{0}_A \ket{0}_B \ket{\psi} \right) \right] &\equiv  \ket{\zeta} \mathcal{A}^{(\zeta)} \mathcal{B}^{(\zeta)} \ket{\psi} \, .
\label{eq:proj-zeta}
\end{align}
The explicit form of the states $\mathcal{B}^{(\zeta)} \mathcal{A}^{(\zeta)} \ket{\psi}$  and $\mathcal{A}^{(\zeta)} \mathcal{B}^{(\zeta)} \ket{\psi}$ can be directly extracted from Eqs.~\eqref{eq:BA-operations-1} and \eqref{eq:AB-operations-1} using the definitions \eqref{eq:proj-zeta}. We find that, for each postselection, the final state of the system formed by the few-level systems and target assumes the form:
\begin{equation}
\ket{{\rm A}\prec{\rm B}} \mathcal{B}^{(\zeta)} \mathcal{A}^{(\zeta)}  \ket{\psi} + \ket{{\rm B}\prec{\rm A}} \mathcal{A}^{(\zeta)} \mathcal{B}^{(\zeta)}  \ket{\psi} \, .
\label{eq:before-diagonal}
\end{equation}
The result is a quantum switch in $\mathcal{H}_A \otimes \mathcal{H}_B \otimes \mathcal{H}_\text{targ}$ controlled by the path of A. If the input target state has a vanishing projection in the state $\ket{e_1}$, i.e. $\alpha_1=0$, both orders of operations produce the same result, and the switch is trivial. The postselection $\zeta=3$, for which no detector clicked, selects states with $\alpha_1 \neq 0$, and the final state is independent of the other components $\alpha_i$, $i=2,\dots,5$, of the input state. After normalization, it is then always given by equation~\eqref{eq:switch-composite}. This postselection thus allows us to restrict to the nontrivial part of the quantum switch.

For any input target state, a measurement in the diagonal basis $\ket{{\rm A}\prec{\rm B}}\pm \ket{{\rm B}\prec{\rm A}}$ can be performed on the final state \eqref{eq:before-diagonal} in order to transfer the superposition of orders into a superposition of states in $\mathcal{H}_A \otimes \mathcal{H}_B \otimes \mathcal{H}_\text{targ}$, resulting in a state of the form
\begin{equation}
\mathcal{B}^{(\zeta)} \mathcal{A}^{(\zeta)}  \ket{\psi} \pm  \mathcal{A}^{(\zeta)} \mathcal{B}^{(\zeta)}  \ket{\psi} \, .
\end{equation}
In special cases, one can alternatively perform a measurement in a diagonal basis that includes the states of the few-level systems in order to transfer the superposition of orders into a superposition of target states, as in the cases previously discussed where the input state is a basis vector $\ket{e_i}$. This is convenient, in particular, for the most relevant case, to our purposes, of an input state $\ket{e_1}$.

\section{Discussion}

We have introduced a protocol for the implementation of a quantum switch in a gravitational system. Instead of considering classical agents A and B operating on a target moving on a superposition state of the gravitational field, we allowed the agents to be quantum systems, with A in a path superposition state, on a fixed curved background geometry produced by a central mass. Proper times along distinct paths are then entangled with the paths. With a careful choice of paths, we constructed a protocol that mirrors the relevant features of the protocol for a gravitational quantum switch proposed in \cite{tbell}. A test of Bell's inequality for temporal order can be implemented with two entangled copies of the agents and target.

In our protocol, the order of the operations is not entangled with the spacetime metric, which is classical, but with paths of a quantum system in this fixed curved background. Its realization would then consist of a test of quantum mechanics on curved spacetimes \cite{Pikovski2015,LAMMERZAHL,Kiefer1991}, the limit of quantum field theory on curved spacetimes with negligible particle creation or annihilation and nonrelativistic speeds. This physical regime has not yet been probed experimentally, and our results provide a tool for testing the frequently adopted formulation of time-evolution on a curved spacetime leading to equation~(\ref{eq:internal-evolution-superposition}).

The quantum switch has been realized experimentally in non-gravitational systems~\cite{ReviewExp}. In such experiments, one does not keep track of the proper times at which agents perform their operations. If $\mathcal{A}$ is applied at distinct proper times of A for the orders $\mathcal{A} \mathcal{B}$ or $\mathcal{B} \mathcal{A}$, then measuring the time of the operation would in fact destroy the superposition of the order of operations. In our case, one agent is assumed to be equipped with an internal clock and apply its operation only at a prescribed time. This ensures that the influence of gravity on proper times along the distinct paths is the underlying effect allowing for the superposition of orders to occur.

Experiments that attest quantum phenomena due to the gravity of Earth in the Newtonian regime have already been made~\cite{Collela,Strelkov}. Time dilation is a dominant general relativistic correction to Newtonian gravity, and can be observed even for a height difference of $1 \, {\rm m}$~\cite{Chou2010}. A natural next step would be the exploration of superposition and entanglement of quantum clocks taking time dilation into account, an issue that has been theoretically explored \cite{Zych2011,Zych_2012,Terno2015,Rivera_Tapia_2020,RelHOM,Roura} and simulated with magnetic fields~\cite{Margalit}, but for which an experimental test with the gravitational field is still missing. With the progress on techniques for manipulating path superposition states at macroscopic scales~\cite{Dickerson2013,Kovachy2015,Hannover}, such tests might provide a path for the observation of quantum effects in gravitational systems, and our results include the quantum switch in a list of possible experiments aimed in this direction.

\chapter{Conclusions}
In this master’s thesis, we presented an introductory overview on the area of indefinite orders, focusing on scenarios in which gravity plays a role.  

In the first part, we have seen how to characterize causal relations in a theory-independent way and learned that quantum evolutions, even when constructed from basic principles with no reference to a time parameter, are always compatible with a causal structure (see chapter \ref{Chap CausalityinQT}). In GR, although causal structure is also definite, it is only established after solving the theory's equations, rather than given a priori. A first step in making the two theories merge is to construct a quantum framework in which order is not assumed from the beginning either. In chapter \ref{Chap Process Matrix}, we constructed the (bipartite) process matrix formalism, which does this by considering that QT predictions are valid only inside local laboratories. This generates a new class of evolutions with indefinite order, a typical example being the quantum switch. 

In the second part, we developed the formalism of classical and quantum clocks (chapter \ref{Chap Clocks}), which is valuable on its own for understanding how systems behave in the interface of quantum mechanics and gravity. Moreover, it is used to define the operational notion of events in spacetime, which is made with respect to the proper time displayed by ideal clocks. We then demonstrate, based on this notion and under minimal assumptions, that a superposition state of positions of a spherical mass could serve as the control of a quantum switch (chapter \ref{Chap:QGravitySwitch}). The notion of event in this problem has a natural ambiguity in which a situation where a mass is in a superposition next to a clock is equivalent to another where the mass is fixed and the clock is in a superposition involving the same relative distances. Using this analogy, we proposed a quantum switch using quantum clocks on a classical Schwarzschild spacetime. The realization of the protocol in the gravity of Earth constitutes a test of the predictions of quantum mechanics on curved spacetimes.

The area of indefinite orders is relatively new and extremely active. Some topics appearing in this work are right now under discussion. To highlight one of them, we have the debate on whether the table-top quantum switch\cite{ReviewExp} should be regarded as a simulation. As we commented, the work in \cite{Voji} argues that the switch \emph{can} be described as causally-separable process if one considers the extra spacetime events that are being used for Alice and Bob and not accounted for. However, it is argued that there exists a broader operational sense in which the order is indeed indefinite, if input and output Hilbert spaces of the laboratories are defined as those corresponding to time delocalized systems~\cite{Oreshkov2019timedelocalized}. Over the last years, indefinite causal order has been certified for quantum switch experiments via the theory-dependent test of measuring a causal witness~\cite{Goswami,Rubino}, and also there have been claims of theory-independent certifications, including a version of the Bell's theorem for temporal order in an optical table~\cite{SemiDeviceIndp,RubinoAgain}. Reference~\cite{ChiribellaDesigualdade} suggests a device-independent certification of indefinite order with a quantum switch, but points out that the use of table-top switches to violate their inequalities would not represent ``an interesting notion of indefinite causality'', because they include delayed measurements: Alice's and Bob's measurement results are only read in the intersection of their future lightcones.
Recently, the authors of reference~\cite{SectorialConstraints} presented an argument, without making reference to spacetime events, stating that indefinite order is realized only in a weak sense in the experiments. More comprehensive discussions about this can be found in references~\cite{Voji,Vilasini:2022ist,Oreshkov2019timedelocalized,SectorialConstraints}. 

Another topic of interest is the operational characterization of events through ideal clocks~\cite{Zych,ClocksClassicalandQuantum}. We have discussed the application of this idea to the examples of the quantum switch in a quantum spacetime and of the quantum switch with delocalized clocks in a classical spacetime. Several works have analyzed the operational view with which Einstein constructed the theory of relativity, using classical rods and clocks, and further explored it with quantum clocks leading to the notion of quantum reference frames---see~\cite{FramesOfRef1984, FlaminiaRefFrames, Esteban2020quantum,FlaminiaEinsteinsEquivRefFrames,Giacomini2021spacetimequantum} and references therein for approaches with and without the presence of gravitating systems.

We have also seen, in chapter \ref{Chap Process Matrix}, for instance, that there are several references about indefinite orders relating to the fields of quantum information processing, computation and thermodynamics. But there are only a few works concerning direct treatment of gravitational setups, or at least exploring the analogy coming from the equivalence principle, to analyze how indefinite order could arise in these scenarios~\cite{tbell, Rindler,RindlerDaZych,QSonEarth}. We expect that, with the development of techniques to describe systems on superpositions of spacetimes at low energies\cite{Foo_2021,BlackHolesSuperp}, this will change in the near future. 

Another possible research direction consists in the analysis of the differences between the three quantum switches presented in this thesis, namely, the table-top experimental switch~\cite{Goswami}, the switch in a quantum gravity scenario~\cite{tbell} and the switch in the gravity of Earth~\cite{QSonEarth}. When it comes to probabilities, all of them produce the same results.
However, it is important to understand where are the probabilities coming from, that is, what are the events or laboratories being considered. In optical implementations, the objects which realize Alice's and Bob's operations are positioned in an interferometer in a way that the target either has the possibility of encountering them at  two different spatial locations (see Fig.~\ref{fig:ExpQSProcopio}) or at two different times (Fig.~\ref{fig:ExpQSGoswami}) as seen from the reference frame of the optical table. Then, the input Hilbert spaces of the agents' laboratories are a tensor product of two Hilbert spaces associated with the target system at different times~\cite{Oreshkov2019timedelocalized}.
 For example, if we talk about Alice's ``laboratory'', this works as if Alice opens the door of her laboratory at two distinct time windows during the experiment, even if the system only uses one of them in each round. If she had a clock, she could measure the time at which her operation was applied and learn which time window was used. But, in these implementations, such a measurement would destroy the path superposition state and is not made. So, the two time (or space) windows are needed for the protocol. 
 In the gravitational switches, the events at which Alice and Bob realize their operations are specified with respect to the agents' clocks. Therefore, not only the superposition is not destroyed by acknowledging the times of the operations, but the agents are indeed expected to apply the operation at a definite proper time. Furthermore, the mechanism behind the gravitational switches is time dilation, which makes order indefinite independently of the specific nature of the systems involved. The quantum gravity description is yet different from the others, because any two local systems near a a mass in a superposition state would attest indefinite orders between their events. In the case of the switch in the gravity of Earth, only those clocks which were put in a superposition of positions identify indefinite order. And in the photonic case, indefinite order can only be attested by the specific photon that goes through the beam splitter (see the supplementary information of~\cite{tbell}).

In chapter~\ref{Chap QSonEarth} we presented our contribution to the area~\cite{QSonEarth} and its conclusions. Some topics to explore in future works concern reviewing the definitions of event used in the literature, the analysis of order relations in other classes of quantum  spacetimes, the formulation of protocols corresponding to other causally non-separable processes, besides the quantum switch, and the study of order relations and causality in protocols on top of quantum spacetimes which have no correspondence with any protocol on a classical spacetime.



\bibliographystyle{ieeetr}
\bibliography{ThesisIOQMGravMain}

\includepdf{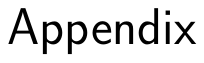}

%
%


\begin{apendicesenv}

\chapter{A model for the trigger}\label{App:trigger}

Let us describe a concrete implementation of a trigger satisfying the condition \eqref{eq:trigger-condition} discussed in section \ref{sec:Modelfortheops} of the main text. We model the trigger as a harmonic oscillator $\mathcal{H}_\mathsmaller{\tiny \VarClock}=L^2(\mathbb{R})$ with free Hamiltonian
\begin{equation}
H_0 = - \frac{\hbar^2}{2m} \nabla^2 + \frac{m \omega^2}{2} q^2
\label{eq:free-hamiltonian-trigger}
\end{equation}
and period
\begin{equation}
T = \frac{2 \pi}{\omega} = 4 \tau^* \, .
\label{eq:alarm-time}
\end{equation}
In the proposed protocol, the trigger plays the role of a clock that is programmed to change the state of the six-level system $\mathcal{H}_A$ of the agent A from $\ket{A_0}$ to $\ket{A_1}$ at a time $\tau^*$. We store the information about the predetermined time $\tau^*$ in the period $T$ of the oscillator through equation~\eqref{eq:alarm-time}.

The interaction of the oscillator with the system $\mathcal{H}_A$ is described by a Hamiltonian $H_{int}$ that is nonzero only on the subspace generated by the relevant states $\ket{A_0},\ket{A_1} \in \mathcal{H}_A$. Let $\sigma_x$ be the first Pauli matrix on this subspace,
\begin{equation}
\sigma_x = \begin{pmatrix}
0 & 1 \\ 1 & 0 
\end{pmatrix} \, .
\end{equation}
The interaction Hamiltonian is defined as
\begin{equation}
H_{int} = P_\Delta \otimes V_0 \, \sigma_x   \, , \qquad V_0 >0 \, ,
\end{equation}
where $P_\Delta$ is the orthogonal projection onto the region $x \in [0,\Delta]$, which acts on the wavefunction $\phi(x)$ of the oscillator according to
\begin{equation}
P_\Delta \phi(x) = \begin{cases}
\phi(x)  \, , & \text{if } x \in [0,\Delta] \\
0 \, , & \text{else} \, .
\end{cases}
\end{equation}
The interaction is nontrivial only when the oscillator is in the region $[0,\Delta]$, which we call the interaction zone. The full Hamiltonian of the system is
\begin{equation}
H = H_0 \otimes \bm{1}+ H_{int} \, .
\end{equation}

Let $\ket{\pm}$ be the eigenvectors of $\sigma_x$ with eigenvalues $\pm 1$,
\begin{equation}
\sigma_x \ket{\pm} = \pm \ket{\pm} \, .
\end{equation}
Then,
\begin{align*}
( H_0 \otimes \bm{1} + H_{int}) (\ket{\phi} \ket{+}) = \left[ (H_0 + V_0 P_\Delta) \ket{\phi}\right] \ket{+} \, , \\ 
( H_0 \otimes \bm{1} + H_{int}) (\ket{\phi} \ket{-}) = \left[ (H_0 - V_0 P_\Delta) \ket{\phi}\right] \ket{-} \, .
\end{align*}
Hence, for $\ket{\pm}$, the wavefunction evolves under a Hamiltonian
\[
H_0 + V_\Delta^\pm \, , \quad V_\Delta^\pm = \begin{cases}
\pm V_0 \, , & \text{if } x \in [0,\Delta] \, , \\
0 \, , & \text{else} \, .
\end{cases}
\]
For $\ket{+}$, the oscillator encounters a potential barrier; for $\ket{-}$, it encounters a potential well. The general case is a superposition of these situations.

Let $\ket{\alpha}$ be a coherent state of the trigger, $a \ket{\alpha}=\alpha \ket{\alpha}$, where $a$ is the annihilation operator of the harmonic oscillator. The trigger is prepared in a coherent state $\ket{\raisebox{-2pt}{\footnotesize \VarClock}\, ;\tau=0}=\ket{\alpha_0}$, where
\begin{equation}
\alpha_0 = \frac{A}{\sqrt{2} \sigma} \, , \qquad \sigma = \sqrt{\frac{\hbar}{m\omega}} \, , \qquad A = \frac{2 \Delta V_0}{\pi \hbar \omega} \, .
\label{eq:parameters}
\end{equation}
The parameter $\sigma$ is the width of the wavepacket. The coherent state describes a configuration of maximum positive displacement for an oscillation of amplitude $A$. The system $\mathcal{H}_A$ is prepared in the state $\ket{A_0}$ at $\tau=0$. The state of the composite system will be represented by $\ket{\Phi(\tau)}$.

We assume that $A \gg \Delta \gg \sigma$. The inequality $\Delta \gg \sigma$ means that the width of the wavepacket is much smaller than the width of the interaction region, i.e., that the state of the oscillator is well localized with respect to the potential step. The condition $A \gg \Delta$ means that the oscillator is initially far away from the interaction zone. Its evolution is thus initially determined by the free Hamiltonian $H_0$. As a result, it remains a coherent state $\ket{\raisebox{-2pt}{\footnotesize \VarClock}\, ;\tau}=\ket{\alpha(\tau)}$, where $\alpha(\tau)=\alpha_0 e^{-i\omega \tau}$, until it reaches the interaction zone. As the average position of such a coherent state is simply
\begin{equation}
\langle x \rangle = A \cos \omega t \, ,
\end{equation}
the wavepacket reaches the boundary of the interaction zone at $x=\Delta$ with a speed $v \sim \omega A$ after a time $\Delta \tau \simeq \tau^* - \epsilon$, where $\epsilon \sim \Delta/v$.  From $A, \Delta \gg \sigma$ and equation~\eqref{eq:parameters}, we also find that the energy of the wavepacket satisfies
\begin{equation}
\frac{m\omega^2 A^2}{2} \gg V_0 \, ,
\end{equation}
i.e., the energy of the wavepacket is much larger than the potential step $V_0$. We can then neglect the reflection of the wavepacket by the potential step and adopt the approximation of perfect transmission.

For $\tau < \tau^* - \epsilon$, the interaction Hamiltonian is negligible, since the wavepacket is outside the interaction zone:
\[
H_{int} \left( \ket{\alpha(t)} \ket{\chi} \right) = 0 \, , \quad \text{for } \, \tau<\tau^* -\epsilon \, ,
\]
for any $\ket{\chi} \in \mathcal{H}_A$. Therefore, the six-level system remains at the initial state $\ket{A_0}$ for $\tau<\tau^* -\epsilon$, and the first condition in equation~\eqref{eq:trigger-condition} is satisfied. The wavepacket then enters the interaction zone and crosses it in a time interval $\Delta \tau \simeq \Delta/v=\epsilon$. During this time, we have $P_\Delta \ket{\raisebox{-2pt}{\footnotesize \VarClock}\, ;\tau}\simeq \ket{\raisebox{-2pt}{\footnotesize \VarClock}\, ;\tau}$, so that 
\begin{align*}
H \ket{\Phi} &=  (\bm{1} \otimes H_0 + H_{int}) \ket{\Phi} \nonumber \\
	&\simeq (\bm{1} \otimes H_0 + V_0 \, \sigma_x \otimes \bm{1}) \ket{\Phi} \, ,\quad \text{for } \, \tau \in [\tau^*-\epsilon, \tau^*] \, .
\end{align*}
The time evolution generated by this Hamiltonian can be integrated exactly. We find that, for $\tau \in [\tau^*-\epsilon, \tau^*]$, the system evolves according to
\begin{equation}
\ket{\Phi(\tau)} = \left\{ e^{-iV_0 \sigma_x [\tau-(\tau^*-\epsilon)]/ \hbar} \ket{A_0} \right\} \ket{\alpha(\tau)} \, .
\end{equation}
At $\tau=\tau^*$, when the wavepacket reaches the opposite boundary of the interaction zone at $x=0$, we have
\begin{equation}
\ket{\Phi(\tau^*)} = \ket{A_1} \ket{\alpha(\tau^*)} \, ,
\end{equation}
and the second condition in equation~\eqref{eq:trigger-condition} is also satisfied, showing that the model described in this Appendix satisfies the properties required of the trigger.

Let us note that the scattering of a Gaussian wavepacket by a potential step is studied in \cite{Bernardini,Norsen} in the regime where the wavepacket has a small width in comparison with the step. The evolution of the wavepacket in this regime, which is usually not discussed in basic textbooks, can be described as a process involving multiple instantaneous scatterings with the boundaries of the potential step. The incident wavepacket branches into two wavepackets when it meets the potential step, with one branch corresponding to the transmitted wave and the other describing the reflected component. The transmitted wavepacket branches again into two new wavepackets when it meets the opposite boundary of the potential step, and so on. The initial wavepacket evolves in this manner into an infinite train of successive wavepackets, both reflected and transmitted, of progressively smaller amplitudes, as described in \cite{Bernardini}. We adopted the approximation of perfect transmission, valid for large energies, for the calculations above, so that the wavepacket crosses the interaction zone without branching into superpositions of localized states. While it crosses the potential $V_0 \sigma_x$ in the interaction zone $[0,\Delta]$, it induces a rotation of the state of $\mathcal{H}_A$, and the parameters of the model can be adjusted so that the initial state $\ket{A_0}$ evolves into the state $\ket{A_1}$, as required.

\end{apendicesenv}












\end{document}